\documentclass[thesis,backref]{thesis-umich}
\usepackage{amsthm} % must be loaded before newtxmath, for theorems and related environments
\usepackage{newtxmath} % Changes math mode from computer modern to Times... because IEEE, thought they publish the best journals, is dumb as bricks when it comes to typesetting, and now I've gotten used to having extra script letters without having to load any other packages
\usepackage[euler]{textgreek} % upright greek characters
\usepackage{enumitem} % Changes enumerate environments
\usepackage{stmaryrd} % allows blackboard bold square brackets
\usepackage{hhline} % double hline without interupting vertical lines in tables
\usepackage{IEEEtrantools} % allows IEEEeqnarray
\usepackage{colortbl} % Allows change of hline color via \arrayrulecolor{color}\hline
\usepackage{cancel} % allows \cancel command, which puts a slash through characters
\usepackage{notoccite}% prevents cites in captions from misnumbering your references
\usepackage{rotating} % allow landscape figures
\usepackage{pdflscape} % landscape mode
\usepackage{afterpage} % ????? used with pdflscape here but not sure why https://tex.stackexchange.com/questions/433411/floating-sideways-table-on-a-rotated-page-in-pdf-viewer/433414#433414
\usepackage{multicol} % For multiple columns in the acknowledgmenets
\usepackage{lipsum}
\usepackage{changepage} % margin manipulation

\usepackage{algorithm} % algorithm environment
\usepackage{algpseudocode} % algorithmic environment
\definecolor{commentgray}{gray}{0.4}

\newcommand{\evec}{\mathbf{e}}
\newcommand{\fvec}{\mathbf{f}}
\newcommand{\gvec}{\mathbf{g}}

\newcommand{\rvec}{\mathbf{r}}

\newcommand{\uvec}{\mathbf{u}}

\newcommand{\xvec}{\mathbf{x}}
\newcommand{\yvec}{\mathbf{y}}

\newcommand{\etavec}{\textnormal{\textbf{\texteta}}}

\newcommand{\phivec}{\textnormal{\textbf{\textphi}}}

\newcommand{\partialfrac}[2]{\frac{\partial #1}{\partial #2}}
\newcommand{\jacobian}[2]{\partialfrac{#1}{#2}}
\newcommand{\norm}[1]{\left\lVert#1\right\rVert}
\DeclareMathOperator*{\argmin}{arg\,min}

\DeclareMathOperator*{\mean}{mean}

\DeclareMathOperator*{\RMS}{RMS}
\DeclareMathOperator*{\range}{range}

\newcommand{\real}{\mathbb{R}}

\newcommand{\integer}{\mathbb{Z}}
\usepackage[normalem]{ulem}
\theoremstyle{definition}
\newtheorem{definition}{Definition}[chapter]
\theoremstyle{plain}
\newtheorem{theorem}{Theorem}[chapter]
\newtheorem{corollary}{Corollary}[theorem]
\newtheorem{lemma}{Lemma}
\theoremstyle{remark}
\newtheorem{remark}{Remark}[chapter]

\newcommand{\eqcomment}[1]{%
  \text{\phantom{(#1)}} \tag{#1}
}

\newcommand{\eqlabel}[1]{\addtocounter{equation}{-1}\refstepcounter{equation}\label{#1}}
\makeatletter
\newcommand{\clonelabel}[2]{\@bsphack
  \expandafter\ifx\csname r@#2\endcsname\relax
  \else\protected@write\@auxout{}{\string\newlabel{#1}%
    {\csname r@#2\endcsname}}%
  \fi
  \expandafter\ifx\csname r@#2@cref\endcsname\relax
  \else\protected@write\@auxout{}{\string\newlabel{#1@cref}%
    {\csname r@#2@cref\endcsname}}%
  \fi
  \@esphack}
\makeatother
\newcounter{subeq}

\newcommand{\suspend}[1]{
\newcounter{#1}\setcounter{#1}{\value{enumi}}
}
\newcommand{\resume}[1]{
\setcounter{enumi}{\value{#1}}
}
\newcommand{\refRmu}{\ref{R1}}
\newcommand{\refRlin}{\ref{R2}}
\newcommand{\refRTI}{\ref{R3}}
\newcommand{\refRsmooth}{\ref{Rsmooth}}

\newcommand{\directAILC}{(\ref{eq:ILCclassic}), (\ref{eq:gammaprior})}
\newcommand{\originalmod}{(\ref{eqn:sys})} % equation number for model definition at beginning of NILC description
\newcommand{\linsys}{(\ref{eq:linall})} % linearization of originalmod about a trajectory
\newcommand{\SAILControl}{(\ref{eq:ILCclassic}), (\ref{eq:gammanew})}

\newcommand{\sysDef}{(\ref{eq:sysDefPWA})}
\newcommand{\baseAssumptions}{\ref{C:reachable}-\ref{C:muc}}
\newcommand{\exclusiveSwitching}{(A5.10)}
\newcommand{\controlModel}{(\ref{eq:monox})-(\ref{eq:monoy})}

\newcommand\blfootnote[1]{%
  \begingroup
  \renewcommand\thefootnote{}\footnote{#1}%
  \addtocounter{footnote}{-1}%
  \endgroup
}

\newcommand{\twoliner}[1]{
\begin{singlespace}
\noindent
#1
\end{singlespace}
}

\makeatletter
\renewcommand*\env@matrix[1][\arraystretch]{%
  \edef\arraystretch{#1}%
  \hskip -\arraycolsep
  \let\@ifnextchar\new@ifnextchar
  \array{*\c@MaxMatrixCols c}}
\makeatother

\newcommand{\muAM}{\textmu-AM}
\newcommand{\ejet}{e-jet printing}

\newcommand{\AILC}{NILC}
\newcommand{\nonmin}{non-minimum}
\newcommand{\nonminphase}{\ac{NMP}}
\newcommand{\ILILC}{ILILC}
\floatname{algorithm}{Procedure}
\makeatletter 
 
\@addtoreset{algorithm}{chapter} 
\makeatother

\newcommand{\PWA}{PWA}
\newcommand{\NILC}{NILC}
\newcommand{\NMP}{NMP}

\newcommand{\shift}{\mathcal{T}}
\newcommand{\cv}{CV}
\newcommand{\vol}{\mathcal{V}}
\newcommand{\inputcoef}{b}
\newcommand{\reset}{\psi}
\newcommand{\reseto}{\reset_0}
\newcommand{\reseti}{\reset_1}

\newcommand{\uDim}{{n_u}}
\newcommand{\xDim}{{n_x}}
\newcommand{\yDim}{{n_y}}
\newcommand{\planeVec}{p}

\newcommand{\planeQuant}{n_P}
\newcommand{\offsetNum}{b}
\newcommand{\offsetVec}{\beta}
\newcommand{\state}{x}
\newcommand{\uMIMO}{u}
\newcommand{\yMIMO}{y}
\newcommand{\sysvec}{\begin{bmatrix}\state(k)\\\uMIMO(k)\end{bmatrix}}
\newcommand{\sysvecT}{\begin{bmatrix}\state(k)^T, &\uMIMO(k)^T\end{bmatrix}^T}
\newcommand{\setQuant}{20}
\newcommand{\expQuant}{50}

\newcommand{\gLength}{949}
\newcommand{\RAM}{\SI{16}{\giga\byte}}
\newcommand{\CPU}{\SI{4}{\giga\hertz}}
\newcommand{\casadiTime}{29 seconds}
\newcommand{\matlabTime}{1.3 hours}
\newcommand{\convZero}{8}
\newcommand{\convMean}{12}
\newcommand{\convEnv}{14}
\newcommand{\modelErrVec}{e_{\theta}}
\newcommand{\modelError}{\norm{\modelErrVec}_2}
\newcommand{\region}{Q}
\newcommand{\regionQuant}{{|\region|}}
\newcommand{\qdx}{q}
\newcommand{\learnMat}{L}
\newcommand{\tdx}{\ell}
\newcommand{\thetamod}{\hat{\theta}}

\newcommand{\mfin}{m_{\text{final}}}
\newcommand{\idx}{\kappa} % arbitrary time indexing
\newcommand{\xmod}{\hat{x}}
\newcommand{\flin}{\bar{f}}
\newcommand{\fmod}{\hat{f}}
\newcommand{\fmodsim}{\tilde{f}}
\newcommand{\fetasim}{\tilde{f}_{\eta}}

\newcommand{\ymod}{\hat{y}}
\newcommand{\yvecmod}{\hat{\yvec}}
\newcommand{\ypreview}{\hat{\mathscr{y}}}
\newcommand{\hmod}{\hat{h}}
\newcommand{\hlin}{\bar{h}}
\newcommand{\gtrue}{\gvec}
\newcommand{\glin}{\bar{\gvec}}
\newcommand{\gmod}{\hat{\gvec}}
\newcommand{\etamod}{\hat{\eta}}
\newcommand{\etamodsim}{\tilde{\eta}}
\newcommand{\etavecsim}{\tilde{\etavec}}
\newcommand{\Asim}{\tilde{A}}
\newcommand{\Bsim}{\tilde{B}}
\newcommand{\force}{c}
\newcommand{\angtruth}{\psi}
\newcommand{\Mp}{M_p}
\newcommand{\Mpmod}{\hat{M}_p}
\newcommand{\Mc}{M_c}
\newcommand{\Mcmod}{\hat{M}_c}
\newcommand{\dptruth}{d_p}
\newcommand{\dctruth}{d_c}
\newcommand{\dcmod}{\hat{d}_c}
\newcommand{\dpmod}{\hat{d}_p}
\newcommand{\grav}{\mathscr{g}}
\newcommand{\kmod}{\hat{\kappa}}
\newcommand{\unoise}{\omega_\force}
\newcommand{\ynoise}{\omega_y}
\newcommand{\rprcv}{r^*} % effective reference
\newcommand{\param}{\modelErrVec}
\newcommand{\parammod}{\hat{\param}}
\newcommand{\convset}{\mathcal{C}} % set of convergent trials
\newcommand{\convRate}{\mathcal{R}} % transient convergence rate
\newcommand{\stab}{{\mathscr{s}}}
\newcommand{\unstab}{{\mathscr{u}}}
\newcommand{\evalq}{n_\stab}

\newcommand{\kappadx}{_\kappa}
\newcommand{\kdxplus}[1]{_{k+#1}}
\newcommand{\kdxminus}[1]{_{k-#1}}
\newcommand{\kdx}{_k}
\newcommand{\dmu}{\mu_g} % Dynamic relative degree
\newcommand{\Qset}{\region}
\newcommand{\Asum}{\mathbf{A}}
\newcommand{\Bsum}{\mathbf{B}}
\newcommand{\Csum}{\mathbf{C}}
\newcommand{\Dsum}{\mathbf{D}}
\newcommand{\Fsum}{\mathbf{F}}
\newcommand{\Gsum}{\mathbf{G}}
\newcommand{\Msum}{\mathbf{M}}
\newcommand{\sumqQ}{\sum_{q=1}^{\regionQuant}}
\newcommand{\locvec}{\delta}
\newcommand{\sigvec}{\delta^*}
\newcommand{\sigvecset}{\Delta^*}
\newcommand{\offsetvec}{\offsetVec}
\newcommand{\inv}{\overline}
\newcommand{\prevC}{\mathcal{C}}
\newcommand{\prevD}{\mathcal{D}}
\newcommand{\prevG}{\mathcal{G}}
\newcommand{\anticaus}{\Psi}
\newcommand{\decoup}{\tilde}
\newcommand{\statedc}{\decoup{\state}}
\newcommand{\Asumdc}{\decoup{\Asum}}
\newcommand{\Bsumdc}{\decoup{\Bsum}}
\newcommand{\Fsumdc}{\decoup{\Fsum}}
\newcommand{\Pdc}{\decoup{P}}
\newcommand{\extracts}{\mathscr{I}^\stab}
\newcommand{\extractu}{\mathscr{I}^\unstab}
\newcommand{\refr}{r}
\newcommand{\controltotal}{\force}
\newcommand{\controlfb}{y}
\newcommand{\Acm}{\hat{A}} % control model
\newcommand{\Bcm}{\hat{B}}
\newcommand{\Ccm}{\hat{C}}
\newcommand{\Dcm}{\hat{D}}
\newcommand{\xcm}{\hat{x}}
\newcommand{\statemod}{\xmod}
\newcommand{\Asummod}{\hat{\Asum}}
\newcommand{\Bsummod}{\hat{\Bsum}}
\newcommand{\Fsummod}{\hat{\Fsum}}
\newcommand{\Csummod}{\hat{\Csum}}
\newcommand{\Dsummod}{\hat{\Dsum}}
\newcommand{\Gsummod}{\hat{\Gsum}}
\newcommand{\ldx}{\tdx}
\newcommand{\uLift}{\uvec}
\newcommand{\rLift}{\rvec}
\newcommand{\yLift}{\yvec}
\newcommand{\yLiftmod}{\hat{\yLift}}
\newcommand{\gLift}{\gvec}
\newcommand{\ginvLift}{\hat{\gLift}^{-1}}

\title{
Hybrid Systems, 
Iterative Learning Control, 
and
Non-minimum Phase
}

\author{Isaac A. Spiegel}

\orcid{0000-0002-4415-9190}

\email{ispiegel@umich.edu}

\department{Mechanical Engineering}

\frontispiece{\includegraphics[scale=1.3]{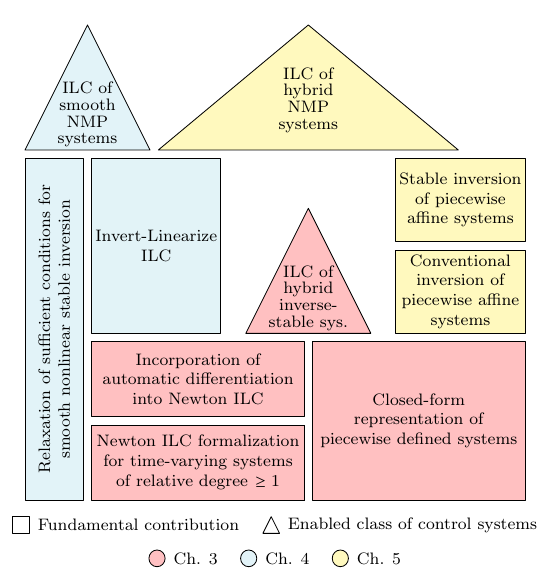}}

\frontpagestyle{1}

\dedication{ %
To those who value this,
and to those who benefit from it.
}

\acknowledgments{
It is difficult to find the right words for this acknowledgement. 
There are so many people I want to address in so many different ways.

This dissertation would not have been possible without funding from the University of Michigan Department of Mechanical Engineering, the National Science Foundation, and the National Institute of Standards and Technology. Equally essential are the members of my doctoral committee: Kira Barton, Tulga Ersal, Chinedum Okwudire, and Necmiye Ozay.
And I am especially grateful for all the collaborators that have contributed to this research: Kira Barton, Tom Oomen, Nard Strijbosch, Robin de Rozario, Patrick Sammons, and Tom van de Laar.

As I conclude this doctoral program, I have also been reflecting on a bigger picture that sprawls beyond the walls of graduate schools and my time within them.
To me, this dissertation has not been 5 years, but rather 13-20 years in in the making.
Because I see this as the culmination of my life to date, I keep thinking back over 
my life's full span.
There's a lot to be humble about.

I bring this 
up because over 
time
I've accrued a great list people to whom I'm deeply indebted, but have not repaid and probably cannot repay, and that's a humbling feeling.
To thank in words is not enough.
These are debts of friendship, time, effort, grace, pain, forgiveness, and love.
I wish I could give you all what you are worth or at least 
give you the feeling deep inside that you are worth more than can be given, and that you are loved in kind.
I realize this is 
a bit dramatic, but what can I do?
So much of me has been shaped by you.
I'm carried by the gifts you've given me.
And I know that you are carried by gifts given to you by others.
There are all these gifts begetting gifts, rippling outward.
So if it is not within my power to repay you, I hope to at least pay it forward as best I can.

You've made my life special.
I will do my best by you, and by those I meet as time goes on.

\vfill

\noindent
So here's to 
those that shaped me,
from the beginning:

\vfill
\pagebreak

\begin{adjustwidth}{0in}{-0.3in}
\begin{multicols}{3}
\noindent
Lisa Spiegel
\\
David Spiegel
\\
Daniel Spiegel
\\
Breanna Spiegel
\\
Lauren Rugge
\\
Jessica Rugge
\\
Sean Kennedy
\\
Dick Rugge
\\
Marilyn Stamper
\\
Hannah Rosenfeld
\\
Colleen Rosenfeld
\\
Marc Rosenfeld
\\
Ian Rosenfeld

\twoliner{Master Mark W. Pattison, \\ \emph{Sir}}
\noindent
Master Danielle Page-Pattison
\\
Kevin Shaw
\\
Dani Gradisher
\\
Lauren Shaw
\\
Anita Hayworth
\\
Dave Shaw
\\
Henry Sweat
\\
Pauline Disch
\\
Alex Sweat
\\
Jola Prosceno
\\
Mark Sweat
\\
Ryan Giza
\\
John McBlair
\\
Master Matt Goodsell
\\
Merl Goodsell
\\
Master Adam Rosenberg
\\
Amanda Rosenberg
\\
Master David Lam
\\
Dan Gehlhaar
\\
Heidi Hill
\\
Michael Eckstein
\\
Dave Lambillotte
\\
Nick Lococo

\twoliner{Erika Shaw, \\ \emph{Honorary High Admiral V}}
\noindent
Matthew Eliceiri
\\
Margie Wilhelm
\\
Judy Holland
\\
Linda Rohmund
\\
Jon Rasmussen
\\
Steven Walker
\\
Eric Smith
\\
Tyler Workinger
\\
Christian Workinger
\\
Lynn Light
\\
Sally Ptak
\\
Sydney Spiegel
\\
Ezra Spiegel
\\
Charlene Spiegel
\\
Charles Spiegel
\\
Lily-Jean Ciccotelli
\\
Tony Ciccotelli
\\
Steve Ciccotelli
\\
Fred Ciccotelli
\\
Eleonore Stump
\\
Donald Stump
\\
Nathan Stump
\\
Aaron Stump
\\
Monica Stump
\\
Cathy Noord
\\
Thomas Noord
\\
Mick Noord
\\
Francesca Noord
\\
Jake Bert
\\
Jason Marshal
\\
McKenna Taylor
\\
Andrew Rappolt
\\
Ben Love
\\
Jen Wilson

\twoliner{Matt Golman,\\ \emph{Founder}}

\twoliner{Asa Puckette,\\ \emph{Co-Founder}}

\twoliner{Alec Asperslag,\\ \emph{Honorary High Admiral I}}

\twoliner{Grant Lawrence Thompson,\\ \emph{Honorary Vice Admiral I}}
\noindent
Elena Clémençon-Charles
\\
Megan Gaffney
\\
Andrea Reyes
\\
Ana Reyes
\\
Micael Maya-Peinl
\\
Julian Noble
\\
Natalie Van Valkenberg
\\
Geoff Van Valkenberg
\\
Alex Van Valkenberg
\\
Maria Van Valkenberg
\\
William Van Valkenberg
\\
Max Van Valkenberg
\\
George Stimson
\\
Rie Tsuboi
\\
Darlene Blanchard
\\
Justin Conn
\\
Michael Santos
\\
Tatiana Roy
\\
Suki Berry
\\
Annie Tarabini
\\
Lonnie Safarik
\\
Amy Kuo
\\
Emma Lindley

\twoliner{Ariel Jones,\\\emph{Honorary High Admiral III}}

\twoliner{Stan Austin,\\\emph{Honorary High Admiral IV}}
\noindent
Sammy Ness
\\
Savonnah Turner
\\
Jaden Pratt
\\
Mike Shook
\\
Blaze Newman
\\
Ryan Cardenes
\\
Caroline Pollock
\\
Scott Huntley
\\
Jason Berend
\\
Karen des Jardins
\\
Stephen des Jardin
\\
Darryl Walton
\\
Bobbie Walton
\\
Nick Foote
\\
Kevin Brice
\\
Willie Saake
\\
Karen Saake
\\
Andrew Dalager
\\
Gabe Cemaj
\\
Andy Colletta
\\
Mark Leon
\\
Dana Pede
\\
Betty Huang
\\
Kyla Wilson
\\
Victoria Ly
\\
Olivia Perry
\\
Linda Lam
\\
Logan Mercer
\\
Zack Mayeda
\\
Cassidy Mayeda
\\
Laura O'Hagan
\\
Dustin Atlas
\\
Ben Atlas
\\
Tory Bader
\\
Evan Wong
\\
Zyad Hammad
\\
Megan Bradley
\\
Sean Holcomb
\\
Greg Jackson
\\
Yukina Tanaka
\\
Haruka Nomura
\\
Satoru Takagi
\\
Ali Candlin
\\
Zoë Winkworth
\\
Michael Reza Farzi
\\
Sydney Lefabvre
\\
Emily Doherty
\\
Michael McCutchen
\\
Lara Harding
\\
Alyx Barbeau
\\
Delani Davis
\\
Drew Spiller
\\
Haley Kovacs
\\
Claire Li
\\
Jessica Resnick
\\
Adam Resnick
\\
Andy Packard
\\
Emily Jensen
\\
Oliver O'Reilly
\\
James Casey
\\
Amando Miller
\\
Chris Bulpitt
\\
Cynthia Tan
\\
Danny Wilson
\\
Louis Malito
\\
Bardia Ganji
\\
Roshena Macpherson
\\
John Wallace
\\
Ben Chen
\\
Brian Graf
\\
Nicole Schauser
\\
Derek Chou
\\
Pol Llado
\\
Alex Cuevas
\\
Mandy Huo
\\
Marc Russel
\\
Stephen Rhodes
\\
Taito Nabeshima
\\
Hiroo Yugami
\\
Fumitada Iguchi
\\
Makoto Shimizu
\\
Asaka Kohiyama
\\
Kyotaka Konno
\\
Hiroaki Kobayashi
\\
Mari Suzuki
\\
Kunihiko Yanagisawa
\\
Hiroki Sato
\\
Kasemchai Chaiprasobphol
\\
Syo Onodera
\\
Taro Fukushige
\\
Jun Sakai
\\
Yuta Fujiwara
\\
Yoshikazu Shibata
\\
Ryusuke Mihara
\\
Takahiro Kumagai
\\
Shoya Murayama
\\
Shoma Onuki
\\
Daniel Kenji Pederson
\\
Lea Rossander
\\
Frej Rossander
\\
Eivind Rossander
\\
Jenny Rossander
\\
Tatsuki Momose
\\
Sayaka Ono
\\
Shogo Onishi
\\
Nick Gustafson
\\
Charlie Heckroth
\\
Lei Lee
\\
Perry Hobbs
\\
Yuri Hobbs
\\
Nick Ong
\\
Juri Mizuki
\\
Haruka Ochisai
\\
Meng Han
\\
Robert Baumann
\\
Gracie Burdeos
\\
Oskar Södergren
\\
Kaori Ohrui
\\
Marian Lumongsod
\\
Ryan Thompson
\\
Takuya Yashima
\\
Mizuki Temma
\\
Lauri Pokka
\\
Samuli Koivu
\\
Philip Grajetzki
\\
Megumi Sugano
\\
Midori Hirose
\\
Heraldo Stefanon
\\
Christopher Tacub
\\
Amrita Srinivasan
\\
Ammama
\\
Tathappa
\\
Ganesh Srinivasan
\\
Terri Srinivasan
\\
Gita Srinivasan
\\
Michael Srinivasan
\\
Raghu Srinivasan
\\
Lisa Tran
\\
Melis Şahinöz
\\
Narayanan Kidambi
\\
Katie McLaughlin
\\
Xiao-Yu Fu
\\
CJ Linton
\\
Meg Brennan
\\
Kira Barton
\\
Alex Shorter
\\
Efe Balta
\\
Ilya Kovalenko
\\
Shreyas Kousik
\\
Miguel Saez
\\
Zheng Wang
\\
Deema Totah
\\
Mike Quann
\\
Ding Zhang
\\
Delaramm Afkhami
\\
Nazanin Farjam
\\
Anne Gu
\\
Joaquin Gabaldon
\\
Meghna Menon
\\
Max Wu
\\
Dory Yang
\\
Yaqing Xu
\\
Loubna Baroudi
\\
Max Toothman
\\
Mingjie Bi
\\
Tyler Toner
\\
Yassine Qamsane
\\
Lai-Yu Leo Tse
\\
Christopher Pannier
\\
Berk Altın
\\
Max Toothman
\\
Bo Fu
\\
Allison White
\\
C. David Remy
\\
Robin Rodríguez
\\
Ahmet Mazacioglu
\\
Nils Smit-Anseeuw
\\
Audrey Sedal
\\
Mario Medina
\\
Suhak Lee
\\
Matthew Porter
\\
Rachel Vitali
\\
Patrick Holmes
\\
Dan Bruder
\\
Kim Ingraham
\\
Misaki Nozawa
\\
Hasnaa Rabbat
\\
Remy Pelzer
\\
Nico Mbolamena
\\
Sherry Lin
\\
Shuyu Long
\\
Tom Oomen
\\
Nard Strijbosch
\\
Noud Mooren
\\
Nic Dirkx
\\
Robin de Rozario
\\
Lennart Blanken
\\
Enzo Evers
\\
Fahim Sakib
\\
Camiel Beckers
\\
Robert van der Weijst
\\
Joey Reinders
\\
Chyannie Amarillio
\\
Rishi Mohan
\\
Daniel Veldman
\\
Frans Verbruggen
\\
Tom van de Laar
\end{multicols}
\end{adjustwidth}

}

\preface{

Technology advances. As it does so, both the complexity of 
our 
machines and the performance demands upon them grow.
To meet these performance demands, ever more sophisticated automatic control schemes are developed for these machines.
Most controllers are designed based on knowledge---i.e. a mathematical model---of a machine's dynamics, and the performance achievable by such controllers is usually a function of how well the dynamical model represents the true system behavior.
The simultaneous increases in system complexity and performance requirements thus clash when the system complexity exceeds the ability of a mathematical modeling framework to capture.

The development of new, more flexible dynamical modeling frameworks (hybrid systems) is a clear response to this problem, one to which the controls engineering community has taken 
heartily
\cite{Tabuada2009,Heemels2001,Cassandras2008}.
However, improved model fidelity is only half the battle. 
There must also be \emph{compatibility} between a 
controller synthesis method and a model framework.
In other words, a second clash occurs when model complexity exceeds the ability of a control framework to utilize.

The main purpose of this dissertation is to bridge the gap between high fidelity modeling and controller synthesis for systems required to perform a task repetitively, such as manufacturing robots producing many copies of the same product. 
This gap-bridging is done by contributing new mathematical tools to both the theory of modeling and the theory of controller synthesis for repetitive systems (\ac{ILC}). 

There are three primary contributions. First is a new mathematical representation of hybrid systems that admits several mathematical operations crucial to \ac{ILC} synthesis, but that had been undefined for prior hybrid system representations.
The second and third contributions deal with the particularly challenging case of dynamical models that have unstable inverses (often called \ac{NMP} models).
Because many \ac{ILC} syntheses use model 
inversion
\cite{Bristow2006},
such \ac{NMP} models can cause instability of the overall controlled system.
To combat this problem for both hybrid and non-hybrid systems, the second contribution is a new \ac{ILC} framework that enables the incorporation of alternative notions of model inversion with \ac{ILC}, circumventing the instability problem.
Finally, the third contribution is a 
method for deriving such an alternative, stable inverse specifically for a class of hybrid \ac{NMP} models.
Thus, the three primary contributions culminate in the ability to perform \ac{ILC} with hybrid models, even when the models have unstable inverses.

These results address challenges faced by a broad array of systems, from pick-and-place robots to piezoactuators \cite[ch. 6]{DeWit1996}, \cite{Hogan2020,Schitter2002}.
However, the original impetus for this research is 
\ac{e-jet} Printing,
which epitomizes the challenges arising from the tandem growth of performance demands and system complexity.
E-jet printing, like commonplace inkjet printing, revolves around the ejection of liquid from a nozzle onto a target surface. 
Unlike inkjet printing, 
\ac{e-jet} 
printing is driven by the delicate interaction of applied electric fields with particles in the liquid meniscus at the nozzle tip.
See Figure \ref{fig:ejetintro} for illustration.
This gives e-jet the capacity for submicron resolution patterning for printed optics and electronics 
\cite{Cho2020,Afkhami2020}. 
But e-jet printers are sensitive machines, and require automatic control to perform reliably. Unfortunately, the ejection process is too small and too fast for online control of the fluid flow \emph{during} ejection to be practical. Thus, the clearest path to control e-jet printing is by analyzing performance \emph{after} each ejection and seeking to improve performance from ejection to ejection, i.e. \ac{ILC}.

However, the system electrohydrodynamics are complicated enough that no traditional control-oriented modeling effort has been able to capture the complete ejection process's dynamics. Thus began research on the hybrid modeling of e-jet printing 
and the \ac{ILC} of hybrid systems, resulting in the first primary contribution.
The hybrid \ac{e-jet} modeling research is provided in this dissertation as 
an 
application-focused
contribution motivating the theoretical work.

Because the complexity and sensitivity of e-jet printing produces substantial confounding factors, preliminary physical implementation of this 
new
theory was attempted on a comparatively simple inkjet printhead positioning system. 
The \ac{ILC} scheme in focus failed to produce stable control.
This spurred the research into the \ac{ILC} failure mechanisms, leading to the discovery of the \ac{ILC} scheme's fundamental inability to leverage models with unstable inverses, which the printhead positioning system happened to have. The second and third major contributions address this problem.

In other words, this dissertation contributes new broadly applicable control-theoretic tools born from challenges encountered in practice on physical systems.

\begin{figure}
    \centering
    \includegraphics[scale=0.25]{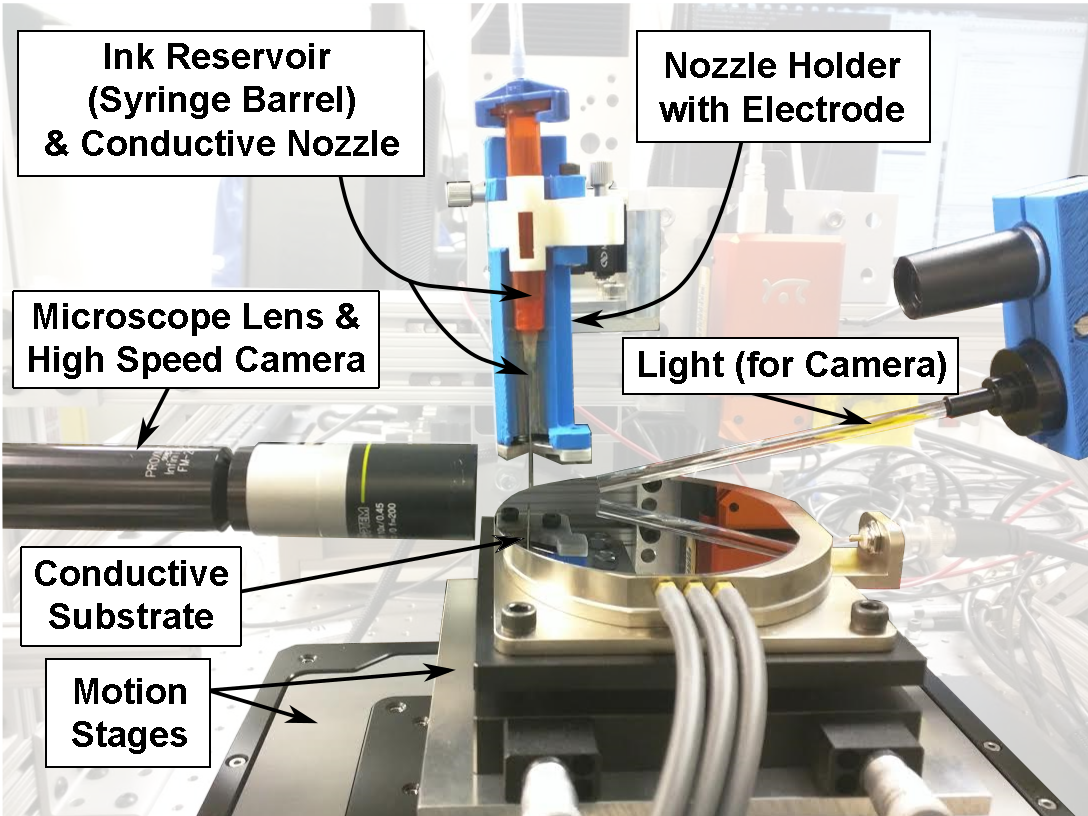}
    \\
    \vspace{0.2in}
    \includegraphics[scale=0.52]{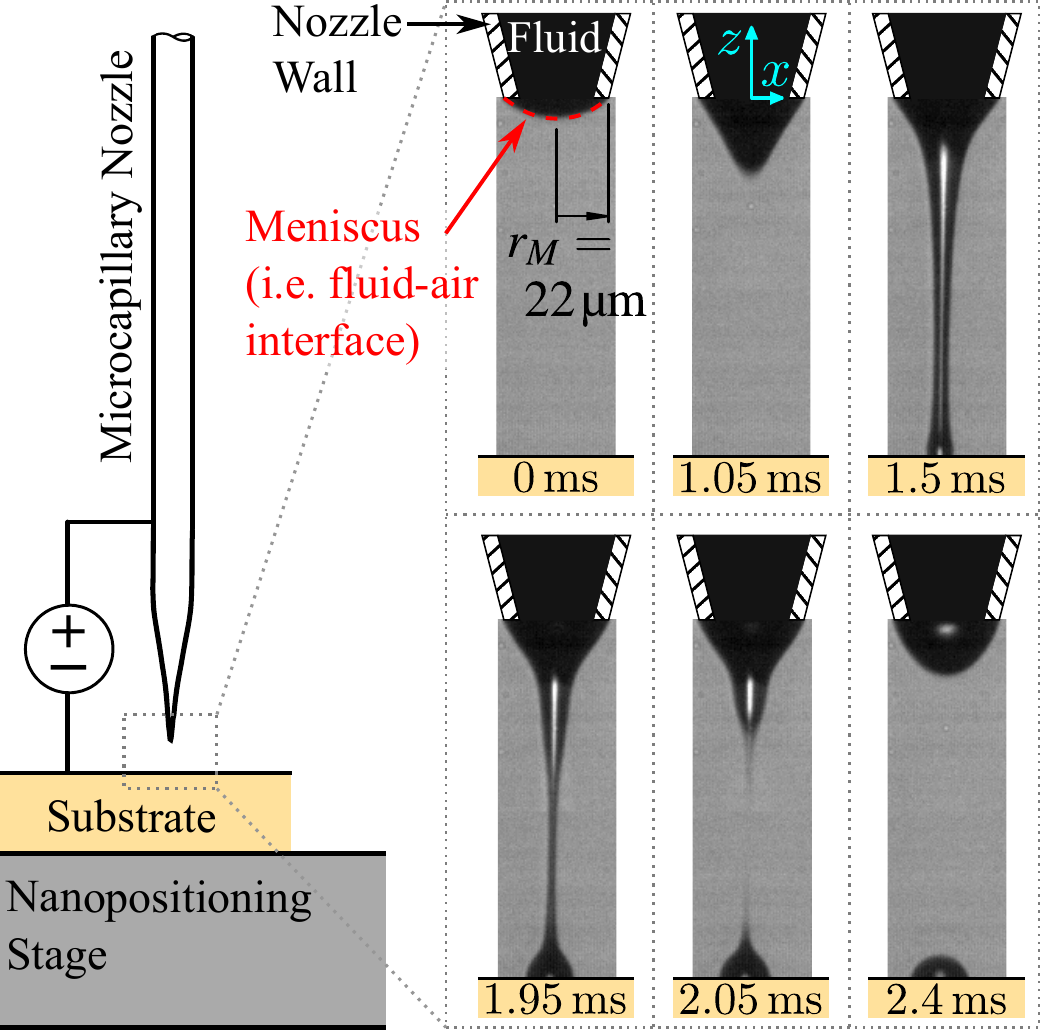}
    \caption{
    \emph{Top:}
    Photograph of the e-jet printer at the University of Michigan. The conductive nozzle and substrate are connected by a high voltage amplifier to apply electric fields to the fluid at the nozzle tip.
    \emph{Bottom:}
    Schematic and time lapse photography 
    illustrating the ejection process
    (with \SI{0}{\milli\second} set at the beginning of a voltage pulse).
    The meniscus base radius $r_M$, determined by the nozzle outer diameter, is \SI{22}{\micro\meter} here, and in general can be as small as $\SI{1}{\micro\meter}$. The speed of the ejection is determined largely by the fluid viscosity, and can be as fast as \SI{10}{\micro\second} for inviscid fluids such as water. The \SI{2.4}{\milli\second} process shown here is for a fluid 300 times the viscosity of water, i.e. about the viscosity of castor oil. Further details are given in Chapter \ref{ch:2}.
    }
    \label{fig:ejetintro}
\end{figure}
}
\committee{ %
Associate Professor Kira Barton, 
Chair
\\
Associate Research Scientist Tulga Ersal%, University of Michigan
\\
Associate Professor Chinedum Okwudire%, University of Michigan
\\
Associate Professor Necmiye Ozay%, University of Michigan
}

\chair{Kira Barton}

\abbreviations{
 \acro{AM}{Additive Manufacturing}
 \acro{uAM}[\textmu-AM]{Microscale Additive Manufacturing}
 \acro{CV}{Control Volume}
 \acro{DoD}{Drop on Demand}
 \acro{e-jet}{Electrohydro\-dynamic Jet}
 \acro{FDM}{Fused Deposition Modeling}
 \acro{ILC}{Iterative Learning Control}
 \acro{ILILC}{Invert-Linearize Iterative Learning Control}
 \acro{LTI}{Linear Time Invariant}
 \acro{MLD}{Mixed Logical Dynamical}
 \acro{MPC}{Model Predictive Control}
 \acro{NILC}{Newton Iterative Learning Control}
 \acro{NMP}{Non-minimum Phase}
 \acro{ODE}{Ordinary Difference/Differential Equation}
 \acro{PD}{Proportional-Derivative}
 \acro{PWA}{Piecewise Affine}
 \acro{PWD}{Piecewise Defined}
 \acro{ZPETC}{Zero Phase Error Tracking Control}
}

\abstract{Hybrid systems have steadily grown in popularity over the last few decades because they ease the task of modeling complicated nonlinear systems. 
Legged locomotion, robotic manipulation, and additive manufacturing are representative examples of systems benefiting from hybrid modeling.
They are also prime examples of repetitive processes; gait cycles in walking, product assembly tasks in robotic manipulation, and material deposition in additive manufacturing.
Thus, they would also benefit substantially from Iterative Learning Control (ILC),
a class of feedforward controllers for repetitive systems
that achieve
high performance in output reference tracking
by learning from the errors of past process cycles.
However, the literature is bereft of ILC syntheses from hybrid models.
The main thrust of this dissertation is to provide a broadly applicable theory of ILC for deterministic, discrete-time hybrid systems, i.e. piecewise defined (PWD) systems.

A type of ILC called Newton ILC (NILC) serves as the foundation for this mission due to its admittance of an unusually broad range of nonlinearities.
Preventing the synthesis of NILC from hybrid models is the fact that
contemporary hybrid modeling frameworks do not admit closed-form function composition of a single state transition formula capturing the complete hybrid system dynamics.
This dissertation offers a new, closed-form PWD modeling framework to solve this problem.

However, NILC itself is not without flaw. This dissertation's research reveals that it generally fails to converge when synthesized from models with unstable inverses (i.e. non-minimum phase (NMP) models), a class that includes flexible-link robotic manipulators.
Thus, to fulfill the goal of providing the most broadly applicable control theory possible, improvement to NILC must be made to avoid the operation that causes divergence when applied to NMP systems (a particular matrix inversion).

Stable inversion---a technique for generating stable state trajectories from unstable systems by decoupling their stable and unstable modes---is identified as a valuable tool in this endeavor.
This concept is well-explored for linear time invariant systems,
but 
stable inversion for hybrid systems has not been explored 
by the prior art.
Thus, to focus the research, this dissertation specifically examines piecewise affine (PWA) systems (a subset of PWD systems) for the study of NMP hybrid system control.
For PWA systems (and their PWD superset), in addition to a lack of stable inversion, a general, closed-form solution to the conventional inversion problem is also absent from the literature. Having a closed-form conventional inverse model is a prerequisite for stable inversion, but inversion of PWA models is nontrivial because the uniqueness of PWA system inverses is not guaranteed as it is for ordinary affine systems.
Therefore, to achieve the first ILC of a hybrid system with an unstable inverse, 
theory for both
conventional inversion and stable inversion must be delivered for PWA systems.

In summary, the three main gaps addressed by this dissertation are (1) the lack of compatibility between existing hybrid modeling frameworks and ILC synthesis techniques, (2) the failure of NILC for NMP systems, and (3) the lack of inversion and stable inversion theory for PWA systems. These issues are addressed by (1) developing a closed-form representation for PWD systems, (2) developing a new ILC framework informed by NILC but free of matrix inversion, and (3) deriving conventional and stable model inversion theories for PWA systems.}
\begin{document}

\chapter{Introduction}
\label{ch:1}
\acresetall

\section{Motivation: E-jet Printing}

The main contributions to hybrid systems and control theory are presented in Chapters \ref{ch:3}-\ref{ch:5}. 
These chapters validate their theoretical contributions with simulations on a diverse set of example systems derived from mechanics principles or data collected from physical systems. Additionally, the simulations feature numerous types of model errors in order to stress test the proposed controllers.
However, it is desirable to further ground theoretical controls research in the needs of 
application-based
research.
To this end Chapter \ref{ch:2} presents \ac{e-jet} printing in detail, specifically the first efforts to create end-to-end \ac{ODE}-based models of the fluid flow from the nozzle to the substrate over the course of an ejection.

Two models are derived. The first is focused on extending physics-based nonlinear modeling as much as possible to achieve an end-to-end dynamical model of the flow rate of fluid through the nozzle outlet. The second is focused on capturing the volume of fluid actually deposited on the substrate. Additionally, the second model seeks to accomplish this while limiting the system nonlinearities to a \ac{PWA} definition of the dynamics.

As explained thoroughly in the chapter, this modeling research is a direct response to the
future
desire to achieve high performance control of deposited droplet volume.
In this endeavor, automatic control is required due to the high uncertainty in system behavior, and \ac{ILC} is the clearest choice because physical sensing limitations preclude real-time feedback control. The desire for \ac{ILC} motivates the desire for an \ac{ODE}-based model for controller synthesis.
As demonstrated in Chapter \ref{sec:physicsModeling}, hybridness is ultimately necessary to capture the end-to-end process dynamics. 
This is in spite of contributing substantial improvements to the physics-based smooth nonlinear electrohydrodynamics 
model,
resulting in a five-fold increase to the range of well-modeled meniscus deformation.
Thus, in addition to raw scientific contributions to \ac{e-jet} modeling, Chapter \ref{ch:2} concretely establishes the need for a theory merging \ac{ILC} and hybrid systems.

Note however, that this dissertation does not contain a physical implementation of the developed \ac{ILC} theory on an \ac{e-jet} printer.
This is because the main objective here is to deliver a foundational theory of \ac{ILC} for hybrid systems, and \ac{e-jet} printing presents additional challenges to \ac{ILC} beyond model hybridness.
These challenges are explained in the description of future work, Section \ref{sec:futurework}.
For a foundational \ac{ILC} theory such issues were deemed less critical than the treatment of the \ac{NMP} behavior described later in this chapter, which fundamentally impacts the stability of \ac{ILC} for many systems.
In other words, while \ac{e-jet} printing is representative of systems requiring iterative learning for control and requiring hybridness for end-to-end \ac{ODE}-based modeling, more \ac{e-jet}-printing-specific research is 
necessary before safe and effective automatic control is achievable.

\section{Aim 1: ILC of Piecewise Defined Systems}

Concretely defined,
\ac{ILC} is the process of learning an optimal feedforward control input over multiple trials of a repetitive process based on 
feedback measurements from previous trials. 
ILC is used when typical real-time-feedback and/or feedforward control techniques yield too much output tracking error (as determined by case-specific criteria) because of their reactive nature or model error, respectively. 
Notable
past applications
include robot-assisted stroke rehabilitation \cite{Freeman2015}, high speed train control \cite{Yu2018}, and laser additive manufacturing \cite{Rafajlowicz2019}, all of which use nonlinear 
models.
In fact, while the majority of ILC literature focuses on linear systems, the prevalence of nonlinear dynamics in 
real-world systems
has motivated the development of numerous ILC syntheses from discrete-time nonlinear models 
\cite{Jang1994,Saab1995,Wang1998,Sun2003}.

Nonlinear modeling can be extremely challenging.
Chapter \ref{ch:2} illustrates that there are repetitive systems for which even extensive traditional (i.e. continuous) nonlinear modeling is insufficient for capturing a process's full dynamics.
Hybrid system modeling offers a more flexible formal modeling framework.

The term ``hybrid systems'' encompasses a wide variety of modeling frameworks that are composed of a set of traditional dynamical models---e.g. systems of 
\acp{ODE}---and a set of rules regarding switching between which of these ``component'' models is governing the system state evolution at a particular point in 
time \cite{Heemels2001}.
The set of switching rules can usually be interpreted as a discrete event system
\cite{Cassandras2008,Tabuada2009}.
Each discrete state, or ``location,'' is associated with one of the component dynamical models. The transitioning between locations is conditioned on the dynamical states and control inputs.
This dissertation specifically considers deterministic hybrid systems, in which satisfaction of any transition condition enforces a location switch at the moment in time the condition is satisfied. This is in opposition to frameworks based on transition guard conditions and location invariant conditions, which may \emph{allow} a transition if satisfied, but do not typically \emph{enforce} switching.
Deterministic switching behavior (and state resets) can in general be captured by piecewise definition of a state space system via augmentation of the state dimension \cite{Torrisi2004}. Thus, most deterministic hybrid system frameworks are a subset of \ac{PWD} systems \cite{Heemels2001}.

In the past couple decades, hybrid systems have become broadly popular because their ability to stitch simple dynamical systems together to produce complicated state trajectories has greatly eased the task of modeling many physical systems. 
Examples include cyber-physical systems in general \cite{Tabuada2009}, automobile driver behavior \cite{Kim2004}, power systems \cite{Hiskens2000}, legged locomotion \cite{Hiskens2001}, conveyor systems \cite{Saez2017}, and additive manufacturing \cite{Spiegel2017}.

Clearly, there is overlap between repetitive systems and systems well-modeled as hybrid systems: gait cycles in legged locomotion are repetitive, as are many manufacturing processes. In fact, of myriad robotics-related uses, manufacturing is one of the primary fields in which ILC is applied \cite{Bristow2006}.
However, the combination of ILC and hybrid modeling is absent from the literature.
This is the main gap of Chapter \ref{ch:3}.

The first step to achieving ILC of \ac{PWD} systems is to choose a particular ILC scheme to build off of.
However, literature on the ILC of nonlinear systems can trend towards hyper-specialization with respect to the system 
model, making it more restrictive than desired.
 Even amongst the more general literature such as \cite{Jang1994,Saab1995,Wang1998,Sun2003,Chi2008} mentioned above,
nearly all published ILC 
theory for discrete-time nonlinear systems
feature at least one of the following model restrictions
\begin{enumerate}[label=(R\arabic*),leftmargin=*]
\item
\label{R1}
relative degree of either 0 or 1 \cite{Saab1995,Wang1998},
\item
\label{R2}
affineness in the input \cite{Jang1994,Saab1995,Wang1998,Chi2008},
\item
\label{R3}
time-invariance \cite{Jang1994,Sun2003}, and
\item
\label{Rsmooth}
smoothness of the state transition formula and output functions (Lipschitz continuity at the most relaxed) \cite{Jang1994,Saab1995,Wang1998,Sun2003,Chi2008}.
\end{enumerate}

This is problematic because many practical systems violate these constraints.
\refRmu{} may be violated in the position control of myriad systems including piezoactuators \cite{Shieh2008}, motors \cite{Hackl2013}, robotic manipulators \cite{Geniele1997}, and vehicles \cite{Munz2011}.
\refRlin{} may be violated by piezoactuators \cite{Shieh2008}, electric power converters \cite{Escobar1999}, wind energy systems \cite{DeBattista2004}, magnetic levitation systems \cite{Gutierrez2005}, e-jet printing (Chapter \ref{sec:physicsModeling}) and flexible-link manipulators \cite{Geniele1997}.
\refRTI{} may be violated by any feedforward-input-to-output model of systems using both feedforward and feedback control, as is often done for robotic manipulation \cite{Khosla1988}.
Finally, and of primary concern here, \refRsmooth{} may be violated by gain switching feedback control systems (e.g. for motor control \cite{Kalman1955}), power converters \cite{Escobar1999}, legged locomotion \cite{Garcia1998}, e-jet printing (Chapter \ref{ch:2}), and robotic manipulation \cite{Hogan2020}.
The fact that many of these example systems violate multiple
restrictions 
illustrates that it can be challenging 
to find a model-based ILC synthesis scheme appropriate for 
many real-world applications.
Indeed, flexible-link manipulators violate all four, and they are relevant to the fast and cost-effective automation of pick-and-place and assembly tasks as well as to the control of large structures such as cranes \cite[ch. 6]{DeWit1996}. 
Such application spaces would
benefit from having
a versatile ILC scheme free from \refRmu-\refRsmooth.

Additionally, while
ILC seeks to converge to a satisfactorily low error, this learning is not immediate, and trials executed before the satisfactory error threshold is passed may be seen as costly failures from the perspective of the process 
specification.
It is thus desirable to develop ILC schemes that converge as quickly as possible.

There is one published \ac{ILC} scheme that meets the need for versatility and speed:
the application of Newton's root finding algorithm to a complete finite error time series (as opposed to individual points in time).
This technique was first proposed in \cite{Avrachenkov1998}, and is called \ac{NILC} here.
\AILC's synthesis procedure and convergence analysis are unusually broad in that they are free of \refRmu-\refRsmooth{} \cite{Avrachenkov1998}.
Additionally, 
Newton's method has been shown to deliver faster convergence in ILC than more basic schemes such as P-type ILC \cite[ch. 5]{Xu2003}.

However, while the convergence conditions of \ac{NILC} do not preclude their application to hybrid systems, the synthesis of \ac{NILC} requires a closed-form lifted system model.
This lifted model is a vector-input-vector-output function taking in the control input time series and outputting the system output time series.
To construct this model from state space systems requires function composition and differentiation of a single closed-form state transition formula capturing the model's entire dynamics.
Existing hybrid system formalisms rarely possess such monolithic state transition formulas, and the ones that do cannot explicitly nest calls to them via function composition (further details on this point are given in the introduction to Chapter \ref{ch:3}).
Thus the approach to addressing the first main gap, i.e. the first main contribution, is the development of a closed-form state space representation of \ac{PWD} systems. This enables \ac{ILC} of a large swath of hybrid systems.

To support this main contribution, 
Chapter \ref{ch:3} also delivers two ancillary contributions related to the implementation of \ac{NILC}.
The original \ac{NILC} literature \cite{Avrachenkov1998} assumes an appropriate lifted model is given. Synthesis of this model from state-space systems requires careful handling of the system relative degree, which has been neglected in subsequent works leveraging \ac{NILC}, such as \cite{Lin2006}. 
Neither have time-varying nonlinear system dynamics been considered in the lifted system derivation.
Thus, the first ancillary contribution is an explicit formalization of \ac{NILC} for time-varying nonlinear systems of any relative degree $\geq 1$.
Additionally, the differentiation in Newton's method has been challenging for past authors due to computational cost.
The second ancillary contribution is the incorporation of automatic differentiation (see \cite{Andersson2018,Naumann2008}) into \ac{NILC} implementation, dramatically reducing this cost.

\section{Aim 2: Versatile, Fast ILC of NMP Nonlinear Systems}

\ac{NILC} has one of the least restrictive sets of sufficient conditions for convergence published in the prior art.
However, Chapter \ref{ch:4} reveals 
that when synthesized from models with unstable
inverses\footnote{
More specifically: models for which the inverse's linearization about the input trajectory is unstable. See Section \ref{sec:NILCfailure} for details.
},
\AILC{} typically generates control signals that diverge 
to enormous
magnitudes.
In other words \AILC{} is not compatible with these models, which are often called 
\ac{NMP}
models. 
This is problematic because a number of important physical systems are well represented by \ac{NMP} models.
Examples include
piezoactuators \cite{Schitter2002}, electric power converters \cite{Escobar1999}, wind energy systems \cite{DeBattista2004}, DC motor and tachometer assemblies \cite{Awtar2004}, and flexible-link manipulators \cite{Geniele1997}.
Thus, incompatibility of \ac{NILC} and \ac{NMP} models is the main gap of Chapter \ref{ch:4}.

Note that
the 
full original definition of \ac{NMP} refers to the property of a frequency response function having the minimum possible phase change from $\omega=0$ to $\omega\rightarrow\infty$ for a given magnitude trajectory. For \ac{LTI} systems this is achieved if and only if the system and its inverse are causal and stable.
In other words, strictly proper systems cannot be minimum phase regardless of the stability of their inverses.
However, because the lack of causality is rarely an obstacle in feedforward control when the entire reference is known 
in advance,
\ac{NMP} is often used as jargon for inverse instability (equivalently, as an abbreviation for ``\nonmin{} phase \emph{zero dynamics}'') in the feedforward control community, and thus in this dissertation as well.

For linear 
 models with unstable inverses,  
a
common way to obtain feedforward control signals 
is
to systematically synthesize 
approximate dynamical  
models
with stable inverses 
by individually changing the model zeros and poles, e.g. \ac{ZPETC} \cite{Tomizuka1987}.
However, it is difficult to prescribe analogous systematic approximation methods for nonlinear models
because the poles and zeros do not necessarily manifest as distinct binomial factors that can be individually inverted or modified in the system transfer function.

An alternative 
is to 
harness the fact that a scalar difference equation that is unstable 
when evolved forward in time from an initial condition
is stable if evolved backwards in time from a terminal 
condition.
If the stable and unstable modes of a system are decoupled and evolved in opposite 
directions, a stable total trajectory can be obtained.
This process is called 
stable inversion.
For linear systems on a bi-infinite timeline, with boundary conditions at time $\pm\infty$, stable inversion gives an exact solution to the output tracking problem posed by the unstable inverse model.
In practice on a finite timeline, a high-fidelity approximation is obtained by ensuring the reference is designed with sufficient room for pre- and post-actuation, i.e. with a ``flat'' beginning and end.
Additionally, unlike ILC, stable inversion alone cannot account for model error.
To address this,
\cite{Zundert2016} 
details
stable inversion and presents an ILC scheme for linear systems that 
incorporates
a 
process similar to stable inversion.

The main contribution of Chapter \ref{ch:4} is to fill the gap of \ac{NILC}'s failure for \ac{NMP} systems by developing a new 
nonlinear systems
\ac{ILC} framework, \ac{ILILC},
which has the ability to incorporate stable inversion into the controller synthesis.
\ac{ILILC} also retains \ac{NILC}'s advantages in speed and broad applicability 
because it preserves the fundamental learning structure on which the \ac{NILC} convergence analysis is founded.

Because stable inversion for hybrid systems does not appear in the prior art, Chapter \ref{ch:4} focuses on smooth nonlinear systems. Consequently, a supporting contribution is made to the theory of smooth nonlinear discrete-time stable inversion to expand its domain of applicability to include representations of systems under both feedback and feedforward control. Specifically, a relaxed set of sufficient conditions is proven for the convergence of the stable inversion procedure.

\section{Aim 3: 
(Stable) Inversion of Piecewise Affine Systems}

The aforementioned lack of stable inversion theory for Hybrid systems provides an obvious final gap that must be filled to achieve the first \ac{ILC} of a hybrid system with unstable inverse dynamics.
However, such a gap statement belies the fact that there is no published general solution to the conventional inversion of \ac{PWD} systems, let alone stable inversion.
To focus the research on solving both the inversion and stable inversion problems for a class of hybrid system, Chapter \ref{ch:5} considers a subset of \ac{PWD} systems called \ac{PWA} systems.

\ac{PWA} systems are simply \ac{PWD} systems with the component dynamics restricted to affine models.
Examples
include current transformers \cite{Ferrari-Trecate2003}, one-sided spring supports \cite{Bonsel2004}, gain switching \cite{Kalman1955}, and e-jet printing (Chapter \ref{sec:controlModeling}).
As with hybrid systems in general, 
the mathematical rigor provided by the \PWA{} framework facilitates analysis and control theory development for these systems.
Examples include stabilizing state feedback control \cite{Mignone2000} and model reference adaptive control \cite{DiBernardo2013}.

To date, such research has focused
primarily on feedback control.
Feedforward control
has not been thoroughly addressed in the \PWA{} literature. In addition to preventing the control of systems with specific needs for feedforward control, this gap  inhibits the implementation of existing feedback control theory that requires feedforward control components. Indeed, \cite{VandeWouw2008} presents a solution to the output reference tracking problem for a class of \PWA{} systems using both feedback and feedforward control elements, but does not present a method to compute the feedforward signal. The validation is instead limited to a master-slave synchronization example in which the feedforward input to the master system is known 
in advance.

In other words, while Chapter \ref{ch:5}'s contributions of rigorous theory for the inversion and stable inversion of \ac{PWA} systems do yield the first \ac{ILC} of an \ac{NMP} hybrid system, they also have broader ramifications for the control of \ac{PWA} systems.

\section{Contribution Synergy}

One can see that the control-theoretic contributions of Aims 1-3 are not wholly independent of one another.
Instead, they build off one another and share responsibility for enabling the \ac{ILC} of different classes of systems.
Figure \ref{fig:gardenBasic} gives a graphical representation of this synergy 
by treating the fundamental contributions as building blocks that work together to hold up both each other and the specific classes of systems for which they enable control.

\begin{figure}
    \centering
    \includegraphics[scale=1.3]{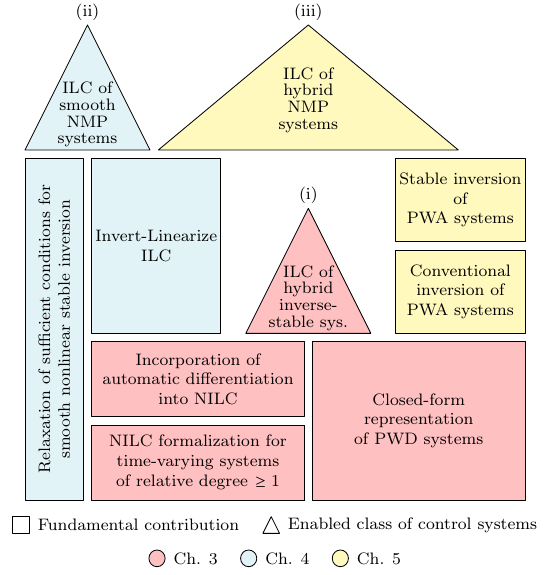}
    \caption{
    Castle of Control Contributions.
    Rectangular ``building blocks'' represent the fundamental control-theoretic contributions.
    Conical ``spires'' represent a class of nonlinear, potentially time-varying systems supported by the underlying theoretical contributions.
    Note that more elevated rectangular building blocks also depend on the building blocks beneath them.
    For example, theory for the stable inversion of \ac{PWA} systems \emph{requires} theory for the conventional inversion of \ac{PWA} systems. 
    The validation of the theory in each chapter is executed via control of an example system from the corresponding spire. Practical target applications for each class include (i) piecewise mass-spring-dampers (e.g. ankle-foot orthosis emulation \cite[ch. 3-4]{Totah2020}), (ii) cart-and-pendulum systems (e.g. bridge and tower cranes \cite{Fatehi2014,Solihin2010}), and (iii) micro-positioning systems (See Chapter \ref{sec:val}).
    }
    \label{fig:gardenBasic}
\end{figure}

\section{A Note On Notation}

Chapter \ref{ch:2} prioritizes the modeling of \ac{e-jet} printing, while Chapters \ref{ch:3}-\ref{ch:5} prioritize hybrid systems and \ac{ILC} theory. Thus, other than the most basic conventions, effort to unify the notation between these two sections has been limited. Overloaded symbol definitions occur, but are highly distinguished by context. For example $Q$ is used to indicate volumetric flow rate in Chapter \ref{ch:2} and a set of subspaces of the real vector space $\real^n$ in Chapters \ref{ch:3}-\ref{ch:5}.

Within Chapters \ref{ch:3}-\ref{ch:5} notation is unified with one exception. Chapters \ref{ch:3}-\ref{ch:4} represent the time argument in functions of discrete time as $x(k)$, with $k$ as the time step index. Chapter \ref{ch:5} uses the subscript notation $x_k$. This is done to compact a number of otherwise very long expressions.

\chapter{
Repetitive Processes Needing Hybrid Modeling:
\texorpdfstring{\\}{ }
Electrohydrodynamic Jet Printing Studies
}
\blfootnote{
Content of this chapter also published as:
\\
\indent\indent
I. A. Spiegel, P. Sammons and K. Barton, ``Hybrid Modeling of Electrohydrodynamic Jet Printing,'' in \emph{IEEE Transactions on Control Systems Technology}, vol. 28, no. 6, pp. 2322-2335, Nov. 2020, \href{https://doi.org/10.1109/TCST.2019.2939963}{https://doi.org/10.1109/TCST.2019.2939963}
\textcopyright IEEE 2020. Reprinted with permission.
\\
\indent\indent
I. A. Spiegel, T. van de Laar, T. Oomen and K. Barton, ``A Control-Oriented Dynamical Model of Deposited Droplet Volume in Electrohydrodynamic Jet Printing.'' in \emph{Proceedings of the ASME 2020 Dynamic Systems and Control Conference}. Virtual: ASME, 2020. \href{https://doi.org/10.1115/DSCC2020-3238}{https://doi.org/10.1115/DSCC2020-3238}
\textcopyright ASME 2020. Reprinted with permission.
}
\label{ch:2}

\ac{AM}
is a growing class of 
processes that fabricate components in a layer-by-layer fashion.
However, 
several obstacles 
inhibit the
widespread adoption of AM.  Chief among these 
is the lack of appropriate process descriptions for both implementing satisfactory process planners and integrating process feedback control to enable repeatable, accurate part fabrication.  A major reason for this is that many AM processes are governed by complex physical phenomena, e.g. the melting, heat transfer and solidification processes in Laser Engineered Net Shaping (LENS) and Selective Laser Melting (SLM) \cite{LENSphysics}, and the jetting and binder-burnout processes in Binder Jetting \cite{binderjetting}, resulting in strongly nonlinear mappings between process inputs and process outputs.
Because of these nonlinearities, models that may be useful for control design are only applicable in small regions of the operating space, 
accounting for only short periods within the total process.
Therefore, in order to enable control that addresses AM processes in a holistic manner, there is a need to develop models that are capable of describing the complex, interconnected, physical phenomena of AM processes.

One important subdivision of AM that has 
 attracted significant attention is 
\ac{uAM}.
 \muAM{} comprises a number of processes
characterized by their ability to produce feature sizes on the order of \SI{100}{\nano\meter} to \SI{100}{\micro\meter}.  
Processes typically classed as \muAM{} 
include ink jet \cite{inkjetrev} and other direct-write technologies such as dip-pen nanolithography \cite{dippennano}, and stereolithography \cite{microSTL}.  
Applications for these processes include patterning, printed electronic components, and biological and pharmaceutical devices.
For such applications, \muAM{} offers potential advantages in operating cost, speed, waste reduction, buildable geometries, and integration of fabricated components with larger systems compared to more traditional fabrication techniques involving lithography, spin-coating, or gas deposition \cite{amtech,ejetelectronic,ejetsensor}.
However, many \muAM{} technologies suffer from having coarser resolutions and smaller portfolios of usable materials than traditional technologies.

Electrohydrodynamic jet (e-jet) printing is a relatively nascent \muAM{} technology that seeks to overcome both of these challenges, and already has proven applications in printed electronics such as thin film transistors \cite{ejetelectronic} and resistive sensors \cite{ejetsensor}, and in biotechnology such as patterned cell cultures \cite{ejetbiotech}.  
In e-jet printing, a microcapillary nozzle filled with a polarizable ink is suspended above a flat substrate and an electric potential is applied between the two.  
The resultant electric  field induces a stress in the liquid surface. When this stress reaches a critical value,
a 
thin
jet of ink issues from the liquid surface towards the substrate.  
Once a 2D pattern of liquid has been deposited on the substrate, a process such as sintering or curing solidifies the liquid. More liquid can then be deposited atop the previously solidified liquid to build 3D structures \cite{Han2015JMP}.

Unlike back-pressure-driven processes such as inkjet and aerosol jet printing, 
the physical phenomenon driving jetting in \ac{e-jet} printing is localized to the fluid surface.
Consequently, 
\ac{e-jet} can print liquids with viscosities multiple orders of magnitude greater than inkjet printing's max viscosity \cite{Jang2013},
allowing the use of new polymers and metallic nanoparticle solutions with higher concentrations of the active ingredient.
Also due to the surface-localization of the electrohydrodynamics, 
while inkjet droplets pinch off the meniscus near the nozzle tip and thus have diameters approximately equal to that of the nozzle, electrohydrodynamic jets thin considerably as they extend to strike the substrate. This enables submicron line widths and gap sizes compared to inkjet printing's $\sim$\SI{10}{\micro\meter} minimum \cite{Park2007}.

Reliably fulfilling the potential for submicron resolution
requires closed-loop control for the rejection of disturbances introduced by variations in nozzle shape, fluid properties, and environmental factors like temperature and humidity \cite{Sammons2017}.
However, the computer-vision-based measurement process, \cite{Spiegel2017}, is too slow compared to the jetting process for traditional real-time feedback control to be practical. Thus,
recent research 
has focused on closing the loop in the iteration domain rather than the time domain.
Specifically, this research uses measurements from previous trials of an \ejet{} task in conjunction with dynamical models of nominal system behavior to inform the feedforward control input signal for the subsequent trial \cite{Afkhami2021,Wang2018a,Altin2019}, a technique known as iterative learning control (ILC).

Thus far, this research has been limited to using models of droplet spreading over a substrate to determine the droplet volume necessary for achieving a desired final print topography. It has not used knowledge of the dynamics between the applied voltage signal and the volume of fluid ultimately ejected from the 
nozzle, instead assuming perfectly known static relationships.
Dynamical models 
relating applied voltage to fluid flow
could be used 
to decouple the problems of learning the correct volume to deposit and learning the applied voltage signal necessary to achieve that volume, potentially decreasing the number of trials
required to achieve satisfactory performance.

For a model to be compatible with ILC it must be founded on 
ordinary differential or difference equations (ODEs), but
the 
majority of current modeling efforts for e-jet printing can be broadly classified as either 
partial-differential-equation-driven finite element/volume methods
or static scaling laws.  In the former category, a significant amount of work has been aimed at studying the fluid-air surface profile as a function of electric field and, to a lesser extent, specific material properties such as fluid conductivity \cite{higuera2013,wei2013,collins2008Nature}.  Additionally, there has been work towards developing models of the spreading and coalescence 
of e-jet-printed droplets on substrates \cite{pannier2017,carter2014}, as well as developing relationships between input parameters, (e.g. applied voltage), and deposited droplet properties (e.g. contact angle and volume).  In the latter category, several authors proposed static scaling law relationships between process inputs such as applied voltage and certain process outputs such as the frequency at which jets issued from the fluid \cite{choi2008,collins2013,chen2006}.  While both the high fidelity simulation models and the static scaling models provide benefits for some applications, they are not suitable for use in model-based feedback control algorithms.  In particular, there is a lack of compact models that are capable of holistically describing the distinct dynamic regimes in e-jet printing. An attractive framework for accomplishing this task is that of hybrid dynamical systems,
as they enable the capture of complicated, varying dynamics while still maintaining a foundation in ODEs rather than in partial differential equations.

After some necessary further technical details on \ac{e-jet} printing in Section \ref{ejet}, this chapter presents two hybrid modeling frameworks for \ac{e-jet} printing.
First, Section \ref{sec:physicsModeling} presents a continuous-time model with nonlinear components that highlights the necessity of hybridness for ODE-based modeling of \ac{e-jet} printing by using as much physics-based modeling as possible and then capturing the remaining unmodeled parts of the ejection process with data driven techniques. These physics- and data-driven models are linked by a hybrid framework. This constitutes the first-ever complete model of the ejection process based on ODEs. However, this model only captures fluid flow at the nozzle outlet, and does not explicitly output the final droplet volume deposited on the substrate. Thus, Section \ref{sec:controlModeling} presents a discrete-time, piecewise affine (PWA) model featuring simplified flow dynamics and a new framework for capturing the final droplet volume.
Both sections feature empirical validation of their models.

\section{Electrohydrodynamic Jet Printing}
\label{ejet}
\label{sec:ejetBack}

The conventional e-jet printing setup requires an ink-filled emitter, typically a conductive (or conductively coated) microcapillary nozzle, and an attractor, typically a flat, conducting, grounded substrate.  The emitter is positioned vertically above the attractor and an electric potential is applied between the two components.
Figure \ref{fig:ejetschematic} gives a schematic of this configuration, labeling 
both 
the important geometric constants of the printer setup and two important dynamical process variables: volumetric flow rate of ink out of the nozzle, $Q$, and the position of the meniscus tip in space, $h$.
Volumetric flow rate $Q$ is important because of its obvious relevance to the fluid volume ultimately deposited on the substrate. 
Meniscus position
$h$ is important both because it captures important ``milestones'' in the ejection process (e.g. the transition between dynamic regimes, impingement of the jet on the substrate, and the breaking of the jet) and because its dynamics are coupled to those of $Q$ by capturing 
the 
change in the system's capacitor geometry.

\begin{figure}
    \centering
    \includegraphics[scale=0.5]{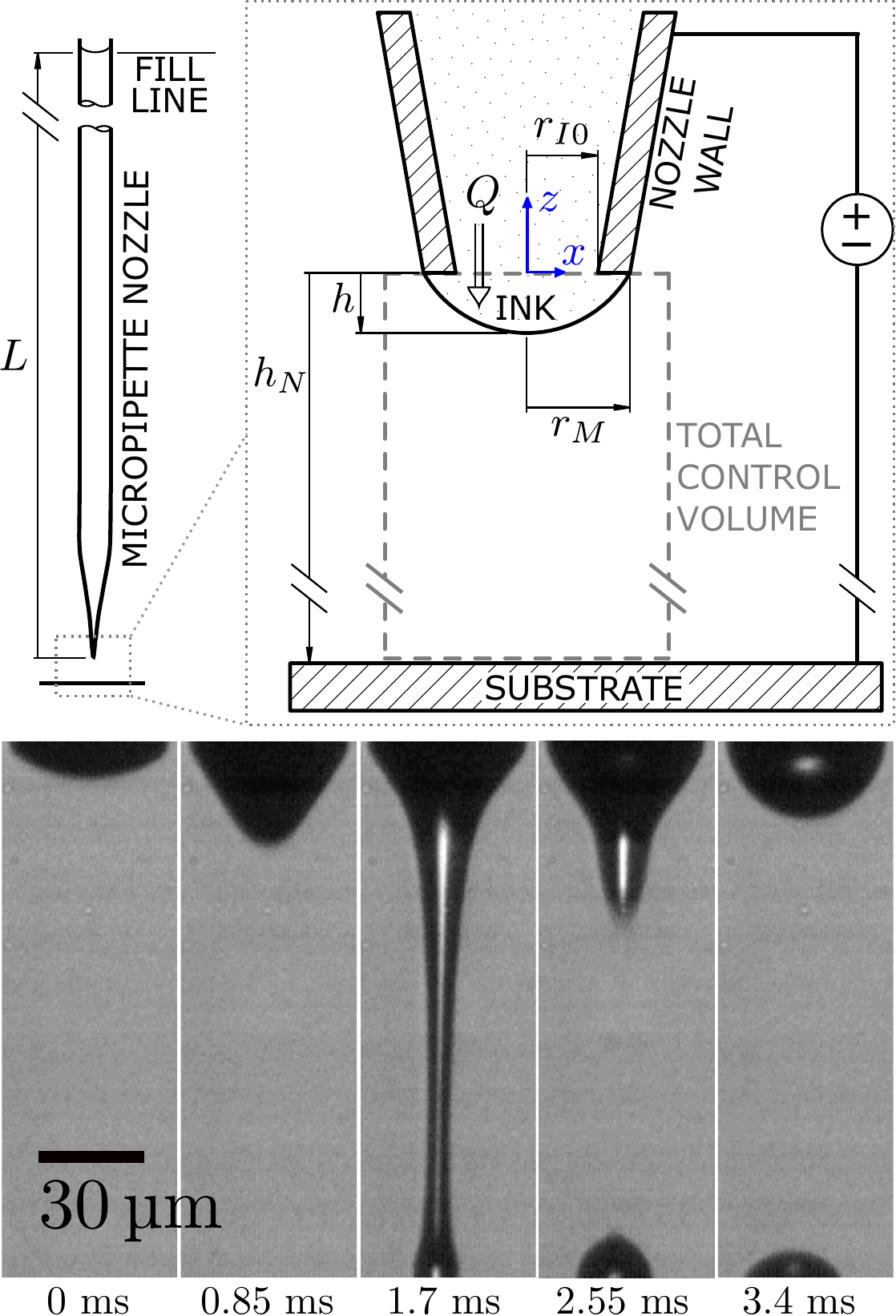}
    \caption{(\emph{Top}) Schematic of e-jet printing setup. Dimensions represent signed displacements in the direction of the single-sided arrows with respect to the inertial $(x,z)$ coordinate system. The volumetric flow rate $Q$, labeled via block arrow, is also a signed quantity, but its sign is with respect to the control volume denoted by the dashed box. (\emph{Bottom}) Time lapse photography of an ejection with time stamps from the rising edge of a voltage pulse.}
    \label{fig:ejetschematic}
\end{figure}

The electric field induced between the liquid meniscus and the attractor causes the meniscus to deform. This deformation can be stable at low voltages, and as applied voltage increases the equilibrium shape of the meniscus changes from a spherical cap to a sharp point known as a Taylor cone. If the applied voltage is high enough, the meniscus becomes unstable and a jet of ink issues from the tip of the Taylor cone towards the attractor.
While the voltage difference between the emitter and attractor is maintained and there exists material contiguity between the emitter and attractor, a redistribution of charge occurs in the fluid until 
the electrically induced surface stress becomes weak enough that the natural liquid surface tension
causes the jet to retract, leaving a droplet of ink on the attractor.

If voltage is held high after this ejection, charge will again accumulate at the meniscus and another ejection will occur.
The indefinite repetition of this cycle at constant high voltage is termed DC printing.
\ac{DoD}
printing is an alternative method in which distinct pulses of length $T_p$ and high voltage $V_h$ are used against a constant low voltage bias $V_l$ to control the timing and size of droplet deposition \cite{Spiegel2017}.
DoD printing can be further subdivided into subcritical printing and the complementary supercritical printing.
Subcritical printing, upon which this work focuses, is defined as \ac{DoD} printing in which $T_p$ is short enough that the high voltage pulse falls before the natural cessation of the jet can begin \cite{Spiegel2017}. In other words, subcritical printing is where each pulse corresponds to a single ejection which is stopped artificially by the falling edge of the pulse. (Supercritical printing is any \ac{DoD} printing that is not subcritical, and is not considered in this work).

The remaining sections in this chapter subdivide the ejection process itself into a set of distinct dynamic regimes.
These regimes are the locations of the hybrid model.
For both hybrid models in this chapter, the division of the total ejection process into partial processes is based on physically significant events. However, because the priorities of the two models differ, so do the location definitions. Section \ref{sec:physicsModeling} partitions the system dynamics about the extension of the meniscus tip beyond its maximum stable elongation, while Section \ref{sec:controlModeling} partitions the system dynamics about impingement of the jet on the substrate.

\section{First-Principles-Based Modeling}
\label{sec:physicsModeling}
To leverage as much as possible existing knowledge of \ac{e-jet} printing physics, this section divides the ejection process into the following three partial processes.
The ``build-up'' regime describes the initial deformation of the meniscus into a Taylor Cone when voltage is stepped high. 
The ``jetting'' regime describes the development of the jet at the Taylor cone tip, its approach towards the attractor, the fluid flow while the jet is fully developed, and the retraction of the jet back to the Taylor cone.
Finally, the ``relaxation'' regime describes the settling of the meniscus from a Taylor cone back to a stable equilibrium position while the voltage is low.
A complete ejection consists of switching from build-up to jetting to relaxation.
This process, and the distinction between dynamic regimes and actuation methods, is illustrated in Figure \ref{fig:ejetcycle}.

\begin{figure}
    \centering
    \includegraphics[scale=0.9]{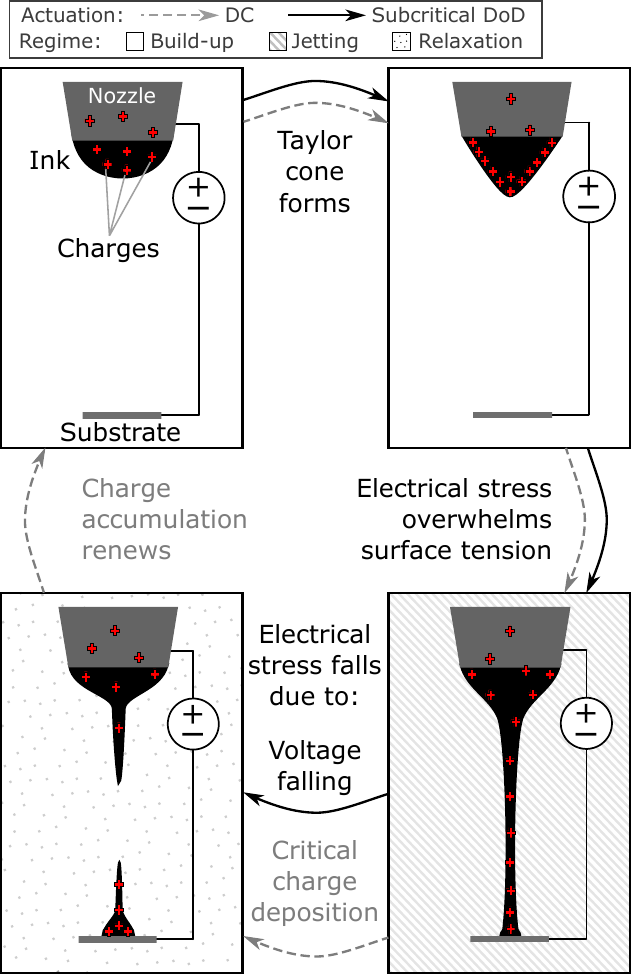}
    \caption{Ejection process in terms of 
    the 
    physical phenomena
    at the liquid-air interface. Distinction is drawn between the process under a constant DC applied voltage, for which ejections begin and end repeatedly under the natural high voltage electrohydrodynamics, and under a single subcritical pulse, which yields a single ejection whose cessation is induced by the falling of the pulse.}
    \label{fig:ejetcycle}
\end{figure}

The rest of the section is structured as follows.  
Section \ref{hybridFrame} defines the mathematical representation of e-jet as a hybrid dynamical system, and derives the individual model components.  
Section \ref{datProc} describes the experimental system, the procedure for collecting and processing data, and the system identification techniques for defining the data-driven portions of the model.
Section \ref{modelVal} presents a validation of the identified model.
Finally, Section \ref{conc} provides a conclusive summary of the section.

\subsection{Hybrid System Model Framework}
\label{hybridFrame}
A hybrid dynamical system is a synthesis of a discrete event system and a dynamical system governed by differential or difference equations. 
Because this section emphasizes the physics-driven modeling of \ac{e-jet} printing, continuous-time differential equations are
used\footnote{
This section is the exception to the rule of using discrete-time systems in this dissertation. 
This section thus uses a somewhat different hybrid system formalization than Section \ref{sec:controlModeling} onwards.
}.
Like discrete event systems, hybrid systems are often formalized as a set, called a hybrid automaton, in which each element 
describes a different feature of the system behavior.
Here,
the hybrid automaton definition of Cassandras and Lafortune \cite{Cassandras2008} serves as a basis for this work's formalism, in which a hybrid system is given as a 10-tuple $\mathcal{G}=(P,X,U,Y,f,g,\phi,\psi,p_0,\xvec_0)$.  $P$ is a set of discrete states or ``locations,'' $X=\mathbb{R}^n$ is the dynamical state space, $U=\mathbb{R}^m$ is the dynamical control input space, $Y=\mathbb{R}^q$ is the dynamical output space, $f:P \times X \times U \rightarrow X$ is a vector field denoting the dynamics of each location, $g \colon P \times X \rightarrow Y$ is the arithmetic map from the dynamical states to the output space in a given location, $\phi : P \times Y \rightarrow P$ is the transition function determining the discrete state based on the dynamical output (usually an inequality condition), $\psi : P\times P\times X \rightarrow X$ is the reset function which can instantaneously change the dynamical state when a transition occurs, and $p_0$, $\xvec_0$ are initial conditions of the discrete and dynamical states.
This definition is reminiscent of a \ac{PWD} nonlinear system.
However, here the automaton structure helps guide the modeling effort, and the use of continuous time rather than discrete time complicates the translation of resets into a strict \ac{PWD} framework, which usually relies on a notion of sample period \cite{Torrisi2004}.

Furthermore, this
definition features one significant structural simplification when compared to most general contemporary hybrid systems, which imposes a limitation on system behavior.
It does not allow for uncertain discrete transitions because it does not have invariants on the locations or guards on the transitions. (Invariants and guards allow for separate, potentially overlapping, restrictions on the viability of each location and transition). Instead, the transition function $\phi$ enforces 
a one-to-one mapping from the dynamic outputs and the current location to the next location. 
This limitation is imposed to keep the scope of the mathematical framework comparable to the physical modeling objectives of this work---in which system stochasticity is not considered.

In this work, the e-jet system is modeled with three locations, $P=\{p_1,p_2,p_3\}$ corresponding to build-up, jetting, and relaxation, respectively. Thus, at any point in time, the state of the system is described by $p(t)\in P$ and the dynamical state vector $\xvec(t) \in X$, and is driven by the input vector $\uvec(t) \in U$ via the differential equations
\begin{align}
    \dot{\xvec}(t) 
    \hspace{-0.75mm}
    &= 
    \hspace{-0.75mm}
    f(p(t),\xvec(t),\uvec(t)) 
    \hspace{-0.75mm}
    = 
    \hspace{-0.75mm}
    \begin{cases}
    f_1(\xvec(t),\uvec(t)) & p(t)=p_1 \\
    f_2(\xvec(t),\uvec(t)) & p(t)=p_2
    \\
    f_3(\xvec(t),\uvec(t)) & p(t)=p_3
    \end{cases}
    \hspace{-3mm}
    &
    \\
    \yvec(t) 
    &= 
    g(p(t),\xvec(t)) =
    \begin{cases}
    g_1(\xvec(t)) & p(t) = p_1 \\
    g_2(\xvec(t)) & p(t) = p_2 \\
    g_3(\xvec(t)) & p(t) = p_3
    \end{cases}
    &
    \\
    \label{eq:instantTrans}
    p(t) 
    &= 
    \phi(p_-(t),\yvec(t))
    &
\end{align}
Note that (\ref{eq:instantTrans}) 
captures
the instantaneousness of discrete
transitions, with $p_-(t)$ 
being
the value of $p(t)$ prior to $\phi$ evaluation.

In this work the outputs are the physical system values, $\yvec(t)=[h(t), \, Q(t)]^T$.
The exact definitions of $\xvec$ and $\uvec$ depend on the details of the differential equations to be derived,
but for the sake of clarity, they are preemptively defined as
% % ----------------------------------------------
% % ONE COLUMN
% % ----------------------------------------------
% \begin{align}
%     \xvec(t) & = \begin{bmatrix} \delta h(t) & \delta Q(t) & \delta \dot{Q}(t) \end{bmatrix}^T =
%     \begin{cases}
%     \begin{bmatrix}
%     h(t)-0 & Q(t)-0 & \dot{Q}(t)
%     \end{bmatrix}^T
%     & p(t)\in \{p_1,\,p_3\}
%      \\
%     \begin{bmatrix}
%     h(t)-h_N & Q(t)-Q_f & \dot{Q}(t)
%     \end{bmatrix}^T
%     & p(t)=p_2
%     \end{cases}\\
%     %
%     \uvec(t) & = u(t) =
%     \begin{cases}
%     V(t)^2 - 0^2   & p(t)\in \{p_1,\,p_3\} \\
%     V(t)^2 - V_h^2 & p(t)=p_2
%     \end{cases}
% \end{align}
% ----------------------------------------------
% ONE COLUMN v2
% ----------------------------------------------
\begin{align}
    \xvec(t) & = \begin{bmatrix} \delta h(t) \\ \delta Q(t) \\ \delta \dot{Q}(t) \end{bmatrix} =
    \begin{cases}
    \begin{bmatrix}
    h(t)-0 \\ Q(t)-0 \\ \dot{Q}(t)
    \end{bmatrix}
    & p(t)\in \{p_1,\,p_3\}
     \\
    \begin{bmatrix}
    h(t)-h_N \\ Q(t)-Q_f \\ \dot{Q}(t)
    \end{bmatrix}
    & p(t)=p_2
    \end{cases}\\
    \uvec(t) & = u(t) =
    \begin{cases}
    V(t)^2 - 0^2   & p(t)\in \{p_1,\,p_3\} \\
    V(t)^2 - V_h^2 & p(t)=p_2
    \end{cases}
\end{align}
% ----------------------------------------------
% TWO COLUMN
% ----------------------------------------------
% \begin{align}
%     \xvec(t) & = \begin{bmatrix} \delta h(t) & \delta Q(t) & \delta \dot{Q}(t) \end{bmatrix}^T = \\
%     & \, \, \begin{cases}
%     \begin{bmatrix}
%     h(t)-0 & Q(t)-0 & \dot{Q}(t)
%     \end{bmatrix}^T
%     & p(t)\in \{p_1,\,p_3\}
%     \nonumber \\
%     \begin{bmatrix}
%     h(t)-h_N & Q(t)-Q_f & \dot{Q}(t)
%     \end{bmatrix}^T
%     & p(t)=p_2
%     \end{cases}\\
%     %
%     \uvec(t) & = u(t) =
%     \begin{cases}
%     V(t)^2 - 0^2   & p(t)\in \{p_1,\,p_3\} \\
%     V(t)^2 - V_h^2 & p(t)=p_2
%     \end{cases}
% \end{align}
where $h_N$ is the displacement of the substrate from the nozzle outlet and $Q_f$ is the flow rate
at the end of the voltage pulse.
This model structure allows for the dynamical states $\delta h$ and $\delta Q$ to be equal to the outputs $h$ and $Q$ during build-up and relaxation (i.e. when the models are physics-based), and equal to deviations from output equilibria $h_N$ and $Q_f$ during jetting (i.e. when the model is a black box).

A graphical representation of the hybrid automaton for an e-jet ejection is given in Figure \ref{fig:hybridAutomaton}.

The remainder of this section focuses on deriving $f_1$, $f_2$, $f_3$, $\phi$, and $\psi$.
For all these derivations, it is assumed that nozzle and the substrate have zero velocity with respect to one another, and that the substrate is flat and clean (i.e. free of pre-deposited substances).

\begin{figure}
    \centering
    \includegraphics[scale=0.625]{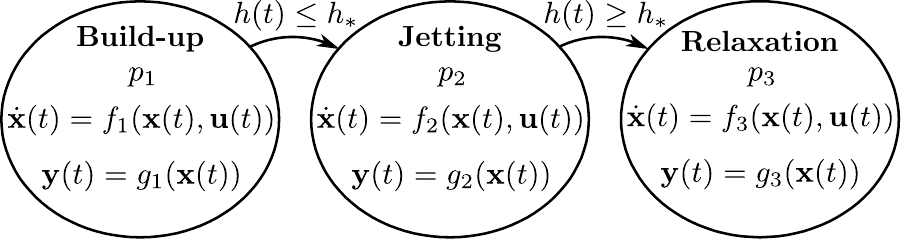}
    \caption{Schematic of the system as a hybrid automaton. The transition inequalities are labeled explicitly on the automaton edges for completeness, and are derived in section \ref{sec:Transition}.}
    \label{fig:hybridAutomaton}
\end{figure}

\subsubsection{Build-up \& Relaxation}
\label{sec:BuildUp}
Physically, the build-up location represents the 
accumulation of charge and mass in the meniscus under high voltage leading up to jetting. The relaxation state represents the 
convergence of the meniscus to an equilibrium shape at low voltage.
Despite involving nontrivial changes to the meniscus volume and shape, throughout both of these states the liquid's form factor is approximately that of a typical pendant droplet,
whose deformations under 
applied electric fields have been the subject of much study \cite{Larrondo1981,Joffre1982,Harris1993}.
This prior art enables these locations to be modeled from a perspective of physical first-principles. Such a model was first proposed by Wright, Krein, and Chato \cite{Wright1993} in 1993, further developed by Yang, Kim, Cho and Chung \cite{Yang2014} in 2014, and serves as a foundation for the physics-driven modeling in this work.

The premise of Wright's dynamic model is a 1-dimensional (along the nozzle axis) Newtonian force balance across the plane of the nozzle outlet (i.e. across the boundary of the control volume via which $Q$ is defined). Algebraic manipulation of this force balance yields the derivative of flow rate as a function of pressures at the plane. This is represented as
\begin{equation}
     \dot{Q} = \frac{\pi r_{I0}^2}{\rho L}P_{net} \label{flowGod}
\end{equation}
where $r_{I0}$ is the inner radius of the nozzle outlet, $\rho$ is the mass density of the ink, $L$ is the length of the fluid column, and $P_{net}$ is the net pressure difference across the nozzle outlet, equal to the sum of pressure changes arising from the dominant physical phenomena in the system. These physical phenomena are
the pressure due to gravity $P_g$ (taken to be hydrostatic pressure),
the pressure due to viscous flow $P_\mu$ (assumed equal to Hagen-Poiseuille flow in the nozzle),
the pressure due to surface tension $P_\gamma$ (from the Young-Laplace equation applied to the meniscus tip), 
and 
the electric pressure due to applied field $P_E$ (equal to the electrical energy density of the system at the meniscus tip \cite{Calero1988,Woodson1968}).
Mathematically, these pressures are given by
\begin{align}
    P_g &= \rho g L \label{Pg}\\
    P_\mu(t) &= -\frac{8\mu L Q(t)}{\pi r_I^4} \label{PmuWright}\\
    P_\gamma(t) &= -\frac{2\gamma}{R(t)} \label{eq:Pgamma} \\
    P_E(t) &= \frac{1}{2}\varepsilon E(t)^2 \label{PE}
\end{align}
where $g$ is the gravitational constant, $\mu$ is the dynamic viscosity of the ink, $r_I$ is a measure of the nozzle shaft inner radius, $\gamma$ is the liquid-air surface tension coefficient of the ink, $R$ is the radius of curvature of the meniscus tip, $\varepsilon$ is the permittivity of the medium between the meniscus and the ground plate (usually air), and $E$ is the electric field at the liquid-air interface at the meniscus tip.

This model provides a concrete ODE-based representation of flow rate dynamics in e-jet printing.
However, Wright does not perform any 
empirical validations of this model.
The model variant presented by Yang in \cite{Yang2014} sees some validation, but is limited by the fact that 
it relies on a time-varying variable that is not explicitly modeled, and whose value must be gleaned from measured data and updated at each time-step or approximated as constant.
Indeed, while the current work uses equations (\ref{flowGod}-\ref{PE}) virtually unchanged as a model framework,
several of their constitutive elements are substantially modified to facilitate the model's practical implementation.
Specifically, $P_\mu(t)$, $R(t)$, and $E(t)$ are addressed in this work.

In past works $P_\mu$ is given by the Hagen-Poisuille equation, which describes the pressure drop through a long tube of constant cross-section due to shear forces, and uses $r_I=r_{I0}$. However, the nozzles used in e-jet printing do not have constant cross-sections.
To achieve the small outlet diameters required, the nozzles are made from pulled-glass micropipettes, which do not have easily defined axial cross-sections.
Thus, this work treats $r_I$ in (\ref{PmuWright}) as an \emph{effective} inner tube radius, 
\begin{equation}
r_{I}=c_{r_I}r_{I0}
\end{equation}
where $c_{r_I}$ is a correction coefficient to be found via a system identification process described in section \ref{sysID}. This modeling strategy keeps the shear term physically meaningful while enabling a degree of flexibility required for dealing with uncertain nozzle geometry.

Both $R(t)$ and $E(t)$ depend on the shape of the meniscus. For the purpose of modeling $R(t)$ in the surface tension term, equation (\ref{eq:Pgamma}), Wright and Yang both assume the meniscus is a spherical cap. This assumption is clearly only valid for small deformations of the meniscus under an electric field. As the meniscus sharpens into a Taylor cone, it achieves a much smaller $R(t)$ than can be captured with a spherical model.

Wright does not offer a physics-based model for $E(t)$, but Yang models the meniscus (assumed to be conducting) and nozzle together as the one half of a two-sheeted hyperboloid of revolution. To fully determine this hyperboloid the following constraints are imposed:
\begin{enumerate}[label=(C2.\arabic*),leftmargin=*]
    \item 
    \label{C1}
    the surface vertex position is fixed to the meniscus tip,
    \item 
    \label{C2}
    the surface center of symmetry is fixed to the  substrate,% and
    \item
    \label{C3}
    $R(t)$ is equal to the meniscus tip radius of curvature measured from photographs.
\end{enumerate}
These constraints enable the use of an analytical solution for electric field, but \ref{C3} makes the model dependent on an enormous quantity of measured data, as each time step of every voltage step (characterized by the combination of $V_l$ and $V_h$) requires a different value of $R(t)$. Yang circumvents this need by considering only small meniscus deformations, up to roughly 40\% increases in $h(t)$ from an initial condition where the spherical cap assumption holds, 
and assuming $R(t)$ is constant. Additionally, to make up for differences between the theoretical electrode configuration described above and the physical system, the theoretical electric field is multiplied by a time-varying model adjustment coefficient, which itself is a function requiring the system identification of two parameters.

The current work presents a model seeking to capture large meniscus deformations
(greater than 200\% increases in $h(t)$)
while reducing the need for measured data to define the model. This is done by modeling the meniscus as a paraboloid of revolution for all dynamical equation derivations (rather than the mixed spherical and hyperboloidal paradigm), and by leveraging knowledge of the force balance at low voltage equilibrium.

To fully determine the paraboloid at any point in time, only two constraints are necessary: \ref{C1} and
\begin{enumerate}[label=(C2.\arabic*),leftmargin=*]
    \setcounter{enumi}{3}
    \item
    \label{C4}
    The surface intersects the nozzle-ink interface (i.e. the outer edge of the nozzle outlet).
\end{enumerate}
Under these constraints, the radius of curvature and theoretical electric field are given by \cite{Kaiser2004}
\begin{align}
    R(t) &= \frac{r_M^2}{-2h(t)} \label{eq:R} \\
    E_t(t) &= \frac{2V(t)}{R(t) \ln \left( \frac{2(h(t)-h_N)}{R(t)} \right)}
\end{align}
The adjusted electrical field to be plugged into equation (\ref{PE}) is chosen to be
\begin{align}
    E(t) &= c_{E_{eq}} E_{t,eq} + c_{\delta E}\delta E_t(t) \label{Eadjusted} \\
    \delta E_t(t) &= E_t(t) - E_{t,eq}
\end{align}
where $E_{t,eq}$ is the unadjusted field at low voltage
equilibrium, $\delta E_t$ is the change in unadjusted field from $E_{t,eq}$, and the correction coefficients $c_{E_{eq}}$ and $c_{\delta E}$ are both constants. Only $c_{\delta E}$ need be found via system identification involving measured timeseries data. $c_{E_{eq}}$ can be found analytically from the fact that at equilibrium the sum of all forces (equivalent to $P_g+P_\mu+P_\gamma+P_E$ in this case) must be zero. This relation yields
\begin{equation}
c_{E_{eq}} = -\frac{r_M \ln \left(\frac{4 h_{eq} (h_N-h_{eq})}{r_M^2}\right) \sqrt{-g L \rho  r_M^2-4 \gamma  h_{eq}}}
{2 \sqrt{2\varepsilon} h_{eq} V_l}
\label{eq:cEeq}
\end{equation}
where the only empirical information required is the meniscus position at low voltage equilibrium, $h_{eq}$, which is a single datum as opposed to a model parameter that must be regressed on a data set.
Thus the paraboloidal meniscus approximation enables modeling of larger meniscus deformations by capturing the sharpening of the meniscus as it grows in size via equation (\ref{eq:R}). Equation (\ref{eq:R}) also eliminates the need for $R(t)$ to be defined by measured timeseries data, and basing the electric field adjustment scheme on perturbation from equilibrium via (\ref{Eadjusted}) and (\ref{eq:cEeq}) eliminates a parameter in $E(t)$ requiring system identification, thereby reducing the model's dependence on empirical data to be fully defined.

The final point to be addressed in the physics-driven modeling of flow rate is the difference between the build-up and relaxation states, which revolves around $c_{\delta E}$.
The high voltage during build-up, which pushes the system well beyond the region of attraction of any stable equilibria, makes $\delta E_t$ a significant overapproximation, necessitating $c_{\delta E}<1$.
However, in relaxation, when the system is under a low voltage yielding a stable equilibrium, the changes in unadjusted field are smaller. Reducing them is both unnecessary for convergence to equilibrium, and can cause surface tension to unrealistically overwhelm electric stress. Thus during relaxation, $c_{\delta E}$ is simply set to 1.

This fully defines the flow rate dynamics for build-up and relaxation, leaving the dynamics of meniscus position to be derived. To do this, an equation for the volume of the fluid outside the nozzle may be differentiated. This yields $Q$ as a function of $\dot{h}$, which can then be rearranged.
 
The volume equation is that of a paraboloid truncated at the nozzle outlet:
\begin{equation}
    \mathcal{V}(t)=-\frac{\pi}{2}r_M^2h(t)
\end{equation}
The time derivative of this equation, rearranged, yields
\begin{equation}
    \dot{h}(t) = \frac{-2}{\pi r_M^2}Q(t) \label{hdot_BuildUp}
\end{equation}

Equations (\ref{flowGod}) and (\ref{hdot_BuildUp}) can be used to describe the evolution of flow rate and meniscus position during build-up, but two small formalities must be addressed before it can be incorporated with the hybrid model. First, it is a model in $Q$ and $h$ whereas the hybrid system definition has dynamical states of $\delta Q$ and $\delta h$. Second, this model is not second-order in $Q$, while the hybrid system definition is.

Both of these issues are remedied by choosing
\begin{equation}
g_1(\xvec(t))=g_3(\xvec(t))=
\begin{bmatrix}
1 & 0 & 0 \\ 0 & 1 & 0
\end{bmatrix}
\xvec(t) \label{g1}
\end{equation}

This implies that during build-up, $\delta Q=Q$, $\delta h=h$.
Thus, the final form of $f_1$ and $f_3$ is 
\begin{align}
    &\hspace{-0.1mm}
    f_{1,3}(\mathbf{x},\mathbf{u}) = 
    \begin{bmatrix}
    \frac{-2}{\pi r_M^2}\delta Q \\
    \frac{\pi r_{I0}^2}{\rho L}\left(
    \rho g L
    -
    \frac{8\mu L}{\pi 
    r_I
    ^4}\delta Q
    +
    \frac{4\gamma}{r_M^2}\delta h
    +
    \frac{\varepsilon}{2}
    E^2
    \right)
    \\
    0
    \end{bmatrix} \label{eq:f1}
    \\
    &E = -\sqrt{ \frac{-8h_{eq}\gamma-2gLr_M^2\rho}{\varepsilon r_M^2} }
    -
    \frac{4 c_{\delta E}\delta h\sqrt{u}}{r_M^2\ln\left( \frac{4\delta h(h_N-\delta h)}{r_M^2} \right)}
    \label{eq:Eexpand}
\end{align}
where time argument $(t)$ has been dropped for compactness, $c_{\delta E}<1$ for $f_1$, and $c_{\delta E}=1$ for $f_3$.

Finally, because the actuating voltage pulses in e-jet typically rise while the system is at some low-voltage equilibrium, the initial conditions of the system are within the build-up state. Here, we assume simulations to start from a stationary meniscus. Thus
\begin{equation}
    \xvec_0 = \begin{bmatrix}
    h_{eq} & 0 & 0
    \end{bmatrix}^T
    \,
    ,
    \quad
    p_0=p_1
\end{equation}

\subsubsection{Jetting}
\label{sec:jetting}

The physics-driven models $f_1$ and $f_3$ 
describe 
meniscus deformations up to the critical Taylor cone.
However, as the jet forms at the cone tip, the capacitor geometry, charge flow, and fluid flow become more complicated, and these models cease to represent the system's dynamics. At this point, the system transitions to the data-driven model of jetting dynamics. 
This model is developed by finding differential equations yielding signal shapes similar to those of empirical data and fitting the parameters of those differential equations to measurements.

Classification of the empirical signal profiles is dependent on the actuation method.
DC 
printing 
and 
frequently 
supercritical \ac{DoD} printing 
can
yield highly nonlinear dynamics in that flow rate will fall and the jet will cease automatically due to critical charge ejection despite a constant input.
In subcritical printing, however, the falling edge of the voltage pulse is directly responsible for the fall in flow rate and jet cessation. This causality enables the input/output dynamics to be well-captured by a linear time invariant (LTI) model defined in terms of change from the locally steady high-voltage flow, i.e. from the flow with a fully developed contiguous jet. Thus, this work limits the scope of the jetting model to subcritical printing.

Based on qualitative observations of measured jetting signals, a second order LTI system is chosen for flow rate, with model coefficients to be found by least squares regression. For mathematical consistency with the physics-driven models, the continuous-time model
\begin{equation}
    \delta \ddot{Q}(t) = a_{Q1}\delta \dot{Q}(t) + a_{Q0} \delta Q(t) + b_Q u(t)
\end{equation}
is desired. However, because jetting is a high-speed, small-scale phenomenon, sampling can be relatively coarse and noisy. Regression performance is thus significantly improved when using the discrete time model given by
\begin{equation}
    \delta Q(k+2) = \tilde{a}_{Q1}\delta Q(k+1) + \tilde{a}_{Q0} \delta Q(k) + \tilde{b}_Q u(k)
\end{equation}
were $k$ is the discrete time index.

To get the continuous-time model coefficients $a_{Q1}$, $a_{Q0}$, and $b_Q$ from the regressed discrete time coefficients $\tilde{a}_{Q1}$, $\tilde{a}_{Q0}$, and $\tilde{b}_Q$, the poles and final values of the two systems are matched. First, the two poles of the discrete time system are converted to their continuous-time equivalents via
\begin{equation}
    p_{si}=\frac{\ln(p_{zi})}{T_s} \quad i\in\{1,2\}
\end{equation}
where $p_{si}$ is a continuous-time pole, $p_{zi}$ is a discrete time pole, and $T_s$ is the sampling period. Solving the characteristic equation of the continuous-time system then yields
\begin{equation}
    a_{Q1}=p_{s1}+p_{s2} \qquad a_{Q0}=-p_{s1}p_{s2}
\end{equation}
For the input coefficient, the final value theorem expressions at low voltage for the continuous and discrete time systems can be equated to yield
\begin{align}
    \nonumber
    \lim_{s\rightarrow 0}s
    \frac{\delta Q(s)}{U(s)}
    \frac{(V_l^2-V_h^2)}{s}&=
    \lim_{z\rightarrow 1}(z-1)
    \frac{\delta Q(z)}{U(z)}
    \frac{z(V_l^2-V_h^2)}{z-1} \\
    b_{Q} &= \frac{\tilde{b}_Q a_{Q0}}{\tilde{a}_{Q1}+\tilde{a}_{Q0}-1}
\end{align}
where $\frac{\delta Q()}{U()}$ is the transfer function from input to $\delta Q$.

The above analysis fully defines the flow rate model for jetting, and one may note that it is independent of meniscus position. However, the meniscus position is still important in the jetting model, particularly during jet retraction, because it will trigger the transition to the relaxation location.

Inspection of measured meniscus position signals during jet retraction suggest it may be approximated by exponential decay, i.e. a first order LTI system. However, the retraction does not begin until some time after the falling edge of the voltage pulse. This is accounted for by adding a delay on the input. 
This system will be identified in much the same way as flow-rate, with regression and delay identification being performed on the discrete time system
\begin{equation}
    \delta h (k+1) = \tilde{a}_{h0}\delta h(k) + \tilde{b}_h u(k-d)
\end{equation}
where $d$ is the delay in time steps between the step down in voltage and the beginning of retraction. The continuous-time model parameters arising from pole matching and final value theorem matching, respectively, are
\begin{equation}
    a_{h0}=\frac{\ln(\tilde{a}_{h0})}{T_s} \qquad
    b_h=\frac{\tilde{b}_h a_{h0}}{\tilde{a}_{h0}-1}
\end{equation}
which are then plugged into the continuous-time form of the system, given by
\begin{equation}
    \delta \dot{h}(t)=a_{h0}\delta h(t) + b_h \shift^{-dT_s} u(t)
\end{equation}
where $\shift$ is the shift operator, making $\shift^{-dT_s} u(t)=u(t-dT_s)$.

With this, $f_2$ is fully defined as
\begin{equation}
    \dot{\xvec}(t)=f_2(\xvec(t),u(t))=
    % \\
    \begin{bmatrix}
    a_{h0} & 0      & 0 \\
    0      & 0      & 1 \\
    0      & a_{Q0} & a_{Q1}
    \end{bmatrix}
    \xvec(t)
    +
    \begin{bmatrix}
    b_h\shift^{-dT_s}\\0\\b_Q
    \end{bmatrix}
    u(t)
\end{equation}

$g_2$ will account for the fact that unlike the physics-driven models, the jetting model is not derived in terms of absolute $Q$ and $h$. Instead, it is assumed that in subcritical \ac{DoD} printing, the flow at the falling edge of the voltage pulse is fully developed and locally steady. Thus
\begin{align}
    \yvec(t)=g_2(\xvec(t))&=
    \begin{bmatrix}
    1 & 0 & 0 \\
    0 & 1 & 0
    \end{bmatrix}
    \xvec+
    \begin{bmatrix}
    h_N \\ Q_f
    \end{bmatrix}
\end{align}
where $Q_f$ is the empirically determined flow rate at $t=T_p$.

\subsubsection{Transition \& Reset}
\label{sec:Transition}
This work uses analysis of system stability to concretely define the points in the ejection process at which transitions between the physics- and data-driven models should occur.
As described in section \ref{ejet}, up to a critical magnitude of deformation each meniscus shape is a stable equilibrium of the system for a particular input. Deformation in excess of that critical shape marks the transition from build-up to jetting. Likewise, the complementary condition---the point during jet retraction at which the meniscus shape becomes a stable equilibrium for some applied voltage---is used to mark the transition from jetting to relaxation.

This notion of the transition condition was first proposed by the authors in \cite{Spiegel2017}, which leverages the physics-based work of Yarin, Koombhongse, and Reneker on deriving the range of stable equilibria in terms of meniscus shape \cite{Yarin2001}.
In this work, Yarin models the meniscus as a hyperboloid in a prolate spheroidal coordinate system $(\eta,\xi)$ with the horizontal axis free to move between the nozzle and substrate, equivalent to replacing \ref{C2} and \ref{C3} with \ref{C4} in the surface definition. This results in hyperboloids that are underdetermined given the information available for dynamic simulation, but can fit the true meniscus shape more closely if extra information is available at a given point.

In the prolate spheroidal coordinate system, $\xi \in [0,1]$ denotes a hyperbola, and $\eta \in [ 1,\infty )$ denotes an ellipse (see Figure \ref{fig:hyperbCoords}). Yarin shows that regardless of surface tension coefficient, the maximum stable deformation of the meniscus corresponds to the shape of the critical hyperbola $\xi_{\ast} \approx 0.834$. This extra datum fully determines the surface of the critical Taylor cone.

\begin{figure}
    \centering
    \includegraphics[scale=0.825]{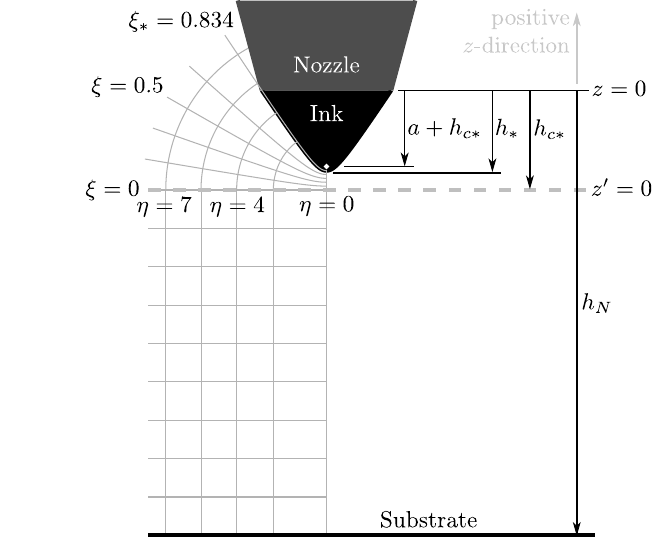}
    \caption{To-scale schematic of a critically deformed meniscus and the prolate spheroidal coordinate system (with horizontal axis at $z'=0$) used to describe the geometry of the meniscus and its local electric field (ellipses) and electric potential contours (hyperbolae). The focus is marked by the white dot at $z'=a$. For $z'<0$ the field and potential contours are assumed ``Cartesian.''
    Note that by the $z$-axis sign convention, $h$-values are negative, and $a$ is positive.}
    \label{fig:hyperbCoords}
\end{figure}

In \cite{Spiegel2017}, $\xi$ is estimated from measured data, and the jetting state is isolated from the build-up state empirically. For simulations independent of measured data, the transition condition must be defined in terms of the model's dynamical states. In other words, for this work a critical meniscus position $h_*$ must be found in terms of $\xi_*$.

The derivation of this relationship starts with the transformation equations between the 2D cross section of the prolate spheroidal coordinate system and a Cartesian coordinate system $(x',z')$ with the same horizontal and vertical axes:
\begin{align}
    \eta &= \frac{1}{2a}\left( \sqrt{x'^2+\left(z'+a\right)^2}+\sqrt{x'^2+\left(z'-a\right)^2} \right) \label{etaraw} \\
    \xi &= \frac{1}{2a}\left( \sqrt{x'^2+\left(z'+a\right)^2}-\sqrt{x'^2+\left(z'-a\right)^2} \right) \label{xiraw} \\
    z' &= a \eta \xi \label{z'raw}
\end{align}
While both coordinate systems are two-dimensional, three equations are needed to relate them. This is because in addition to $\xi$ and $\eta$, the prolate spheroidal coordinate system is also parameterized by the distance $a$ from the horizontal $z'$-axis to the focus of the coordinate system.
The application of \ref{C4} to equation (\ref{xiraw}) and \ref{C1} to equation (\ref{z'raw}) yields
\begin{align}
    \xi_\ast &= \frac{ \sqrt{r_M^2+\left(a-h_{c*}\right)^2}-\sqrt{r_M^2+\left(-a-h_{c*}\right)^2}}{2a} \label{eq:YarinHyp1}
    \\
    h_\ast &= a \xi_\ast + h_{c\ast}  \label{eq:YarinHyp2}
\end{align}
where $h_{c\ast}$ is the displacement of the prolate spheroidal coordinate system from the inertial coordinate system $(x,z)$.

As is, the system is underdetermined, having unknowns $a$, $h_\ast$, and $h_{c\ast}$. This is remedied by reformulating the theory in \cite{Yarin2001} to yield $h_{c\ast}$ as a function of $h_\ast$ and constant system properties via
\begin{align}
h_{c\ast} &= \sqrt{\alpha^2+2\alpha h_\ast}-\alpha \\ 
\textrm{where} \quad
\alpha &= \frac{2\pi \gamma h_N^2 \left( 1-\xi_\ast^2 \right)\left(\ln \frac{1+\xi_\ast}{1-\xi_\ast} \right)^2}{\left(C V_\ast\right)^2},
\end{align}
$C\approx \SI{1.05e-5}{\kilogram\tothe{\frac{1}{2}}\meter\tothe{\frac{1}{2}} \per \second \per \volt}$ is the conversion from volts to statvolts in SI units, and $V_\ast$ is the maximum non-jetting applied voltage. $V_*$ is empirically determined, providing the extra datum required to fully determine the critical hyperboloid.

Now solving the system (\ref{eq:YarinHyp1}), (\ref{eq:YarinHyp2}) will yield four sets of solutions for $h_\ast$, and $a$.
Three of these solutions are spurious: two for returning meniscus positions inside the nozzle and another for yielding $\eta$ outside of its domain when plugged into equation (\ref{etaraw}).
This leaves a single solution for the maximum stable meniscus length, thereby defining the transitions into and out of jetting via the transition function
\begin{align}
p_+&=\phi(p_-,\yvec) = 
\begin{cases} 
p_1 & h \geq h_\ast \land p_-=p_1 \\
p_2 & h < h_\ast \\
p_3 & h \geq h_\ast \land p_-\neq p_1 
\end{cases} \\
\textrm{where } \,
h_\ast &=  \alpha - \frac{\beta}{1-\xi_\ast^2} - 
\sqrt{ \frac{2\alpha \left( 
 \alpha \left( 1-\xi_\ast^2 \right)-\beta \right)
}
{1-\xi_\ast^2} } 
\\ \textrm{and } \,
\beta &=\sqrt{\left( 1-\xi_\ast^2 \right)\left(  \alpha^2\left( 1-\xi_\ast^2 \right) - r_M^2\xi_\ast^2 \right) } 
\end{align}
The ($-$) subscript indicates the state value preceding transition and the ($+$) subscript indicates the state value after transition. The transition is instantaneous.

While the outputs $\yvec$ representing the absolute flow rate and meniscus position should be continuous over transitions between locations, the dynamical states $\xvec$ must undergo a reset at each transition because the build-up and jetting dynamics use different set points from which $\delta h$ and $\delta Q$ are defined.

The reset's objective is thus to ensure $\yvec$ is continuous despite discontinuities in $\xvec$. For $\delta h$ and $\delta Q$, this amounts to subtracting and adding the output offsets of the jetting model when entering and exiting the jetting state, respectively. For $\delta \dot{Q}$, the flow rate dynamics in equations (\ref{eq:f1}-\ref{eq:Eexpand}) can be used to set the initial $\delta \dot{Q}$ when entering jetting, and the state can be set to zero when exiting jetting, as it is unused by the build-up dynamics.

Thus
% --------------------------
% ONE COLUMN
% --------------------------
\begin{equation}
    \xvec_+ = \psi(p_-,p_+,\xvec_-)=
    % =
    \begin{cases}
    \xvec_- +
    \begin{bmatrix}
    -h_N \\ -Q_f \\
    \begin{bmatrix}
    0&1&0
    \end{bmatrix}
    f_1(\xvec_-,\uvec_-)
    \end{bmatrix}
    & 
    \begin{bmatrix}
    p_- \\ p_+
    \end{bmatrix}
    =
    \begin{bmatrix}
    p_1 \\ p_2
    \end{bmatrix}
    % (p_-=p_1) \land (p_+=p_2)
    \\
    \xvec_- +
    \begin{bmatrix}
    h_N \\ Q_f \\ 
    \begin{bmatrix}
    0&-1&0
    \end{bmatrix}
    f_2(\xvec_-,\uvec_-)
    \end{bmatrix}
    & 
    \begin{bmatrix}
    p_- \\ p_+
    \end{bmatrix}
    =
    \begin{bmatrix}
    p_2 \\ p_3
    \end{bmatrix}
    % (p_-=p_2) \land (p_+=p_3)
    \end{cases}
\end{equation}
% --------------------------
% TWO COLUMN
% --------------------------
% \begin{multline}
%     \xvec_+ = \psi(p_-,p_+,\xvec_-)=
%     \\
%     \begin{cases}
%     \xvec_- +
%     \begin{bmatrix}
%     -h_N \\ -Q_f \\
%     \begin{bmatrix}
%     0&1&0
%     \end{bmatrix}
%     f_1(\xvec_-,\uvec_-)
%     \end{bmatrix}
%     & 
%     \begin{bmatrix}
%     p_- \\ p_+
%     \end{bmatrix}
%     =
%     \begin{bmatrix}
%     p_1 \\ p_2
%     \end{bmatrix}
%     \\
%     \xvec_- +
%     \begin{bmatrix}
%     h_N \\ Q_f \\ 
%     \begin{bmatrix}
%     0&-1&0
%     \end{bmatrix}
%     f_2(\xvec_-,\uvec_-)
%     \end{bmatrix}
%     & 
%     \begin{bmatrix}
%     p_- \\ p_+
%     \end{bmatrix}
%     =
%     \begin{bmatrix}
%     p_2 \\ p_3
%     \end{bmatrix}
%     \end{cases}
% \end{multline}

Now, with the dynamics of all locations and the process for transitioning between them defined, a complete end-to-end model of the e-jet cycle under subcritical actuation is achieved. The following sections discuss the experimental procedures and data processing involved in model verification.

\subsection{Experimental Methods}
\label{datProc}

\subsubsection{Data Collection and Processing}
\label{expMeth}
There are two main components of the process used for collecting data for system identification and validation of the hybrid e-jet model proposed above: the physical system used to generate high speed videos and image processing code that is used to extract signals of interest from the raw high speed videos.

The experimental setup consists of an e-jet printer (custom built at the University of Michigan), and a high speed camera (Vision Research, Phantom V9.0) which are automated and synchronized via the drivers of an X-Y-Z nanopositioning stage (Aerotech, “Planar\textsubscript{\emph{DL}}”) and software written in the Aerotech A3200 Motion Composer Integrated Development Environment.
The high speed camera is fixed with a 20x microscope lens assembly yielding \SI{0.65}{\micro\meter\per\pixel} resolution.
Norland Optical Adhesive 81, a UV curable monomer-photocatalyst mixture, is used as the ink for all experiments in this work.
A silicon wafer serves as the substrate. The microcapillary nozzle is supplied by Word Precision Instruments (product TIP30TW1), and sputter coated in house with gold-palladium alloy for conductivity.

Important physical parameters of the printer setup and video capture process are tabulated in table \ref{table:ExpParams}. Twenty single-droplet ejections are recorded for each element of a set $\Omega$ of ten $(T_p,V_h)$ pairs, listed in table \ref{table:Omega}.
All trials begin from the low voltage equilibrium meniscus position over a clean region of substrate with no prior fluid depositions.
For each pair, 10 of the 20 recordings are designated as training data, and the remaining as validation data.
These pairs are chosen to approximately cover the range of voltages within the subcritical regime for the given printer configuration and ink, as well as to offer some variation in $T_p$ where possible (the subcritical regime may be ``narrow'' at particularly low or high values of $V_h$). 

The experimental setup and procedure (other than the specific $(T_p,V_h)$ pairs tested) is identical to those of \cite{Spiegel2017}, which may be referenced for greater detail.
Additionally, Figure \ref{fig:ejetintro} contains a photograph of the experimental system.

\begin{table}
\centering
\caption{Experimental Setup and Process Parameters}
\label{table:ExpParams}
\setstretch{1.2}
\begin{tabular}{|l|l|l|}
\hline
Parameter                                                                     & Symbol & Value                                  \\ \hhline{|=|=|=|}
Nozzle outlet inner radius                                                             & $r_{I0}$        & $\SI{15}{\micro\meter}$                 \\ \hline
\begin{tabular}[l]{@{}l@{}} Nozzle outlet outer radius\\ (meniscus radius)\end{tabular}
& $r_M$           & $\SI{21.8}{\micro\meter}$               \\ \hline
Fluid column height                                                                    & $L$             & $\SI{5}{\centi\meter}$                  \\ \hline
Substrate position                                                                     & $h_N$           & $\SI{150}{\micro\meter}$                 \\ \hline
Ink Density                                                                            & $\rho$          & $\SI{1200}{\kilo\gram\per\meter\cubed}$ \cite{Pannier2015} \\ \hline
Ink Dynamic Viscosity                                                                  & $\mu$           & $\SI{0.3}{\pascal\second}$ \cite{noa81}               \\ \hline
Surface Tension Coefficient
& $\gamma$        & $\SI{0.039}{\newton\per\meter}$ \cite{pannier2017}       \\ \hline
Low Voltage                                                                            & $V_l$           & $\SI{525}{\volt}$                       \\ \hline
Sample Period                                                                          & $T_s$           & $\SI{50}{\micro\second}$
\\ \hline
\end{tabular}
\end{table}

\begin{table}
\centering
\caption{The set $\Omega$ of $\omega=(T_p,V_h)$ pairs}
\label{table:Omega}
\setstretch{1.2}
% \vspace{-0.5mm}
% 
% \begin{tabular}{|c|c|}
% \hline
% $V_h$ {[}V{]} & $T_p$ {[}$\SI[detect-weight=true]{}{\milli\second}${]}
% \\ 
% \hhline{|=|=|}
% 1100       & 2.0     \\ \hline
% 1150       & 2.0     \\ \hline
% 1200       & 2.0     \\ \hline
% 1250       & 2.0     \\ \hline
% 1300       & 1.5     \\ \hline
% \end{tabular}
% \qquad
% \begin{tabular}{|c|c|}
% \hline
% $V_h$ {[}V{]} & $T_p$ {[}$\SI[detect-weight=true]{}{\milli\second}${]}
% \\ 
% \hhline{|=|=|}
% 1300       & 1.8     \\ \hline
% 1300       & 2.0     \\ \hline
% 1300       & 2.3     \\ \hline
% 1350       & 2.0     \\ \hline
% 1370       & 2.0     \\ \hline
% \end{tabular}
\begin{tabular}{|c||c|c|c|c|c|c|c|c|c|c|}
\hline
    $V_h$ {[}V{]} &  1100 & 1150 & 1200 & 1250 & 1300 & 1300 & 1300 & 1300 & 1350 & 1370
    \\ \hline
     $T_p$ {[}$\SI[detect-weight=true]{}{\milli\second}${]} 
     & 2.0 & 2.0 & 2.0 & 2.0 & 1.5 & 1.8 & 2.0 & 2.3 & 2.0 & 2.0 \\\hline 
\end{tabular}
\end{table}

Once the videos are captured, $h$ and $Q$ signals are extracted for each trial via an image processing protocol consisting of substrate position identification, ink-nozzle interface identification, edge finding, region area/volume computation, and numerical differentiation. These operations are also described in detail in \cite{Spiegel2017}. The only procedure change from \cite{Spiegel2017} is that where the previous work discarded trials for which the algorithm failed to identify the ink-nozzle interface (a task made difficult by the dullness of the feature corners), this work retains all recorded video and uses manual interface identification for trials where the algorithm failed.

\subsubsection{System Identification}
\label{sysID}
For the physics-driven locations, there are two major parameters that must be found via system identification: $c_{r_I}$ and $c_{\delta E}$.
The system identification is done via the minimization of flow rate error in simulations performed over a mesh of parameter test values.

Each parameter is given a set of 100 linearly spaced test values on the ranges $c_{r_I}\in [1,35]$ and $c_{\delta E}\in [0.3,0.6]$, yielding 10,000 models in total.
For $c_{r_I}$, this range allows $r_I$ to vary from the nozzle outlet radius to approximately the radius of the main body of the micropipette.
For $c_{\delta E}$, this range is found by trial and error with boundaries chosen such that no $r_I$ in the aforementioned range yields simulation timeseries sufficiently approximating the empirical signals under holistic qualitative assessment.

Each of the 10,000 models is given a scalar error value, $Err$, by the the metric
% --------------------
% ONE COLUMN
% --------------------
\begin{equation}
Err(c_{\delta E},c_{r_I}) = 
\mean_{\omega\in\Omega,i\in I} 
\left(
\RMS_{k\in K_{p_1}(\omega,i)}
\left(
Q_{\omega,i}(kT_s) - \hat{Q}_{\omega,c_{\delta E},c_{r_I}}(kT_s)
\right)
\right)
\end{equation}
% --------------------
% TWO COLUMN
% --------------------
% \begin{multline}
% Err(c_{\delta E},c_{r_I}) = \\
% \mean_{\omega\in\Omega,i\in I} 
% \left(
% \RMS_{k\in K_{p_1}(\omega,i)}
% \left(
% Q_{\omega,i}(kT_s) - \hat{Q}_{\omega,c_{\delta E},c_{r_I}}(kT_s)
% \right)
% \right)
% \end{multline}
where $i\in I=\{1,2,\cdots,10\}$ is the index of the training data recording for a given $\omega$, $K_{p_1}(\omega,i)$ is the set of time steps in the build-up state for the $(\omega,i)$ recording, $Q_{\omega,i}$ is the measured flow rate trajectory of the $(\omega,i)$ recording, and $\hat{Q}_{\omega,c_{\delta E},c_{r_I}}$ is the simulated flow rate trajectory generated with the given test values of $c_{\delta E}$ and $c_{r_I}$, and the input dictated by $\omega$. These simulations are performed using a fourth order Runge-Kutta method with step-size $T_s=\SI{100}{\nano\second}$.

In addition to $c_{r_I}$ and $c_{\delta E}$, the low-voltage equilibrium position $h_{eq}$, the maximum non-jetting voltage $V_*$, and the input delay $d$ for the jetting dynamics of meniscus position must be determined empirically. $h_{eq}$ is found by averaging the empirical $h(0)$ over all training trials. $V_*$ is set to \SI{1000}{\volt} based on the results in \cite{pannier2017}, which uses the same ink and printer setup as the current work. For each $\omega\in\Omega$, $d$ is the median number of time steps during jetting for which $V(t)=V_l$ and $|\delta h(t)|\leq\texttt{tol}$, where \texttt{tol} is a tolerance close to 0.

For the jetting location, this work assumes that \begin{enumerate}[label=(A2.\arabic*),leftmargin=*]
\item
\label{A1}
each $(T_p,V_h)$ pair requires a different set of LTI model parameters, and
\item
\label{A2}
these parameter sets must be derived independently from one another.
\end{enumerate}
The LTI models for each $(T_p,V_h)$ pair are thus derived from only the 10 training recordings corresponding to that pair.
The discrete-time LTI model parameters for $Q$, $\tilde{a}_{Q0}$, $\tilde{a}_{Q1}$, and $\tilde{b}_Q$, are to be fit by basic least squares regression. For $h$, however, a standard least squares regression runs the risk of returning an LTI model whose steady state value under low voltage is near $h_*$ but does not cross $h_*$. Such a model would prevent a transition from jetting to relaxation from ever happening. Thus the system identification for meniscus position will constrain the least squares optimization such that the final value theorem applied to the model under low voltage yields
\begin{equation}
    h(t\rightarrow \infty) \geq 0.95 h_*
\end{equation}
Because the optimization is carried out over the decision variables $\tilde{a}_{h0}$ and $\tilde{b}_h$, this inequality constraint is realized as
\begin{equation}
    \begin{bmatrix}
    h_N - 0.95 h_* & V_h^2 - V_l^2
    \end{bmatrix}
    \begin{bmatrix}
    \tilde{a}_{h0} \\ \tilde{b}_h
    \end{bmatrix}
    \leq
    h_N - 0.95h_*
\end{equation}
which can be implemented via MATLAB's \texttt{lsqlin} function.

\subsection{Model Validation}
\label{modelVal}
This section first presents the results of the system identification processes described above, and uses these results to assess this work's approach to integrating data-driven components into e-jet modeling. Then, an error breakdown of the fully defined model against the validation data is presented, and analysis is given.

\subsubsection{System Identification Results}
\begin{figure}
    \centering
    \includegraphics[scale=0.9]{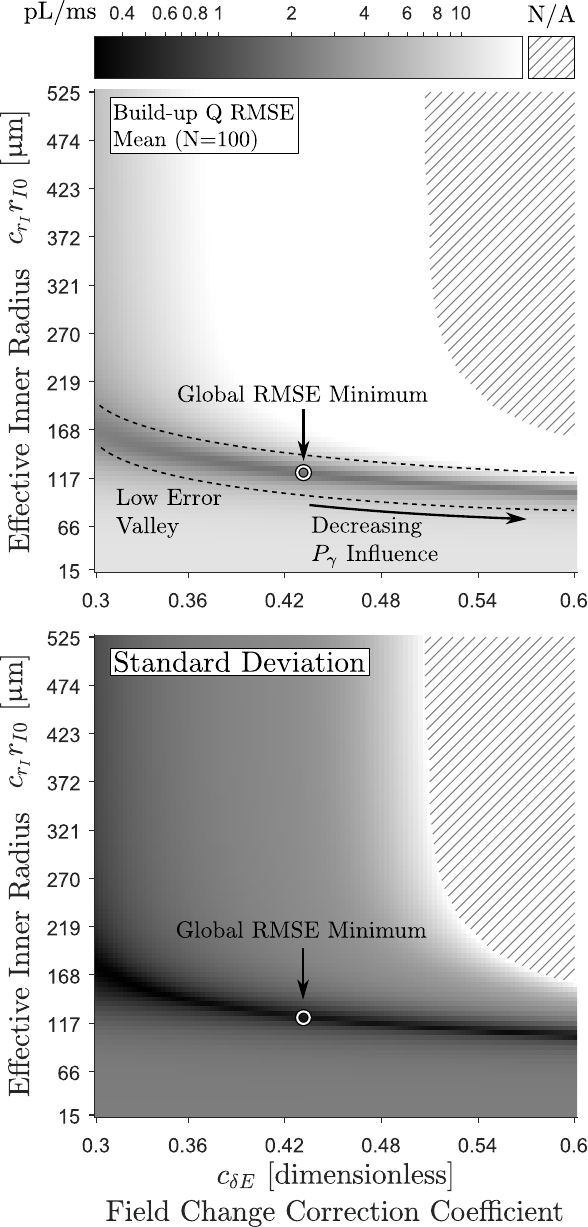}
    \caption{Training error, $Err$, and standard deviation of $Q$ in the build-up state with respect to $c_{\delta E}$ and $c_{r_I}$. The color bar across the top of the figure applies to both plots, and gives the logarithmic relationship between pixel brightness and the magnitude of the mean and standard deviation in picoLiters per millisecond (\SI{}{\pico\liter\per\milli\second}). The hatched regions represent $(c_{\delta E},c_{r_I})$ combinations corresponding to \texttt{Inf} error values.
    The circled point in both images represents the global minimum $Err$ found in this analysis.
    }
    \label{fig:BuildUpErr}
\end{figure}
\begin{figure}
    \centering
    \includegraphics[scale=0.57]{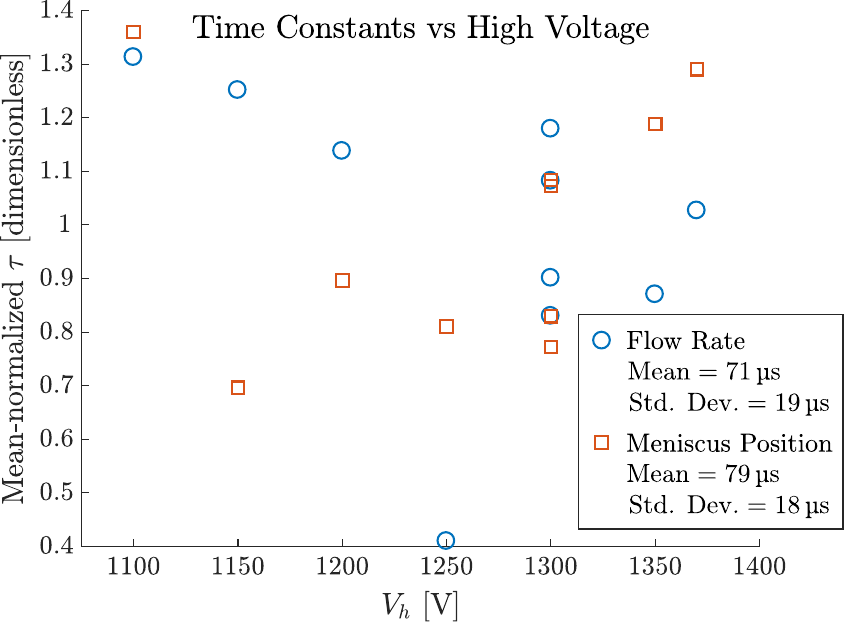}
	\includegraphics[scale=0.57]{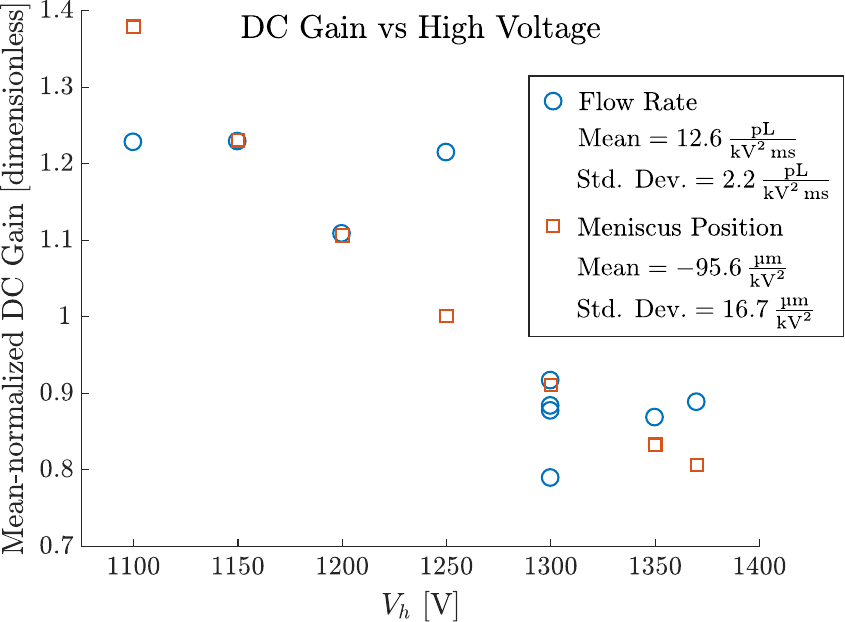}    
	\caption{
    Time constants ($\tau$) and DC Gains ($G$) of $Q$ and $h$ jetting models versus $V_h$. 
    For comparability between $Q$ and $h$, $\tau$ and $G$ are normalized by their mean value.
    No correlation is apparent between $\tau$ and $V_h$. 
    However, DC Gain and $V_h$ are clearly correlated.
    }
    \label{fig:DCGainAndTau}
\end{figure}

Figure \ref{fig:BuildUpErr} gives the error and standard deviation of the build-up $Q$ models over the 100-by-100 array of $c_{r_I}$ and $c_{\delta E}$ choices. The parameter values yielding minimum error are $c_{r_I}=7.87$ and $c_{\delta E}=0.44$, which are applied to all $(T_p,V_h)$ pairs for the given ink and printer configuration. This value of $c_{r_I}$ yields an effective inner shaft radius of $r_I=\SI{118}{\micro\meter}$. This is nearly an order of magnitude larger than the nozzle outlet inner radius, $r_{I0}=\SI{15}{\micro\meter}$. However, it is still significantly closer to $r_{I0}$ than it is to the capillary radius preceding the taper ($\sim400$-$\SI{500}{\micro\meter}$), which makes up $\sim 90\%$ of the fluid column. This verifies the necessity of using $c_{r_I}$ to find an effective radius for the fluid column in the physics-based modeling, as the true nozzle outlet radius and mean column radius would both fall outside of the low error valley seen in Figure \ref{fig:BuildUpErr}, regardless of choice of $c_{\delta E}$.

For the choice of $c_{\delta E}$ itself, the low error valley apparent in Figure \ref{fig:BuildUpErr} may tempt one to believe that the global minimum within this valley is not substantially different from other points along some ``minimum trajectory'' across the mesh, and that choosing a point in the valley yielding the least aggressive model adjustment (i.e. $c_{\delta E}$ closest to $1$) may be a better modeling choice. However, in fact the model behavior changes as $c_{\delta E}$ increases and $c_{r_I}$ correspondingly decreases. Specifically, the pressure due to surface tension ($P_\gamma$) loses influence over the flow rate dynamics compared to the pressures due to electric field ($P_E$) and shear forces ($P_\mu$), as can be seen by inspection of equation (\ref{eq:f1}). This degrades the model's ability to capture the system's transient response to the step increase in applied voltage, which is governed largely by the balance of $P_E$ and $P_\gamma$ due to $P_\mu$ being zero at $t=0$. Thus, while the error increases incurred by choosing $c_{\delta E}$ closer to 1 may be relatively small, this choice has nontrivial ramifications for the physical meaning of the model, which justifies the selection of $c_{\delta E}$ based on global minimum RMSE.

To assess the system identification approach taken for the jetting location, specifically the necessity of assumptions \ref{A1} and \ref{A2}, Figure \ref{fig:DCGainAndTau} illustrates how the transient and steady state behavior of the LTI models vary with applied voltage. This is done by plotting the time constants, $\tau$, and DC gains, $G$, for each $\omega$ against $V_h$.
$\tau$ and $G$ come from the roots of a system's characteristic polynomial (which determine the rate of exponential decay of transient responses in stable LTI systems) and the final value theorem applied to a system with a unit step input, respectively, and are given by
\begin{align}
    \tau_h&=\frac{-1}{a_{h0}} & \tau_Q &=\frac{-2}{a_{Q1}}  \\
    G_h &=  \frac{-b_h}{a_{h0}}  & G_Q &= \frac{-b_Q}{a_{Q0}}
\end{align}
Pearson correlation coefficients, $r$, and the corresponding $\mathfrak{p}$-values (stylized with a fraktur font to distinguish $\mathfrak{p}$ from the hybrid system's discrete state, $p(t)$) of the $\tau$ and $G$ magnitudes versus $V_h^2$ are given in table \ref{table:correlation}.
$r\in [-1,1]$ quantifies the degree to which two variables are linearly correlated (and the direction in which they are correlated).
$\mathfrak{p}$-values indicate the probability that the measured $r$ could arise from randomness given a true relationship of zero correlation. Thus, correlation coefficients associated with high $\mathfrak{p}$-values are marked irrelevant.
$r$ and $\mathfrak{p}$ are computed with respect to $V_h^2$ because the input to the dynamic system is based on squared voltage.

\begin{table}
\centering
\caption{$\mathfrak{p}$-values and correlation coefficients for LTI Metrics
vs. $V_h^2$
}
\label{table:correlation}
\vspace{-0.5mm}
\begin{tabular}{|c|c|c|}
\hline
  LTI Metric          & $\mathfrak{p}$                 & $r$      \\
\hhline{|=|=|=|}
$|\tau_h|$    & $0.67$            & Irrelevant \\
$|\tau_Q|$    & $0.23$             & Irrelevant \\
$|G_h|$& \SI{7e-8}{} & $-0.99$      \\
$|G_Q|$& $0.002$             & $-0.85$    \\\hline
\end{tabular}
\end{table}

The fact that the $\tau$ and $G$ values are far from constant across all $V_h^2$ clearly demonstrates the validity of \ref{A1}. \ref{A2} is more interesting. Time constant shows no correlation with $V_h^2$. 
This means that there is unlikely to be a meaningful linear relationship between $\tau$ and $\omega$, and that some independent system identification may be necessary for each $\omega$.
However, DC Gain exhibits strong correlation. This is a desirable result because it suggests the possibility of reducing the strictness of \ref{A2} by incorporating the dependence of DC Gain on $\omega$ into the framework of jetting-state dynamics, which may reduce the burden of system identification in future work.

\subsubsection{Total Model Validation}

To represent the results of the complete model, Figure \ref{fig:finalErr} gives a normalized root mean squared error (NRMSE) breakdown of the model for each 
location
with respect to the measured data across all $(T_p,V_h)$ pairs from the rising edge of the voltage pulse, $t=0$, to $t=\SI{14.65}{\milli\second}$.
For each
location
$p\in \{p_1,p_2,p_3\}$ and each output $y \in \{h,Q\}$ the NRMSE is calculated by
\begin{align}
\nonumber
    NRMSE(p,y) = \frac
    {
    \RMS\limits_{\omega\in\Omega,i\in I,k\in K_{p}(\omega,i)} 
    y_{\omega,i}(kT_s) - \hat{y}_\omega(kT_s)  
    }
    {\range\limits_{\omega\in\Omega,i\in I,k\in K_p(\omega,i,p)}
    y_{\omega,i}(kT_s)
    }
\end{align}
where $\hat{y}_\omega$ is the simulated output trajectory generated with the input dictated by $\omega$, and $K_p(\omega,i)$ is the set of time-steps for which the physical system trial $(\omega,i)$ is in state $p$.

\begin{figure}
    \centering
    \includegraphics[scale=1]{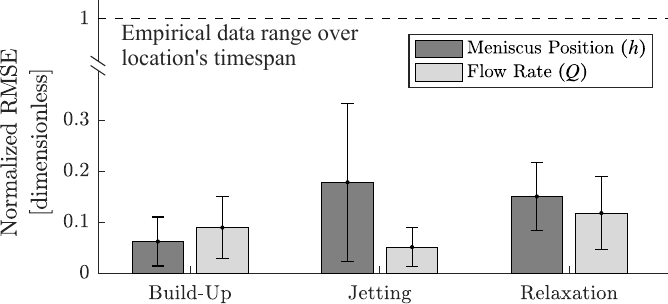}
    \caption{Complete model error given as the RMSE normalized by the range of measured data within each 
    location.
    Each bar is computed from 100 measured timeseries  averaging
    23 points each for build-up,
    34 points for Jetting, and
    237 points for relaxation. Error bars represent plus/minus one standard deviation.}
    \label{fig:finalErr}
\end{figure}

\begin{table}
\centering
\caption{
Transition Timing Error Mean ($e_\mu$) \& Standard Deviation ($e_\sigma$)
}
\label{table:TimingErr}
\setstretch{1.2}
\vspace{-0.5mm}
\begin{tabular}{|c|c|c|}
\hline
Transition            & $e_\mu \, [\SI{}{\milli\second}]$ & $e_\sigma \, [\SI{}{\milli\second}]$ \\ \hhline{|=|=|=|}
Build-up to Jetting   & $-0.26$                        & $0.12$                            \\ \hline
Jetting to Relaxation & $-0.04$                        & $0.05$                            \\ \hline
\end{tabular}
\end{table}

\begin{figure}
    \centering
    \includegraphics[scale=0.9,trim={0.4cm 1.75cm 1.15cm 5.75cm},clip]{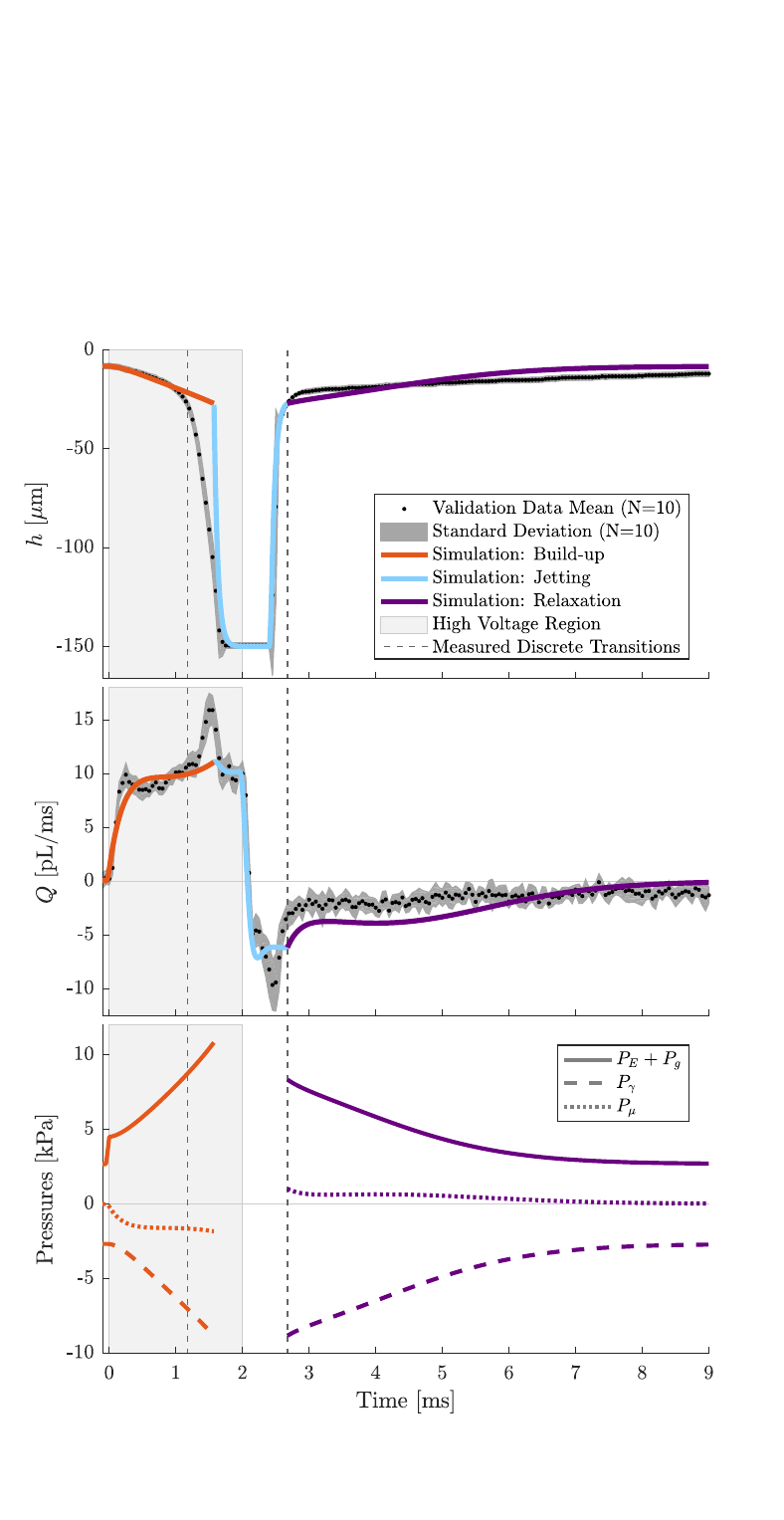}
    \caption{Plot of the simulation and measured data for $h$, $Q$, and the individual pressure components contributing to $Q$ for $T_p=\SI{2}{\milli\second}$ and $V_h=\SI{1150}{\volt}$. $P_E$ and $P_g$ are combined because $P_g$ is constant and they are both always positive, meaning that their sum should be balanced against $P_\gamma$ with the assistance of $P_\mu$.}
    \label{fig:timeseries01}
\end{figure}

Error of state transition times is reported in table \ref{table:TimingErr}. These errors are given as a simple mean and standard deviation, rather than an RMSE-style metric. This is done in order to preserve the sign of the error, and thus indicate whether the simulation transitions early (positive error) or late (negative error).

\begin{sloppypar}
Finally, a representative timeseries plot showing the simulated output trajectories 
against 
the measured trajectories is given in Figure \ref{fig:timeseries01}.
\end{sloppypar}

A qualitative assessment of Figure \ref{fig:timeseries01} shows that despite some notably erroneous features, the overall model reproduces the empirical timeseries satisfactorily. This holistic satisfaction is supported quantitatively by the NRMSE values presented in Figure \ref{fig:finalErr}, which illustrates that the 
range-normalized timeseries errors average to only 
11\% 
for both outputs across all 
locations
(where the average is evenly weighted with respect to the 
locations,
not individual points in time).
However, the standard deviations on these NRMSEs are relatively large, such that the maximum sum of standard deviation and NRMSE reaches 33\%.
While the average sum of the NRMSE and standard deviation is still just 18\%,
indicating that overall the error is acceptable despite some nontrivial spread,
the standard deviations still warrant discussion.

A key reason the standard deviation is large is that the error within a single timeseries is not evenly distributed. Instead, there are areas of small and large errors.
Take, for example, 
the meniscus position during jetting,
which has the largest NRMSE and standard deviation.
Figure \ref{fig:timeseries01} shows that towards the end of the jetting state, there is very little error between the measured and simulated $h$ trajectories. However at the beginning of the measured trajectory's jetting state (upon which the NRMSE computation is based), there is a brief period of large error while the simulated system is still in the build-up state.
As can be seen from table \ref{table:TimingErr}, the simulations' delayed transition into jetting (and thus the large error in $h$ during the period where the 
location
of the simulation and experiment are mismatched) is consistent, and is likely the cause of the relatively large standard deviations. In fact, if the error analysis is restricted to the time span in which the measured and simulated trajectories are in the same
location,
the sum of the standard deviation and NRMSE for $h$ in jetting decreases by 23\%.

While
location
mismatch is not the only factor contributing to the observed errors, it is worth focusing on in particular not only because of its numerical impact illustrated above, but also because it can be ascribed to a specific modeling assumption: that of the meniscus's shape.
While the paraboloidal shape assumption is clearly more accurate than the spherical cap assumption, it still increasingly overestimates the volume of the meniscus as the meniscus elongates, with the critical meniscus volume being overestimated by 16\% on average (compare to 75\% for the spherical cap). This means the model requires a greater liquid volume change per unit change in $h$ than the true system. Consequently, for the same $Q$ the simulated $h$ in this work grows more slowly than the physical system the nearer $h$ draws to $h_*$.
On top of this, any retardation of $h$ due to geometric model mismatch is amplified through the positive coupling of $h$ and $Q$. In other words, the geometric mismatch also causes a reduction in $\dot{Q}$, which further slows meniscus elongation beyond the direct effect of geometric mismatch on $\dot{h}$. This combination of factors ultimately results in delayed transitions to jetting.

Note that a hyperboloidal shape assumption is not intrinsically better than the paraboloidal model in this regard. Indeed, while combining constraints \ref{C1}, \ref{C2}, and \ref{C4} can be used to produce fully determined hyperboloids for use in the dynamic equations, these hyperboloids are virtually identical to the mathematically simpler paraboloids used in this work. This is illustrated in Figure \ref{fig:hypVparab}.
In other words, the ability of hyperboloids to capture the sharpness of the critical meniscus is contingent on the relaxation of \ref{C2}, and the alternative constraint used to derive the transition condition ($\xi=\xi_*$) is only valid at the critical meniscus.

\begin{figure}
    \centering
    \includegraphics[scale=0.7]{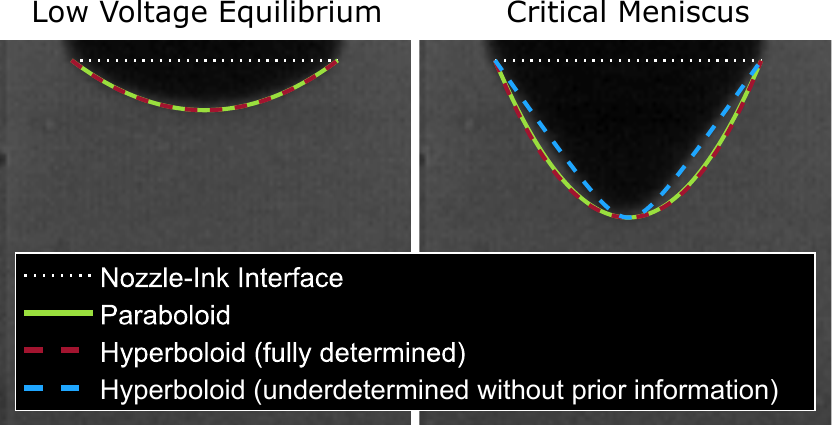}
    \caption{Paraboloidal and hyperboloidal electrode geometry approximations overlaid on nozzle/meniscus photographs at a low voltage equilibrium (left), and at the critical Taylor cone preceding jetting 
    (right). The ``underdetermined'' hyperboloid is given only at the critical meniscus because the necessary 
    prior
    information (in this case, the value of $\xi$) is only known for the critical meniscus (see section \ref{sec:Transition}).
    }
    \label{fig:hypVparab}
\end{figure}

Thus, while there is room for future studies to polish the model minutiae, the above results demonstrate the efficacy of the given hybrid system framework, along with this work's contributions to each of the component models, in capturing the end-to-end dynamics of an e-jet printing process.

\subsection{Physics-Focused Hybrid E-jet Modeling Conclusion}
\label{conc}
In summary, this section delivers the first end-to-end ODE-driven model of an e-jet printing process.
This is achieved by the combination of three major types of contribution.
First, a hybrid systems framework is presented for combining multiple partial process models into an end-to-end model.
Second, the scope of prior partial process models is expanded by means such as increasing the sophistication of geometric modeling, leveraging equilibrium information to structure dynamical equations, and analysis of the necessary areas in which to incorporate data-driven modeling elements. Lastly, transitions between the hybrid system's component models are determined via consideration of the partial process's stability.

\section{Piecewise Affine Modeling}
\label{sec:controlModeling}

The physics-focused model of Section \ref{sec:physicsModeling}
does not explicitly model the deposited fluid volume, instead only modeling the volumetric flow rate of fluid out of the nozzle.
This is in large part because past works have considered droplet volume ill-defined until the jet breaks, at which point the droplet volume was considered constant.
Thus,
there remains a gap 
in the
satisfaction of the requirements that a model be both control-oriented (i.e. ODE-based) and explicitly output deposited droplet volume.

The main contribution of the present section is
a hybrid system model framework for \ejet{} that bridges this gap.
Specifically, this section 
\begin{itemize}[leftmargin=*]%[label=(C2.\arabic*)]
\item
defines the droplet volume as a dynamical state variable that may evolve over time,
\item
presents a new division of the ejection process into partial processes to facilitate droplet volume modeling,
\item
proposes and experimentally validates a mapping between nozzle flow rate and deposited droplet volume 
enabled by the preceding bullets, and
\item
presents a new computer vision technique for taking consistent droplet volume measurements from high speed microscope video.
\end{itemize}

The remainder of the section is organized as follows. 
Section \ref{sec:model} presents the  model 
(i.e. the first three bullets).
Section \ref{sec:methods} presents the experimental methods for measurement, system identification, and model verification, including 
the final bullet.
Section \ref{sec:results} presents and discusses the 
model verification results.
Finally, concluding remarks are given in section \ref{sec:conclusion}.

\subsection{Dynamical Deposited Volume Model}
\label{sec:model}

\subsubsection{Droplet Volume Definition}
\label{sec:dropVolDef}
The plane of the nozzle outlet and the solid substrate surface provide obvious boundaries for a 
\ac{CV}
through which the total volume of fluid outside the nozzle,
$\mathcal{V}(t)$,
and 
the total flow rate through the nozzle outlet,
$Q(t)$,
may be analyzed and modeled.
Similarly, to analyze deposited droplet volume,
$\mathcal{V}_d(t)$,
as a dynamically evolving variable a droplet 
\cv{}
must be defined.
This \cv{} cannot be the same as the total fluid \cv{} because only a fraction of the cumulative flow out of the nozzle up until the jet breaks is deposited on the substrate.
The remainder of the fluid is retracted back into the nozzle under the power of surface tension after the jet breaks and the nozzle-connected fluid body and substrate-connected fluid body become disjoint.

This work introduces a \cv{} with an upper boundary at the $z$-coordinate $h_b$,  where the jet ultimately pinches closed and breaks into two disjoint fluid bodies,
as shown in Figure \ref{fig:CV}.
This \cv{} allows measurements of droplet volume to be made as time series data while remaining consistent with prior notions of droplet volume in that after the jet breaks the volume of fluid in the droplet \cv{} remains constant (assuming negligible evaporation).

\begin{figure}
    \centering
    \hspace*{0.4in}
    \includegraphics[scale=0.6]{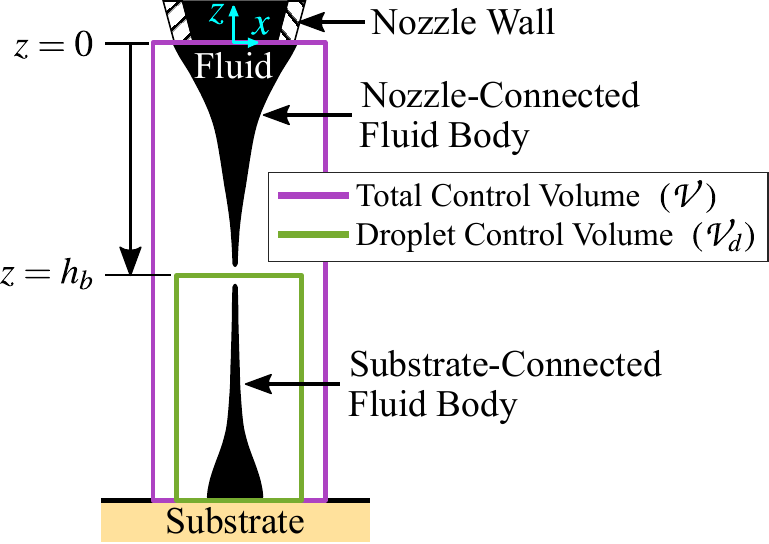}
    \caption{Schematic of the two \cv{}s used in this work superimposed over an illustration of a jet immediately after breaking.}
    \label{fig:CV}
\end{figure}

To avoid dramatic increases in complexity, this work's model does not explicitly use $h_b$.
The jet break point
is only used to facilitate defining droplet volume as a time series signal and for extracting time series measurements of droplet volume from microscope videos.
Theoretical derivation of the jet break point's position is thus beyond the scope of this work, and it is estimated independently for each material ejection
as described in Section \ref{sec:ID}.

\subsubsection{Hybrid System Architecture}
Ultimately, this work's model is given as the cascading of two discrete-time state-space systems with state transition formulas of the form $f:X\times U\times T\rightarrow X$ where $X$ is the state vector space, $U$ is the input vector space, and $T$ is the time vector, all over the field $\real$. These two systems are the input-to-nozzle-flow-rate model
\begin{equation}
    \begin{bmatrix}
    Q(t+T_s)
    \\
    Q(t+2T_s)
    \end{bmatrix}
    =f_Q\left(
    \begin{bmatrix}
    Q(t)
    \\
    Q(t+T_s)
    \end{bmatrix}
    ,\,
    V(t)^2-V_l^2
    ,\,\,
    t
    \right),
\end{equation}
and the nozzle-flow-rate-to-droplet-volume model
\begin{equation}
    \begin{bmatrix}
    \mathcal{V}(t+T_s)
    \\
    \mathcal{V}_d(t+T_s)
    \end{bmatrix}
    =f_{\mathcal{V}_d}\left(
    \begin{bmatrix}
    \mathcal{V}(t)
    \\
    \mathcal{V}_d(t)
    \end{bmatrix}
    ,
    \begin{bmatrix}
    Q(t)
    \\
    Q(t+T_s)
    \end{bmatrix}
    ,\,\,
    t
    \right).\textit{}
\end{equation}
$T_s$ is the sample period in seconds, $t$ is the time from the rising edge of the voltage pulse in seconds, and $V(t)$ is the applied voltage signal.

Both $f_Q$ and $f_{\mathcal{V}_d}$ are piecewise defined to capture switching between partial process dynamics and to capture state resets---functions that execute upon certain switches and alter the dynamical states before the first evaluation of the newly active partial process dynamics.

In 
Section \ref{sec:physicsModeling}
the division of the material ejection process into partial processes 
was done to maximize the use of physics-driven first principles 
model components, and was based on the stretching of the meniscus beyond its maximum stable non-jetting equilibrium extension.

This work presents an alternate breakdown of the material ejection process designed to facilitate modeling of the deposited droplet volume. This breakdown revolves around whether or not there exists a contiguous fluid stream between the nozzle and substrate, as the cessation of this contiguity is synonymous with the cessation of flow into or out of the droplet \cv. Specifically, the complete process is broken into an ``approach'' stage, a ``contiguity'' stage, and a ``retraction'' stage.
To better focus on the mapping between $Q$ and $\vol_d$, instead of modeling the meniscus tip position dynamics this work assumes that the timing of jet impingement and breaking are determined solely by the pulse parameters $V_l$, $V_h$, and $T_p$.
Switching is thus governed by time: with the rising edge of the voltage pulse set to $t=0$, transition from approach to contiguity occurs when $t$ exceeds $t_c(V_l,V_h,T_p)$ and transition from contiguity to retraction occurs when $t$ exceeds $t_r(V_l,V_h,T_p)$, where $t_c$ and $t_r$ are identified from data for each set of pulse parameters as described in Section \ref{sec:ID}.

The lone reset in the system is applied to $\mathcal{V}_d$ upon the switch from approach to contiguity. The complete model architecture is thus visualized by Figure \ref{fig:architecture}.

\begin{figure}
    \centering
    \includegraphics[width=3.25in]{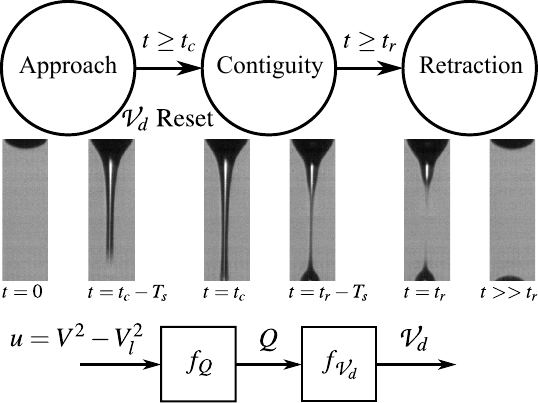}
    \caption{
    System Architecture. \emph{Top:} Automaton
    illustrating the
    timed switching behavior of the system model and the reset determining the initial contiguity droplet volume.
    Each automaton location is accompanied by photographs of the first and final samples of the corresponding partial process from a video with $V_h=\SI{1300}{\volt}$ and $T_p=\SI{1.5}{\milli\second}$.
    \emph{Bottom:} Block diagram illustrating breakdown of a complete input-to-droplet-volume model into a nozzle flow rate model $f_Q$ and a droplet volume model $f_{\vol_d}$, both of which are piecewise defined to capture the switching and reset behavior of the automaton.
    }
    \label{fig:architecture}
\end{figure}

\subsubsection{Deposited Droplet Volume Model}
\label{sec:dropmodel}
If the shape of the fluid-air interface 
 were constant over time and the fluid incompressible, 
the volumetric flow rate out of the nozzle, $Q$, would be equal to the volumetric flow rate into the droplet control volume, $Q_d$.
However, observation of 
video data
indicates that 
the interface broadens slowly but steadily while voltage is high during contiguity. This implies that $Q_d$ is only some fraction of $Q$. The video observation is corroborated by observation of the extracted time series data, such as that shown in Figure \ref{fig:staticJustification}.
This motivates a simple proportional model between $Q$ and $Q_d$ during contiguity:
\begin{equation}
    Q_d(t) = \inputcoef_{Q_d} Q(t)
    \label{eq:static}
\end{equation}
where $\inputcoef_{Q_d}$ is a constant.

Trapezoidal integration of $Q_d(t)$ in equation (\ref{eq:static}) yields the first-order discrete-time droplet volume model for the contiguity stage
\begin{equation}
    \vol_d(t+T_s) = \vol_d(t) + \frac{T_s}{2}\inputcoef_{Q_d}Q(t+T_s) + \frac{T_s}{2}\inputcoef_{Q_d}Q(t)
\end{equation}
where $\vol_d$ is the droplet volume and $T_s$ is the sample period.

\begin{figure}
    \centering
    \includegraphics[width=4.25in]{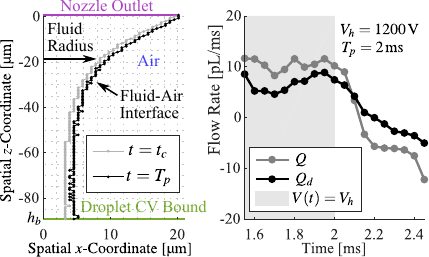}
    \caption{
    Measured data from a particular 
    ejection video illustrating the motivation for a proportional $Q_d$ model.
    \emph{Left:} Half-outline of fluid body (the jet is roughly symmetric about $z$-axis) at the onset of contiguity and at the end of the voltage pulse, representing a $7\%$ increase in volume outside the droplet \cv. \emph{Right:} Flow rate time series data during contiguity illustrating roughly proportional signals between $Q$ and $Q_d$.
    }
    \label{fig:staticJustification}
\end{figure}

Volume is also added to the droplet 
\cv{}
during the approach stage 
when the tip of the meniscus crosses the upper boundary of the droplet 
\cv{}
but has not yet struck the substrate.
However, each ejection video only provides a few samples of this situation, as the jet traverses the distance from the jet break position to the substrate relatively quickly.
It may thus be impractical to identify a dynamical model of volume increase during the approach stage.
Instead, 
this work sets $\vol_d(t+T_s)$ to $\vol_d(t)$ during approach, and uses a reset to give $\vol_d$ an initial condition in the contiguity stage, which accounts for the fluid added to  the droplet \cv{} during approach.

The total volume of the fluid outside the nozzle, $\vol$, at the moment of jet collision with the substrate, $t_c$, may be roughly modeled as the volume of the Boolean union of a 
cylinder and a cone arranged to approximate the fluid body shape.
$\vol_d$
is then some fraction of the cylinder 
volume.
This is equivalent to some fraction, $\reseto$, of the total volume minus the Boolean difference of the cone and the cylinder, $\reseti$. Physics-driven modeling of the jet diameter and break position necessary to explicitly
calculate $\reseto$ and $\reseti$
are beyond the scope of this work, 
but
the structure of the mapping between $\mathcal{V}$ and $\mathcal{V}_d$ at $t=t_c$ arising from this geometric analysis may still be used:
\begin{equation}
    \vol_d^+ = \vol_d^- + \reseto\vol^- + \reseti
    \label{eq:reset}
\end{equation}
where $\reseto$ and $\reseti$ require data-driven identification and the subscripts $+$ and $-$ indicate a state's value before and after reset. $\vol_d^-$ will be 0 unless there was already fluid in the droplet control volume (e.g. if a second pulse is fired over an existing droplet).

Equation (\ref{eq:reset}) requires total volume $\vol$ be captured by the state dynamics, which can be done with a trapezoidal integration of the input $Q$ similar to that of equation (\ref{eq:static}). This addition completes the hybrid model of droplet volume evolution in terms of total flow rate input, which can be given
in totality as
\begin{multline}
    \begin{bmatrix}\vol(t+T_s)\\\vol_d(t+T_s)\end{bmatrix}=
    f_{\mathcal{V}_d}\left(
    \begin{bmatrix}
    	% \vol(t+T_s)\\\vol_d(t+T_s)
    	\vol(t)\\\vol_d(t)
    	\end{bmatrix},
    \begin{bmatrix}
    Q(t)
    \\
    Q(t+T_s)
    \end{bmatrix},\,\,t
    \right)
    =
    \\
    \begin{cases}
    \begin{bmatrix}
    \vol(t) + \mathscr{Q}
    \\
    \vol_d(t)
    \end{bmatrix}
    &
    t<t_c \,\, \lor \,\, t\geq t_r
    \\
    \begin{bmatrix}
    \vol(t) + \mathscr{Q}
    \\
    \reseto\vol(t) + \vol_d(t) + \inputcoef_{Q_d}\mathscr{Q} + \reseti
    \end{bmatrix}
    &
    t_c\leq t < t_c+T_s
    \\
    \begin{bmatrix}
    \vol(t) + \mathscr{Q}
    \\
    \vol_d(t) + \inputcoef_{Q_d}\mathscr{Q}
    \end{bmatrix}
    &
    t_c+T_s \leq t < t_r
    \end{cases}
    \label{eq:volumeModel}
\end{multline}
where
\begin{equation}
    \mathscr{Q}=\frac{T_s}{2}\left(Q(t+T_s) + Q(t)\right)
\end{equation}
and $t<t_c$ corresponds to the approach stage, $t\geq t_r$ corresponds to the retraction stage, $t_c+T_s\leq t<t_r$ corresponds to all but the first time step of the contiguity stage, and $t_c\leq t < t_c+T_s$ corresponds to the first time step of contiguity, in which the reset is applied.

\subsubsection{Nozzle Flow Rate Model}
The main focus of this chapter is the development and validation of the mapping between nozzle flow rate $Q$ and deposited droplet volume $\mathcal{V}_d$.
This could be done by simply injecting measured $Q$ data into equation (\ref{eq:volumeModel}) and assessing the generated $\mathcal{V}_d$ signals against measured droplet volumes.
However, for control there must ultimately be a model 
with input based on applied voltage $V(t)$
rather than $Q(t)$.
To demonstrate the viability of 
equation
(\ref{eq:volumeModel})
for this purpose,
this section presents a simple $V$-to-$Q$ model based on 
Section \ref{sec:physicsModeling},
which is cascaded
with the $\mathcal{V}_d$ model.

The jetting model in 
Section \ref{sec:physicsModeling}
is a second-order LTI system, which may be represented in discrete time as
% % BIND
% \begin{multline}
%     Q(t+2T_s) = 
%     % \\
%     a_{Q,1}(V_l,V_h,T_p)Q(t+T_s) + a_{Q,2}(V_l,V_h,T_p)Q(t) 
%     \\
%     + b_{Q,2}(V_l,V_h,T_p)u(t)
%     \label{eq:priorQ}
% \end{multline}
% % ONLINE
\begin{equation}
    Q(t+2T_s) = 
    % \\
    a_{Q,1}(V_l,V_h,T_p)Q(t+T_s) + a_{Q,2}(V_l,V_h,T_p)Q(t) 
    % \\
    + b_{Q,2}(V_l,V_h,T_p)u(t)
    \label{eq:priorQ}
\end{equation}
where the input $u(t)$ is given as
\begin{equation}
    u(t) = V(t)^2 - V_l^2
\end{equation}
This choice of input is made because the physics-based first principles models of flow rate are driven by the applied voltage squared, and because at the 
low voltage stable equilibrium,
$Q$ should be zero. The model parameters $a_{Q,1}$, $a_{Q,2}$, and $b_{Q,2}$ are identified independently for each pulse definition in 
Section \ref{sec:physicsModeling}.

Section \ref{sec:physicsModeling}
uses equation (\ref{eq:priorQ}) when the nonlinear physics first principles models cease to capture the observed dynamics.
This happens during contiguity and in approach and retraction when the meniscus is sufficiently elongated.
Because the nonlinear models cannot capture the entirety of approach or contiguity, incorporating them into this work's deposited-volume-focused switching framework would substantially complicate the model.
Thus to preserve the model's focus and manage complexity while still accounting for changes in dynamical behavior over the course of ejection, the structure of equation (\ref{eq:priorQ}) is applied to the entire model with separate parameters identified for contiguity and non-contiguity partial processes.
This results in the model
\begin{multline}
    \begin{bmatrix}
    Q(t+T_s) 
    \\
    Q(t+2T_s) 
    \end{bmatrix}
    =
    f_Q\left(
    \begin{bmatrix}
    Q(t) 
    \\
    Q(t+T_s)
    \end{bmatrix},u(t),\,\,t
    \right)=
    \\
    \begin{cases}
    \begin{bmatrix}
    0 & 1 \\ \tilde{a}_{Q,2}
    & \tilde{a}_{Q,1}
    \end{bmatrix}
    \begin{bmatrix}
    Q(t) \\ Q(t+T_s)
    \end{bmatrix}
    +
    \begin{bmatrix}
    0\\\tilde{b}_{Q,2}
    \end{bmatrix}
    u(t)
    &
    t<t_c \,\, \lor \,\, t\geq t_r
    \\
    \begin{bmatrix}
    0 & 1 \\ \overline{a}_{Q,2}
    & \overline{a}_{Q,1}
    \end{bmatrix}
    \begin{bmatrix}
    Q(t) \\ Q(t+T_s)
    \end{bmatrix}
    +
    \begin{bmatrix}
    0\\\overline{b}_{Q,2}
    \end{bmatrix}
    u(t)
    &
    t_c \leq t < t_r
    \end{cases}
    ,
    \label{eq:flowModel}
\end{multline}
where the tilde-topped and overlined parameters are separately identified 
(and have the input arguments $(V_l,V_h,T_p)$ dropped for compactness)
and correspond to the approach and retraction stages and the contiguity stage, respectively.

\subsection{Experimental Methods}
\label{sec:methods}

\subsubsection{Experimental Setup}

The \ac{e-jet} printing setup is identical to that of Section \ref{expMeth}.

Twelve sets of $(V_l,V_h,T_p)$ parameters, referred to as ``experiments,'' are tested, with 20 trials of each experiment being recorded.
All trials begin from the low voltage equilibrium meniscus position over a clean region of substrate with no prior fluid depositions: $V(0)=V_l$, $Q(0)=Q(T_s)=0$, 
$\vol_d(0)=0$, and $\vol(0)$ is the small total fluid volume outside the nozzle at low voltage equilibrium (see Figure \ref{fig:ejetschematic}, \SI{0}{\milli\second}).
Time $t=0$ is defined at the rising edge of the voltage pulse. $V_l=\SI{525}{\volt}$ for all experiments. $V_h$ and $T_p$ values are tabulated in Table \ref{table:exp}.

\begin{table}
\centering
\caption{
Experimental High Voltage and Pulse Width Pairs
}
\label{table:exp}
\renewcommand{\arraystretch}{1.1}
\begin{tabular}{|c|c|}
\hline
$V_h$ {[}V{]} & $T_p$ {[}$\SI[detect-weight=true]{}{\milli\second}${]}
\\ 
\hhline{|=|=|}
1100       & 2.0     \\ \hline
1150       & 2.0     \\ \hline
1200       & 2.0     \\ \hline
1250       & 2.0     \\ \hline
\end{tabular}
\,
\begin{tabular}{|c|c|}
\hline
$V_h$ {[}V{]} & $T_p$ {[}$\SI[detect-weight=true]{}{\milli\second}${]}
\\ 
\hhline{|=|=|}
1300       & 1.5     \\ \hline
1300       & 1.8     \\ \hline
1300       & 2.0     \\ \hline
1300       & 2.3     \\ \hline
\end{tabular}
\,
\begin{tabular}{|c|c|}
\hline
$V_h$ {[}V{]} & $T_p$ {[}$\SI[detect-weight=true]{}{\milli\second}${]}
\\ 
\hhline{|=|=|}
1350       & 2.0     \\ \hline
1370       & 2.0     \\ \hline
1420       & 1.5     \\ \hline 
1470       & 1.5     \\ \hline
\end{tabular}
\end{table}

\subsubsection{High Speed Microscopy \& State Extraction}
\label{sec:compvis}
Each frame of video is a grayscale image containing the nozzle tip, the fluid outside the nozzle, and---if there is fluid near enough to the substrate---a reflection of the fluid off of the substrate.
The image processing protocol used to extract time series measurements of $\vol$ and $\vol_d$ from these images is nearly the same as that 
of
\cite{Spiegel2017}. Edge finding identifies the $(x,z)$ coordinates for the silhouette of the nozzle tip, fluid, and reflection. Corner finding and extremum finding identify the $z$-coordinates of the nozzle-fluid interface and the substrate-fluid interface. Finally, the width of the silhouette at each $z$-coordinate is treated as the diameter of a disk of height equal to the image resolution (i.e. the height of one pixel) for volume determination.
$Q$ and $Q_d$ measurements are numerical derivatives of $\vol$ and $\vol_d$ measurements.
On top of this established procedure, this work introduces a method to harness the reflection for improving volume measurement consistency near the substrate and a method for estimating jet break position, both of which are illustrated in Figure \ref{fig:compvisMethod} and explained in detail below.

In previous works, the reflection was eliminated from volume calculations 
entirely.
However, due to the quantization error associated with fixing the substrate position measurement to a pixel 
edge 
(and possibly other optical or image processing imperfections), this direct calculation leads to nonzero flow of fluid through the bottom of the droplet 
\cv{}
as the droplet spreads.
This change in droplet volume measurement during the retraction stage makes identifying a final droplet volume value difficult.
Thus this work introduces reflection-augmented volume measurement to 
maintain conservation of volume near the substrate.
Given images that extend $\Delta p_r$ pixels below the 
estimated substrate position,
a region $2\Delta p_r$ pixels tall centered on the estimated substrate position is defined as the ``reflection-augmented measurement region.''
The volume of fluid contained in the upper half of this region (i.e. the portion of the region containing direct fluid silhouette rather than reflection) is taken as half the volume computed in the total reflection-augmented measurement region.
In other words, the true fluid volume is taken as the average of the silhouette volume and reflection volume.

\begin{figure}
    \centering
    \includegraphics[scale=0.54]{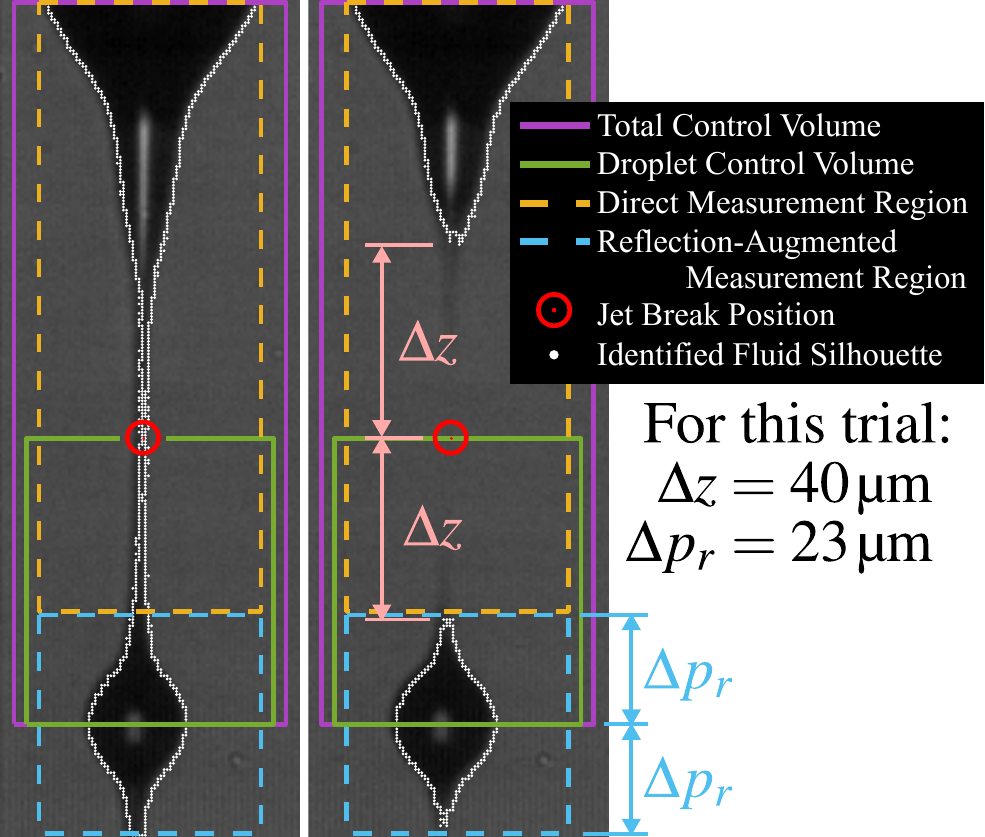}
    \caption{
    Depiction of the two \cv{}s described in Section \ref{sec:dropVolDef}, the distinction between the direct and reflection-augmented measurement regions, and the estimated jet break position superposed on the final contiguity frame and first retraction frame for a particular trial ($V_h=\SI{1300}{\volt}$, $T_p=\SI{1.5}{\milli\second}$). $\Delta z$ is the distance from either fluid body tip to the jet break position.
    }
    \label{fig:compvisMethod}
\end{figure}

The jet break position for each trial is measured from the first video frame in which there are two disjoint fluid bodies (i.e. the first frame of retraction).
The jet break position is estimated to be the midpoint between the tips of these bodies based on the assumption that the initial droplet and meniscus tip velocities and accelerations are equal and opposite at the moment the jet breaks.
This assumption is driven by the fact that the initial jet-breaking and retraction behavior is dominated by surface tension. While this assumption neglects much of the complexity of the true retraction physics, Figure \ref{fig:compvisMethod} suggests it does an acceptable job of identifying the thinnest portion of the jet immediately before breaking, and it circumvents image resolution and noise issues associated with directly computing the thinnest jet point in the final frame of contiguity.

\subsubsection{System Identification}
\label{sec:ID}
\begin{sloppypar}
The parameters to be identified can be grouped into two categories. First are the primary model parameters making up the dynamics of the 
$Q$-to-$\vol_d$
model (\ref{eq:volumeModel}): $\inputcoef_{Q_d}$, $\reset_0$, and $\reset_1$, which are constant over all experiments. Second are the supporting model parameters determining the timing of switching in the volume model (\ref{eq:volumeModel}) and the simulated nozzle flow rate model (\ref{eq:flowModel}): $t_c(V_l,V_h,T_p)$, $t_r(V_l,V_h,T_p)$, $\tilde{a}_{Q,1}(V_l,V_h,T_p)$, $\tilde{a}_{Q,2}(V_l,V_h,T_p)$, $\tilde{b}_{Q,2}(V_l,V_h,T_p)$, $\overline{a}_{Q,1}(V_l,V_h,T_p)$, $\overline{a}_{Q,2}(V_l,V_h,T_p)$, and $\overline{b}_{Q,2}(V_l,V_h,T_p)$, whose values vary with the the experiment parameters. The measured data is divided evenly into training and validation data. For each experiment, 10 trials are reserved for training and 10 for validation.
\end{sloppypar}

In order to keep the timing of switching fixed to a particular sample, for each experiment $t_c$ and $t_r$ are taken as the median time of the first frame of contiguity and retraction over the 10 training trials of that experiment (1 parameter from 10 samples).
All other parameters are identified via least squares regression.
The model coefficients for 
equation
(\ref{eq:flowModel})
are trained independently for each experiment on all the available data in the corresponding process stage (3 parameters from 353 samples for approach and 194 samples for contiguity, on average).
The reset parameters $\reset_0$ and $\reset_1$ are trained on the set of first frames of contiguity (i.e. the $t_c\leq t<t_c+T_s$ sample) from all trials (2 parameters from 120 samples). Finally, $\inputcoef_{Q_d}$ is trained on all contiguity training data (1 parameter from 2337 samples).

\subsubsection{Error Metrics}
This section presents metrics for validating the reflection-augmented volume measurement technique,
and the ability of the total model to predict deposited droplet volume.

The improvement yielded by reflection-augmented volume measurement 
over direct measurement
is 
quantified as 
the mean percent decrease in total variation of the filtered retraction-stage $\vol_d$ timeseries 
between 
the two
techniques.
The total variation $L$ of a time-varying parameter is its total change (as opposed to net change) over a given period of time $T_L$.
In theory, the total variation of $\vol_d$ in the retraction stage is zero,
making its reduction a practical improvement metric.
However, high frequency measurement noise
also contributes to $L$.
To moderate noise's influence, the $\vol_d$ signal is filtered before its total variation is computed.
Thus, the total variation of a trial $j$ is given by
\begin{equation}
    L_j = \sum_{t=t_r}^{t_r+T_L} |\vol_d^f(t+T_s)-\vol_d^f(t)|
    \label{eq:totalVar}
\end{equation}
where $\vol_d^f$ is the filtered volume signal.
The final metric for reflection-augmented volume measurement performance is
\begin{equation}
    \Delta L\% = 100\mean_{j\in\textrm{All Validation Trials}} \frac{L_j^{\textrm{direct}}-L_j^{\textrm{augmented}}}{L_j^{\textrm{direct}}}
    \label{eq:percentDecrease}
\end{equation}.

Here, a Savitzky-Golay filter with a window size of 15 samples is used. $T_L=\SI{10}{\milli\second}$ (200 samples), roughly the time it takes for the droplet to spread and settle to its final shape on the substrate.

The efficacy of the overall model in predicting deposited droplet volume is measured by the mean unsigned error between the modeled $\vol_d(t_r)$ (equal to $\vol_d(t>t_r)$) and the measured final droplet volume. Measured final droplet volume is taken as
\begin{equation}
    \vol^{\textrm{final}}_{d,j} = \mean_{t\in \left[t_r,t_r+\SI{10}{\milli\second}\right]} \vol_{d,j}^{\textrm{meas}}(t)
\end{equation}
where $\vol_{d,j}^{\textrm{meas}}(t)$ is the measured droplet volume time series for a particular trial $j$,
making the mean unsigned error
\begin{equation}
    e_{\vol_d} = \mean_{j\in J}|\vol^{\textrm{final}}_{d,j}-\vol_{d,j}(t_r)|
    \label{eq:dropErr}
\end{equation}
This metric is evaluated over multiple sets of trials $J$. In addition to an aggregate $e_{\vol_d}$ in which $J$ contains 
the validation trials 
of
all but one experiment (that of lowest $V_h$, see Section \ref{sec:dropletVal} for discussion of this exclusion), individual $e_{\vol_d}$ values are computed for each experiment. This is done to examine how model performance changes with the pulse parameters.
Additionally, for each of these sets $J$, both an $e_{\vol_d}$ using $\vol_{d}(t_r)$ generated from injecting measured nozzle flow rate into equation (\ref{eq:volumeModel}) and an $e_{\vol_d}$ using $\vol_{d}(t_r)$ generated from nozzle flow rate simulated by equation (\ref{eq:flowModel}) are computed. This is done to enable both focus on the quality of equation (\ref{eq:volumeModel}) and broader consideration of the ultimate needs for a control-oriented \ejet{} model, respectively.
Finally, along with each $e_{\vol_d}$, a corresponding standard deviation of the signed error is presented.

\subsection{Results \& Discussion}
\label{sec:results}

\subsubsection{Droplet Volume Measurement}
Equations (\ref{eq:totalVar}) and (\ref{eq:percentDecrease}) show that the reflection-augmented image processing yields a 42\% decrease in total variation of measured droplet volume time series data compared to direct measurement, with an associated standard deviation of 10\%. 
This substantial performance improvement can be visualized through the example retraction-stage droplet volume time series
in Figure \ref{fig:totalVar},
in which the direct measurement yields a steady decrease while the reflection-augmented measurement yields a relatively constant droplet volume.

\begin{figure}
    \centering
    \includegraphics[width=4.25in,trim={0.345in 0in 0.5in 0.15in},clip]{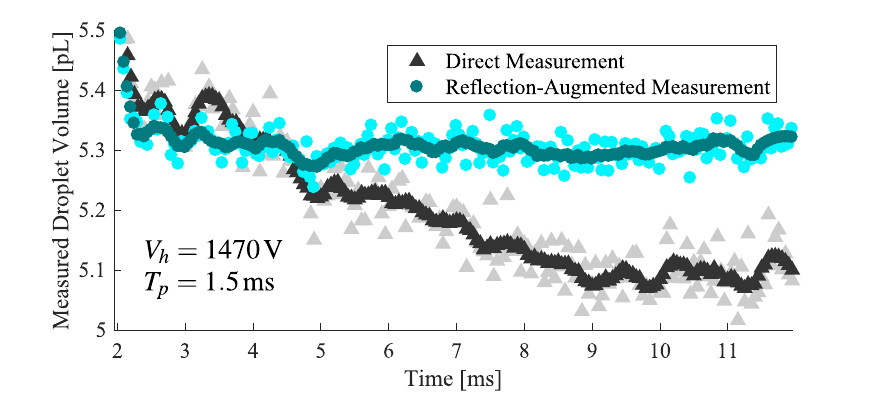}
    \caption{
    Retraction-stage droplet volume measurement taken by direct and reflection-augmented measurement techniques. Light and dark points represent raw and filtered data, respectively. As $\vol_d$ is expected to be constant during retraction, this plot illustrates the reflection-augmented technique's superiority in that it maintains a roughly constant value of \SI{5.3}{\pico\liter} after the transient (i.e. for $t\geq\SI{2.5}{\milli\second}$), while the direct measurement steadily decreases until about $t=\SI{9}{\milli\second}$.
    }
    \label{fig:totalVar}
\end{figure}

However, in both measurement schemes there is a steep transient at the start of retraction. This arises from an inability of these measurement techniques to conserve volume over the collapse of the relatively tall and thin droplet tail (observable in Figure \ref{fig:compvisMethod}, right) into the larger and wider main droplet body.

Thus, these results demonstrate that the reflection-augmented volume measurement scheme is an effective tool that may be useful for future \ejet{} research, but does not address every artifact associated with video-based measurement, which may serve as the subject of future investigations.

\subsubsection{Deposited Droplet Volume Error}
\label{sec:dropletVal}
Figure \ref{fig:3Derror} presents the mean percent error in the final droplet volume 
as computed by equation (\ref{eq:dropErr}),
for each experiment.
The experiment of lowest high voltage ($V_h=\SI{1100}{\volt}$) clearly represents an outlier in this data, having a percent error of 130\% for the predictions driven by measured nozzle flow rate and 170\% for the predictions using simulated nozzle flow rate, more than triple the next highest percent error.

\begin{figure}
    \centering
    \noindent
    \includegraphics[scale=0.7,trim={0in 0in 0in 0in},clip]{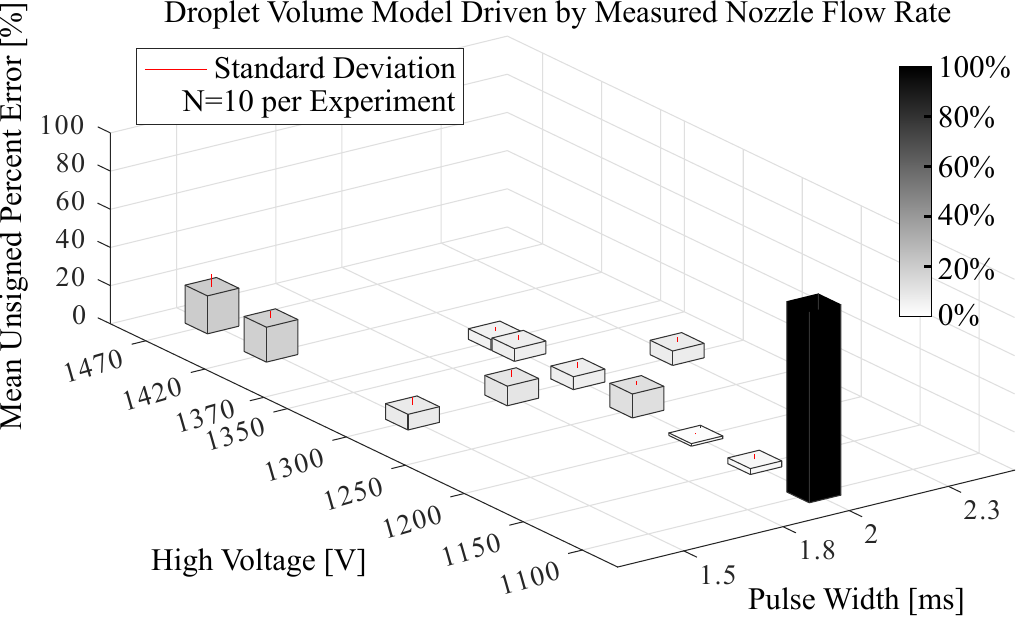}
    \\
    \vspace{4mm}
    \noindent
    \includegraphics[scale=0.7,trim={0in 0in 0in 0in},clip]{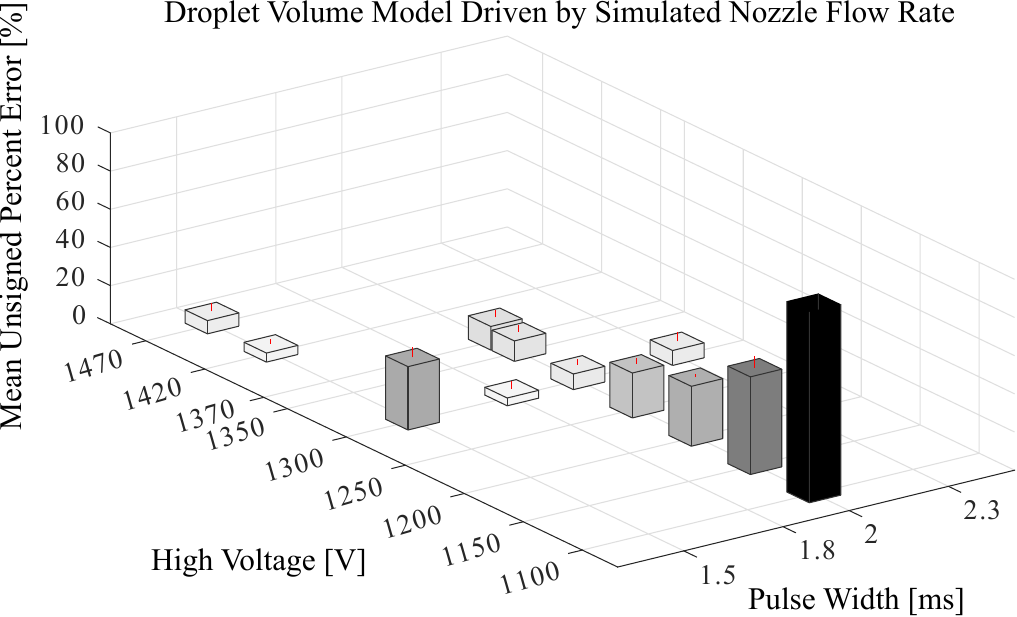}
    \caption{Percent error in final deposited droplet volume, $\vol_d(t_r)$, using measured and simulated $Q$. The lowest voltage experiment exceeds 100\% error in both cases.
    Each bar represents the mean value of $N=10$ samples.
    Measured $Q$ results illustrate the high quality of equation (\ref{eq:volumeModel}) for all but the lowest voltage case. Simulated $Q$ results illustrate increased error associated with increased uncertainty in the cascaded model, motivating future flow rate modeling work.
    }
    \label{fig:3Derror}
\end{figure}

To better discuss this outlying experiment, 
time series plots of 
$\vol_d(t)$
and
a plot of each trial's estimated jet break position, $h_b$, versus 
$V_h$ 
are presented in Figures 
\ref{fig:timeseries}
and 
\ref{fig:breakData}, 
respectively.
From the time series plot, one observes that the reset---the initial step change from zero to non-zero volume---is the most clearly erroneous feature of the low $V_h$ time series.
The reset causes a large overestimation of the initial volume in the contiguity stage that cannot be compensated for by the contiguity dynamics models, which only capture the change in droplet volume from the beginning to the end of contiguity.

The plot of 
$h_b$
in Figure \ref{fig:breakData} lends insight into why this reset error may arise. While the experiments well within the subcritical jetting regime (those from \SI{1150}{\volt} to \SI{1420}{\volt}) show comparable $h_b$ values, the experiment of lowest $V_h$ shows a jet break position markedly closer to the substrate. Because $h_b$ marks the upper boundary of the droplet control volume (a condition necessary for droplet flow rate to be zero after the jet breaks), this lowered jet break position substantially reduces the fraction of total volume that is in the droplet control volume at the first moment of contiguity. This change in the fraction of total volume is not accounted for by the reset model (\ref{eq:reset}), which assumes only the total volume itself is changing (e.g. because of jet diameter variations over applied voltage).

\begin{figure}
    \centering
    \noindent
    \includegraphics[scale=0.9,trim={0.4in 0.2in 0.6in 0.1in},clip]{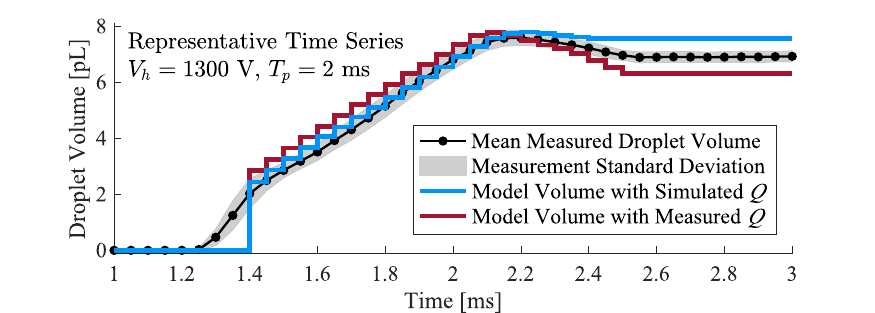}
    \\
    \noindent
    \includegraphics[scale=0.9,trim={0.4in 0in 0.6in 0in},clip]{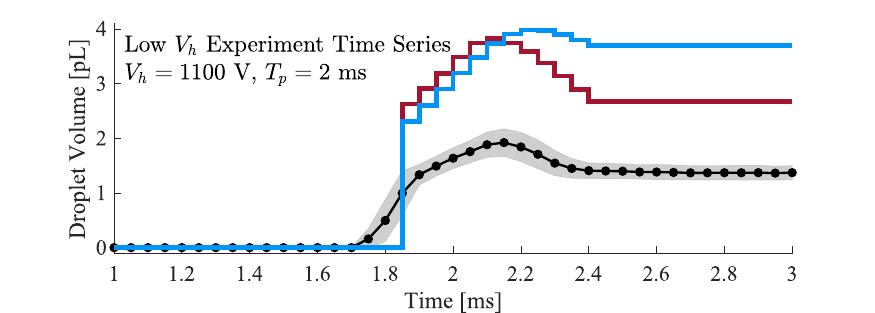}
    \caption{
    Time series plots of droplet volume $\vol_d$ for a representative experiment and the experiment of lowest high voltage $V_h$. Plotted measured data is the mean of the validation data ($N=10$ samples for each time series) with an envelope of plus or minus the standard deviation. The data suggests that the reset is the main source of error in low $V_h$ experiments.
    }
    \label{fig:timeseries}
\end{figure}

\begin{figure}
    \centering
    \includegraphics[scale=0.8,trim={0.28in 0in 0.52in 0.05in},clip]{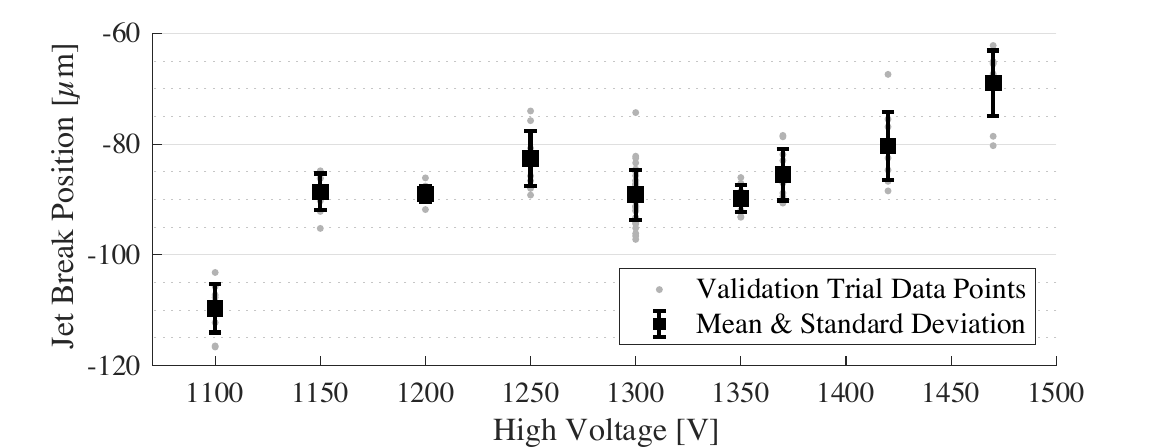}
    \caption{
    Jet break positions of each experiment's validation data against high voltage.
    The nozzle outlet is located at \SI{0}{\micro\meter} and the substrate at \SI{-150}{\micro\meter}.
    $N=10$ samples for all high voltages except $V_h=\SI{1300}{\volt}$, for which $N=40$ samples because four pulse widths are tested at $V_h=\SI{1300}{\volt}$.
    The modest spread of data points at $V_h=\SI{1300}{\volt}$ suggests that high voltage (equivalent to the difference between high and low voltage in this data set) has a greater influence on jet break position than pulse width in the subcritical jetting regime.
    } 
    \label{fig:breakData}
\end{figure}

Because the given model structure does not account for the changing jet break position near the boundaries of the subcritical jetting regime, 
the low $V_h$ experiment is deemed to be outside the applicable domain of the model, and is thus removed from the aggregate model error data, given in Figure \ref{fig:aggregate}.
The lower error yielded when the droplet volume model is driven by measured nozzle flow rate illustrates the validity of the $Q$-to-$\vol_d$ model (\ref{eq:volumeModel}).
When equation (\ref{eq:volumeModel}) is driven by the nozzle flow rate simulated by equation (\ref{eq:flowModel}), making a complete model from $V$ to $\vol_d$, the error increases.
This is due to increased model uncertainty associated with equation (\ref{eq:flowModel}) and its cascading with equation (\ref{eq:volumeModel}). While reducing this model uncertainty will be an important future endeavor, these results demonstrate the foundation of a
dynamical $V$-to-$\vol_d$ model that may be integrated with iterative learning control 
for the sake of \ejet{} control.

\begin{figure}
    \centering
    \includegraphics[scale=0.8,trim={0.28in 0.1in 0.78in 0.2in},clip]{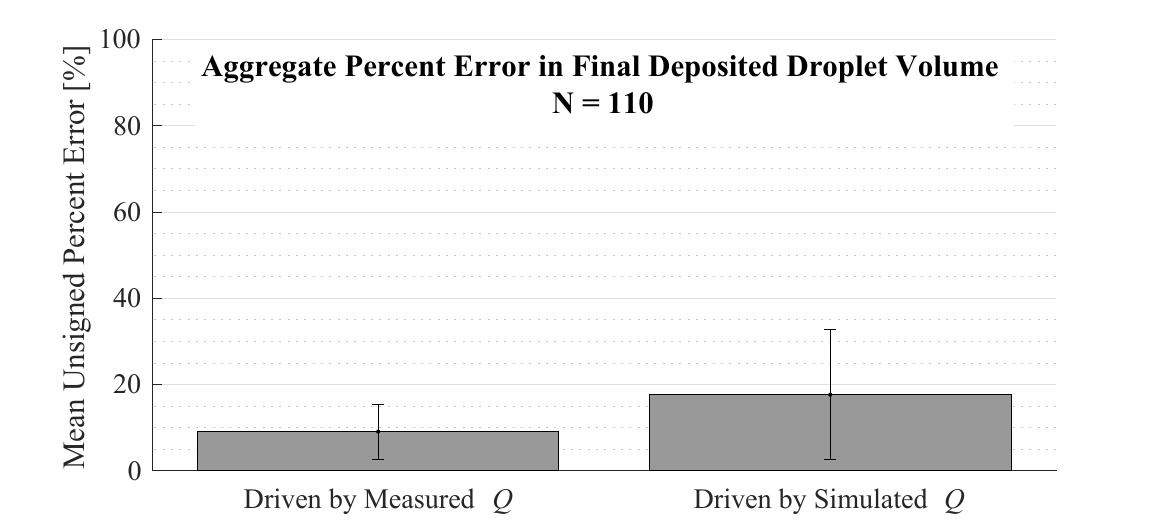}
    \caption{Mean unsigned percent error and standard deviation (given by error bars) of final droplet volume over all validation trials except those of lowest high voltage ($N=110$ samples).}
    \label{fig:aggregate}
\end{figure}

\subsection{Piecewise Affine E-jet Modeling Conclusion}
\label{sec:conclusion}
This section presents a hybrid system model framework for dynamical droplet volume modeling in \ejet{} based on contiguity of the fluid jet between the nozzle and the substrate.
This overarching modeling framework involves the contributions of 
several novel model elements whose structures are motivated by physical analysis of the material ejection process.
These contributions include
a new droplet control volume definition enabling the treatment of droplet volume as a dynamically evolving parameter, a dynamical mapping between nozzle flow rate and droplet volume during contiguity, 
and a reset map circumventing the need for dynamical modeling of droplet volume before contiguity, when data is of limited availability.
Additionally, this section contributes a new computer vision method for extracting time series volume data from videos that, while simple, substantially improves the consistency of volume measurements.
These contributions are validated with physical experiments that show good model performance in the interior of the subcritical jetting regime, but reveal system behaviors not captured by the model towards the boundaries of the subcritical jetting regime, specifically unmodeled changes in the jet break position.

The presented work thus represents an important step towards model-based control of deposited droplet volume in \ejet. To further pursue this goal,
future work will 
focus on
refining the reset model to account for the entire subcritical jetting regime, and on 
finding an ideal balance between simplicity and fidelity in
applied-voltage-to-nozzle-flow-rate
models.
Finally, relaxation of the time-based switching assumption and reintroduction of meniscus position as the switching criterion, and thus as a dynamical state, will substantially improve the flexibility and utility of the droplet volume modeling framework.

\section{Overarching E-jet Printing Conclusion}

This chapter has contributed two hybrid models for e-jet printing: one focused on fidelity to the first principles of the physical system and the other on capturing the final volume of the droplet on the substrate.

Specifically, the former contributes new geometrical and equilibrium analysis, extending the amount of the ejection process that can be modeled by first principles and reducing model reliance on measured data.
Additionally, Section \ref{sec:physicsModeling} introduces and validates hybrid modeling as a means to merge physics-driven modeling with data-driven modeling to produce an ODE-based model of the end-to-end electrohydrodynamic ejection process.

The latter hybrid model of Section \ref{sec:controlModeling} introduces an alternative division of the total ejection process into partial process and defines a new \ac{CV} for droplet volume. These contributions enable dynamical modeling of the droplet volume ultimately deposited on the substrate via \ac{PWA} framework. 

The bigger picture delivered by this chapter is that \ac{e-jet} printing exemplifies a class of physical 
systems
for which hybrid modeling is thus far the only path to ODE-based modeling, and thus to control-oriented modeling.
E-jet printing also represents a class of systems for which \ac{ILC} is necessary to improve performance because of the impracticality of real-time feedback.
Thus, while it is certainly not the only motivation, e-jet printing provides a concrete motivation for the integration of hybrid systems theory and ILC theory, upon which the subsequent chapters focus.
Finally, because e-jet’s modeling
and control challenges are shared by numerous other AM
technologies, the validation of this modeling philosophy may
lower the boundary for the development of similar models for
these AM processes.
\ac{FDM}, perhaps the most ubiquitous \ac{AM} technology, serves as a prime example. While formal hybrid modeling has not been attempted, 
\cite{Plott2018}
identifies several distinct regimes for the filament deposition dynamics. As these regimes arise due to the physical state of the printhead and filament, these dynamics may be well-unified by a hybrid framework.

\chapter{
Enabling ILC of Hybrid Systems:
\texorpdfstring{\\}{ }
Closed-Form Hybrid System Representation
}
\blfootnote{
Content of this chapter also published as:
\\
\indent\indent
I. A. Spiegel and K. Barton, ``A Closed-Form Representation of Piecewise Defined Systems and their Integration with Iterative Learning Control,'' in \emph{2019 American Control Conference (ACC)}, Philadelphia, PA, USA, 2019, pp. 2327-2333, \href{https://doi.org/10.23919/ACC.2019.8814823}{https://doi.org/10.23919/ACC.2019.8814823}
\textcopyright IEEE 2019. Reprinted with permission.
}
\label{ch:3}

In most cases, even when hybrid system control takes inspiration from preexisting control theories, significant work must be done to redevelop the theory specifically for hybrid systems \cite{DiBernardo2013,Zhao2008,VandeWouw2006}.
One reason for this difficulty is that the hybrid systems are, in general, mathematically represented as a type of automaton \cite{Tabuada2009,Cassandras2008} rather than as a closed-form system of ODEs as is typical for control systems. This lack of a closed-form system representation prevents the direct application of many analytical mathematical operations that would ordinarily accompany dynamical system analysis and controller design (e.g. function composition, Jacobian).

The importance of 
performing such operations on piecewise functions, and thus the importance of having closed form representations of piecewise functions,
was originally identified by the nonlinear circuit theory community well before the modern notion of hybrid systems were developed. In fact, in 1977 Chua and Kang published a canonical, closed-form representation of piece-wise linear functions to close this gap \cite{Chua1977}. However, their representation was not designed with dynamical systems in mind, and thus has several features making it incompatible with hybrid systems analysis and control. Most importantly, it is based on interpolation between the breakpoints of the piecewise functions, precluding it from representing nonlinear dynamics within a particular discrete state. Secondly, rather than representing each ``piece'' of a piecewise function independently, 
\cite{Chua1977}
treats each ``piece'' as a superposition of components from all ``pieces'' corresponding to lesser values of the independent variable. Finally, higher dimensional functions are represented by nesting functions of 1 or 2 dimensions, which results in structures far more convoluted than the state-space models employed by today's control engineers.

In the more contemporary literature, Bemporad and Morari's seminal 
\ac{MLD}
Systems seek to 
provide a more analytical hybrid system representation, which is
composed entirely of systems of algebraic equalities and inequalities \cite{Bemporad1999}. However, 
like the systems of \cite{Chua1977}, \ac{MLD} systems
cannot contain nonlinear dynamics within a discrete state. More importantly, while MLD systems integrate readily with control frameworks that involve online optimization, such as 
\ac{MPC},
they can be difficult to integrate with other classes of controllers because they require the solution of a mixed integer program at each time step to determine the system's discrete state.

This restriction on controller options is particularly problematic when real time feedback is unavailable, or when model errors are substantial enough to prevent the fulfillment of performance goals. 
In such cases, 
\ac{ILC}
is attractive.
However, to date no implementation of ILC has been made with hybrid system models.
This is in part because the mathematical operations required to synthesize the controller cannot be performed on contemporary hybrid system representations.

The primary contribution of the present chapter is to deliver a closed-form representation of a particular class of hybrid systems: 
\ac{PWD} 
systems (i.e. a generalization of the popular 
\ac{PWA}
system class).
This closed-form representation is shown to bridge the gap between system representation and control via the application of ILC to an example hybrid system.
Minor contributions are made to the selected ILC method (applicable to both hybrid and non-hybrid systems):
it is reformalized to enable application
to systems of any relative degree, and a novel recommendation is made vis-à-vis implementation in order to facilitate direct application of the control theory and to improve controller scalability.
No special modifications are made for the hybrid nature of the example system, thereby illustrating the utility of the closed-form representation.

The rest of the chapter is organized as follows. 
Section \ref{sec:closedForm} presents the closed-form piecewise defined system representation, i.e. this work's main contribution, via a proposition and constructive proof. 
Section \ref{sec:NILC} details the iterative learning controller to be used. Section \ref{sec:validation} presents the example hybrid system, the methods of controller performance analysis, and presents and discusses the results of the simulation experiments. Finally, Section \ref{sec:conc} provides concluding remarks and suggests future work.

\section{Closed Form Piecewise Defined Systems}
\label{sec:closedForm}

\begin{definition}[\ac{PWD} System]
A discrete-time \ac{PWD} system is a system defined by
\begin{equation}
    \begin{aligned}
    \state(k+1) &= f_\qdx(\state(k),\uMIMO(k),k) \\
    \yMIMO(k) &= h_\qdx(\state(k),\uMIMO(k),k)
    \end{aligned}
    \quad
    \textrm{for}
    \begin{bmatrix}
    \state(k) \\ \uMIMO(k)
    \end{bmatrix}
    \in
    \region_\qdx
    \label{eqn:openForm}
\end{equation}
where 
$k$ is the discrete time index,
$\state\in\mathbb{R}^\xDim$ is the state vector,
$\uMIMO\in\mathbb{R}^\uDim$ is the control input,
$\yMIMO\in\mathbb{R}^\yDim$ is the output vector,
$\qdx\in\{1,2,\cdots,\regionQuant\}$,
$\region_\qdx \in
\region=
\{\region_1, \region_2, \cdots, \region_{\regionQuant}\}$ is a convex polytope (i.e. an intersection of half-spaces) in $\mathbb{R}^{\xDim+\uDim}$,
$f_\qdx:\mathbb{R}^\xDim\times\mathbb{R}^\uDim\times\integer\rightarrow\mathbb{R}^\xDim$ is a closed-form function representing the potentially nonlinear state dynamics in $\region_\qdx$,
and
$h_\qdx:\mathbb{R}^\xDim\times\mathbb{R}^\uDim\times\integer\rightarrow\mathbb{R}^\yDim$ is the potentially nonlinear closed-form output function in $\region_\qdx$. 
All polytopes in $\region$ are disjoint, and their union is equal to $\mathbb{R}^{\xDim+\uDim}$.
\end{definition}

The following theorem and constructive proof constitute the primary contribution of this chapter.

\begin{theorem}[Closed-Form \ac{PWD} System Representation]
\label{prop:main}
Any PWD system (\ref{eqn:openForm}) can be represented in closed form as
\begin{equation}
\begin{aligned}
    \state(k+1) &= f\left(\state(k),\uMIMO(k),k\right)\\
    \yMIMO(k) &= h\left(\state(k),\uMIMO(k),k\right)
\end{aligned}
\end{equation}
where $f$ and $h$ explicitly encapsulate both the component dynamics and switching behavior of the PWD system.
\end{theorem}

\begin{proof}
Because $\region_\qdx$ are convex polytopes, their boundaries are hyperplanes, which can be represented via
\begin{equation}
    \planeVec^T
    \sysvec
    =
    \offsetNum
\end{equation}
where $\planeVec\in\mathbb{R}^{\xDim+\uDim}$ describes the orientation of the hyperplane and $\offsetNum\in\mathbb{R}$ is an offset. To capture the partitioning of $\mathbb{R}^{\xDim+\uDim}$ into $\regionQuant$ regions by $\planeQuant$ hyperplanes in the system dynamics, this work introduces the auxiliary logical state vector
\begin{align}
    \locvec(k) &= f_\locvec(\state(k),\uMIMO(k)) =
    H \left(
    P\sysvec - \offsetVec
    \right) 
    \label{eqn:delta}
    \\
    P &= \begin{bmatrix}
    \planeVec_1 & \planeVec_2 & \cdots & \planeVec_{\planeQuant}
    \end{bmatrix}^T
    \nonumber
    \\
    \offsetVec &= \begin{bmatrix}
    b_1 & b_2 & \cdots & b_{\planeQuant}
    \end{bmatrix}^T
    \nonumber
\end{align}
where $H$ is the element-wise Heaviside step function\footnote{
Some communities may not consider the Heaviside function to be closed-form. Many others do \cite{Chicurel-Uziel2001,Dona2014,Hetnarski1993,Kim2012,Ochmann2011,Bhattacharyya2010,Peng1998}.
This classification as closed-form is supported by its seamless integration into most symbolic math software, in which its derivative is well-defined.
% Here it is considered closed-form because of its unimpeachable compatibility with the mathematics of dynamical system analysis and control.
% and because the selector functions in this work can account for any convention for its value at the origin.
}
with convention $H(0)=1$, 
and the relationship between $\regionQuant$ and $\planeQuant$ depends on the exact configuration of the hyperplanes.
The $i^\textrm{th}$ element of $\locvec$ indicates on which side of the $i^\textrm{th}$ hyperplane
the system vector $\sysvecT$ falls (with points on hyperplanes being included in the space corresponding to $\planeVec_i^T\sysvecT-b_i>0$, see Remark \ref{rem:sgnvH}). In this manner, each $\region_\qdx$ uniquely corresponds to a particular value of $\locvec$, which will be denoted $\sigvec_\qdx$.

As the system evolves $\locvec$ will be compared to each $\sigvec_\qdx$ via a set of selector functions given by
\begin{equation}
    K_\qdx(\locvec(k)) = 0^{\norm{\sigvec_\qdx-\locvec(k)}} =
    \begin{cases}
    1 & \sigvec_\qdx  = \locvec(k) \\
    0 & \sigvec_\qdx \neq \locvec(k)
    \end{cases}
    \label{eqn:Kronecker}
\end{equation}
Note that this equation is equivalent to the Kronecker delta of $0$ and $\norm{\sigvec_\qdx-\locvec(k)}$ for any norm, but is left in the given zero exponential form for explicitness.

With these selector functions, the original state dynamics and outputs can be represented by the closed-form equations
\begin{equation}
\begin{aligned}
    \state(k+1) &= 
    f(\state(k),\uMIMO(k),k) =
    \sumqQ f_\qdx\left(\state(k),\uMIMO(k),k\right)K_\qdx\left(f_\locvec\left(\state(k),\uMIMO(k)\right)\right) 
\\
    \yMIMO(k) &=
    h(\state(k),\uMIMO(k)) = \sumqQ h_\qdx(\state(k),\uMIMO(k),k)K_\qdx(f_\locvec(\state(k),\uMIMO(k)))
\end{aligned}
\label{eqn:closedForm}
\end{equation}
This completes the closed-form representation of (\ref{eqn:openForm}).
A block diagram of the representation's structure is given in Figure \ref{fig:closedFormBlock}
\end{proof}

\begin{figure}
    \centering
    \includegraphics{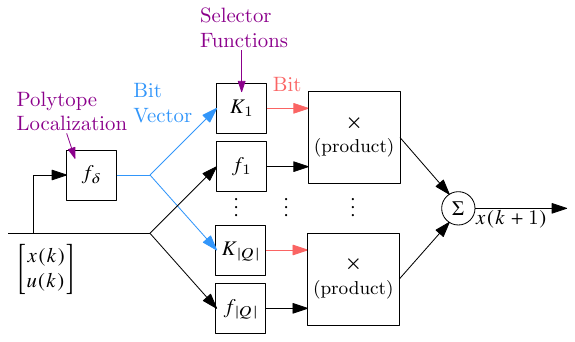}
    \caption{Block diagram of the closed-form \ac{PWD} system representation}
    \label{fig:closedFormBlock}
\end{figure}

\begin{remark}[Continuity]
This representation makes no assumptions regarding the continuity of either the state transition formula $f$ or the output function $h$ over the switching hyperplanes. The representation may model both continuous and discontinuous hybrid systems.
\end{remark}

\begin{remark}[Hyperplane Inclusion in Polytopes]
\label{rem:sgnvH}
The use of the Heaviside function in equation (\ref{eqn:delta}) systematically includes each hyperplane in the polytopes lying to a particular side. In other words, a hyperplane serving as a boundary/face for two adjacent polytopes is included in only one of them, such that one polytope is closed on that face and the other polytope is open. To which side the hyperplane belongs can be chosen by manipulation of $P$ and $\offsetVec$.
For example, ($P=[\,0,\,1\,],\,\offsetVec=0$) and ($P=[\,0,\,-1\,],\,\offsetVec=0$) both define the horizontal axis in $\real^2$.
Under (\ref{eqn:delta}), the former definition yields the same value of $\locvec$ ($\locvec=1$) for points on the hyper plane \emph{and} in top half of $\real^2$, while the bottom half of $\real^2$ \emph{excluding} the horizontal axis yields $\locvec=0$.
The latter definition yields $\locvec=1$ for the hyperplane \emph{and} the \emph{bottom} half of $\real^2$, while the top half of $\real^2$ \emph{excluding} the hyperplane yields $\locvec=0$.
In this manner, the hyperplane can be included in either the top or bottom polytope, but not both. The inability to include the hyperplane in both polytopes is desirable because if the polytopes have a non-null intersection there is uncertainty regarding the state dynamics in the intersection.

In many cases this 
method of determining which polytope contains a face (i.e. which polytope is closed) and which polytope is open at the same face
is
amply flexible, and ensures that there are an equal number of polytopes and $\sigvec$ values. However, if a system designer desires greater flexibility, the Heaviside function can be replaced with the signum function. In this case, each polytope interior and each face uniquely corresponds to a different $\locvec$, and the designer must assemble the complete polytope from an interior and the desired faces. For example, if $\region_\qdx$ corresponds to the interior and face associated with $\sigvec_{\qdx,1}$ and $\sigvec_{\qdx,2}$, then $K_\qdx(\locvec)=0^{\norm{\sigvec_{\qdx,1}-\locvec(k)}\norm{\sigvec_{\qdx,2}-\locvec(k)}}$.
\end{remark}

\begin{remark}[Location Non-Convexity and Switching Condition Nonlinearity]
Despite the convexity requirements of the $\mathbb{R}^{\xDim+\uDim}$ partitioning, and consequential affine-in-$\state$ requirement 
of the switching behavior governed by equations (\ref{eqn:delta}) and (\ref{eqn:Kronecker}),
the representation (\ref{eqn:closedForm}) can encompass a wide range of complex system structures.
There need not be a one-to-one correspondence between the quantity of discrete dynamic regimes and the quantity of polytopes in system (\ref{eqn:closedForm}). Indeed, while each polytope $\region_\qdx$ must be convex, one may construct many interesting nonconvex topologies by stitching together adjacent polytopes $\region_e$ and $\region_d$ by setting $f_e=f_d$, $h_e=h_d$. This is illustrated in Figure \ref{fig:nonConvex}.
Additionally, nonlinear switching conditions such as 
$\sin(\state^{i})>0$ (where $\state^{i}$ is the $i^\textrm{th}$ element of $\state$)
can obviously be accommodated by simply making the nonlinear expression of $\state$ into a new state. This of course increases the dimension of the system, but it notably does \emph{not} increase the dimension of the \ac{NILC} problem, which depends solely on the outputs and quantity of samples in a trial timeseries.
\end{remark}

\begin{figure}
    \centering
    \vspace{0.1in}
    \includegraphics[scale=0.6]{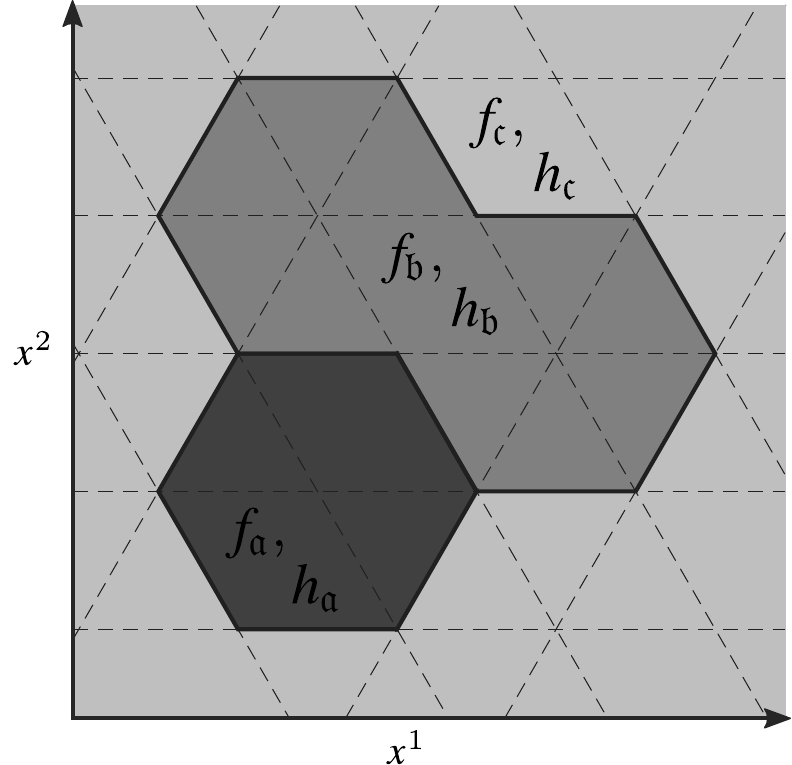}
    \caption{An $\mathbb{R}^2$ topology representable by (\ref{eqn:closedForm}) with 14 hyperplanes (dashed lines) having 1 convex 
    ($\mathfrak{a}$)
    and 2 nonconvex 
    ($\mathfrak{b}$ and $\mathfrak{c}$)
    dynamic regimes.}
    \label{fig:nonConvex}
\end{figure}

\section{Newton Iterative Learning Control}
\label{sec:NILC}

\ac{NILC}, first introduced in terms of abstract Banach space operators by \cite{Avrachenkov1998}, is the use of the Newton-Raphson root-finding algorithm to derive a trial-varying learning matrix $\learnMat_\tdx$ for the classical \ac{ILC} law
\begin{equation}
\label{eq:ILCclassic}
    \uvec_{\tdx+1} = \uvec_\tdx + \learnMat_\tdx \evec_\tdx
\end{equation}
where $\uvec$ is a control input timeseries vector, or ``lifted'' vector, $\evec$ is 
a lifted error vector,
and $\tdx\in\{0,1,\cdots\}$ is the trial index\footnote{
In the literature, $j$ or $k$ is usually used for the trial index. Here,
$\ell$ is used for the trial index because $i$ and $j$ will be used for matrix element indexing, $k$ is used for the discrete time index, $t$ is avoided to prevent confusion with continuous time, and $\ell$ is the next letter in the alphabet and thus commonly used for indexing.
}.

\ac{NILC} was first formalized for discrete-time state space systems by \cite{Lin2006}.
However, the treatment of \ac{NILC} in \cite{Lin2006} is limited in that it only considers time-invariant systems and is only applicable to systems with a relative degree, $\mu$, of 1.
For discrete-time systems, the relative degree of an output is the number of time steps that must transpire before the explicit representation of the output, in terms of inputs and initial conditions, contains any input \cite{Sun2001}. 
In many cases this notion of relative degree is adequate for the synthesis of \ac{NILC} from \ac{PWD} systems. However, there are also \ac{PWD} models for which alternative notions of relative degree are useful for managing issues related to the switching behavior of hybrid systems. A rigorous treatment of such issues and the definition of an alternative relative degree is given in Chapter \ref{ch:5}.
For now, the present
section presents an \ac{NILC} framework generalized for time-varying discrete-time SISO systems of any relative degree $\geq 1$. 
For all systems, the purpose of incorporating knowledge of the relative degree into \ac{NILC} synthesis is to guarantee that $\learnMat_\tdx$ is well-defined regardless of model relative degree. Specifically, $\learnMat_\tdx$ being well-defined is contingent on $\learnMat_\tdx$ being invertible, and any notion of relative degree yielding this invertibility is acceptable.

Consider the SISO, discrete-time 
time-varying 
nonlinear 
model
\begin{IEEEeqnarray}{RL}
\eqlabel{eqn:sys}
\IEEEyesnumber
\IEEEyessubnumber*
\xmod_{\tdx}(k+1) &= \fmod\left(\xmod_\tdx(k),u_\tdx(k),k\right)
\label{eq:modelstate}
\\
\ymod_\tdx(k) &= \hmod(\xmod_\tdx(k))
\label{eq:modelout}
\end{IEEEeqnarray}
where 
$\xmod\in\real^\xDim$, 
$u\in\real$, 
$\ymod\in\real$, 
$f:\mathbb{R}^\xDim\times\mathbb{R}\times\integer\rightarrow\mathbb{R}^\xDim$, and
$h:\mathbb{R}^\xDim\rightarrow\mathbb{R}$.
Hats, $\hat{}\,$, are used to emphasize that \originalmod{} is an imperfect model of some true system,
though it is assumed that the control input and initial condition are perfectly known.

Let the system describe a repetitive process with finite duration. This translates to the assumptions:
\begin{enumerate}[label=(A3.\arabic*),leftmargin=*]
\item 
\label{A1c3}
A trial must have a duration of $N$ time steps, \\
$k\in\{0,1,\cdots,N\}$.
\item
\label{A2c3}
The initial conditions are trial invariant, \\
$\xmod_\tdx(0)=\state_0 \quad \forall \tdx$.
\item
\label{A3c3}
The desired output values, $\refr(k)\in\mathbb{R}$, are trial invariant.
\suspend{ch3Ass1}
\end{enumerate}

By \ref{A1c3}, the input and output trajectories can be represented as timeseries vectors of the same length for every trial:
\begin{align}
    \yLiftmod_\tdx &= 
    \begin{bmatrix}
    \ymod_\tdx(\mu) & \ymod_\tdx(\mu+1) & \cdots & \ymod_\tdx(N)
    \end{bmatrix}^T \in \mathbb{R}^{N-\mu+1}
    \\
    \uvec_\tdx &= 
    \begin{bmatrix}
    u_\tdx(0) & u_\tdx(1) & \cdots & u_\tdx(N-\mu)
    \end{bmatrix}^T \in \mathbb{R}^{N-\mu+1}
    \label{eq:uvec}
\end{align}
Note that $\yLiftmod_\tdx$ is time-shifted forward by the relative degree $\mu$, which is a function of the specific definitions of $f$ and $h$.

By \ref{A2c3} and (\ref{eqn:sys}), if $\state_0$ is known then $\yLiftmod_\tdx$ is entirely a function $\uvec_\tdx$ given by 
$\gmod:\mathbb{R}^{N-\mu+1}\rightarrow\mathbb{R}^{N-\mu+1}$
\begin{IEEEeqnarray}{RL}
\eqlabel{eq:gq}
\IEEEyesnumber
\IEEEyessubnumber*
    \yLiftmod_\tdx^i &= \gmod^i\left(\uLift_\tdx\right)=\ymod_\tdx(\mu+i-1)
    \\
    \label{eq:ymodk}
    \ymod_\tdx(k) &= \hmod\left(\fmod^{(k-1)}(\uvec_\tdx)\right) 
    \qquad k\in\{\mu,\mu+1,\cdots,N\}
\end{IEEEeqnarray}
where the non-parenthetical
superscript $i$ denotes the $i^\textrm{th}$ element of a vector, indexing from 1, 
and the 
parenthetical superscript $(k)$ denotes function composition of the form
\begin{align}
    \fmod^{(k)}(\uvec_\tdx)&=\fmod(\xmod_\tdx(k),u_\tdx(k),k)
    =
    \fmod\left(\fmod\left(\vphantom{\fmod}\cdots,u_\tdx(k-1),k-1\right),u_\tdx(k),k\right)
    \label{eqn:recursion}
\end{align}

\begin{sloppypar}
The recursion of the state dynamics $f$ expressed by (\ref{eqn:recursion}) has a terminal condition of
$\fmod(\xmod_\tdx(0),u_\tdx(0),0)$.
Because 
$\xmod_\tdx(0)=\state_0$ is known 
in advance
and the time argument is determined by the element index of the lifted representation, $\yvecmod_\tdx$ is a function of only $\uvec_\tdx$. Note that because the first element of $\yvecmod_\tdx$ is $\ymod_\tdx(\mu)$ it explicitly depends on $\uLift^1=u_\tdx(0)$.
\end{sloppypar}

Similarly, by \ref{A3c3}, the reference
\begin{equation}
    \rLift = 
    \begin{bmatrix}
    \refr(\mu) & \refr(\mu+1) & \cdots & \refr(N) 
    \end{bmatrix}^T \in\mathbb{R}^{N-\mu+1}
\end{equation}
is fixed, causing the output error timeseries, $\evec_\tdx$ to be approximable by a function of only $\uvec_\tdx$. This relationship is given by
\begin{equation}
    \evec_\tdx(\uvec_\tdx) = \rLift - \yLift_\tdx 
    = \rLift - \gtrue(\uLift_\tdx)
    \approx \rLift - \gmod(\uvec_\tdx)
    \label{eqn:err}
\end{equation}
where $\gtrue:\real^{N-\mu+1}\rightarrow\real^{N-\mu+1}$ represents the unknowable dynamics of the true system, and
$\yLift_\tdx$ is the measured value of this true system output.
Newton's method can thus be applied to iteratively find an argument $\uvec_\tdx$ bringing
$\evec_\tdx(\uvec_\tdx)$
toward 0. This naturally yields a control law of the form 
(\ref{eq:ILCclassic}) with
\begin{equation}
    \learnMat_\tdx=\left(\jacobian{\gmod}{\uvec}(\uvec_\tdx)\right)^{-1}
    \label{eq:gammaprior}
\end{equation}
where $\jacobian{\gmod}{\uvec}$ is the Jacobian (in numerator layout) of $\gmod$ with respect to $\uvec$  as a function of $\uvec$.

However, \directAILC{} is not strictly Newton's method applied to $\evec_\tdx$ because $\gmod(\uvec_\tdx)$ is, like all models, only an approximation of the true physical dynamics $\gtrue(\uLift_\tdx)$.
Because of this, convergence of $\yLift_\tdx$ to $\rLift$ is linear (rather than quadratic) and 
guaranteed if the following 
sufficient conditions
are satisfied \cite{Avrachenkov1998}.
\begin{enumerate}[label=(A3.\arabic*),leftmargin=*]
\resume{ch3Ass1}
    \item
    \label{A4c3}
    The true dynamics $\gtrue(\uvec_\tdx)$ are continuously differentiable 
    and their Jacobian $\jacobian{\gtrue}{\uLift}$ is Lipschitz continuous with respect to $\uLift$
    within some ball around the solution trajectory $\uLift_\text{soln}$.
    
    \item
    The inverse of the model Jacobian (i.e. the learning matrix) always has a bounded norm: 
    \\
    $\norm{L_\tdx}_2<\varepsilon_1\in\real_{>0} \,\,\forall\,\tdx$    
    
    \item
    \label{A6c3}
    The learning matrix is sufficiently similar to the inverse of the true lifted system Jacobian:
    \\
    $\norm{I-L_\tdx\jacobian{\gtrue}{\uvec}(\uvec_\tdx)}_2<1 \,\,\forall\,\tdx$
    
    \item
    \label{A7c3}
    The initial guess $\uvec_0$ is in some basin of attraction around the solution $\uvec_\text{soln}$ (guaranteed to exist by \ref{A4c3}-\ref{A6c3}),\\
    $\norm{\uvec_0-\uvec_\text{soln}}_2<\varepsilon_2\in\mathbb{R}_{>0}$
    \suspend{ch3Ass2}
\end{enumerate}
Note that the continuous differentiability and Lipschitz continuity conditions are on the sufficient condition for the true system, not the model.

In the implementation of this controller, the Jacobian as a function of $\uvec_\tdx$ only needs to be derived once 
(in advance
of trial 0). Still, large sample quantities (i.e. long trial durations and/or high temporal resolution) have historically made deriving this matrix a significant computational burden.
This has forced past authors to approximate $\jacobian{\gmod}{\uLift}$ using techniques ranging from the 
coarse
and simple chord method \cite{Avrachenkov1998}, to more accurate but computationally expensive methods like gradient-based optimization \cite{Freeman2014}.
However, advances in automatic differentiation techniques, such as those in the software tool CasADi \cite{Andersson2018}, now enable rapid derivation of Jacobian functions that are exact to nearly machine precision \cite{Naumann2008}.
This largely eliminates the need for Jacobian approximation methods in many scenarios where a system model is available.

Finally, note that the class of systems (\ref{eqn:openForm}) treated by the closed-form representation in this work is 
broader
than the class of systems (\ref{eqn:sys}) to which the controller may be applied, and that (\ref{eqn:sys}) itself is broader than the class of systems guaranteed to satisfied \ref{A4c3}-\ref{A6c3}.
However, this work chooses to present a broad closed-form representation both because there may be other controllers for which the representation is useful, and because there exists systems of class (\ref{eqn:openForm}) that do converge under the given controller and are not contained by narrower popular subclasses such as piecewise affine and affine-in-the-input systems. An example of such a system is given in Section \ref{sec:validation}.

\section{Validation}
\label{sec:validation}

This section demonstrates both the closed form system representation and \ac{NILC} through the simulated control of a nonlinear mass-spring-damper system over an increasing degree of mismatch between the truth model and control model.

\subsection{Model and Software Implementation}
\label{sec:model3}
\begin{figure}
    \centering
    \includegraphics{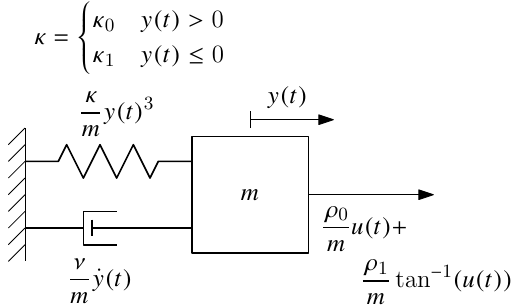}
    \caption{Mass-spring-damper system used for validation. The spring stiffness is a piecewise defined function of the spring extension and the applied force is a nonlinear function of actuator voltage $u$.}
    \label{fig:MSD}
\end{figure}

The truth model is given by the continuous-time mass-spring-damper system pictured in Figure \ref{fig:MSD}. The system's equation of motion is
\begin{align}
    \ddot{y}(t) &= 
    -\frac{\kappa}{m}y(t)^3 - \frac{\nu}{m}\dot{y}(t) + 
    \frac{\rho_0}{m} u(t) + \frac{\rho_1}{m}\tan^{-1}(u(t))
    \label{eqn:truth}
    \\[0.5em]
    \kappa &=
    \begin{cases}
    \kappa_0 & y(t) > 0 \\
    \kappa_1 & y(t) \leq 0
    \end{cases}
    \label{eqn:out_and_switch}
\end{align}
where $y(t)$ is the displacement of the mass from the neutral position (positive in direction of spring extension), $u(t)$ is an applied actuator voltage,  $m$ is the mass, $\kappa$ is the stiffness coefficient, $\nu$ is a damping coefficient, and $\rho_0$, $\rho_1$ are constant coefficients mapping from applied voltage to applied force. The stiffness coefficient depends on whether the system is in compression (softer) or extension (stiffer), and the spring hardens with displacement in either direction due to the cubic exponent on $y(t)$ in (\ref{eqn:truth}).
The ``measurements'' given to the controller after each trial are generated from this model using Runge Kutta 4-step integration and a step size of $T_{s}=\SI{0.001}{\second}$.

\begin{sloppypar}
The controller itself requires a discrete time model, which is derived from (\ref{eqn:truth}) using the forward Euler method. This yields 
output and switching
equations identical to those of the continuous-time system, but the state dynamics become
\end{sloppypar}
\begin{align}
    \hat{x}(k+1) &= \begin{bmatrix}\hat{y}(k+1) & \hat{y}(k+2)\end{bmatrix}^T 
    \end{align}
    where
    \begin{align}
    \hat{y}(k+2) &= 
    2\hat{y}(k+1) - \hat{y}(k) + 
    \frac{\hat{T}_{s}^2}{\hat{m}}\left( 
    - \hat{\kappa} \hat{y}(k)^3 - \frac{\hat{\nu}}{\hat{T}_{s}} \hat{y}(k+1) +
     \hat{\rho}_0 u(k) +  \hat{\rho}_1 \tan^{-1}(u(k))
    \right)
\end{align}
where $\hat{T}_{s}=\SI{0.01}{\second}$, an order of magnitude coarser than the truth model. 
This model is nonlinear in the states and the input, has relative degree $\mu=2$, 2 dynamic regimes and
can be represented in closed form with one auxiliary variable. The hyperplane corresponding to this auxiliary 
variable
is given by
\begin{equation}
    P = \begin{bmatrix}
    -1 &  \phantom{-}0 & \phantom{-}0 \\
    \end{bmatrix}
    \qquad
    \offsetVec = \begin{bmatrix}
    0 
    \end{bmatrix}
\end{equation}
which gives rise to 2 polytopes.

Note that because neither the switching function nor the output function in equation (\ref{eqn:out_and_switch}) depends on the input, and because the system is SISO, this system's closed from representation (\ref{eqn:closedForm}) is equivalent to (\ref{eqn:sys}), the system form accepted by the ILC.

Assumptions \ref{A1c3}-\ref{A3c3} are enforced as follows.
Each trial lasts 9.5 seconds (approximately $3\pi$ seconds), making the length of $\gmod(\uvec_\tdx)$ equal to $N-\mu+1=949$.
The trial invariant initial conditions are 0 for both $\hat{y}(k)$ and $\hat{y}(k+1)$.
The trial invariant reference is chosen to be
\begin{equation}
    % \yvec_d
    \refr(k)
    =\begin{cases}
    \sin((k-1)\hat{T}_{s}) & k>0 \\
    0                     & k=0
    \end{cases}
\end{equation}
where the delay in the sinusoid and the leading 0 ensure that the reference does not demand that the controller alter the initial conditions.

For the derivation of the control law \directAILC{} from the closed form system representation, this work uses the open source software tool CasADi. CasADi is chosen because its combination of symbolic framework, sparse matrix storage, and automatic differentiation make for the rapid computation of exact $\gmod(\uvec_\tdx)$ and $\jacobian{\gmod}{\uLift}$ functions. However, the Heaviside function in CasADi uses the trinary ``0.5-at-origin'' convention rather than the binary ``1-at-origin'' convention,
making it equivalent to signum in the context of determining 
on which side of a hyperplane a point lies. 
Because of this, selector functions in CasADi must be written by assembling the polytope interiors and faces as in Remark \ref{rem:sgnvH}. Here this assembly is done to be equivalent to a ``1-at-origin'' Heaviside convention.

\subsection{Methods}
\label{sec:methods3}

There are two simulation design objectives. The first is to assess the utility of the closed form system representation in enabling the control of a hybrid system via a controller developed for non-hybrid systems. The second is to assess the sensitivity of the combination of \ac{NILC} with this system representation. Here, ``sensitivity'' is characterized by the likelihood that the controller will diverge given a particular degree of modeling error.
Both of these goals are accomplished with simulations wherein the controller is synthesized via increasingly erroneous models and its convergence behavior is analyzed.

In this study, a ``simulation'' is an attempt to produce $\rLift$ within 20 process trials under \ac{NILC} for a particular set of mismatched truth and control model parameters.
The control model is identical for all simulations, with parameters given by the vector
\begin{align}
    \thetamod &=
    \begin{bmatrix} 
    \hat{m} & \hat{\rho}_0 & \hat{\rho}_1 & \hat{\nu} & \hat{\kappa}_0 & \hat{\kappa}_1
    \end{bmatrix}^T
    \\&=
    % =
    \begin{bmatrix}
     \phantom{i}1 &
     \phantom{i}1 & 
     \phantom{i}1 & 
     \phantom{i}1 & 
     \phantom{i}4 &
     \phantom{i}1\phantom{n}
    \end{bmatrix}^T
\end{align}
The truth model parameters are 
perturbations of the control model parameters by a random relative error. Mathematically this is
given by the random vector
\begin{equation}
    \theta=\left(
    1_{6\times 1}
    \modelErrVec^T\odot I+I
    \right)\thetamod
\end{equation}
where $1_{6\times 1}\in\mathbb{R}^6$ is a vector with every element equal to 1, $\modelErrVec\in\mathbb{R}^{6}$ is a random vector, $\odot$ is the matrix Hadamard product, and $I$ is the identity matrix. Each element of $\modelErrVec$ is the (positive or negative) relative error
\begin{equation}
    \modelErrVec^{i}=\frac{\theta^{i}-\thetamod^{i}}{\thetamod^{i}}
\end{equation}
and is regenerated for each simulation.
A scalar parameter describing the degree of modeling error can thus be given by $\modelError$.

Simulations are organized with respect to increasing $\modelError$. There are \setQuant{} sets of \expQuant{} simulations each. Sets are defined by bounds on $\modelError$, and each simulation in a set is characterized by a random $\modelErrVec$ such that $\modelError$ falls within the specified bounds. These bounds are given by 0.05 increments from 0 to 1.
For the sake of comparison, one additional simulation is run with zero modeling error ($\modelError=0$).
To ensure the model mismatch is truly 0, the truth model in this simulation is taken to be the forward Euler discrete model rather than a Runge Kutta integration of the continuous model.

The efficacy of the controller in a simulation is quantified by the normalized root mean square error (NRMSE) between $\rLift$ and $\yvec_{20}$. The RMSE is normalized by the amplitude of $\rLift$, which in this work is 1. Thus the normalization does not alter the numerics here, but it formally nondimensionalizes the results.
$\yvec_\tdx$ is said to have ``completely'' converged to $\rLift$ if this NRMSE is less than 0.005.

Holistic analysis of the integration of the closed-form system representation and ILC is based on the spread of NRMSE values over all sets and the average trajectory of NRMSE versus trial number $\tdx$ for all completely convergent simulations. Comparisons between simulation sets are made by the percentage of simulations within each set that completely converge.

\subsection{Results \& Discussion}
\label{sec:results3}

For the controller derivation, on a desktop computer with \RAM{}
of RAM and a \CPU{} CPU, CasADi performs the repeated function composition to acquire $\gmod$ and the differentiation to acquire the \gLength-by-\gLength{} Jacobian $\jacobian{\gmod}{\uLift}$ in \casadiTime{}. 
Once these functions are derived for the first time, before the first learning operation, they need not be derived again, and can be called ordinarily to generate the trial-varying learning matrix for each subsequent trial.
MATLAB symbolic toolbox is used as a traditional symbolic math tool for comparison on the same task. MATLAB fails to complete the task, running out of memory after \matlabTime{}. This illustrates the substantial computational advantage of automatic differentiation in this context, and its ability to enable direct application of Newton's method in ILC.

\begin{figure}
    \centering
    \includegraphics[scale=0.84]{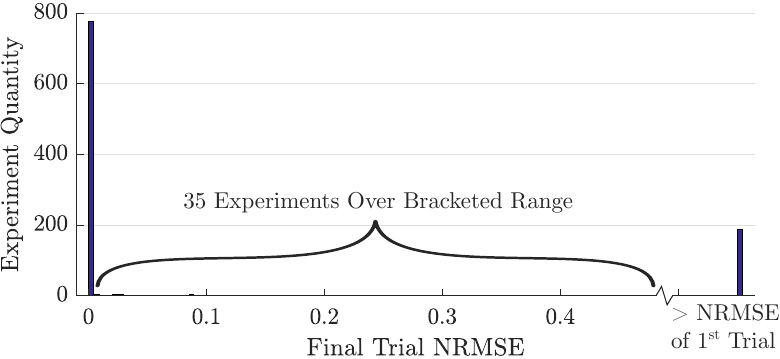}
    \\
    \hspace{1mm}
    \includegraphics[scale=1.12]{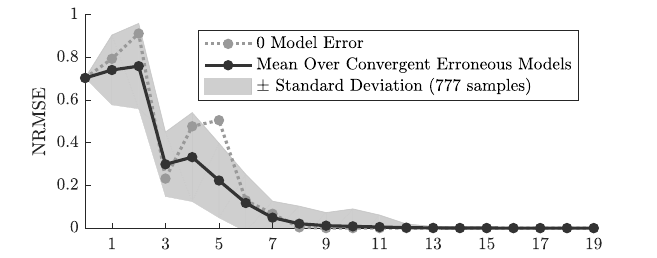}
    \\
    \hspace{0.9mm}
    \includegraphics[scale=1.12]{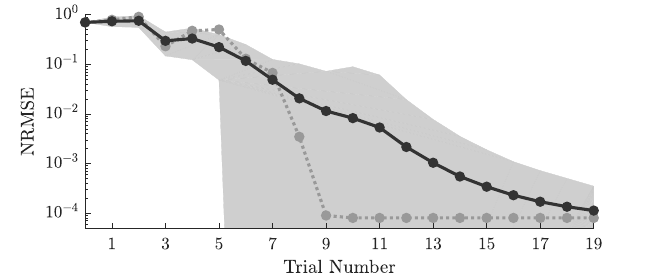}    \caption{\emph{Above:} Histogram of final-trial NRMSE for all experiments. All but the highest error bin are defined by a 0.005 range of NRMSE. \emph{Middle:} Average trajectory of NRMSE vs. Trial number for all completely convergent simulations. \emph{Bottom:} Same on a logarithmic scale, for perspective.}
    \label{fig:convTraj}
\end{figure}

To confirm the efficacy of the combination of \ac{NILC} and the closed-form hybrid system representation, Figure \ref{fig:convTraj} presents the trajectory of output convergence over trial iteration.
The controller reaches the complete convergence threshold in \convZero{} trials under 0 model error, in \convMean{} trials for the mean trajectory over all convergent simulations with erroneous models, and in \convEnv{} trials for a simulation with one standard deviation greater NRMSE than the mean trajectory. This illustrates that while increases in model error can slow convergence, this retardation is modest.
Additionally, it can be noted that nearly all experiments (96.5\%) either completely converge or diverge entirely, yielding final trial NRMSE values in excess of the zero-input NRMSE.
This is illustrated by the large gap in the Figure \ref{fig:convTraj} histogram between the final-trial RMSEs.
The remaining experiments are either almost convergent (likely requiring several more trials to completely converge), or are divergent but with the ``blowing up'' limited to very few points at the very end of the timeseries, thereby having only a small affect on NRMSE.
Finally, to supplement above statistical analysis and to show that the converged simulations indeed yield qualitatively reasonable input trajectories, the timeseries data for a representative simulation is shown in Figure \ref{fig:repTS}.

\begin{figure}
    \centering
    \includegraphics[scale=0.8]{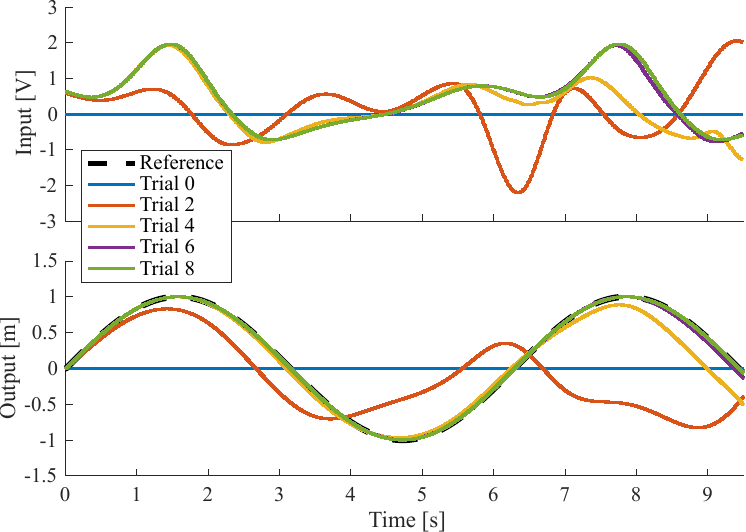}
    \caption{Timeseries evolution over multiple trials for a simulation with $\modelError=0.27$. This simulation took 8 trials to completely converge.}
    \label{fig:repTS}
\end{figure}

The results in the preceding paragraph follow the convergence behavior predicted by \cite{Avrachenkov1998} for a generic input-output model under \ac{NILC} with static modeling error. The adherence of this work's results to prior predictions illustrates that, as expected, the hybrid nature of the example system does not intrinsically necessitate special accommodation in the controller. Instead, all that is needed to apply the controller is the closed-form input-output model, the construction of which is enabled by the closed-form piecewise defined system representation presented.

\begin{figure}
    \centering
    \includegraphics[scale=0.8]{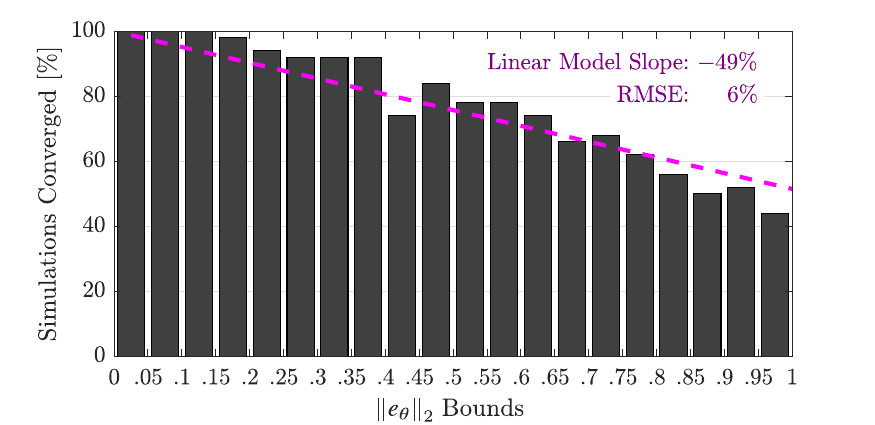}
    \caption{Percentage of experiments that converge for each 0.05 range of relative model error from 0 to 1. The dashed line is a least squares model of the decay in the probability of convergence as model error increases.}
    \label{fig:convProb}
\end{figure}

Finally, to practically evaluate \ref{A6c3}, Figure \ref{fig:convProb} gives the percentage of experiments that converge within each 
experiment set.
Up to $\modelError=0.15$, 100\% of the experiments converge. Beyond a relative model error of 0.15, the convergent experiment percentage decays in a reasonably linear fashion. Taking the conservative assumption that the probability of experiment convergence begins to decay for model error greater than 0, a linear least squares model yields a decay rate of $-49\%$ convergence probability per unit increase in relative model error (or $-0.49$ percent convergence probability per unit increase in percent model error), with an RMSE of $6\%$ between the decay model and data.

Naturally, the convergence probability of other systems may behave differently under increasing $\modelError$, and the above analysis clearly does not constitute a definitive theory. However, because the set of values of $\modelErrVec$ for which \ref{A6c3} is satisfied is usually impossible to compute in practice (because $\gtrue$ is unknown), it is important to have practical references for understanding system convergence. Because mass-spring-damper-like oscillators are ubiquitous across most fields of engineering, the above analysis serves as the first such point of reference for \ac{NILC}.

\section{Conclusion}
\label{sec:conc}
This chapter contributes a closed-form representation of piecewise defined dynamical systems.
The utility of having this representation is demonstrated by the application to a hybrid system of an iterative learning controller based on Newton's method, which was hitherto impossible because of the controller's need for a dynamical system model supporting function composition and Jacobian operations.

Additionally, the controller derivation itself is formally generalized for systems of relative degree greater than 1, 
and for time-varying systems. Additionally,
 automatic differentiation is used to enable faster and more accurate implementation of Newton's method than prior works, which required coarser and/or more computationally expensive approximation of the model Jacobian.

Future work for the closed-form system representation revolves around expanding the representation to
facilitate specification of more complicated switching logic and dynamical state resets upon discrete transitions.
While all discrete-time hybrid automaton topologies may be captured by piecewise definition \cite{Torrisi2004}, doing so often requires augmentation of the state vector and may make model synthesis more challenging.

\chapter{
Nonlinear Systems with Unstable Inverses:
\texorpdfstring{\\}{ }
Invert-Linearize ILC and Stable Inversion
}
\blfootnote{
Content of this chapter also published as:
\\
\indent\indent
I. A. Spiegel, N. Strijbosch, T. Oomen and K. Barton, ``Iterative Learning Control with Discrete-Time Nonlinear Nonminimum Phase Models via Stable Inversion,'' in \emph{International Journal of Robust and Nonlinear Control.} pp. 1-22. Aug. 2021,
\href{https://doi.org/10.1002/rnc.5726}{https://doi.org/10.1002/rnc.5726}
\textcopyright John Wiley \& Sons Ltd 2021. Reprinted with permission.
}
\label{ch:4}

Recall from Chapter \ref{ch:1} that \ac{NILC} is chosen as the foundational \ac{ILC} technique for this dissertation because of its fast convergence rate and because it is synthesizable from a broader range of systems than is considered by other publications  on \ac{ILC} synthesis from nonlinear models.
Specifically, other published \ac{ILC} laws are subject to at least one of the following restrictions on the control model:
\begin{enumerate}[label=(R\arabic*)]
\item
having relative degree of either 0 or 1 \cite{Saab1995,Wang1998},
\item
being affine in the input \cite{Jang1994,Saab1995,Wang1998,Chi2008},
\item
being time-invariant \cite{Jang1994,Sun2003}, and
\item
being smooth (Lipschitz continuous at the most relaxed) \cite{Jang1994,Saab1995,Wang1998,Sun2003,Chi2008}
\end{enumerate}
from all of which \ac{NILC} is free.
Models used to synthesize \ac{NILC} are, however, subject to the restriction of
\begin{enumerate}[label=(R\arabic*)]
\setcounter{enumi}{4}
\item
having a stable inverse.
\label{R4}
\end{enumerate}
Note that the non-\ac{NILC} prior art \cite{Jang1994,Saab1995,Wang1998,Sun2003} 
also suffers from \ref{R4}. While the prior art does not explicitly reveal this shortcoming, evidence of it is given in the appendix.
The ultimate goal of this chapter is a new \ac{ILC} framework inheriting the benefits of \ac{NILC} while surmounting this shortcoming.

Recall too that stable inversion is considered as a foundation for the treatment of unstable inverses in this dissertation, and that for linear systems stable inversion can be summarized as the forward-in-time evolution of a system's stable modes and the backward-in-time evolution of its unstable modes.

Extension of stable inversion to nonlinear models
involves additional complexities.
Some of these challenges, e.g. the difficulty of completely decoupling the stable and unstable parts of a nonlinear system, have been addressed by works such as \cite{Devasia1996, Devasia1998} 
for continuous-time systems and \cite{Zeng2000} for discrete-time systems.
However, 
the following
challenges remain.
First, for
nonlinear models, stable inversion requires 
Picard iteration
to solve a fixed-point-finding problem, and the computational complexity of the solution grows exponentially with the number of iterations.
Thus it is desirable to reach a satisfactory point in as few iterations as possible.
Secondly, the
quality,
i.e. proximity to the solution,
of the initial 
guess
strongly influences these early iterates, but the 
initial Picard iterate prescribed
by \cite{Zeng2000}---the zero state 
trajectory---can be improved upon
for many representations of practical systems, such as those employing both feedback and feedforward control. In fact, the convergence proof in \cite{Zeng2000} relies on an assumption that precludes 
many representations of these systems.
First, this prior art assumes that if the state and input are both zero at a particular time step, then the state will be zero at the next time step.
This is not true for most representations of systems employing both feedback and feedforward control
because the reference becomes a time-varying parameter embedded in the dynamics of the closed loop system. The reference thus drives state change via the feedback controller despite the initial state and feedforward input being zero. 
Stable inversion erroneously based on this assumption can have poor performance, and stable inversion has not been proven to converge when this assumption is relaxed.
Secondly,
\cite{Zeng2000} lacks discussion of the translation from the theoretical solution on a bi-infinite timeline to an implementable solution on a finite timeline, let alone a validation of such.
This chapter addresses these challenges.

In short, while the work to date on \AILC{} and stable inversion has made great strides, 
gaps remain between the prior art and
a synthesis scheme for ILC that is fast and 
applicable to a wide variety of models---including nonlinear \nonminphase{} models.
This chapter first contributes mathematical analysis concretely identifying the failure mechanism of \AILC{} when synthesized from 
 models with unstable inverses. 
This leads to the core contribution of a novel ILC framework for controlling nonlinear, \nonminphase{} systems. The key elements of this framework are
\begin{itemize}[leftmargin=*]
    \item
    restructuring of the linearization and model inversion processes in \AILC{} to circumvent issues associated with matrix inversion,
    \item
    reformulation of the model inversion in \AILC{} as stable inversion,
    \item
    proof of stable inversion convergence with relaxed assumptions on state dynamics, enabling treatment of a wider array of feedback control and other time-varying models,
    and
    \item
    development of a structured method for implementing the stable inversion technique proposed in this work.
\end{itemize}
As a final contribution, the proposed framework is validated in simulation on a nonlinear, relative degree 2, time-varying, \nonminphase{} cart-and-pendulum system with model error and process and measurement noise.

The remainder of the chapter is organized as follows. 
Section \ref{sec:back} provides technical details from the prior art in
nonlinear
stable inversion \cite{Zeng2000} necessary to present the novel contributions of the present work. 
Section \ref{sec:AILC} 
presents analysis that justifies
the attribution of a class of \AILC{} failures to 
inverse instability, 
and provides a new 
ILC framework
that enables the circumvention of this failure mechanism by incorporating stable inversion.
Section \ref{sec:StabInv} provides proof of convergence of stable inversion for an expanded class of systems and provides improved methods for practical implementation.
Section \ref{sec:Val} details and discusses the validation of the new 
ILC
framework
with stable inversion through benchmark simulations on an \nonminphase{} cart-and-pendulum system.
This includes demonstration of 
conventional
\AILC's divergence when applied to the same system.
Section \ref{sec:conc4} presents conclusions and areas for future work.

As a final note: because this chapter seeks to present a new \ac{ILC} framework that is general enough for both hybrid systems and non-hybrid nonlinear systems, this chapter does not explicitly discuss hybrid systems.
Furthermore, the stable inversion techniques discussed in this chapter are focused on smooth nonlinear systems. Consequently, the example system used for validation in this chapter is a smooth nonlinear system.
The stable-inversion-based solution to inverse instability will be brought to bear on hybrid systems in
Chapter \ref{ch:5}.

\section{Nonlinear Stable Inversion Background}
\label{sec:back}

The first step of stable inversion is 
deriving the conventional inverse.
To synthesize a minimal inverse system representation, first assume \originalmod{} is in the normal form
\begin{IEEEeqnarray}{RLR}
\eqlabel{eq:simmuall}
\IEEEyesnumber
\IEEEyessubnumber*
\label{eq:simmuminus}
    \xmod^{i}(k+1) &= \xmod^{i+1}(k)
    &
    \qquad
    i<\mu
\\
\xmod^{i}(k+1) &= \fmod^{i}\left(\xmod(k),u(k),k\right)
    &
    \qquad
    i\geq\mu
    \label{eq:simmu}
    \\
    \label{eq:simoutput}
    \ymod(k)&=\xmod^1
    &
\end{IEEEeqnarray}
where $\xmod(0)=0$, and the superscripts $i$ indicate the vector element index, starting from 1.
Note the ILC trial index subscript $\tdx$ is omitted in this section,  
as stable inversion on its own does not involve incrementing $\tdx$. 
Equation (\ref{eq:simmuminus}) captures the time delay arising from the system relative degree, while equation (\ref{eq:simmu}) captures the remaining system 
dynamics.
One method of deriving 
this normal form from a system not in normal form is given in
\cite{Eksteen2016}.

Given this normal form,
use (\ref{eq:simoutput}) to
replace the first $\mu$ state variables with output variables via
\begin{align}
    \xmod^{i}(k) &= \ymod(k+i-1) \qquad i\leq\mu
    \label{eq:sub1}
\end{align}
Similarly, replace the $\mu$\textsuperscript{th} state variable incremented by one time step 
(i.e. the left side of (\ref{eq:simmu}) for $i=\mu$)
with an output variable via
\begin{align}
    \xmod^\mu(k+1)&= \ymod(k+\mu)
    \label{eq:sub2}
\end{align}
These substitutions are made to facilitate the inversion of system
(\ref{eq:simmuall}),
as the inverse of a system with relative degree $\mu\geq 1$ is necessarily acausal with dependence on some subset of $\{\ymod(k),\,\ymod(k+1),\cdots,\ymod(k+\mu)\}$ at each time step $k$.
For notational compactness, define the $\ymod$-preview vector $\ypreview(k)\coloneqq[\ymod(k),\cdots,\ymod(k+\mu)]^T$.
Then inverting (\ref{eq:simmu}) with $i=\mu$ yields the conventional inverse output function
\begin{equation}
    u(k) = \fmod^{\mu}{\vphantom{\fmod}}^{ -1 }\left(
    \begin{bmatrix} \xmod^{\mu+1}, & \cdots, & \xmod^{n_x}  \end{bmatrix}^T,
    \ypreview(k),
    k\right)
    \label{eq:invoutx}
\end{equation}
where $\fmod^{\mu}{\vphantom{\fmod}}^{ -1 }$ is the inverse of $\fmod^{\mu}$, i.e. (\ref{eq:simmu}, $i=\mu$) solved for $u(k)$.
This output equation is substituted into 
(\ref{eq:simmu}) with $i>\mu$
along with (\ref{eq:sub1})-(\ref{eq:sub2})
to yield the entire inverse state dynamics
\begin{IEEEeqnarray}{RL}
\eqlabel{eq:invall}
\IEEEyesnumber
\IEEEyessubnumber*
\etamod(k+1) &= \fmod_{\eta}\left(\etamod(k),\ypreview(k),k\right)
    \label{eq:invstate}
    \\
    u(k)&=\fmod^{\mu}{\vphantom{\fmod}}^{-1}\left(\etamod(k),\ypreview(k),k\right)
    \label{eq:invout}
\end{IEEEeqnarray}
where $\etamod\in\real^{n_{\eta}}$ ($n_{\eta}=n_x-\mu$) is the inverse state vector defined
\begin{equation}
    \etamod^i(k)\coloneqq\xmod^{\mu+i}(k)
\end{equation}
and $\fmod_\eta:\real^{n_{\eta}}\times\real^{\mu+1}\times\integer\rightarrow\real^{n_{\eta}}$ is the inverse state dynamics 
\begin{equation}
    \fmod_\eta^i(\etamod(k),\ypreview(k),k)\coloneqq \fmod^{i+\mu}(\xmod(k),u(k),k)
\end{equation}

Next, this inverse system is to be similarity transformed to decouple the stable and unstable modes of its linearization about the initial condition. Consider the Jacobian
\begin{equation}
A=\jacobian{\fmod_{\eta}}{\etamod}\left(\etamod=0,\ypreview=\ypreview^\dagger,k=0\right)
\end{equation}
where $\ypreview^\dagger$ is the solution to $\fmod_{\eta}(0,\ypreview^\dagger,0)=0$. 
Then let $V$ be the similarity transform matrix such that
\begin{equation}
    \Asim=V^{-1}AV=
    \begin{bmatrix}
    \Asim_\stab & 0
    \\
    0
    & \Asim_\unstab
    \end{bmatrix}
    \label{eq:transform}
\end{equation}
where $\Asim_\stab\in\real^{\evalq\times \evalq}$ has all eigenvalues inside the unit circle, and $\Asim_\unstab\in\real^{n_{\eta}-\evalq\times n_{\eta}-\evalq}$ has all eigenvalues outside the unit circle. This can be satisfied by deriving the real block Jordan form of $A$. 
The corresponding inverse system state dynamics are
\begin{equation}
    \etamodsim(k+1)=\fetasim\left(\etamodsim(k),\ypreview(k),k\right)
    \coloneq V^{-1}\fmod_{\eta}\left(V\etamodsim(k),\ypreview(k),k\right)
    \label{eq:fetasim}
\end{equation}
where the tilde on $\fetasim$ indicates application to $\etamodsim$ rather than $\etamod$.
Note that despite using a linearization-derived linear similarity transform, (\ref{eq:fetasim}) describes the same nonlinear time-varying dynamics as (\ref{eq:invstate}), but with the linear parts of the stable and unstable modes decoupled.

If \originalmod{} 
has an unstable inverse, 
then (\ref{eq:fetasim}) is unstable and $\etamodsim(k)$ will be unbounded as $k$ 
increases.
However, given an infinite timeline in the positive and negative direction, the equation
\begin{equation}
\label{eq:implicitsoln}
    \etamodsim(k) = \sum_{\idx=-\infty}^{\infty}
    \phi(k-\idx)\left(\fetasim\left(\etamodsim(\idx-1),\ypreview(\idx-1),\idx-1\right)-\Asim\etamodsim(\idx-1)\right)
\end{equation}
where
\begin{align}
    \phi(k) = \begin{cases}
    \begin{bmatrix}
    \Asim^k_\stab & 0_{\evalq\times n_\eta-\evalq} \\
    0_{n_\eta-\evalq\times \evalq} & 0_{n_\eta-\evalq\times n_\eta-\evalq}
    \end{bmatrix}
    & k>0
    \\[10pt]
    \begin{bmatrix}
    I_{\evalq\times \evalq} & 0_{\evalq\times n_\eta-\evalq} \\
    0_{n_\eta-\evalq\times \evalq} & 0_{n_\eta-\evalq\times n_\eta-\evalq}
    \end{bmatrix}
    & k=0
    \\[10pt]
    \begin{bmatrix}
    0_{\evalq\times \evalq} & 0_{\evalq\times n_\eta-\evalq} \\
    0_{n_\eta-\evalq\times \evalq} & -\Asim_\unstab^k
    \end{bmatrix}
    & k<0
    \end{cases}
    \label{eq:phi}
\end{align}
is an exact, bounded solution to (\ref{eq:fetasim}) provided the right hand side of (\ref{eq:implicitsoln}) exists for all $k\in\integer$. However, (\ref{eq:implicitsoln}) is implicit, and thus cannot be directly evaluated. A fixed-point problem
solver---past work uses Picard iteration---must 
be used to find $\etamodsim$, and sufficient conditions for the solver convergence and solution uniqueness must be determined.

\section{
Novel 
ILC
Analysis and Development
}
\label{sec:AILC}

\subsection{Failure for 
Models with Unstable Inverses
}
\label{sec:NILCfailure}
The 
\AILC{}
scheme \directAILC{} provides convergence of $\evec_\tdx$ to 0 in theory.
However, this assumes perfect computation of the matrix inversion in (\ref{eq:gammaprior}).
In practice, the precision to which $\left(\jacobian{\gmod}{\uvec}(\uvec_\tdx)\right)^{-1}$ can be accurately computed is directly dependent on the condition number of $\jacobian{\gmod}{\uvec}(\uvec_\tdx)$.
If the condition number 
of a matrix
is large enough, the values computed for 
its inverse
may become arbitrary, and their order of magnitude may grow directly with the order of magnitude of the condition number \cite[ch. 3.2]{Belsley1980}, \cite{Rump2009}. This ``blowing up'' of the matrix inverse can cause divergence of \directAILC.

Large $\jacobian{\gmod}{\uvec}$ condition numbers have been previously observed for \nonminphase{} linear systems, both time-invariant 
\cite{Chu2013},
\cite[ch. 5.3-5.4]{Moore1993} and time-varying \cite{Norrlof2002}, \cite[ch. 4.1.1]{Dijkstra2004}. 
The fact that the minimum singular value of $\jacobian{\gmod}{\uvec}(\uvec_\tdx)$ decreases with increases in the system frequency response function magnitude at the Nyquist frequency \cite{Lee2000} contributes to this ill-conditioning.
For linear systems, this magnitude is directly dependent on the zero 
magnitudes, and thus on the inverse systems' 
stability.
These phenomena generalize 
to nonlinear systems because the Jacobian evaluated at a particular input trajectory, $\jacobian{\gmod}{\uvec}(\uvec^*)$, is equal to the constant matrix $\jacobian{\glin}{\uvec}$ where $\glin$ is the lifted input-output model of the linearization of \originalmod{} about the trajectory $\uvec^*$.

To illustrate this equality, first consider that the elements of $\jacobian{\gmod}{\uvec}(\uvec^*)$ are given by (\ref{eq:ymodk}) and the chain rule as
\begin{align}
\label{eq:jacobianelements}
    \partialfrac{\ymod(k)}{u(j)}(\uvec^*) =
    \partialfrac{\hmod}{\xmod}\left( \fmod^{(k-1)}(\uvec^*) \right)
    \partialfrac{\fmod^{(k-1)}}{\uvec}(\uvec^*)
    \partialfrac{\uvec}{u(j)}
\end{align}
where
\begin{align}
    \partialfrac{\uvec}{u(j)} = \begin{bmatrix}
    0_{1\times j} & 1 & 0_{1\times N-\mu+1-j}
    \end{bmatrix}^T
\end{align}
and $\partialfrac{\hmod}{\xmod}$ is a row vector.

Then consider the linearization of \originalmod{} about 
$\uvec^*$:
\begin{IEEEeqnarray}{RL}
\eqlabel{eq:linall}
\IEEEyesnumber
\IEEEyessubnumber*
\label{eq:flin}
    \delta \xmod(k+1) &=\flin\left(\delta\xmod(k),\delta u(k),k\right)
    \\
    &= \nonumber
    \jacobian{\fmod}{\xmod}\left(\xmod^*(k),u^*(k),k\right)\,\delta \xmod(k) +
    \jacobian{\fmod}{u}(\xmod^*(k),u^*(k),k)\,\delta u(k)
    \\
    \label{eq:hlin}
    \delta \ymod(k) &= \hlin(\delta\xmod(k)) 
    =
    \jacobian{\hmod}{\xmod}(\xmod^*(k))\,\delta\xmod(k)
\end{IEEEeqnarray}
where
$\xmod^*(k)=\fmod^{(k-1)}(\uvec^*)$ and the $\delta$ notation denotes $\delta\xmod(k)=\xmod(k)-\xmod^*(k)$ for $\xmod$ and similar for $u$.

Lifting \linsys{} in the same manner as \originalmod{} yields the output perturbation as a function of the input perturbation time series $\delta\uvec$ via
\begin{align}
\label{eq:deltaymod}
    \delta \ymod (k) = \jacobian{\hmod}{\xmod}\left(\fmod^{(k-1)}(\uvec^*)\right)\flin^{(k-1)}(\delta\uvec)
\end{align}
Because of \linsys's linearity, $\flin^{(k-1)}(\delta\uvec)$ can be explicitly expanded as
\begin{align}
\label{eq:linfunccomp}
    \flin^{(k-1)}(\delta\uvec)=\left(\prod_{\idx=0}^{k-1}\partialfrac{\fmod^{(\idx)}}{\fmod^{(\idx-1)}}(\uvec^*)\right)\delta\xmod(0) + \partialfrac{\fmod^{(k-1)}}{\uvec}(\uvec^*)\delta\uvec
\end{align}
where $\prod$ is ordered with the factor of least $\ell$ on the right and the factor of greatest $\ell$ on the left.
The terminal condition of the recursive function composition is $\fmod^{(-1)}=\xmod(0)$.
From (\ref{eq:deltaymod}) and (\ref{eq:linfunccomp}) it is clear that the elements of $\jacobian{\glin}{\delta\uvec}$ are given by
\begin{align}
    \partialfrac{\delta\ymod(k)}{\delta u(j)} = \jacobian{\hmod}{\xmod}\left(\fmod^{(k-1)}(\uvec^*)\right)
    \jacobian{\fmod^{(k-1)}}{\uvec}(\uvec^*)
    \jacobian{\delta\uvec}{\delta u(j)}
\end{align}
which is equal to (\ref{eq:jacobianelements}) because $\jacobian{\delta\uvec}{\delta u(j)}=\partialfrac{\uvec}{u(j)}$ 
due to the identical structures (\ref{eq:uvec}) of $\uvec$ and $\delta \uvec$ with respect to $u$ and $\delta u$ time indexing.
Thus, if \originalmod{} is such that its linearization \linsys{} is 
unstable 
it will suffer ill-conditioning and  
(\ref{eq:gammaprior}) may be so difficult to compute in practice that attempts to do so yield a matrix with large 
erroneous
elements. Such a learning gain matrix may in turn cause $\uvec_{\tdx+1}$ to contain large 
erroneous
elements, causing the learning law to diverge.

Therefore, for the learning law (\ref{eq:ILCclassic}) to converge for a 
system with an unstable inverse 
in practice, a learning matrix synthesis that does not require matrix inversion of $\jacobian{\gmod}{\uvec}(\uvec_\tdx)$ is desired.

\subsection{Alternative Learning Matrix Synthesis}
\label{sec:gammanew}
To circumvent issues associated with inverting 
$\jacobian{\gmod}{\uvec}(\uvec_\tdx)$
this work introduces a new learning matrix definition seeking to satisfy the requirements 
\ref{A4c3}-\ref{A6c3}
in the spirit of Newton's method, but without the matrix inversion requirement of (\ref{eq:gammaprior}). The new learning matrix is given by
\begin{equation}
    L_\tdx = 
    \jacobian{\gmod^{-1}}{\yvecmod}(\yvec_\tdx)
    \label{eq:gammanew}
\end{equation}
where $\gmod^{-1}:\real^{N-\mu+1}\rightarrow\real^{N-\mu+1}$ is a lifted model of the inverse of \originalmod{}.
This makes 
$\jacobian{\gmod^{-1}}{\yvecmod}$
a function of the output of \originalmod{}, 
namely $\yvecmod_\tdx$. As stated in Section 
\ref{sec:NILC},
$\yvecmod_\tdx$ is merely a prediction of the accessible, measured output $\yvec_\tdx$. Hence $\yvec_\tdx$ is used as the input to 
$\jacobian{\gmod^{-1}}{\yvecmod}$.
In short, this work proposes using the linearization of the inverse of \originalmod{} rather than the inverse of the linearization,
and thus the new framework \SAILControl{} will be referred to as ``Invert-Linearize ILC'' (\ILILC).

A direct method of inverting \originalmod{} is to solve
\begin{equation}
    \ymod_\tdx(k+\mu) = \hmod(
    \fmod^{(k+\mu-1)}(\uvec_\tdx)
    )
\end{equation}
for $u_\tdx(k)$, and substitute the resulting function of 
$\left\{\ymod_\tdx(k),\,\ymod_\tdx(k+1),\cdots,\ymod_\tdx(k+\mu) \right\}$
into (\ref{eq:modelstate}).
However, if \originalmod{} 
has an unstable inverse, 
this method of inversion will yield unbounded states $\xmod_\tdx(k)$ as $k$ increases. Thus, $\gmod^{-1}$ is derived via stable inversion rather than direct inversion.
Note, though, that (\ref{eq:gammanew}) also admits the use of other stable approximate inverse models for $\gmod^{-1}$ should they be available.

\section{
Novel Stable Inversion Development
}
\label{sec:StabInv}
This section proves a relaxed set of sufficient conditions for the convergence of Picard iteration to the unique solution to the stable inversion problem, i.e. the unique solution to (\ref{eq:implicitsoln}) from Section \ref{sec:back}.
This enables stable inversion---and thus ILC---for a new class of system representations capturing simultaneous feedback and feedforward control.
Additionally, a new initial Picard iterate prescription is given to suit the broadened scope of stable inversion, and a procedure for practical implementation is described.
This procedure enables the derivation of $\gmod^{-1}$.

\subsection{Fixed-Point Problem Solution}
The standard Picard iterative solver \cite[ch. 9]{Agarwal2000} for (\ref{eq:implicitsoln}) is 
\begin{equation}
    \etamodsim_{(m+1)}(k) =
    \sum_{\idx=-\infty}^{\infty}
    \phi(k-\idx)\left(\fetasim\left(\etamodsim_{(m)}(\idx-1),\ypreview(\idx-1),\idx-1\right) - \Asim\etamodsim_{(m)}(\idx-1) \right)
    \label{eq:picardinf}
\end{equation}
where the parenthetical subscript $(m)\in\integer_{\geq 0}$ is the Picard iteration index.

To prove that (\ref{eq:picardinf}) converges to a unique solution, \cite{Zeng2000} makes the 
assumptions\footnote{
The continuous-time literature also makes these assumptions
\cite{Devasia1996,Devasia1998}
}
that
\begin{enumerate}[label=(Z\arabic*)]
\item
\label{(Z2)}
$\fmod(0,0,k)=0$ $\forall$ $k$, and
\item
\label{(Z3)}
$\etamodsim_{(0)}(k)=0$ $\forall$ $k$.
\end{enumerate}
The 
first assumption
is violated for many representations of systems incorporating both feedback and feedforward control. An example of such a system is given in Section \ref{sec:Val}, where $u$ is the feedforward control input and the feedback control is part of the time-varying dynamics of $\fmod$. This feedback control influences $\xmod$ regardless of whether or not $u(k)=0$. While there may often be a change of variables that enables satisfaction of \ref{(Z2)}, (\ref{eq:simmuminus}-\ref{eq:simoutput}) already imposes constraints on the states and outputs, and for many systems it is unlikely for there to exist a change of variables satisfying both assumptions.

Furthermore, while for systems satisfying \ref{(Z2)}, \ref{(Z3)} may 
be the zero-input state trajectory, this is 
untrue
for systems violating \ref{(Z2)}. For these systems, the zero state trajectory \ref{(Z3)} is essentially arbitrary, and may 
degrade the quality of low-$m$ Picard iterates if far from the solution trajectory.
This seriously jeopardizes convergence because the computational complexity of the Picard iteration solution grows exponentially with the number of iterations. It is thus desirable to reach a satisfactory solution in as few iterations as possible, i.e. it is desirable to have high-quality low-$m$ iterates. 

Thus, this work presents a new set of sufficient conditions for the unique convergence of (\ref{eq:picardinf}) that relaxes 
\ref{(Z2)}, \ref{(Z3)}.
Before proof of this, several definitions are presented.

\begin{definition}[Lifted 
Matrices and Third-Order Tensors]
Given the vector and matrix functions of time $a(k)\in\real^{n}$ 
and
$B(k)\in\real^{n\times n}$, the corresponding lifted matrix and third order tensor are given by upright bold notation: $\mathbf{a}\in\real^{n\times\mathcal{K}}$ and $\mathbf{B}\in\real^{n\times n\times\mathcal{K}}$. $\mathcal{K}$ is the time dimension, and may be $\infty$. Elements of the lifted objects are $\mathbf{a}^{i,k}\coloneqq a^i(k)$ and $\mathbf{B}^{i,j,k}\coloneqq B^{i,j}(k)$.
\label{def:lift}
\end{definition}

\begin{definition}[Matrix and Third-Order Tensor Norms]
$\norm{\cdot}_\infty$ refers to the ordinary $\infty$-norm when applied to vectors, and is the matrix norm induced by the vector norm when applied to matrices (i.e. the maximum absolute row sum). Additionally, the entry-wise $(\infty,1)$-norm is defined for the matrices and third-order tensors $\mathbf{a}$ and $\mathbf{B}$ from Definition \ref{def:lift} as
\begin{align}
\norm{\mathbf{a}}_{\infty,1}
\coloneqq \sum_{k\in\mathcal{K}}\norm{a(k)}_\infty
\qquad
\norm{\mathbf{B}}_{\infty,1}\coloneqq \sum_{k\in\mathcal{K}}\norm{B(k)}_{\infty}
\end{align}
\end{definition}

\begin{definition}[Local Approximate Linearity \cite{Devasia1996,Zeng2000}]
$\fetasim$ is locally approximately linear in $\etamodsim(k)$ and its $\etamodsim(k)=0$ dynamics,
 in a closed $s$-neighborhood around $(\etamodsim(k)=0,\fetasim(0,\ypreview(k),k)=0)$, 
 with Lipschitz constants $K_1, K_2>0$ if $\exists s>0$ such that for any vectors
\begin{itemize}
    \item
    $a(k)$, $b(k)\in\real^{n_\eta}$ with $\norm{\cdot}_\infty\leq s$ $\forall k$, and
    \item
    $\mathscr{a}(k)$, $\mathscr{b}(k)\in\real^{\mu+1}$ such that $\norm{\fetasim(0,\mathscr{a}(k),k)}_\infty$, $\norm{\fetasim(0,\mathscr{b}(k),k)}_\infty\leq s$ $\forall k$
\end{itemize}
the following is true $\forall k$
\begin{multline}
    \left\lVert
    \left(
    \fetasim(a(k),\mathscr{a}(k),k)-Aa(k)
    \right) 
    -
    \right.
    \left.
    \left(
    \fetasim(b(k),\mathscr{b}(k),k)-Ab(k)
    \right)
    \right\rVert_\infty
    \\
    \leq
    K_1\norm{
    a(k)-b(k)
    }_\infty+
    K_2\norm{
    \fetasim(0,\mathscr{a}(k),k)-\fetasim(0,\mathscr{b}(k),k)}_\infty
    \label{eq:lal}
\end{multline}
\label{def:lal}
\end{definition}

With these definitions 
a new set of sufficient conditions for Picard iteration convergence 
 may be established.

\begin{theorem}
The Picard iteration (\ref{eq:picardinf}) converges to a unique solution to (\ref{eq:fetasim}) if the following sufficient conditions are met.
\begin{enumerate}[label=(A4.\arabic*),leftmargin=*]
    \item 
    \label{(C5)}
    $\norm{\etavecsim_{(0)}}_{\infty,1}\leq s$
    \item
    \label{(C5b)}
    $\forall k$ $\exists\ypreview(k)=\ypreview^\dagger(k)$ such that $\fetasim(0,\ypreview^\dagger(k),k)=0$
    \item
    \label{(C6)}
    $\fetasim$ is locally approximately linear in the sense of (\ref{eq:lal})
    \item
    \label{(C7)}
    $K_1\norm{\phivec}_{\infty,1}<1$
    \item
    \label{(C8)}
    $\frac{\norm{\phivec}_{\infty,1}K_2
    \norm{
    \tilde{\fvec}_\eta(0,\ypreview)
    }_{\infty,1}
    }{1-\norm{\phivec}_{\infty,1}K_1}\leq s$
\end{enumerate}
where $
\norm{
    \tilde{\fvec}_\eta(0,\ypreview)
    }_{\infty,1} =
    \sum_{k=-\infty}^\infty \norm{\fetasim(0,\ypreview(k),k)}_\infty
$
and $\phivec$ are defined by Definition \ref{def:lift}; i.e. $\phivec$ is the lifted tensor version of (\ref{eq:phi}).
\label{thm:picard}
\end{theorem}
\begin{proof}
This proof shares the approach of \cite{Zeng2000} in establishing the Cauchy nature of the Picard sequence. It is also influenced by the proofs of Picard iterate local approximate linearity for continuous-time systems in \cite{Devasia1996}.

Proof 
that (\ref{eq:picardinf}) converges to a unique fixed point begins with 
an induction showing
that $\etamodsim_{(m)}(k)$ remains in the locally approximately linear neighborhood $\forall$ $k$, $m$. The base case of this induction is given by \ref{(C5)}. Then under the premise 
\begin{equation}
\norm{\etavecsim_{(m)}}_{\infty,1}\leq s
\label{eq:premise}
\end{equation}
the induction proceeds as follows.
Here, ellipses indicate the continuation of a line of mathematics.

By 
the Picard iterative solver 
(\ref{eq:picardinf}):
\begin{multline}
    \norm{\etavecsim_{(m+1)}}_{\infty,1}=
    \\
    \sum_{k=-\infty}^\infty\left\lVert
    \sum_{\idx=-\infty}^\infty
    \phi(k-\idx)
    \right.
    \left.
    \vphantom{\sum_{\idx=-\infty}^\infty}
    \left(
    \fetasim\left(\etamodsim_{(m)}(\idx-1),\ypreview(\idx-1),\idx-1\right)
    -A\etamodsim_{(m)}(\idx-1)
    \right)
    \right\rVert_\infty
    \cdots
\end{multline}
By the triangle inequality:
\begin{multline}
    \cdots \leq \sum_{k=-\infty}^\infty \sum_{\idx=-\infty}^\infty
    \left\lVert
    \phi(k-\idx)
    \vphantom{\fetasim}
    \right.
    \left.
    \left(
    \fetasim\left(\etamodsim_{(m)}(\idx-1),\ypreview(\idx-1),\idx-1\right)
    -A\etamodsim_{(m)}(\idx-1)
    \right)
    \right\rVert_{\infty}
    \cdots
\end{multline}
By the fact that for matrix norms induced by vector norms $\norm{Ba}\leq\norm{B}\norm{a}$ for matrix $B$ and vector $a$:
\begin{multline}
    \cdots \leq \sum_{k=-\infty}^\infty \sum_{\idx=-\infty}^\infty
    \norm{
    \phi(k-\idx)
    }_\infty
    \norm{
    \fetasim\left(\etamodsim_{(m)}(\idx-1),\ypreview(\idx-1),\idx-1\right)
    -A\etamodsim_{(m)}(\idx-1)
    }_{\infty}
    \cdots
\end{multline}
\begin{multline}
    \cdots = \sum_{\idx=-\infty}^\infty
    \left\lVert
    \fetasim\left(\etamodsim_{(m)}(\idx-1),\ypreview(\idx-1),\idx-1\right)
    \vphantom{
    -A\etamodsim_{(m)}(\idx-1)
    }
    -A\etamodsim_{(m)}(\idx-1)
    \right\rVert_{\infty}
    \sum_{k=-\infty}^\infty
    \norm{
    \phi(k-\idx)
    }_\infty
    \cdots
\end{multline}
By the fact that $\sum_{k=-\infty}^\infty\norm{\phi(k-\idx)}_\infty$ has the same value $\forall \idx$
\begin{equation}
    \cdots = \norm{\phivec}_{\infty,1}
    \sum_{\idx=-\infty}^\infty
    \norm{
    \fetasim\left(\etamodsim_{(m)}(\idx-1),\ypreview(\idx-1),\idx-1\right)
    -A\etamodsim_{(m)}(\idx-1)
    }_\infty
    \cdots
\end{equation}
By \ref{(C5b)}:
\begin{multline}
    \cdots = \norm{\phivec}_{\infty,1}
    \sum_{\idx=-\infty}^\infty 
    \left\lVert
    \left(
    \fetasim\left(\etamodsim_{(m)}(\idx-1),\ypreview(\idx-1),\idx-1\right)
    -A\etamodsim_{(m)}(\idx-1)
    \right)
    \right.
    \\
    \left.
    -
    \left(
    \fetasim\left(0,\ypreview^\dagger(\idx-1),\idx-1\right)
    -A(0)
    \right)
    \right\rVert_\infty
    \cdots
\end{multline}
By \ref{(C6)}:
\begin{equation}
    \cdots 
    \leq
    \norm{\phivec}_{\infty,1} \sum_{\idx=-\infty}^\infty
    K_1\norm{\etamodsim_{(m)}(\idx-1)}_\infty +
    K_2\norm{\fetasim\left(0,\ypreview(\idx-1),\idx-1\right)}_\infty
    \cdots
\end{equation}
\begin{equation}
    \cdots =
    \norm{\phivec}_{\infty,1}\left(
    K_1\norm{\etavecsim_{(m)}}_{\infty,1}
    +
    K_2\norm{\tilde{\fvec}_\eta(0,\ypreview)}_{\infty,1}
    \right)
    \cdots
\end{equation}
By \ref{(C7)}, both sides of \ref{(C8)} can be multiplied by the denominator in \ref{(C8)} without changing the inequality direction. Thus by (\ref{eq:premise}) and algebraic rearranging of \ref{(C8)}
\begin{equation}
    \cdots \leq \norm{\phivec}_{\infty,1}\left(
    K_1s+K_2\norm{\tilde{\fvec}_\eta(0,\ypreview)}_{\infty,1}
    \right)
    \leq
    s
\end{equation}
$\therefore$ $\norm{\etavecsim_{(m)}}_{\infty,1}\leq s$ $\forall m$. Because $\norm{\etavecsim_{(m)}}_{\infty,1}\geq \norm{\etamodsim_{(m)}(k)}_\infty$ $\forall k$, this implies that $\etamodsim_{(m)}(k)$ is within the locally approximately linear neighborhood $\forall$ $m$, $k$.

To show that (\ref{eq:picardinf}) converges to a unique fixed point, define
\begin{equation}
    \Delta\etamodsim_{(m)}(k)\coloneqq \etamodsim_{(m+1)}(k) - \etamodsim_{(m)}(k)
\end{equation}
Then, by a nearly identical induction
\begin{equation}
    \norm{\Delta\etavecsim_{(m+1)}}_{\infty,1} \leq \norm{\phivec}_{\infty,1}K_1\norm{\Delta\etavecsim_{(m)}}_{\infty,1}
\end{equation}
By \ref{(C7)} 
\begin{equation}
    \lim_{m\rightarrow\infty} \norm{\Delta\etavecsim_{(m)}}_{\infty,1}=0
\end{equation}
which implies
\begin{equation}
    \lim_{m\rightarrow\infty}\norm{\Delta\etamodsim_{
(m)}(k)}_\infty=0\,\,\forall k
\end{equation}
\begin{sloppypar}
$\therefore$ $\forall k$ the sequence $\{\etamodsim_{m}(k)\}$ is a Cauchy sequence, and thus the fixed point 
$\etamodsim(k)=\lim_{m\rightarrow\infty}\etamodsim_{(m)}(k)$
is unique.
\end{sloppypar}
\end{proof}

Neither the preceding presentation nor the nonlinear stable inversion prior art \cite{Zeng2000} explicitly discusses the intuitive foundation of stable inversion: evolving the stable modes of an inverse system forwards in time from an initial condition and evolving the unstable modes backwards in time from a terminal condition. Unlike for linear time invariant (LTI) systems, this intuition is not put into practice directly for nonlinear systems because 
the similarity transforms that completely decouple the stable and unstable modes of linear systems do not necessarily decouple the stable and unstable modes of nonlinear systems.
However, the same principle underpins this work. This is evidenced by the fact that the intuitive LTI stable inversion is recovered from (\ref{eq:implicitsoln}) when $\fmod$ is LTI, as illustrated briefly below.

For LTI $\fmod$, $\fmodsim$ takes the form
\begin{align}
    \etamodsim(k+1)&=
    \Asim\etamodsim(k) + \Bsim\ypreview(k)
    \\
    \begin{bmatrix} \etamodsim_{\stab}(k+1) \\ \etamodsim_{\unstab}(k+1) \end{bmatrix}
    &=
    \begin{bmatrix}
    \Asim_\stab & 0 \\ 0 & \Asim_\unstab 
    \end{bmatrix}\begin{bmatrix}
    \etamodsim_{\stab}(k) \\ \etamodsim_{\unstab}(k)
    \end{bmatrix} + \begin{bmatrix}
    \Bsim_\stab \\ \Bsim_\unstab
    \end{bmatrix}\ypreview(k)
\end{align}
Then the implicit solution (\ref{eq:implicitsoln}) becomes the explicit solution
\begin{align}
    \etamodsim(k) &= \sum_{\idx=-\infty}^\infty \phi(k-\idx)\Bsim\ypreview(\idx-1)
    \\
    \begin{bmatrix}
    \etamodsim_\stab(k) \\ \etamodsim_\unstab(k)  
    \end{bmatrix}
    &= \begin{bmatrix}
    \sum_{\idx=-\infty}^{k}\Asim^{k-\idx}_\stab\Bsim_\stab\ypreview(\idx-1)
    \\
    -\sum_{\idx=k+1}^\infty \Asim_\unstab^{k-\idx}\Bsim_\unstab\ypreview(\idx-1)
    \end{bmatrix}
    \\
    &=
    \begin{bmatrix}
    \Asim_\stab\etamodsim_\stab(k-1) + \Bsim_\stab\ypreview(k-1)
    \\
    \Asim_\unstab^{-1}\etamodsim_\unstab(k+1)-\Asim_\unstab^{-1}\Bsim_\unstab\ypreview(k)
    \end{bmatrix}
\end{align}
which is the forward evolution of the stable modes and backward evolution of the unstable modes where the initial and terminal conditions at $k=\pm\infty$ are zero.

\subsection{Initial Picard Iterate 
\texorpdfstring{$\etamodsim_{(0)}$}{ }
Selection and Implementation}
This subsection addresses the need to select a new initial Picard iterate $\etamodsim_{(0)}(k)$ in the absence of \ref{(Z3)}.
Also addressed is
 the fact that (\ref{eq:picardinf}) is a purely theoretical, rather than implementable, solution because it contains infinite sums along an infinite timeline.

In the context of ILC, the learned feedforward control action is often intended to be a relatively minor adjustment to the primary action of the feedback controller. 
Thus, choosing $\etamodsim_{(0)}(k)$ to be the feedback-only trajectory, i.e. the zero-feedforward-input trajectory, is akin to warm-starting the fixed-point solving process. This trajectory is given by
\begin{equation}
\begin{aligned}
    \xmod(k+1) &= \fmod\left( \xmod(k), 0, k \right) \qquad \xmod(0)=0_{n_x}
    \\
    \etamodsim_{(0)}(k) &=     V^{-1}
    \begin{bmatrix}
    0_{\mu\times\mu} 
    & 0_{\mu\times n_{\eta}} \\ 0_{n_{\eta}\times \mu} & I_{n_{\eta}\times n_{\eta}}
    \end{bmatrix}\xmod(k)
    \end{aligned}
    \label{eq:picard0}
\end{equation}
for $k\in\{0,\cdots,N-\mu\}$.

An implementable version of (\ref{eq:picardinf}) is given by
\begin{equation}
    \etamodsim_{(m+1)}(k) =
    \sum_{\idx=1}^{N-\mu+1}
    \phi(k-\idx)
    \left(\fetasim\left(\etamodsim_{(m)}(\idx-1),\ypreview(\idx-1),\idx-1\right) - \Asim\etamodsim_{(m)}(\idx-1) \right)
    \label{eq:implementableStabInv}
\end{equation}
for $k\in\{1,...,N-\mu\}$,
fixing the initial condition
$\etamodsim_{(m)}(0)=0_{n_{\eta}}$ 
$\forall m$.

Note that (\ref{eq:implementableStabInv}) is equivalent to 
assuming $\etamodsim_{(m)}(k)=0$, 
$\ypreview(k)=0$,
and $\fetasim\left(0,0,k\right)=0$
for $k\in(-\infty,-1]\cup[N-\mu+1,\infty)$
and
extracting the $k\in[1,N-\mu]$ elements of $\etamodsim_{(m+1)}(k)$ generated by (\ref{eq:picardinf}).
These assumptions correspond to a lack of control action prior to $k=0$ and a reference trajectory that brings the system back to its zero initial condition 
with enough trailing zeros for the system to settle by $k=N-\mu$.
This is typical of repetitive motion processes, but admittedly may preclude some other ILC applications.

Furthermore, 
for the first Picard iteration ($m+1=1$)
these assumptions yield identical (\ref{eq:implementableStabInv})- and (\ref{eq:picardinf})-generated $\etamodsim_{(1)}(k)$ on $k\in[0,N-\mu]$.
Because output tracking of \nonminphase{} systems in general requires preactuation,
for this range of $k$ to contain a practical control input trajectory there must be sufficient leading zeros in the reference starting at $k=0$.
For the following Picard iterates the theoretical and implementable trajectories are unlikely to be equal, but can be made closer the more leading zeros are included in the reference.

Ultimately, applying (\ref{eq:implementableStabInv}) for any number of iterations $\mfin\geq 1$ yields an expression for each time step of $\etamodsim_{(\mfin)}(k)$ whose only variable parameters are the elements of $\yvecmod$. This is because the recursion calling $\etamodsim_{(\mfin)}(k)$ terminates at the known trajectory $\etamodsim_{(0)}(k)$, and because $\ymod(k)=0$ for $k\in\{0,...,\mu-1\}$ due to the known initial condition $\xmod(0)=0$. The concatenation of these expressions plugged into 
the inverse output function 
(\ref{eq:invout}) yields the lifted inverse system model
\begin{equation}
    \gmod^{-1}(\yvecmod)=\begin{bmatrix}
    \fmod^{\mu}{\vphantom{\fmod}}^{ -1 }\left(\etamodsim_{(\mfin)}(0),\ypreview(0),0\right)
    \\
    \fmod^{\mu}{\vphantom{\fmod}}^{ -1 }\left(\etamodsim_{(\mfin)}(1) ,\ypreview(1),1\right)
    \\ \vdots \\
    \fmod^{\mu}{\vphantom{\fmod}}^{ -1 }\left(\etamodsim_{(\mfin)}(N-\mu),\ypreview(N-\mu),N-\mu\right)
    \end{bmatrix}
    \label{eq:ginv}
\end{equation}
which enables the 
synthesis of the 
\ILILC{} 
learning matrix (\ref{eq:gammanew}).
With this, the complete synthesis 
of \ILILC{} with stable inversion---starting from a model 
in the normal form (\ref{eq:simmuall})---can
be summarized by Procedure \ref{proc:synth}.

\begin{algorithm}
	\caption{\ILILC{} Synthesis with Stable Inversion}
		\label{proc:synth}
\begin{algorithmic}[1]

\State
Derive the minimal state space representation $\fmod_{\eta}$ and $\fmod^{\mu}{\vphantom{\fmod}}^{-1}$ (from (\ref{eq:invall})) of the conventional inverse of (\ref{eq:simmuall}).

\State
Apply similarity transform $V$ (from (\ref{eq:transform})) to derive the inverse state dynamics representation $\fetasim$ (from (\ref{eq:fetasim})) with decoupled stable and unstable linear parts.

\State
\label{step:etasim}
Use (\ref{eq:picard0})-(\ref{eq:implementableStabInv}) to derive the inverse system state $\etamodsim_{(\mfin)}$ as a function of $\yvecmod$ at each point in time $k\in\{0,\cdots,N-\mu\}$.

\State
\label{step:ginv}
Derive the lifted inverse model $\gmod^{-1}$ 
via (\ref{eq:ginv}).

\State
Use an automatic differentiation tool to 
derive $\jacobian{\gmod^{-1}}{\yvecmod}$ as a function of $\yvec$, i.e. the learning matrix $L_\tdx$ from (\ref{eq:gammanew}).

\State
Compute
$L_\tdx = \jacobian{\gmod^{-1}}{\yvecmod}(\yvec_\tdx)$ at each trial for the ILC law (\ref{eq:ILCclassic}).

\Comment{
Steps \ref{step:etasim}-\ref{step:ginv} are greatly facilitated by using a computer algebra system. CasADi can provide this functionality in addition to automatic 
differentiation.
}

\end{algorithmic}
\end{algorithm}

\section{Validation}
\label{sec:Val}

This section presents validation of the fundamental claim that the original \AILC{} fails for 
models with unstable inverses 
and that the 
newly
proposed
\ILILC{} 
framework---when used with stable inversion---succeeds.
Additionally, while the intent of ILC is to account for model error, overly erroneous modeling can cause
violation of \ref{C3}, which 
may cause 
divergence of the ILC law.
Thus this section also probes the performance and robustness of 
\ILILC{} 
with stable inversion over increasing model error in physically motivated simulations.

The 
\ILILC{} law
(\ref{eq:ILCclassic}), (\ref{eq:gammanew}) is applied as a reference shaping tool to a feedback control system (sometimes called ``series ILC''). 
This represents the common scenario of applying a higher level controller to ``closed source'' equipment.
The resultant system \originalmod{} is a nonlinear time-varying system with relative degree $\mu=2$.

Modeling error
is simulated by synthesizing the ILC laws from a nominal ``control model'' of the example system, and applying the resultant control inputs to a set of ``truth models'' featuring random parameter errors and the injection of process and measurement noise.
Finally, to give context to the results for 
\ILILC{}
with stable inversion, identical simulations are run with a benchmark technique that does not require modification for \nonminphase{} systems: gradient ILC.

\subsection{Benchmark Technique: Gradient ILC}
Aside from those using \AILC{}, the authors know of no prior art explicitly addressing ILC  synthesis from discrete-time nonlinear (particularly nonlinear in the input) time-varying models with relative degree greater than 1. However, with automatic differentiation, ``gradient ILC'' is nearly as easily synthesized from this class of 
models as it is from the LTI models it was proposed for in \cite{Owens2009}.

The most straightforward form of Gradient ILC is gradient descent applied to the optimization problem
\begin{equation}
    \argmin_{\uvec}\frac{1}{2}\evec^T\evec
\end{equation}
which yields the ILC law
\begin{equation}
    \uvec_{\tdx+1} = \uvec_\tdx + \gamma \jacobian{\gmod}{\uvec}\left(\uvec_\tdx\right)^T\evec_j
    \label{eq:gradILC}
\end{equation}
where $\gamma>0$ is the gradient descent step size. Note that (\ref{eq:gradILC}) is free of the matrix inversion that historically inhibited the application of \AILC{} to 
systems with unstable inverses.

$\gamma$ is a tuning parameter that influences the performance-robustness trade off of (\ref{eq:gradILC}). 
Reducing $\gamma$ improves the probability that (\ref{eq:gradILC}) will converge for some unknown model error, but may also reduce the rate of convergence. 
For the sake of comparing the convergence rates between gradient ILC and 
\ILILC, here we choose $\gamma$ such that the two methods have comparable probabilities of convergence over the battery of random model errors tested: $\gamma=1.1$.

\subsection{Example System}
\label{sec:exSystem}
\begin{figure}
    \centering
    \includegraphics[scale=0.3]{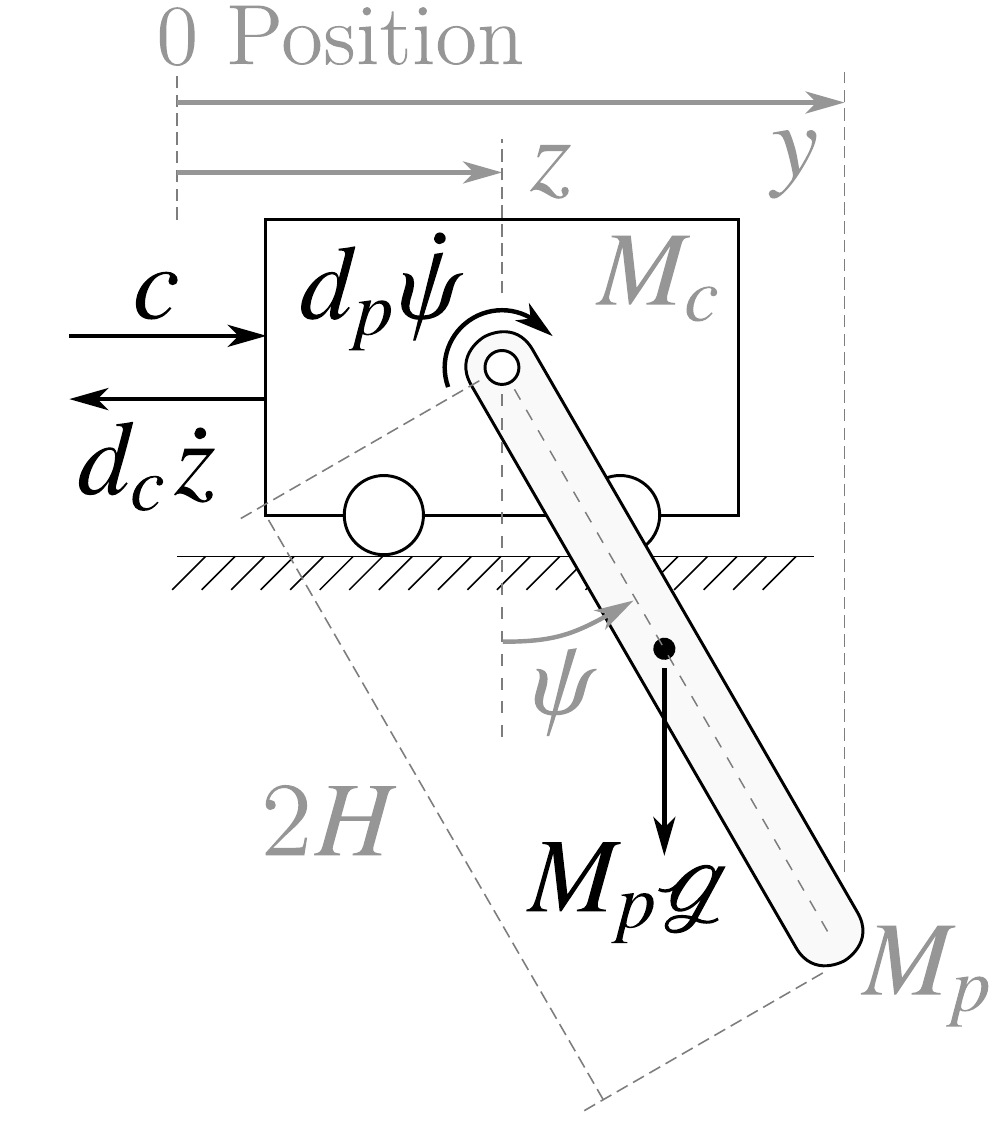}
    \caption{Cart and pendulum system. Dimension, position, and mass annotations are in grey. Force and torque annotations are in black.}
    \label{fig:cartpend}
\end{figure}

Consider the system pictured in Figure \ref{fig:cartpend}, consisting of a pendulum fixed to the mass center of a cart on a rail. 
This subsection presents the first-principles continuous-time equations of motion for this plant, the method for converting these dynamics to the discrete-time normal form (\ref{eq:simmuall}), and the control architecture of the system.

The cart is subjected to an applied force $\force$, and
viscous damping occurs both between the cart and the rail and between the pendulum and the cart.
Equations of motion for this plant 
are given by
\begin{multline}
\label{eq:EOMconttheta}
    \ddot{\angtruth} =
     -3 \left(
     H \Mp \left( \force + \omega_\force \right) \cos (\angtruth ) +
     \dptruth ( \Mc + \Mp ) \dot{\angtruth}+
     \vphantom{\left(\Mc \Mp +\Mp^2\right)}
     \right. \\ \left.
      H^2 \Mp^2 \sin (\angtruth ) \cos (\angtruth ) \dot{\angtruth}^2+
     \grav H \left(\Mc \Mp +\Mp^2\right) \sin (\angtruth ) -
    \right. \\ \left.
    \vphantom{
    \left(\Mc \Mp +\Mp^2\right)
     }
  \dctruth H \Mp \cos (\angtruth )\dot{z}
  \right)
  \frac{1}{ H^2 \Mp \left(
  4 (\Mc+\Mp)-3 \Mp \cos ^2(\angtruth )
  \right) }
  \end{multline}%
\begin{multline}
    \label{eq:EOMcontx}
    \ddot{z} =
    \left(
    4 H \left( \force + \omega_\force \right) +
     3 \dptruth  \cos (\angtruth ) \dot{\angtruth}+
     \right.\\\left.
     4  H^2 \Mp \sin (\angtruth ) \dot{\angtruth}^2+
     3 \grav H\Mp \sin (\angtruth ) \cos (\angtruth ) -
     \right.\\\left. 
     \vphantom{
     4 H \left( \force + \omega_\force \right) +
     3 \dptruth  \cos (\angtruth ) \dot{\angtruth}+
     }
      4 \dctruth  H \dot{z}
     \right)
  \frac{1}{
    H \left(4 (\Mc+\Mp) -3 \Mp \cos ^2(\angtruth )\right)
  }
  \end{multline}
  where
  $\angtruth(k)$ is the pendulum angle, $z(k)$ is the cart's horizontal position, $\grav=\SI{9.8}{\meter\per\second\squared}$ is gravitational acceleration, 
  and the process noise $\omega_\force(k)$ is a random sample from a normal distribution with $0$ mean and standard deviation \SI{3.15e-2}{\newton}. 
  $H$ is the pendulum half-length, $\Mc$ and $\Mp$ are the cart and pendulum masses, and $\dctruth$ and $\dptruth$ are the cart-rail and pendulum-cart damping coefficients, respectively.
  The time argument of $\omega_c$, $\angtruth$, $z$ and their derivatives has been dropped for compactness.

\begin{sloppypar}
The output to be tracked is the pendulum tip's horizontal position, $y$.
Obtaining a discrete-time state space model of this system in the normal form (\ref{eq:simmuall}) requires first a change of coordinates
such that the desired output is a state,
and then discretization. The change of coordinates is
\begin{equation}
    \angtruth = \arcsin\left(\frac{y-z}{2H}\right)
    \label{eq:ChangeOfCoord}
\end{equation}
with
associated derivative substitutions
\begin{align}
    \dot{\angtruth} &= \frac{\dot{y}-\dot{z}}{2H\sqrt{1-\frac{\left(y-z\right)^2}{4H^2}}}
    \\
    \ddot{\angtruth} &=
    \frac{
    \sec\left(\angtruth\right)\left(
    \ddot{y}-\ddot{z} + 2H\sin\left( \angtruth \right)\dot{\angtruth}^2
    \right)
    }{
    2H
    }
    \label{eq:ddang}
\end{align}
Then the equations of motion are solved for in terms of the new coordinates.
In the present case
(\ref{eq:EOMconttheta})-(\ref{eq:ddang}) can be solved for $\ddot{y}(k)$ and $\ddot{z}(k)$ as functions of $y(k)$, $z(k)$, $\dot{y}(k)$, and $\dot{z}(k)$.
Next, forward Euler discretization 
is applied recursively to the equations of motion to reformulate the state dynamics in terms of discrete time increments rather than derivatives, as is required by the normal form. The innermost layer of the recursion is the first derivatives
\begin{equation}
\dot{y}(k)=\frac{y(k+1)-y(k)}{T_s}
\qquad 
\dot{z}(k)=\frac{z(k+1)-z(k)}{T_s}
\end{equation}
where 
the sample period 
$T_s=\SI{0.016}{\second}$ in this case.
These can be plugged into 
$\ddot{y}(k)$ and $\ddot{z}(k)$ to eliminate their dependence on derivatives.
The next---and in this case final---layer is
the forward Euler discretization of the second derivatives.
The outermost layer 
can be rearranged to yield the discrete-time equations of motion 
\begin{equation}
\label{eq:EOMdisc}
\begin{aligned}
    y(k+2)&=\ddot{y}(k)T_s^2+2y(k+1)-y(k)
    \\
    z(k+2)&=\ddot{z}(k)T_s^2+2z(k+1)-z(k)
    ,
\end{aligned}
\end{equation}
which are directly used to define the state dynamics $f$ in terms of the state vector
$x(k)=[y(k),\,y(k+1),\,z(k),\,z(k+1)]^T$.
The explicit expressions of (\ref{eq:EOMdisc}) 
are too long to print here, but can be easily obtained in Mathematica, MATLAB symbolic toolbox, etc. via the algebra 
described 
in (\ref{eq:ChangeOfCoord})-(\ref{eq:EOMdisc}).
\end{sloppypar}

The output must track the reference $r(k)$ given in Figure \ref{fig:ref}. To accomplish this the plant is equipped with a full-state feedback controller modeled as
\begin{align}
    \force(k)&= \kappa_0r^*(k) - 
    \begin{bmatrix}
    \kappa_1 & \kappa_2 & \kappa_3 &\kappa_4
    \end{bmatrix}x(k)
    \\
    \rprcv(k)&= r(k)+u(k)
\end{align}
Here, $\rprcv(k)$ is the effective reference and $u(k)$ is the control input generated by the ILC law. In other words, the ILC law adjusts the reference delivered to the feedback controller to eliminate the error transients inherent to feedback control.
Finally, the error signal input to the ILC law is subject to measurement noise $\omega_y(k)$
\begin{equation}
    e(k)=r(k)-y(k)-\omega_y(k)
\end{equation}
where the noise's distribution has 0 mean and standard deviation \SI{5e-5}{\meter}.

\begin{figure}
    \centering
    \includegraphics[scale=0.4]{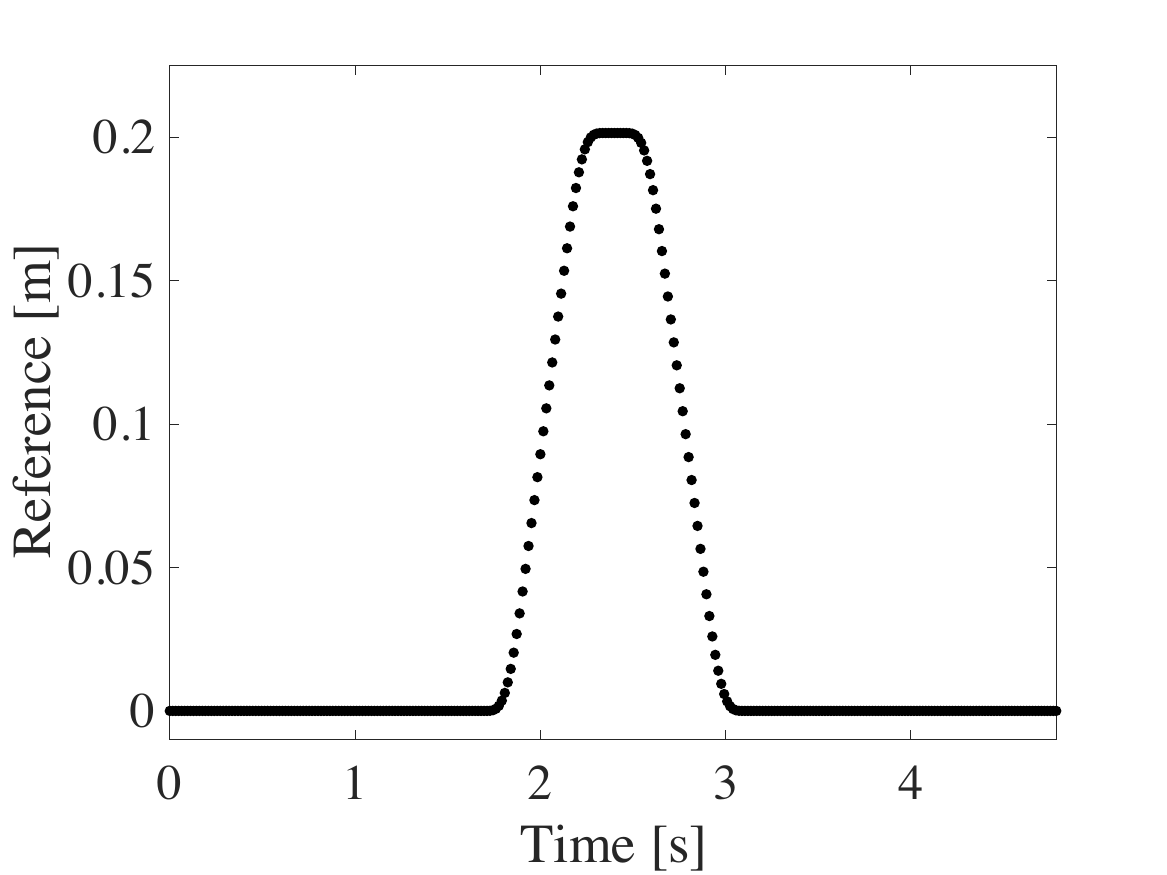}
    \caption{Reference}
    \label{fig:ref}
\end{figure}

The ILC law itself is synthesized from a ``control model'' that is identical in structure to the ``truth model'' presented above, but has $\hat{\omega}_\force=\hat{\omega}_y=0$ and uses the model parameters tabulated in Table \ref{table:cartpend}.
Stable inversion for the synthesis of 
learning matrix 
(\ref{eq:gammanew}) is performed with a single Picard iteration, i.e. $\mfin=1$ in (\ref{eq:ginv}).
To simulate model error, the hatless truth model parameters differ from the behatted control model parameters in a manner detailed in Section \ref{sec:method}.
This ultimately results in the system block diagram given in Figure \ref{fig:block}.

\begin{table}
\centering
\caption{
Cart-Pendulum Control Model Parameters
}
\label{table:cartpend}
\renewcommand{\arraystretch}{1.25}
\begin{tabular}{|l|c|c|}
\hline
Parameter & Symbol & Value 
\\
\hhline{|=|=|=|}
Cart Mass & $\Mcmod$ & \SI{0.5}{\kilo\gram}
\\
\hline
Pendulum Mass & $\Mpmod$ & \SI{0.25}{\kilo\gram}
\\
\hline
Pendulum Half-Length & $\hat{H}$ & \SI{0.225}{\meter}
\\
\hline
Cart-Rail Damping Coefficient & $\dcmod$ & \SI{10}{\kilo\gram\per\second}
\\
\hline
Pendulum-Cart Damping Coefficient & $\dpmod$ & 
\SI{0.01}{\kilo\gram\meter\squared\per\second}
\\
\hline
Full State Feedback Gain 0 & $\kmod_0$ & \SI{630}{}
\\
\hline
Full State Feedback Gain 1 & $\kmod_1$ & \SI{-5900}{}
\\
\hline
Full State Feedback Gain 2 & $\kmod_2$ & \SI{5900}{}
\\
\hline
Full State Feedback Gain 3 & $\kmod_3$ & \SI{-3700}{}
\\
\hline
Full State Feedback Gain 4 & $\kmod_4$ & \SI{4300}{}
\\
\hline
\end{tabular}
\end{table}

\begin{figure}
    \centering
    \includegraphics[scale=1.2]{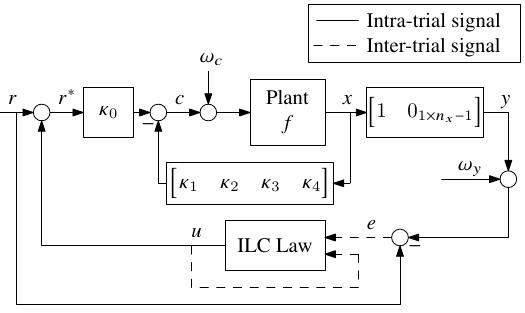}
    \caption{
    System Block Diagram. The control law outputting $u$ is synthesized from the control models defined by the behatted parameters of Table \ref{table:cartpend} and by $\unoise(k)=\ynoise(k)=0$. The plant and controller gain blocks are defined with the truth model parameters generated according to Section \ref{sec:method}.
    Inter-trial signals from trial $\tdx$ are stored and used to compute the input for trial $\tdx+1$.
    }
    \label{fig:block}
\end{figure}

\subsection{Simulation and Analysis Methods}
\label{sec:method}
Let $\parammod\in\real^{10}$ be a vector of the control model parameters in Table \ref{table:cartpend}. Then a truth model can be specified by the vector $\theta$, generated via
\begin{equation}
    \param = \left( 1_{10\times 1} e_{\param}^T\odot I + I \right)\parammod
    \label{eq:truthModlGen}
\end{equation}
where $\odot$ is the Hadamard product and $e_{\param}\in\real^{10}$ is a random sample of a uniform distribution. Under (\ref{eq:truthModlGen}), each element of $e_{\param}$ is the relative error between the corresponding elements of $\param$ and $\parammod$. Thus, $\norm{e_{\param}}_2$ provides a scalar metric for the model error between the control model and a given truth model. The range $\norm{e_{\param}}_2\in[0,0.1]$ is divided into 20 bins of equal width, and 50 truth models are generated for each bin. Both ILC schemes are applied to each truth model with 50 trials, 
and $u_0(k)=0$ $\forall k$.
A full set of 50 trials of one of the ILC laws applied to a single truth model is referred to as a ``simulation.''
The results of these simulations are used to characterize the probability of convergence and rate of convergence of each ILC law.

For each iteration of a simulation, the normalized root mean square error (NRMSE) is given by
\begin{equation}
    \text{NRMSE}_\tdx \coloneqq \frac{\RMS\left(\evec_\tdx\right)}{\norm{\rvec}_\infty}
\end{equation}
A simulation is deemed convergent if there exists $\tdx^*$ such that $\text{NRMSE}_\tdx$ is less than some tolerance for all $\tdx\geq \tdx^*$. This work uses a tolerance of \SI{5e-4}{}, which is close to the NRMSE floor created by noise.

Let $\tdx^{\beta,\tau,\lambda}$ be the minimum $\tdx^*$ for truth model $\tau\in[1,50]$ in bin $\beta\in[1,20]$ under ILC law $\lambda\in\{\text{
\ILILC
},\,\text{gradient ILC}\}$, and
let $\convset$ be the set of all $(\beta,\tau)$ for which both 
\ILILC{}
and gradient ILC converge. Then
the mean transient convergence rate 
\begin{equation}
    \convRate_{\lambda}=\mean_{\convset,\tdx\in[1,\tdx^{\beta,\tau,\lambda}]}
    \left( \frac{\text{NRMSE}^{\beta,\tau,\lambda}_{\tdx}}{\text{NRMSE}^{\beta,\tau,\lambda}_{\tdx-1}} \right)
    \label{eq:convRate}
\end{equation}
offers a numerical performance metric.
Note that \cite{Avrachenkov1998} gives a theoretical convergence analysis for the ILC structure (\ref{eq:ILCclassic}) in general (covering 
\AILC,
\ILILC,
and gradient ILC). This analysis can be used to lower bound performance (i.e. upper bound convergence rate) via multiple parameters computed from the learning matrix $L_\tdx$ and the true dynamics $\gtrue$. The mean transient convergence rate (\ref{eq:convRate}) may thus serve as a specific, measurable counterpart to any theoretical worst-case-scenario analyses performed via the formulas in \cite{Avrachenkov1998}.

Finally, to verify the fundamental necessity and efficacy of 
\ILILC{} for systems with unstable inverses, 
2 trials of
traditional stable-inversion-free \AILC{} \directAILC{}
are applied to each truth model.

\subsection{Results and Discussion}
The condition number of $\jacobian{\gmod}{\uvec}\left(\uvec_0\right)$ is \SI{1e17}{}.
Attempted inversion of this matrix in MATLAB 
yields an inverse matrix with average nonzero element magnitude of \SI{4e13}{} and max element magnitude of \SI{3e16}{}. Consequently, $\uvec_1$ generated by \directAILC{} has an average element magnitude of \SI{2e10}{\meter} and a max element magnitude of \SI{8e11}{\meter}, which is so large that $\yvec_1$ and $\jacobian{\gmod}{\uvec}\left(\uvec_1\right)$ contain \texttt{NaN} elements for all simulations.
Conversely, while some simulations using 
\ILILC{}, i.e. 
\SAILControl{}, diverge due to excessive model error, the majority converge. This validates the fundamental claim that the direct application of Newton's method in \AILC{} is insufficient for 
systems with unstable inverses, 
and that slight modification of the learning matrix and the incorporation of stable inversion addresses this gap.

To accompany the quantitative metric $\norm{e_\theta}_2$,
Figure \ref{fig:solution} 
offers a qualitative sense of the degree of model 
error
in this study by
comparing two representative 
\ILILC{} 
solution trajectories $u_{50}(k)$ with the solution to the $\norm{e_\theta}_2=0$, $\omega_\force(k)=\omega_y(k)=0$ scenario.
The lower-model-error representative solution is from within the range of $\norm{e_\theta}_2$ for which all simulations converged, while the higher-model-error solution comes from a bin in which some simulations diverged.
A more detailed analysis of the boundaries in $\theta$-space determining convergence or divergence of a simulation is beyond the scope of this work.
However, the given trajectories illustrate that even in the conservative subspace defined by the 100\% convergent bins learning bridges a visible performance gap, and that beyond this subspace there are far greater performance gains to be had.

\begin{figure}
    \centering
    \includegraphics[scale=0.9]{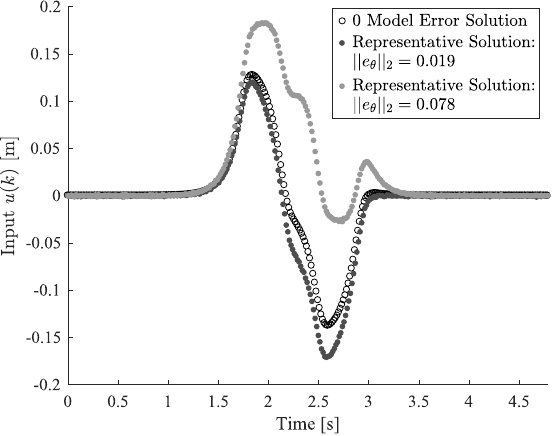}
    \caption{
    Representative input solution trajectories from low- and high-model-error 
    \ILILC{} 
    simulations compared with the solution to the zero-model-error problem.
    The zero-model-error solution is the input trajectory that would be chosen for feedforward control in the absence of learning, and differs notably from both minimum-error trajectories found by 
    \ILILC{}  
    with stable inversion.
    }
    \label{fig:solution}
\end{figure}

Finally, a statistical
comparison
of the performance and robustness of 
\ILILC{} 
with stable inversion and gradient ILC is given in Figure \ref{fig:data}.
The tuning of gradient ILC indeed yields comparable robustness to 
\ILILC,
with 
\ILILC{} 
97\%
as likely to converge as gradient ILC over all simulations.
The convergence rates of the two ILC schemes, however, differ substantially, with gradient ILC taking over 3 times as many trials as 
\ILILC{} 
to converge on average.
The mean transient convergence rate values tabulated in Table \ref{table:convRate} give a more portable quantification of 
\ILILC's 
advantage,
having a convergence rate 
nearly half
that of gradient ILC's.

This analysis confirms that 
\ILILC{} 
with stable inversion is an important addition to the engineer's toolbox because it enables ILC synthesis from nonlinear non-minimum phase models and delivers the fast convergence characteristic of algorithms based on Newton's method.

\begin{figure}
    \centering
    \includegraphics[scale=0.925,trim={0.14in 0in 0.3in 0.1in},clip]{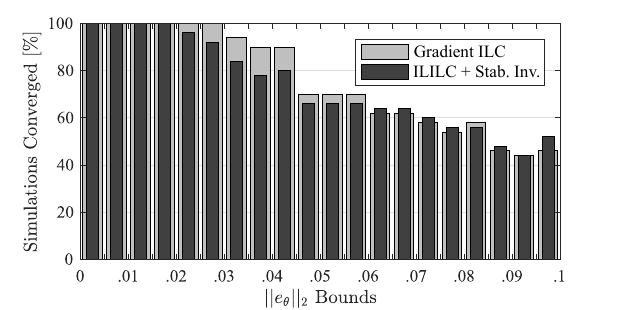}
    \\
    \includegraphics[scale=0.925,trim={0.14in 0in 0.3in 0.1in},clip]{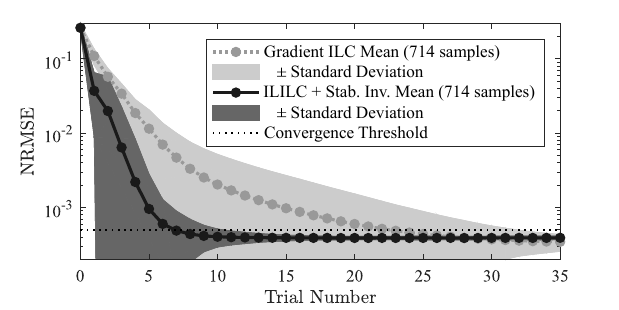}
    \caption{
    \emph{Top:} 
    Histogram giving the percentage of simulations converged 
    in each bin of the model error metric $\norm{e_\theta}_2$.
    \emph{Bottom:}
    Mean value of NRMSE for each ILC trial over all simulations that are convergent for both gradient ILC and 
    \ILILC{}
    with stable inversion.
    This illustrates that for comparable robustness to model error, 
    \ILILC{} 
    converges substantially faster than gradient ILC.
    }
    \label{fig:data}
\end{figure}

\begin{table}
\caption{Transient Convergence Rates for 
\ILILC{}
and Gradient ILC}
\centering
\renewcommand{\arraystretch}{1.25}
\begin{tabular}{|l|c|c|}
\hline
ILC Law           & \multicolumn{1}{l|}{Mean} & \multicolumn{1}{l|}{Standard Deviation} \\ \hhline{|=|=|=|}
Gradient ILC      & 0.76                      & 0.17                                    \\ \hline
\ILILC{}
+ Stab. Inv. & 0.41                      & 0.27                                    \\ \hline
\end{tabular}
\label{table:convRate}
\end{table}

\section{Conclusion}
\label{sec:conc4}
This chapter introduces and validates 
a new
ILC synthesis scheme applicable to nonlinear time-varying 
systems with unstable inverses and 
relative degree greater than 1.
 This is done with the support of nonlinear stable inversion, which is 
 advanced from the prior art via proof of convergence for an expanded class of systems and methods for improved practical implementation.
In all, this results in a new, 
broadly implementable ILC scheme 
displaying
a competitive convergence speed under benchmark testing.

Future work may focus on
further broadening the applicability of 
\ILILC{} 
by relaxing reference and initial condition repetitiveness assumptions, and on
the extension of 
\ILILC{} 
with a potentially adaptive tuning parameter or other means to enable the exchange of some speed for robustness when called for.
Levenberg-Marquardt-Fletcher algorithms may offer one source of inspiration for such work.

\chapter{
Hybrid Systems with Unstable Inverses:
\texorpdfstring{\\}{ }
Stable Inversion of Piecewise Affine Systems
}
\blfootnote{
Content of this chapter also to be submitted as:
\\
\indent\indent
I. A. Spiegel, N. Strijbosch, R. de Rozario, T. Oomen and K. Barton, ``Stable Inversion of Piecewise Affine Systems with Application to Feedforward and Iterative Learning Control,'' in \emph{IEEE Transactions on Automatic Control}.
Copyright may be transferred without notice, after which this version may no longer be accessible.
}
\label{ch:5}

As discussed in Chapter \ref{ch:1}, despite the call for it \cite{VandeWouw2008}, little work has been done on the feedforward control of \ac{PWA} systems. Naturally, a lack of feedforward control research in general implies a lack of research on the more specific problem of feedforward control for \ac{NMP} \ac{PWA} systems.

Thus, while the combination of \ac{ILILC} and stable inversion is promising for the \ac{ILC} of \ac{PWA} systems with \ac{NMP} dynamics, there are two major gaps in the literature that must be filled before such control can be realized.
First, there does not exist any published theory on stable inversion for \PWA{} systems (or hybrid systems of any kind).
Second, there does not exist sufficient literature on 
the 
conventional
closed-form
\PWA{} system inversion 
prerequisite for stable inversion.

The most relevant prior art on \PWA{} system inversion is Sontag's foundational work on piecewise linear systems \cite{Sontag1981}. Sontag proposes that these systems are potentially 
invertible-with-delay, and explains the signal time shifting necessary to accommodate this delay. This serves as a beginning for the concept of the relative degree of a \PWA{} system. However, these concepts require further development to account for the 
fact
that 
a \PWA{} system's apparent relative degree
may change during switching between component models and
the fact
that \PWA{} systems may have multiple inverses. Additionally, Sontag leaves the derivation of the inverse dynamics other than the delay open as a ``nontrivial part'' of the inversion process. Finally, for reference tracking it is desirable to invert a system without delay. This frequently results in anticausal inverses---those in which future reference values are required to compute a current input---but this is no problem for typical feedforward control scenarios where the entire reference is known 
in advance.

Further work of Sontag
\cite{Sontag1982}
delves into the abstract algebra of piecewise linear functions and includes the useful facts that it is decidable whether two sets are isomorphic under piecewise linear functions, that 
the class of
piecewise linear functions 
is
closed under composition, and that if a piecewise linear function has an inverse, it has a piecewise linear inverse. These facts are promising 
and may support dynamical systems research, but are still quite distant from the problem of actually finding the inverse of a dynamical system. Additionally, beyond the suggestion of a direction for future work in abstract algebra, no information is given on the uniqueness of the inverse. Finally, these results do not apply to time-varying systems, which must be considered if one is to perform feedforward control of a system also subject to feedback control with a time-varying reference.

In short, 
the lack of general \PWA{} system inversion theory stymies the output reference tracking control of a broad range of practical systems, especially those with NMP models, for which ILC-based workarounds are currently inapplicable. The present work fills this gap by contributing
\begin{enumerate}
    \item
    formal inversion formulas for a class of \PWA{} systems,
    \item
    sufficient conditions for inverse uniqueness and the corresponding explicit inverse representations,
    \item
    stable inversion 
    theory and implementation
    methods for \PWA{} models with NMP components, 
    and
    \item
    a simulation validation of this theory via the \NILC{} of an NMP \PWA{} system.
\end{enumerate}

Section \ref{sec:sysdef} formally defines the class of \PWA{} systems treated in this work and a concept of relative degree for deterministic hybrid systems in general. Section \ref{sec:exactinverse} proves inverse \PWA{} system formulas and sufficient conditions for inverse uniqueness. Section \ref{sec:stabinv} presents the stable inversion theory for \PWA{} systems. Section \ref{sec:val} presents and discusses the validation of the inversion and stable inversion theory via the \ILILC{} of an NMP system. Finally, Section \ref{sec:conc5} presents conclusions and recommendations for future work.

\section{System Definition}
\label{sec:sysdef}
A variety of similar 
discrete-time 
\PWA{} system definitions appear in the literature. For the sake of 
familiarity and 
simplicity, the following 
time-varying 
system definition uses a representation similar to \cite{Heemels2001} and \cite{Gao2009}.
\begin{definition}[\PWA{} System] A \PWA{} system is given by
\begin{equation}
\begin{aligned}
    \state\kdxplus{1} &= A_{q,k}\state\kdx + B_{q,k}\uMIMO\kdx + F_{q,k}
    \\
    \yMIMO\kdx &= C_{q,k}\state\kdx + D_{q,k}\uMIMO\kdx + G_{q,k}
\end{aligned}
\quad \textrm{for }
    \state\kdx 
    \in Q_q
\end{equation}
where $k\in\integer$ is the time-step index, 
$\state\in\real^{n_x}$, $\uMIMO\in\real^{n_u}$, and $\yMIMO\in\real^{n_y}$ are the state, input, and output vectors,
and $Q_q\in \Qset$ where $\Qset$ is the set of ``locations,'' i.e. a set of disjoint 
regions with union equal to 
$\real^{n_x}$.
Each location $Q_q$ is the union of a set of disjoint 
convex
polytopes.
Here, a
convex
polytope is defined simply as an intersection of half spaces.
Additionally, let the relative degree of the $q$\textsuperscript{th} component model be denoted $\mu_q$ for $q\in\llbracket\, 1,\regionQuant\,\rrbracket$ ($\llbracket$ $\rrbracket$ indicates a closed set of integers).
\end{definition}

To facilitate both mathematical analysis and controller synthesis, the remainder of the 
chapter 
uses the equivalent closed-form representation
\begin{IEEEeqnarray}{RL}
\eqlabel{eq:sysDefPWA}
\IEEEyesnumber
\IEEEyessubnumber*
\label{eq:sysDef_x}
    \state\kdxplus{1} &= 
    \Asum\kdx \state\kdx + \Bsum\kdx u\kdx + \Fsum\kdx
\\
    y\kdx &= \Csum\kdx\state\kdx + \Dsum\kdx u\kdx + \Gsum\kdx
\label{eq:sysDef_y}
\end{IEEEeqnarray}
with the upright, bold, capital letter notation defined as
\begin{align}
\Msum\kdx &\coloneq \sumqQ M_{q,k} K_q( \locvec\kdx )
\label{eq:Mshorthand}
\\
K_q(\locvec\kdx) &\coloneq 0^{
\prod_{i=1}^{|\sigvecset_q|}
\norm{\sigvec_{q,i} - \locvec\kdx}
} = 
\begin{cases}
1 & \locvec\kdx\in\sigvecset_q
\\
0 & \textrm{otherwise}
\end{cases}
\label{eq:selector}
\\
\locvec\kdx=\locvec(\state\kdx) &\coloneqq H\left( P\state\kdx - \offsetvec \right)
\label{eq:localization}
\end{align}
where $\Msum$ and $M_{q,k}$ stand in for any of
$\{\Asum,\Bsum,\Fsum,\Csum,\Dsum,\Gsum\}$,
and 
$\{ A_{q,k},\allowbreak B_{q,k},\allowbreak F_{q,k},\allowbreak C_{q,k},\allowbreak D_{q,k},\allowbreak G_{q,k} \}$, respectively.
$H$ is the Heaviside step function evaluated element-wise on its vector argument,
$K_q$ is the binary-output selector function for the $q$\textsuperscript{th} location,
$\sigvecset_q=\{\sigvec_{q,1},\allowbreak\,\sigvec_{q,2},\allowbreak\,\allowbreak\cdots\}$ is the set of binary vector signatures of the $q$\textsuperscript{th} location,
$P\in\real^{n_P\times n_\state}$ is a matrix consisting of concatenated hyperplane orientation vectors, and $\offsetvec\in\real^{n_P}$ is a vector of hyperplane offsets. 
For more information about closed form representations of piecewise defined systems, 
see Chapter
\ref{ch:3}.

Furthermore,
the following assumptions are made for all systems
\begin{enumerate}[label=(A5.\arabic*),leftmargin=*]
    \item
    \label{C:reachable}
    $\state_0 \in X_0$, where $X_0$ is the set of initial conditions from which all locations $Q_q \in Q$ are reachable in finite time
    \item 
    \label{C:SISO}
    the system is single-input-single-output
    (SISO),
    \\
    $n_u=n_y=1$
    \item
    \label{C:stateswitch}
    switching depends only on the states, not the input
    \item
    \label{C:muc}
    all component models have the same relative degree, $\mu_c$, for all time,
    $\mu_q=\mu_c$ $\forall q\in \llbracket 1,\regionQuant\rrbracket$ and $\forall k$
    \suspend{listA}
\end{enumerate}

Note that while the assumption \ref{C:stateswitch} is implied by (\ref{eq:localization}), the other assumptions are not implied by the system representation 
\sysDef-(\ref{eq:localization}).

Finally, this work introduces the concept of the ``global dynamical relative degree:''
\begin{definition}[Global Dynamical Relative Degree]
\label{def:mu}
The \emph{global dynamical relative degree} of a SISO \PWA{} system
is the smallest number $\dmu\geq 0$ such that 
the explicit expression of $y_{k+\dmu}$ in terms of component state space matrices, selector functions, $\state_k$, and $u_i, \,i\geq k$ contains $u_k$ outside of a selector function for all switching sequences on the interval $\llbracket k,k+\dmu \rrbracket$.
\end{definition}

This $\dmu$ is essentially the traditional relative degree, but neglecting inputs appearing in selector functions. This neglect is introduced to avoid situations in which the only explicit appearance of the input in the output function is within the Heaviside function, leading to a potentially infinite number of input values yielding the same output.
Such non-injectiveness would make inversion unusually challenging. The assumption \ref{C:stateswitch} is made for similar reasons.

\section{Conventional \PWA{} System Inversion}
\label{sec:exactinverse}
This section presents the conventional exact inverses of \PWA{} systems under assumptions \baseAssumptions. 
Conventional inversion is the process of
\begin{enumerate}[label=(Step \arabic*),leftmargin=*]
    \item
    \label{step1}
    obtaining an expression for the previewed output $y\kdxplus{\dmu}$ in terms of $\state\kdx$, $\Msum\kdxplus{\idx}$, and $u\kdxplus{\idx}$ where $\idx\geq 0$,
    \item
    \label{step2}
    solving the previewed output equation for $u\kdx$ in terms of $\state\kdx$, $\Msum\kdxplus{\idx}$, and $u\kdxplus{\idx+1}$, and finally
    \item
    \label{step3}
    taking this expression of $u\kdx$ as the output function of the inverse system, and plugging it in to 
    (\ref{eq:sysDef_x}) to obtain the state transition formula of the inverse system.
\end{enumerate}
Note that $y\kdxplus{\dmu}$ is necessarily an explicit function of $u\kdx$ by Definition \ref{def:mu}.

For systems with $\dmu=0$, this inverse system is unique, and can be expressed explicitly. For systems with $\dmu\geq1$, there may be multiple solutions to the problem of solving $y\kdxplus{\dmu}$ for $u\kdx$ \ref{step2}, and thus the inverse system cannot be expressed explicitly without additional assumptions. This section gives both the general, implicit inverse system for $\dmu\geq 1$ systems and sufficient conditions for the uniqueness of system inversion for $\dmu\in\{1,2\}$ systems along with the corresponding explicit inverse systems.

\subsection{Unique Exact Inversion For 
\texorpdfstring{$\dmu=0$}{mu=0}
}

\begin{lemma}[Relative Degree of 0]
\label{lem:mu0}
The global dynamical relative degree of a reachable SISO \PWA{} system, i.e. a \PWA{} system satisfying \ref{C:reachable} and \ref{C:SISO}, is 0 if and only if the relative degree of all component models are 0 for all time:
\begin{equation}
\label{eq:mu0}
\mu_q=0 \,\,\,\, \forall q\in \llbracket 1, \regionQuant \rrbracket \text{, } \forall k \iff \dmu=0
\end{equation}
\end{lemma}
\begin{proof}
First, the forward implication is proven directly.
\begin{align}
    \mu_q=0 \,\,\,\, 
    \forall q,k
    &\implies 
    D_{q,k} \neq 0 \,\,\,\, 
    \forall q,k
    \\
    D_{q,k} \neq 0 \,\,\,\, 
    \forall q,k
    &\implies
    \Dsum\kdx  \neq 0 \,\,\,\, \forall k
    \\
    \Dsum\kdx  \neq 0 \,\,\,\, \forall k
    &\implies
        \text{(\ref{eq:sysDef_y}) always explicitly contains $u\kdx$}
\\
    \therefore \,\, 
    \mu_q=0 \,\, \forall \,\, q\in Q 
    &\implies 
    \dmu=0
    \label{eq:mu0direct}
\end{align}

Now the backwards implication is proven by proving the contrapositive.
\begin{align}
    \exists q,k \textrm{ s.t. } \mu_q \neq 0
    &\implies
    \exists q,k \textrm{ s.t. } D_{q,k} = 0
\end{align}
By \ref{C:reachable}, there exists some finite sequence of inputs to bring the system to this location $\Qset_q$ at some time step $k$ such that $\Dsum\kdx  = D_{q,k} = 0$, and  (\ref{eq:sysDef_y}) becomes
\begin{equation}
    y\kdx = \Csum\kdx \state\kdx + \Gsum\kdx
\end{equation}
which is not an explicit function of $u\kdx $, meaning $\dmu \neq 0$.
\begin{equation}
\label{eq:mu0contrapositive}
    \therefore 
    \neg \left(
    \mu_q=0 \,\,\,\, \forall q\in \llbracket 1,\regionQuant \rrbracket, \,\,\forall k
    \right)
    \implies
    \neg \left(
    \dmu=0
    \right)
\end{equation}
By (\ref{eq:mu0direct}) and (\ref{eq:mu0contrapositive}), (\ref{eq:mu0}) must be true.
\end{proof}

\begin{theorem}[$\dmu=0$ \PWA{} System Inverse]
\label{thm:mu0inv}
The inverse of a \PWA{} system satisfying \baseAssumptions{} with $\mu_c=0$ is itself a \PWA{} system satisfying \baseAssumptions{} and is given by
\begin{align}
    \state\kdxplus{1} &= 
    \inv{\Asum}\kdx \state\kdx  + \inv{\Bsum}\kdx y\kdx  + \inv{\Fsum}\kdx 
    \\
    u\kdx  &= \inv{\Csum}\kdx  \state\kdx  + \inv{\Dsum}\kdx  y\kdx  + \inv{\Gsum}\kdx
\end{align}
where
\begin{equation*}
\begin{aligned}
    \inv{\Asum}\kdx  &= \Asum\kdx -\Bsum\kdx \Dsum\kdx ^{-1}\Csum\kdx 
    & & &
    \inv{\Bsum}\kdx  &= \Bsum\kdx \Dsum\kdx ^{-1}
    & & &
    \inv{\Fsum}\kdx  &= \Fsum\kdx -\Bsum\kdx \Dsum^{-1}\kdx \Gsum\kdx 
    \\
    \inv{\Csum}\kdx  &= -\Dsum^{-1}\kdx \Csum\kdx
    & & &
    \inv{\Dsum}\kdx  &= \Dsum^{-1}\kdx
    & & &
    \inv{\Gsum}\kdx  &= -\Dsum^{-1}\kdx  \Gsum\kdx 
\end{aligned}
\end{equation*}
such that $y\kdx $ is the input of the inverse system and $u\kdx $ is the output.
\end{theorem}
\begin{proof}
\ref{step1} is satisfied by (\ref{eq:sysDef_y}) and \ref{step2} 
by
\begin{equation}
    u\kdx  = \Dsum\kdx ^{-1}(y\kdx -\Csum\kdx \state\kdx  - \Gsum\kdx )
    \label{eq:umu0}
\end{equation}
Equation (\ref{eq:umu0}) is always well-defined because by Lemma \ref{lem:mu0} ~$\mu_c=0\implies\dmu=0$, which in turn implies $\Dsum\kdx  \neq 0$ $\forall k$, and because $\Dsum$ is always scalar because the system is SISO by \ref{C:SISO}.
Plugging (\ref{eq:umu0}) into (\ref{eq:sysDef_x}) yields the inverse system state transition formula, satisfying \ref{step3}.
\end{proof}
\subsection{Non-unique Exact Inversion For 
\texorpdfstring{$\dmu\geq 1$}{mu>=1}
}
For systems with $\dmu\geq 1$, (\ref{eq:sysDef_y}) does not explicitly contain $u\kdx$ because $\Dsum\kdx = 0$ $\forall k$; this is corollary to Lemma \ref{lem:mu0}, Definition \ref{def:mu}, and \ref{C:muc}. Consequently, \ref{step1} necessitates the derivation of an explicit formula for the output preview $y\kdxplus{\dmu}$, i.e. the output at a time after the current time step $k$. This preview of future output is necessary for deriving an output equation that explicitly depends on the current input $u\kdx$.

\begin{lemma}[$\dmu\geq 1$ \PWA{} System Output Preview]
\label{lem:previewOutput}
Given a \PWA{} system satisfying \baseAssumptions{} with known global dynamical relative degree $\dmu$, the output function with minimum preview such that the function is explicitly dependent on an input term outside of the selector functions for any switching sequence is given by
\begin{equation}
\label{eq:previewOutput}
    y\kdxplus{\dmu} = \prevC\kdx \state\kdx + \prevD\kdx u\kdx + \prevG\kdx + \anticaus\kdx(u\kdxplus{1},\cdots,u\kdxplus{\dmu-1})
\end{equation}
with
\begin{align*}
&\prevC\kdx \coloneq
\Csum\kdxplus{\dmu} \left(\prod_{m=0}^{\dmu-1}\Asum\kdxplus{m}\right)
\qquad
\prevD\kdx \coloneq
\Csum\kdxplus{\dmu}
    \left( \prod_{m=1}^{\dmu-1} \Asum\kdxplus{m} \right) \Bsum_{k}
\\
&\prevG\kdx \coloneq
\Csum\kdxplus{\dmu} \sum_{s=0}^{\dmu-1} \left(
    \left( \prod_{m=s+1}^{\dmu-1} \Asum\kdxplus{m} \right)
    \Fsum\kdxplus{s}
    \right) + \Gsum\kdxplus{\dmu}
\\
&\anticaus\kdx(u\kdxplus{1},\cdots)
\coloneq
\Csum\kdxplus{\dmu}
    \sum_{s=1}^{\dmu-1} 
    \left(
    \left( \prod_{m=s+1}^{\dmu-1} \Asum\kdxplus{m} \right) 
    \Bsum\kdxplus{s} u\kdxplus{s}
    \right)
\end{align*}
where the factors in the products generated by $\prod$ are ordered sequentially by the index $m$. The factor corresponding to the greatest value of the index must be on the left, and the factor corresponding to the smallest value of the index must be on the right. 
For example,
\begin{equation}
    \prod_{m=0}^{2}\Asum\kdxplus{m} \equiv \Asum\kdxplus{2}\Asum\kdxplus{1}\Asum\kdxplus{0} 
    \not\equiv \Asum\kdxplus{0}\Asum\kdxplus{1}\Asum\kdxplus{2}
\end{equation}
because matrices do not necessarily commute.
Additionally, if the lower bound on the product index exceeds the upper bound on the product index (an ``empty product''), then the product resolves to the identity matrix. Similarly, empty sums resolve to $0$. 
For example
\begin{align}
    \prod_{m=1}^{0}\textup{anything} \equiv I
    \qquad \qquad
    \sum_{s=1}^{0}\textup{anything} \equiv 0
\end{align}
\end{lemma}
\begin{proof}
Let the base case of the proof by induction be $\dmu=1$ such that
\begin{equation}
\label{eq:BaseCase}
    y\kdxplus{1} = \Csum\kdxplus{1} \left(
    \Asum\kdx  \state\kdx 
    +
    \Fsum_{k}
    +
    \Bsum\kdx  u\kdx 
    \right)
    +
    \Gsum\kdxplus{1}
\end{equation}
which is achieved equivalently from (\ref{eq:previewOutput}) and from the system definition by plugging (\ref{eq:sysDef_x}) into (\ref{eq:sysDef_y}) incremented by one time step (i.e. plugging the equation for $\state\kdxplus{1}$ into the equation for $y\kdxplus{1}$).
The preview is minimal because, by Definition \ref{def:mu}, $\Csum\kdxplus{1}\Bsum\kdx \neq 0$ for all switching sequences on $\llbracket k,k+1 \rrbracket$, and by Lemma \ref{lem:mu0} $\dmu\geq 1 \implies \Dsum\kdx =0$ $\forall k$. 
In other words, for $\dmu=1$, (\ref{eq:BaseCase}) is always an explicit function of $u\kdx$. Thus, regardless of switching sequence, outputs further in the future, such as $y\kdxplus{2}$, never need to be considered in order to explicitly relate the current input to an output.

Then consider (\ref{eq:previewOutput}) with $\dmu=\nu$ as the foundation of the induction step. To prove (\ref{eq:previewOutput}) holds for $\dmu=\nu+1$, first increment (\ref{eq:previewOutput}) with $\dmu=\nu$ by one time step, yielding
\begin{multline}
    y\kdxplus{\nu+1} =
    \Csum\kdxplus{\nu+1}\left( \prod_{m=0}^{\nu-1}\Asum\kdxplus{1+m}   \right)\state\kdxplus{1}
    +
    \Csum\kdxplus{\nu+1}\left(\prod_{m=1}^{\nu-1}\Asum\kdxplus{1+m}\right)\Bsum\kdxplus{1}u\kdxplus{1}
    +
    \\
    \Csum\kdxplus{\nu+1}\sum_{s=0}^{\nu-1}\left( \left(\prod_{m=s+1}^{\nu-1}\Asum\kdxplus{1+m}\right)\Fsum\kdxplus{1+s} \right) + \Gsum\kdxplus{1+\nu}
    +
    \Csum\kdxplus{\nu+1}\sum_{s=1}^{\nu-1}\left(\left(\prod_{m=s+1}^{\nu-1}\Asum\kdxplus{1+m}\right) \Bsum\kdxplus{1+s}u\kdxplus{1+s}\right)
\end{multline}
This can be simplified by first factoring out $\Csum\kdxplus{\nu+1}$ and adjusting the product $\Pi$ indices to subsume the constant $+1$ in $\Asum\kdxplus{1+m}$:
\begin{multline}
    y\kdxplus{\nu+1} =
    \Csum\kdxplus{\nu+1}
    \left(\left( \prod_{m=1}^{\nu}\Asum\kdxplus{m}   \right)\state\kdxplus{1}
    +
    \left(\prod_{m=2}^{\nu}\Asum\kdxplus{m}\right)\Bsum\kdxplus{1}u\kdxplus{1}
    +
    \sum_{s=0}^{\nu-1}\left( \left(\prod_{m=s+2}^{\nu}\Asum\kdxplus{m}\right)\Fsum\kdxplus{1+s} \right) 
    \right.
    \\
    \left.
    +
    \sum_{s=1}^{\nu-1}\left(\left(\prod_{m=s+2}^{\nu}\Asum\kdxplus{m}\right) \Bsum\kdxplus{1+s}u\kdxplus{1+s}\right)
    \right)
    +
    \Gsum\kdxplus{\nu+1}
\end{multline}
Similarly, the sum $\sum$ indices may be adjusted to subsume the constant $+1$ in $\Fsum\kdxplus{1+s}$, $\Bsum\kdxplus{1+s}$, and $u\kdxplus{1+s}$. Because the sum index $s$ also appears in the lower bound of the products, $m=s+2$, the product index lower bound must also be adjusted.
\begin{multline}
    y\kdxplus{\nu+1} =
    \Csum\kdxplus{\nu+1}
    \left(\left( \prod_{m=1}^{\nu}\Asum\kdxplus{m}   \right)\state\kdxplus{1}
    +
    \left(\prod_{m=2}^{\nu}\Asum\kdxplus{m}\right)\Bsum\kdxplus{1}u\kdxplus{1}
    +
    \sum_{s=1}^{\nu}\left( \left(\prod_{m=s+1}^{\nu}\Asum\kdxplus{m}\right)\Fsum\kdxplus{s} \right) 
    \right.
    \\
    \left.
    +
    \sum_{s=2}^{\nu}\left(\left(\prod_{m=s+1}^{\nu}\Asum\kdxplus{m}\right) \Bsum\kdxplus{s}u\kdxplus{s}\right)
    \right)
    +
    \Gsum\kdxplus{\nu+1}
\end{multline}
Finally, the two input terms (those containing $u$, arising from $\prevD\kdxplus{1}$ and $\anticaus\kdxplus{1}$) can be combined to achieve
\begin{multline}
\label{eq:lem2increment}
    y\kdxplus{\nu+1} = 
    \Csum\kdxplus{\nu+1}
    \left(
    \left(\prod_{m=1}^{\nu}\Asum\kdxplus{m}\right) 
    \state\kdxplus{1}
    +
    \sum_{s=1}^{\nu} \left(
    \left( \prod_{m=s+1}^{\nu} \Asum\kdxplus{m} \right)
    \Fsum\kdxplus{s}
    \right)
    \right.
    \\
    \left.
    +
    \sum_{s=1}^{\nu} 
    \left(
    \left( \prod_{m=s+1}^{\nu} \Asum\kdxplus{m} \right) 
    \Bsum\kdxplus{s} u\kdxplus{s}
    \right)
    \right)
    +
    \Gsum\kdxplus{\nu+1}
\end{multline}
which is a function of $\state\kdxplus{1}$ and potentially $u\kappadx$ for $\kappa \in~ \llbracket k+~1,k+\nu+1 \rrbracket$.
The dependence of $y\kdxplus{\nu+1}$ on the input terms is conditioned on the switching sequence. Definition \ref{def:mu} implies that if $\dmu=\nu+1$ there exists some switching sequence on $\llbracket k,k+\nu+1 \rrbracket$ such that the input coefficients in (\ref{eq:lem2increment}) are zero, i.e.
\begin{equation}
\exists\, \{\state\kappadx\, |\, \kappa \in \llbracket k,k+\nu+1\rrbracket\} 
\quad
\text{ s.t. }
\quad
\forall s\in\llbracket 1,\nu\rrbracket \quad
    \Csum\kdxplus{\nu+1}
    \left( \prod_{m=s+1}^{\nu} \Asum\kdxplus{m} \right) 
    \Bsum\kdxplus{s}
    =
    0
\end{equation}
Thus, to guarantee the expression for $y\kdxplus{\nu+1}$ explicitly contains the input for all switching sequences,
$\state\kdxplus{1}$ in (\ref{eq:lem2increment}) must be expanded (via (\ref{eq:sysDef_x})) to be in terms of $\state\kdx $ and $u\kdx $ explicitly.
The resulting expression can be rearranged as follows:
\begin{multline}
    y\kdxplus{\nu+1} = 
    \Csum\kdxplus{\nu+1}
    \left(
    \left(\prod_{m=1}^{\nu}\Asum\kdxplus{m}\right) 
    \left(
    \Asum\kdx  \state\kdx  + \Bsum\kdx  u\kdx  + \Fsum\kdx 
    \right)
    \right.
    \\
    \left.
    +
    \sum_{s=1}^{\nu} \left(
    \left( \prod_{m=s+1}^{\nu} \Asum\kdxplus{m} \right)
    \Fsum\kdxplus{s}
    \right)
    \right)
    +
    \anticaus\kdx(u\kdxplus{1},\cdots,u\kdxplus{\nu})
    +
    \Gsum\kdxplus{\nu+1}
\end{multline}
\begin{multline}
    y\kdxplus{\nu+1} = 
    \Csum\kdxplus{\nu+1}
    \left(
    \left(\prod_{m=1}^{\nu}\Asum\kdxplus{m}\right) 
    \Asum\kdx  \state\kdx 
    +
    \left(\prod_{m=1}^{\nu}\Asum\kdxplus{m}\right)
    \Bsum\kdx  u\kdx 
    \right.
    \\
    \left.
    +
    \sum_{s=1}^{\nu} \left(
    \left( \prod_{m=s+1}^{\nu} \Asum\kdxplus{m} \right)
    \Fsum\kdxplus{s}
    \right)
    +
    \left(\prod_{m=1}^{\nu}\Asum\kdxplus{m}\right) \Fsum\kdx 
    +
    \Gsum\kdxplus{\nu+1}
    \right)
    % \\
    +\anticaus\kdx(u\kdxplus{1},\cdots,u\kdxplus{\nu})
\end{multline}
\begin{multline}
    y\kdxplus{\nu+1} = 
    \Csum\kdxplus{\nu+1}
    \left(
    \left(\prod_{m=0}^{\nu}\Asum\kdxplus{m}\right) 
    \state\kdx 
    +
    \left(\prod_{m=1}^{\nu}\Asum\kdxplus{m}\right)
    \Bsum\kdx  u\kdx
    \right.
    \\
    \left.
    +
    \sum_{s=0}^{\nu} \left(
    \left( \prod_{m=s+1}^{\nu} \Asum\kdxplus{m} \right)
    \Fsum\kdxplus{s}
    \right)
    \right)
    +
    \Gsum\kdxplus{\nu+1}
    +\anticaus\kdx(u\kdxplus{1},\cdots,u\kdxplus{\nu})
\label{eq:lem2conc}
\end{multline}
Equation (\ref{eq:lem2conc}) is 
equation (\ref{eq:previewOutput}) for $\dmu=\nu+1$, thereby proving the lemma.
\end{proof}

Using the output preview equation (\ref{eq:previewOutput}), a general \PWA{} system inverse for $\dmu\geq 1$ can be found in the same manner as for $\dmu=0$.

\begin{theorem}[General $\dmu\geq 1$ \PWA{} System Inverse]
\label{thm:genInv}
Given any \PWA{} system satisfying \ref{C:reachable}-\ref{C:muc} with known global dynamical relative degree $\dmu \geq ~1$, the inverse system with $u\kdx $ as the output is given by the implicit, anticausal system
\begin{align}
    \state\kdxplus{1} &= 
    \inv{\Asum}\kdx  \state\kdx  
    +
    \inv{\Bsum}\kdx y\kdxplus{\dmu}
    + \inv{\Fsum}\kdx
    -\Bsum\kdx\prevD\kdx^{-1}\anticaus\kdx\left(u\kdxplus{1},\cdots,u\kdxplus{\dmu-1}\right)
    \label{eq:xGen}
    \\
    u\kdx  &=  \inv{\Csum}\kdx\state\kdx+\inv{\Dsum}y\kdxplus{\dmu}+\inv{\Gsum}\kdx
    -\prevD\kdx^{-1}\anticaus\kdx\left(u\kdxplus{1},\cdots,u\kdxplus{\dmu-1}\right)
    \label{eq:uGen}
\end{align}
where
\begin{align*}
\inv{\Asum}\kdx &= \Asum\kdx + \Bsum\kdx\inv{\Csum}\kdx
&
\inv{\Bsum}\kdx &= \Bsum\kdx\inv{\Dsum}\kdx
&
\inv{\Fsum}\kdx &= \Fsum\kdx+\Bsum\kdx\inv{\Gsum}\kdx
\\
\inv{\Csum}\kdx &= -\inv{\Dsum}\kdx\prevC\kdx
&
\inv{\Dsum}\kdx &= \prevD\kdx^{-1}
&
\inv{\Gsum}\kdx&=-\inv{\Dsum}\kdx\prevG\kdx
\end{align*}
\end{theorem}
\begin{proof}
The sole term in
Lemma \ref{lem:previewOutput}'s
(\ref{eq:previewOutput})
containing $u\kdx $ outside of a selector function 
is $\prevD\kdx u\kdx$. The coefficient $\prevD\kdx=~\Csum\kdxplus{\dmu}\left( \prod_{m=1}^{\dmu-1} \Asum\kdxplus{m} \right) \Bsum_{k}$ is always scalar because the system is SISO, \ref{C:SISO}, and always nonzero by Definition \ref{def:mu}.
Thus (\ref{eq:previewOutput}) can be divided by $\prevD\kdx$ and $u\kdx $ can be arithmetically maneuvered onto one side of the equation by itself, yielding (\ref{eq:uGen}).

Equation (\ref{eq:uGen}) is implicit in general because (\ref{eq:previewOutput}) cannot be uniquely solved for $u\kdx$ in general. This is proven by presenting an example in which multiple input trajectories have the same output trajectory. Consider the two-location system
\begin{equation}
\begin{aligned}
    &
    A_{1,k}=\begin{bmatrix}0&1\\0&0\end{bmatrix}
    \quad
    A_{2,k}=\begin{bmatrix}0&2\\0&0\end{bmatrix}
    \quad
    B_{1,k}=B_{2,k}=\begin{bmatrix}0\\1\end{bmatrix}
    \\
    &
    C_{1,k}=C_{2,k}=\begin{bmatrix}1&0\end{bmatrix}
    \\
    &
    P=\begin{bmatrix}0&1\end{bmatrix}
    \qquad
    \offsetvec=1.5
    \qquad
    \sigvecset_1=\{1\}
    \qquad
    \sigvecset_2=\{2\}
\end{aligned}
\end{equation}
with $F$, $D$, and $G$ matrices all equal to zero. Given $\state\kdx=~[\,0,\,\,0\,]^T$, both $u\kdx=2$ and $u\kdx=1$ yield $y\kdxplus{\dmu}=2$. Because there is not a unique solution to (\ref{eq:previewOutput}) for $u\kdx$, there does not exist an explicit formula for the solution (\ref{eq:uGen}).

The system is necessarily anticausal because $u\kdx $ is necessarily a function of $y\kdxplus{\dmu}$ and $\dmu > 0$.
\end{proof}

\begin{remark}[Inverse Implicitness]
Analytically, the implicitness of (\ref{eq:uGen}) arises in $\inv{\Csum}\kdx$ through $\prevC\kdx$, which is a function of $\Csum\kdxplus{\dmu}$ by definition, 
and $\Csum\kdxplus{\dmu}$
is a function of $\state\kdxplus{\dmu}$ by 
(\ref{eq:Mshorthand})-(\ref{eq:localization}). 
Finally, $\state\kdxplus{\dmu}$ is a function of $u\kdx$ via
\begin{multline}
    \state\kdxplus{\dmu} = 
    \left(\prod_{m=0}^{\dmu-1}\Asum\kdxplus{m}\right) 
    \state\kdx 
    +
    \sum_{s=0}^{\dmu-1} \left(
    \left( \prod_{m=s+1}^{\dmu-1} \Asum\kdxplus{m} \right)
    \Fsum\kdxplus{s}
    \right)
    +
    \\
    \sum_{s=1}^{\dmu-1} 
    \left(
    \left( \prod_{m=s+1}^{\dmu-1} \Asum\kdxplus{m} \right) 
    \Bsum\kdxplus{s} u\kdxplus{s}
    \right)
    +
    \left( \prod_{m=1}^{\dmu-1} \Asum\kdxplus{m} \right) 
    \Bsum_{k} u_{k}
\end{multline}
following from $y\kdxplus{\dmu}=\Csum\kdxplus{\dmu}\state\kdxplus{\dmu}+\Gsum\kdxplus{\dmu}$ and Lemma \ref{lem:previewOutput}.
\end{remark}

\begin{remark}[Input Preview and Inter-location Relative Degree]
Note $\anticaus\kdx$ is written as a function of a set of previewed
$u$-values 
The written set of $u$-values is the maximum quantity of $u$-values that may be required by $\anticaus\kdx$. Depending on the switching sequence, fewer previewed $u$-values may be required. 
In fact, by the definition of global dynamical relative degree $\dmu$, there must exist a switching sequence for which no previewed $u$-values are required, because otherwise $\dmu$ would be smaller. In other words, Definition \ref{def:mu} and Lemma \ref{lem:previewOutput} imply
\begin{multline}
\label{eq:uPreviewless}
    \exists\{\state\kappadx \, |\,  \kappa\in\llbracket k,k+\dmu\rrbracket\} \text{ s.t. }
    \\
    \forall s\in\llbracket 1,\dmu-1\rrbracket \,\,\,
    \Csum\kdxplus{\dmu}
    \left(\prod_{m=s+1}^{\dmu
    -
    1}\Asum\kdxplus{m}\right)\Bsum\kdxplus{s}=0 
    \\
    \land \,\,\, \Csum\kdxplus{\dmu}\left(\prod_{m=1}^{\dmu-1}\Asum\kdxplus{m}\right)\Bsum\kdx\neq 0
\end{multline}

For affine time-invariant systems without piecewise definition, there is never required input preview because the expressions claimed equal to zero in (\ref{eq:uPreviewless}) reduce as
\begin{equation}
    \Csum\kdxplus{\dmu}
    \left(\prod_{m=s+1}^{\dmu
    -
    1}\Asum\kdxplus{m}\right)\Bsum\kdxplus{s}
    =
    CA^{\dmu-1-s}B
    \qquad
    s\in\llbracket 1,\,\dmu-1\rrbracket
\end{equation}
which are always zero regardless of state sequence.
However, it is important to emphasize that this is not the case for \PWA{} systems.
Even when the relative degrees of all component models 
are equal to $\mu_c$, \ref{C:muc}, the inter-location relative degree may not be equal to $\mu_c$.
More formally
\begin{align}
    &C_{1,k+1}B_{1,k} = 0 \land C_{2,k+1}B_{2,k} = 0 \centernot\implies C_{1,k+1}B_{2,k}=0
    \end{align}
    and
    \begin{align}
    &C_{1,k+1}B_{1,k} \neq 0 \land C_{2,k+1}B_{2,k} \neq 0 \centernot\implies C_{1,k+1}B_{2,k}\neq0
\end{align}
and this lack of conclusiveness regarding the inter-state relative degree generalizes to larger $\mu_c$. In short, the inter-state relative degree may be either lower or higher than $\mu_c$, and may be different for different switching sequences (with a maximum value of $\dmu$ as given by Definition \ref{def:mu}).
\end{remark}

Because of the implicitness and potential requirement for input preview in the general \PWA{} system inverse,
it is nontrivial to use it for computing input trajectories from output trajectories.
However, there are conditions under which inversion of a \PWA{} system with $\dmu\geq 1$ is unique, and the inverse itself becomes an explicit \PWA{} system as is the case for $\dmu=0$.
The remainder of this section provides such sufficient conditions for the cases of $\dmu=1$ and $\dmu=2$.
\subsection{Unique Exact Inverses For 
\texorpdfstring{$\dmu\in\{1,2\}$}{mu=\{1,2\}}
}
First the ``location-independent output function'' assumption is introduced:
\begin{enumerate}[label=(A5.\arabic*),leftmargin=*]
    \resume{listA}
    \item 
    \label{C:locIndep}
    $\Csum\kdx=C\kdx$, $\Dsum\kdx=D\kdx$, $\Gsum\kdx=G\kdx$, with $C\kdx$, $D\kdx$, $G\kdx$ indicating parameters that potentially vary with time but that are identical $\forall q\in\llbracket 1,\regionQuant\rrbracket$
    \suspend{listB}
\end{enumerate}

\begin{lemma}[Relative Degree of 1]
\label{lem:mu1}
A \PWA{} system satisfying \ref{C:reachable}-\ref{C:locIndep} has a global dynamical relative degree of 1 if and only if the relative degree of all component models are 1 for all time:
\begin{equation}
    \mu_c=1 \iff \dmu=1
    \label{eq:LemMu1}
\end{equation}
\end{lemma}
\begin{proof}
The foreward implication follows directly from \ref{C:muc}, \ref{C:locIndep}, and $\mu_c=1$:
    \begin{align}
        \mu_c=1 &\implies D\kdxplus{1}=0 
        \,\land\, C\kdxplus{1}B_{q,k}\neq 0 \,\,\,\forall q\in\llbracket 1,\regionQuant\rrbracket
        \\
        &\implies
        \Dsum\kdxplus{1}=0 \,\land\, \Csum\kdxplus{1}\Bsum\kdx\neq 0 \,\,\,\forall k
    \end{align}
This is equivalent to implying that the inter-location 
relative degree
can be neither higher nor lower than $\mu_c=1$ under \ref{C:muc} and \ref{C:locIndep}
\begin{equation}
    \therefore \,\,\, \mu_c=1\implies \dmu=1
\end{equation}
The backward implication follows from (\ref{eq:uPreviewless}) (i.e. the combination of Definition \ref{def:mu} and Lemma \ref{lem:previewOutput}), which implies
\begin{align}
    \dmu=1\implies \exists\{\state\kdx,\state\kdxplus{1}\} \text{ s.t. } \Csum\kdxplus{1}\Bsum\kdx\neq 0
\end{align}
Then, by \ref{C:muc}, \ref{C:locIndep}, and the fact that $\dmu=1\implies \Dsum\kdx=0$ by Lemma \ref{lem:mu0},
one finds that
$\Csum\kdxplus{1}\Bsum\kdx=C\kdxplus{1}\Bsum\kdx\neq 0\implies \mu_c=1$. Therefore (\ref{eq:LemMu1}) is true.
\end{proof}

\begin{corollary}[Unique Inverse of $\dmu=1$ \PWA{} Systems]
\label{corollary:mu1}
The inverse of a \PWA{} system satisfying \ref{C:reachable}-\ref{C:locIndep} with 
$\mu_c=1$
is 
given by the following explicit, anticausal \PWA{} system.
\begin{align}
    \state\kdxplus{1} &= 
    \inv{\Asum}\kdx \state\kdx  + \inv{\Bsum}\kdx y\kdxplus{1} + \inv{\Fsum}\kdx 
    \\
    u\kdx  &= \inv{\Csum}\kdx  \state\kdx  + \inv{\Dsum}\kdx  y\kdxplus{1} + \inv{\Gsum}\kdx 
\end{align}
where
\begin{align*}
    \inv{\Asum}\kdx  &= \Asum\kdx +\Bsum\kdx \inv{\Csum}\kdx
    &
    \inv{\Bsum}\kdx  &= \Bsum\kdx \inv{\Dsum}\kdx
    &
    \inv{\Fsum}\kdx  &= \Fsum\kdx +\Bsum\kdx\inv{\Gsum}\kdx
    \\
    \inv{\Csum}\kdx  &= 
    - \inv{\Dsum}\kdx
    C\kdxplus{1}\Asum\kdx 
    &
    \inv{\Dsum}\kdx  &= 
    \left( C\kdxplus{1} \Bsum_{k}\right)^{-1}
    &
    \inv{\Gsum}\kdx  &= 
    -
    \inv{\Dsum}\kdx\left(
    C\kdxplus{1} \Fsum\kdx  + G\kdxplus{1} \right)
\end{align*}
\end{corollary}
\begin{proof}
$\dmu=1$ by Lemma \ref{lem:mu1}.
Plugging \ref{C:locIndep} and $\dmu=1$ into (\ref{eq:uGen}) 
yields
\begin{equation}
    u\kdx  = (C\Bsum\kdx )^{-1} (y\kdxplus{1}-C(\Asum\kdx \state\kdx +\Fsum\kdx )-G)
\end{equation}
which is explicit.
Note that 
$\left(C\Bsum\kdx\right)^{-1}$ is always well defined because it is scalar by \ref{C:SISO} and is nonzero
by $\dmu=1$ and Definition \ref{def:mu}.
\end{proof}

To derive explicit inverses for \PWA{} systems with $\dmu=2$, \ref{C:locIndep} is used along with the new ``output-based switching'' assumption. Assuming $\dmu>0$, this assumption is expressed as
\begin{enumerate}[label=(A5.\arabic*),leftmargin=*]
\resume{listB}
\item
\label{C:outputSwitch}
$P = P_o\Csum\kdx $ and $\offsetvec = \offsetvec_o - P_o\Gsum\kdx $
\suspend{listC}
\end{enumerate}
where $P_o$ and $\offsetvec_o$ contain the orientation vectors and offsets of hyperplanes in the output space $\mathbb{R}^{n_y}$.

\begin{corollary}[Unique Inverse of $\dmu=2$ \PWA{} Systems]
The inverse of a \PWA{} system satisfying \ref{C:reachable}-\ref{C:outputSwitch} with known global dynamical relative degree $\dmu=2$ is given by the following explicit, anticausal \PWA{} system.
\begin{align}
    \state\kdxplus{1} &= 
    \inv{\Asum}\kdx \state\kdx  + \inv{\Bsum}\kdx y\kdxplus{2} + \inv{\Fsum}\kdx 
    \\
    u\kdx  &= \inv{\Csum}\kdx  \state\kdx  + \inv{\Dsum}\kdx  y\kdxplus{2} + \inv{\Gsum}\kdx 
\end{align}
    where
\begin{align*}
    \inv{\Asum}\kdx  &= 
    \Asum\kdx +\Bsum\kdx 
    \inv{\Csum}\kdx 
    &
    \inv{\Bsum}\kdx  &= \Bsum\kdx 
    \inv{\Dsum}\kdx 
    &
    \inv{\Fsum}\kdx  &= \Fsum\kdx +\Bsum\kdx 
    \inv{\Gsum}\kdx 
\end{align*}
\vspace{-40pt}
\begin{align*}
    \inv{\Csum}\kdx  &= 
    -
    \inv{\Dsum}\kdx
    C\kdxplus{2}\Asum\kdxplus{1}\Asum\kdx 
    &
    \inv{\Dsum}\kdx  &= 
    \left( C\kdxplus{2} \Asum\kdxplus{1}\Bsum\kdx \right)^{-1}
\end{align*}
\vspace{-40pt}
\begin{align*}
    \inv{\Gsum}\kdx  &= 
    -
    \inv{\Dsum}\kdx
    \left(
    C\kdxplus{2} \Asum\kdxplus{1}\Fsum\kdx  + C\kdxplus{2}\Fsum\kdxplus{1} + G\kdxplus{2} \right)
\end{align*}
\end{corollary}
\begin{proof}
    Plugging \ref{C:locIndep} and $\dmu=2$ into (\ref{eq:uGen}) yields
\begin{equation}
        u\kdx =
        \left(
        C\kdxplus{2} \Asum\kdxplus{1} \Bsum\kdx 
        \right)^{-1}
        \left(
        y\kdxplus{2} - 
        C\kdxplus{2}
        \Asum\kdxplus{1}\Asum\kdx 
        \state\kdx 
        \right.
        \left.
        -
        C\kdxplus{2}
        \Asum\kdxplus{1}
        \Fsum\kdx 
        -
        C\kdxplus{2}
        \Fsum\kdxplus{1}
        -
        G\kdxplus{2}
        \right)
        \label{eq:mu2ass}
\end{equation}
In general, (\ref{eq:mu2ass}) would be implicit because of $\Asum\kdxplus{1}$ and $\Fsum\kdxplus{1}$'s dependence on $\state\kdxplus{1}$, and thus $u\kdx $, via the selector functions
\begin{equation}
    K_q(\locvec(\state\kdxplus{1})) = 
    0^{
    \prod_{i=1}^{\sigvecset_q}
    \norm{\sigvec_{q,i}-H\left(
    P\left(
    \Asum\kdx \state\kdx  + \Bsum\kdx u\kdx  + \Fsum\kdx 
    \right)
    - \offsetvec
    \right)}}
    \label{eq:implicitK}
\end{equation}
However, 
under
\ref{C:outputSwitch} (in combination with \ref{C:locIndep}) this becomes
\begin{equation}
    K_q(\locvec(\state\kdxplus{1})) = 
    0^{
    \prod_{i=1}^{\sigvecset_q}
    \norm{\sigvec_{q,i}-H\left(
    P_oC\kdxplus{1}\left(
    \Asum\kdx \state\kdx  + \Fsum\kdx 
    \right)
    - \offsetvec_o + P_oG\kdxplus{1}
    \right)}}
    \label{eq:explicitK}
\end{equation}
which is not a function of $u\kdx $ and is thus explicit.

The reduction of (\ref{eq:implicitK}) to (\ref{eq:explicitK}) relies on the fact that $C\kdxplus{1}\Bsum\kdx ~=0$ $\forall k$. This is true for $\dmu=2$ systems under \ref{C:locIndep} because $\dmu=2\implies\mu_c>1$ by Lemma \ref{lem:mu1}.
\end{proof}

\begin{remark}[$\mu_c$, $\dmu$ Relationship]
Note that unlike for relative degrees of 0 and 1, $\mu_c=2\centernot\implies\dmu=2$ under assumptions \ref{C:reachable}-\ref{C:outputSwitch}. If $n_x>2$, there exists systems for which $\mu_c=2$ but $\dmu>2$ due to inter-location dynamics.
\end{remark}

\section{Stable Inversion of \PWA{} Systems}
\label{sec:stabinv}
\subsection{Exact Stable Inversion}

For many \PWA{} systems, evolving the inverse systems derived in Section \ref{sec:exactinverse} forward in time from an initial state 
at time $k=k_0$ 
and with a bounded reference $y\kdxplus{\dmu}=r\kdxplus{\dmu}$
will yield an inverse system trajectory 
$u\kdx$ that is bounded for all $k\geq k_0$
and suitable for feedforward control.
However, inverse \PWA{} system instabilities may arise from NMP component dynamics, 
causing $u\kdx$ to become unbounded under this conventional system evolution scheme, despite the bounded reference.
In such cases, a bounded $u\kdx$ may still be achievable on a bi-infinite timeline via stable inversion.
Formally, the stable inversion problem may be given as follows.
\begin{definition}[\PWA{} Stable Inversion Problem Statement]
\label{def:SIproblem}
Given an explicit inverse \PWA{} system representation
\begin{IEEEeqnarray}{RL}
\eqlabel{eq:inverse}
\IEEEyesnumber
\IEEEyessubnumber*
\state\kdxplus{1} &= \inv{\Asum}\kdx\state\kdx + \inv{\Bsum}\kdx y\kdxplus{\dmu} + \inv{\Fsum}\kdx
\label{eq:inverseState}
\\
u\kdx &=\inv{\Csum}\kdx\state\kdx + \inv{\Dsum}\kdx y\kdxplus{\dmu}  +\inv{\Gsum}\kdx
\label{eq:inverseOutput}
\end{IEEEeqnarray}
and a reference trajectory $y\kdxplus{\dmu}=r\kdxplus{\dmu}$ known for all $k\in\integer$, a two point boundary value problem is formed by (\ref{eq:inverseState}) and the boundary conditions $\state_{-\infty}=\state_\infty=0$.
The solution to the stable inversion problem is the bounded bi-infinite time series $u\kdx\in\real$ $\forall k$,
which is
generated by (\ref{eq:inverseOutput}) and the bounded bi-infinite solution $\state\kdx$ to the boundary value problem.
\end{definition}

The following assumptions on system parameter boundedness and boundary conditions are common in some form across much stable inversion literature.
\begin{enumerate}[label=(A5.\arabic*),leftmargin=*]
\resume{listC}

\item
\label{C:boundedSystem}
There exists a supremum to the norms of the inverse system matrices:
\begin{equation}
\sup_k\norm{\inv{\Asum}\kdx},
\,
\sup_k\norm{\inv{\Bsum}\kdx},
\,
\sup_k\norm{\inv{\Fsum}\kdx},
\,
% \\
\sup_k\norm{\inv{\Csum}\kdx},
\,
\sup_k\norm{\inv{\Dsum}\kdx},
\,
\sup_k\norm{\inv{\Gsum}\kdx} \in \real
\end{equation}
Any vector norm may be used, and the matrix norm is that induced by the vector norm.

\item
\label{C:refdecay}
The reference $y\kdxplus{\dmu}$ and bias terms $\inv{\Fsum}\kdx$, $\inv{\Gsum}\kdx$ decay to zero at the extremities of the bi-infinite time series:
\begin{multline}
    \forall \varepsilon \in \real_{>0}
    \quad
    \exists \eta_1, \eta_2 \in\integer
    \quad
    \text{ s.t. }
    \\
    \norm{y\kdxplus{\dmu}}
    ,
    \norm{\inv{\Fsum}\kdx}
    ,
    \norm{\inv{\Gsum}\kdx}
    <\varepsilon
    \quad
    \forall k\in(-\infty,\eta_1\rrbracket\cup\llbracket \eta_2,\infty)
\end{multline}

\suspend{listD}
\end{enumerate}

Additionally, stable inversion of \PWA{} systems involves two challenges not faced in the stable inversion of linear systems. First, 
the dynamics of all locations and 
the inter-location dynamics must
be
simultaneously 
accounted for when decoupling the stable and unstable system modes. 
Second, there must be a way to manage switching in the two partial system evolutions.
These challenges are manifested in 
the following assumptions.

\begin{enumerate}[label=(A5.\arabic*),leftmargin=*]
\resume{listD}

\item
\label{C:decouplable}
There exists a similarity transform matrix $V\in\real^{n_x\times n_x}$ that decouples the stable and unstable modes of 
(\ref{eq:inverseState}).
Formally this decoupling can be expressed as
\begin{align}
    V \inv{\Asum}\kdx V^{-1}=\begin{bmatrix}[1.5] 
    \Asumdc^\stab\kdx & 0_{n_\stab \times n_\unstab}
    \\
    0_{n_\unstab\times n_\stab} & \Asumdc^\unstab\kdx
    \end{bmatrix}
    \quad \forall k
\end{align}
where 
$n_\stab$ is the number of stable modes, $n_\unstab$ is the number of unstable modes, $n_\stab+n_\unstab=n_x$,
$\Asumdc^\unstab\kdx$ has all eigenvalue magnitudes $>1$ $\forall k$,
and the free systems
\begin{align}
    z\kdxplus{1}^\stab = \decoup{\Asum}\kdx^\stab z\kdx^\stab
    \qquad
    z\kdxplus{1}^\unstab = \left(\decoup{\Asum}\kdx^\unstab\right)^{-1}z\kdx^\unstab
    \label{eq:free}
\end{align}
with 
appropriately 
sized
state vectors $z^\stab$, $z^\unstab$ 
are globally uniformly asymptotically stable about the origin.
\suspend{listE}
\end{enumerate}
\addtocounter{listE}{1}
\begin{enumerate}[label=(A5.\arabic{listE}\alph*),leftmargin=*]
    \item
    \label{C:switchStable}
    Switching is exclusively dependent on the stable modes:
    \begin{equation}
    PV^{-1}=\begin{bmatrix}
    \decoup{P}^\stab  & 0_{n_P\times n_\unstab}
    \end{bmatrix}
    \end{equation}
    \item
    \label{C:switchUnstable}
    Switching is exclusively dependent on the unstable modes and 
    all unstable states 
    arising from (\ref{eq:inverse})
    are reachable in one time step from some predecessor state 
    for all $k$:
    \begin{equation}
    \left(
    PV^{-1}=\begin{bmatrix} 0_{n_P\times n_\stab} &  \decoup{P}^\unstab \end{bmatrix}
    \right)
    \,\,
    \land
    \,\,
    \left(
    \forall k,\,\,\forall \statedc^\unstab\kdxplus{1}\in\mathcal{X}^\unstab\subseteq\real^{n_\unstab}  \,\,\, \text{Pre}(\{\statedc^\unstab\kdxplus{1}\})\neq \emptyset
    \right)
    \end{equation}
    where $\statedc^\unstab\kdx=
    \begin{bmatrix} 0_{n_\unstab\times n_\stab}& I_{n_\unstab\times n_\unstab} \end{bmatrix}
    V\state\kdx$.
    $\mathcal{X}^\unstab$ is 
    a set containing at least all solution values of $\statedc^\unstab\kdxplus{1}$ (see Section \ref{sec:practicalSI} for elaboration).
    $\text{Pre}(\mathcal{X})$ is the set of predecessor states whose one-step successors belong to the set $\mathcal{X}$. 
\end{enumerate}

For detailed theorems on the sufficient conditions for uniform asymptotic stability of systems of the form (\ref{eq:free}), see \cite{Chen2020}.

As implied by the separation of \exclusiveSwitching{} into two opposing assumptions, the challenges associated with switching management precipitate different stable inversion procedures for the stable-mode-dependent 
switching 
and unstable-mode-dependent 
switching 
cases. In general, the trajectory of the modes upon which switching is dependent are computed first. This allows the switching signal for the overall system to be computed and given as an exogenous input to the evolution of the remaining modes.

The theorems for these cases are supported by the following notation for the decoupled system in addition to the above-defined 
$\Asumdc^\stab\kdx$, $\Asumdc^\unstab\kdx$, $\Pdc^\stab$, $\Pdc^\unstab$.
\begin{equation}
\begin{aligned}
    \statedc\kdx^\stab &\coloneqq \extracts V\state\kdx
    &
    \Bsumdc\kdx^\stab &\coloneqq \extracts V\Bsum\kdx
    &
    \Fsumdc\kdx^\stab &\coloneqq \extracts V\Fsumdc\kdx
    \\
    \statedc\kdx^\unstab&\coloneqq \extractu V\state\kdx
    &
    \Bsumdc\kdx^\unstab &\coloneqq \extractu V\Bsum\kdx
    &
    \Fsumdc\kdx^\unstab &\coloneqq \extractu V\Fsumdc\kdx
\end{aligned}
\end{equation}
where
\begin{align}
    \extracts \coloneqq \begin{bmatrix} I_{n_\stab\times n_\stab} & 0_{n_\stab\times n_\unstab} \end{bmatrix}
    \qquad
    \extractu \coloneqq \begin{bmatrix} 0_{n_\unstab\times n_\stab}& I_{n_\unstab\times n_\unstab} \end{bmatrix}
\end{align}

\begin{theorem}[\PWA{} Stable Inversion with Stable-Mode-Dependent Switching]
\label{thm:SIstable}
Given an explicit inverse \PWA{} system (\ref{eq:inverse}) satisfying \ref{C:boundedSystem}-\ref{C:decouplable} and \ref{C:switchStable}, 
the solution to the stable inversion problem exists and can be found by
first computing the stable mode time series $\statedc^\stab\kdx$ and location time series $\locvec\kdx$ $\forall k$ forwards in time via
\begin{align}
    \locvec\kdx &= H\left( \decoup{P}^\stab\statedc^\stab\kdx  -\offsetvec \right)
    \label{eq:locstable}
    \\
    \statedc^\stab\kdxplus{1} &= \Asumdc^\stab\kdx\statedc^\stab\kdx + \Bsumdc^\stab\kdx y\kdxplus{\dmu} + \Fsumdc^\stab\kdx
    \label{eq:forwardstable}
\end{align}
The location time series being now known, the unstable mode time series $\statedc^\unstab\kdx$ can be computed backwards in time via
\begin{equation}
    \statedc^\unstab\kdx = \left(\Asumdc^\unstab\kdx\right)^{-1}\left(\statedc^\unstab\kdxplus{1}-\Bsumdc^\unstab\kdx y\kdxplus{\dmu} - \Fsumdc^\unstab\kdx\right)
    \label{eq:backwardstable}
\end{equation}
with $\locvec\kdx$ input directly to 
the selector functions in
(\ref{eq:Mshorthand}).
Finally the solution $u\kdx$ is computed via (\ref{eq:inverseOutput}) with
$\state\kdx=V^{-1}\left[\left(\statedc^\stab\kdx\right)^T,\,\,\left(\statedc^\unstab\kdx\right)^T\right]^T$.
\end{theorem}
\begin{proof}
The prescribed formula represents a solution to the stable inversion problem because
\begin{itemize}[leftmargin=*]
\item
by \ref{C:switchStable}, (\ref{eq:locstable}) is equivalent to (\ref{eq:localization}), 
\item
by \ref{C:decouplable}, the concatenated evolutions of (\ref{eq:forwardstable}) and (\ref{eq:backwardstable}) are equivalent to (\ref{eq:inverseState}), and
\item
by \ref{C:boundedSystem} and \ref{C:refdecay}, (\ref{eq:forwardstable}) and (\ref{eq:backwardstable}) decay to the form of (\ref{eq:free}) in the limits as $k$ approaches $\infty$ or $-\infty$, and thus by \ref{C:decouplable} the boundary conditions at these limits are satisfied.
\end{itemize}
The solution is guaranteed to exist because
\begin{itemize}[leftmargin=*]
\item
(\ref{eq:locstable})-(\ref{eq:backwardstable}) and (\ref{eq:inverseOutput}) are all explicit functions with all variables in the right-hand side known due to the order of time series computation, 
and
\item
the outputs of (\ref{eq:locstable})-(\ref{eq:backwardstable}) exist because the system parameters and input signals are bounded by \ref{C:boundedSystem} and \ref{C:refdecay}, and $\Asumdc^\unstab\kdx$ is guaranteed invertible by the eigenvalue condition of \ref{C:decouplable}.
\end{itemize}
\end{proof}

\begin{theorem}[\PWA{} Stable Inversion with Unstable-Mode-Dependent Switching]
\label{thm:SIunstable}
Given an explicit inverse \PWA{} system (\ref{eq:inverse}) satisfying \ref{C:boundedSystem}-\ref{C:decouplable} and \ref{C:switchUnstable}, the solution to the stable inversion problem exists and can be found in the following manner.
First solve the implicit backward-in-time evolution of the unstable modes, (\ref{eq:backwardstable}), for $\statedc^\unstab\kdx$ at each time step using any of the applicable algorithms (e.g. brute force search over all locations or computational geometry methods 
\cite{Rakovic2006},
see Section \ref{sec:practicalSI} for elaboration).
For each time step at which one of the potentially multiple solution values of $\statedc^\unstab\kdx$ is chosen (any selection method is valid), the location vector $\locvec\kdx$ may be computed by
\begin{equation}
    \locvec\kdx=H\left(\Pdc^\unstab\statedc^\unstab\kdx-\offsetvec\right)
\end{equation}
The location time series being computed, $\locvec\kdx$ may be directly plugged in to
the selector functions in
(\ref{eq:Mshorthand}) to make the forward-in-time evolution of the stable modes, (\ref{eq:forwardstable}), explicit such that it can be evaluated at each time step for $\statedc^\stab\kdx$.
The solution $u\kdx$ is then computed via (\ref{eq:inverseOutput}) with $\state\kdx=V^{-1}\left[\left(\statedc^\stab\kdx\right)^T,\,\left(\statedc^\unstab\kdx\right)^T\right]^T$, as in Theorem \ref{thm:SIstable}.
\end{theorem}
\begin{proof}
\begin{sloppypar}
The prescribed formula represents a solution to the stable inversion problem for the same reasons as Theorem \ref{thm:SIstable}, but with the first proposition of \ref{C:switchUnstable}---i.e. $PV^{-1}=
\left[
0_{n_P\times n_\stab} 
,\,\,
\decoup{P}^\unstab
\right]
$---used in place of \ref{C:switchStable}. The solution is guaranteed to exist because
\end{sloppypar}
\begin{itemize}[leftmargin=*]
\item
the second proposition of \ref{C:switchUnstable} guarantees the solution set of (\ref{eq:backwardstable}) is non-empty,
\item
the \PWA{} nature of the original inverse system (\ref{eq:inverse}) enables application of existing algorithms guaranteed to find $\text{Pre}(\{\statedc^\unstab\kdxplus{1}\})$ and thus solve (\ref{eq:backwardstable})
\cite{Rakovic2006},
and
\item
with solutions to (\ref{eq:backwardstable}) chosen $\forall k$, the remaining equations are explicit with bounded outputs for the same reasons as in Theorem \ref{thm:SIstable}.
\end{itemize}
\end{proof}

Note that while
\cite{Rakovic2006} focuses on \PWA{} systems with time-invariant components, 
only the one-step predecessor set need be computed at each time step. Thus, because a time-varying system is indistinguishable from a time-invariant system over a single time step, the algorithms of \cite{Rakovic2006} are still applicable.

\subsection{Practical Considerations}
\label{sec:practicalSI}

The most immediate issue with Theorems \ref{thm:SIstable} and \ref{thm:SIunstable} is that, while they provide an exact solution to the stable inversion problem, their procedures cannot be implemented because of the infinite nature of the time series involved.
This issue applies to past works on stable inversion as well, and the same means of addressing the issue is taken here.
Namely, an approximate solution is 
obtained
by prescribing a finite reference $\refr\kdxplus{\dmu}$ for $k\in\llbracket 0,\,\,N-\dmu \rrbracket$ and strictly enforcing the boundary conditions on only the 
initial/terminal states of the stable/unstable mode evolution. In other words, $\state^\stab_{0}=0$ and $\state^\unstab_{N-\dmu}=0$ but $\state^\stab_{N-\dmu}$ and $\state^\unstab_{0}$ may be nonzero.

In
general the closer $\state^\stab\kdx$ and $\state^\unstab\kdx$ come to decaying to zero by $k=N-\dmu$ and $k=0$, respectively, the higher quality the approximation of the $u\kdx$ time series will be. In other words, the closer $u\kdx$ comes to returning $y\kdxplus{\dmu}=\refr\kdxplus{\dmu}$ when input to the original system from which the inverse (\ref{eq:inverse}) was derived.
To achieve this high quality approximation, one may specify $\refr\kdxplus{\dmu}$ to begin and end with a number of zero elements to allow space for the $u\kdx$ time series to contain the pre- and post-actuation typically necessary for the control of NMP systems.
The number of zero elements required to achieve a satisfactorily low error is case dependent. One typical heuristic is to ensure that the durations of the leading and trailing zeros are approximately equal to the system settling time.

In addition to the practical need for finite references with leading and trailing zeros, the case of unstable-mode-dependent switching warrants special attention regarding implementation.

First, consider methods to solve
the implicit equation
(\ref{eq:backwardstable}) at each time step. Any method will consist of two parts: identifying a set of valid solutions and then choosing one of them.
Choosing a solution 
can be formalized as the minimization of some cost function, $\norm{\state^\unstab\kdxplus{1}-\state^\unstab\kdx}$ being a straightforward and universally applicable option. If the inverse system (\ref{eq:inverse}) has exclusively unstable modes, input-based costs such as $\norm{u\kdx}$ and $\norm{u\kdxplus{1}-u\kdx}$ may also be used.
Note that any such optimization is combinatorial, i.e. 
the decision variable can only take on values from a particular finite set.
For the problem considered here, the cardinality of this 
set 
is at most
the number of locations $\regionQuant$ in the system.
This is because each location contains at most one solution to (\ref{eq:backwardstable}) due to the eigenvalue condition in \ref{C:decouplable} making $(\Asumdc^\unstab\kdx)^{-1}$ full rank and thus one-to-one.

Having an upper bound of one solution per location leads to a direct method for deriving the set of valid solutions to (\ref{eq:backwardstable}). For each location $\Qset_q\in\Qset$, compute
\begin{equation}
    \left(\decoup{A}^\unstab_{q,k}\right)^{-1}
    \left(\statedc^\unstab\kdxplus{1}-B_{q,k}^\unstab \refr\kdxplus{\dmu}-F_{q,k}^\unstab\right)
\end{equation}
and check whether the result lies in $\Qset_q$. If so, the result is a solution to (\ref{eq:backwardstable}). Naturally, logic relating to solution selection criteria may be incorporated to reduce computational cost, e.g. checking locations in order of proximity to the current location to avoid checking all locations in the case that the solution selection cost function is something like $\norm{\state^\unstab\kdxplus{1}-\state^\unstab\kdx}$.
Alternatively, $\text{Pre}\left(\left\{ \statedc^\unstab\kdxplus{1} \right\}\right)$ may be derived whole using the computational-geometry-supported algorithms of 
\cite{Rakovic2006}.
As noted in 
\cite{Rakovic2006},
the algorithm of least cost may be case-dependent.

Finally, consider verification of the existence of a solution to the stable inversion problem with unstable-mode-based switching.
The verification method recommended here is to first verify \ref{C:boundedSystem}-\ref{C:decouplable} and the first proposition of \ref{C:switchUnstable} directly, then run the procedure given in Theorem \ref{thm:SIunstable} for finding a solution.
If \ref{C:boundedSystem}-\ref{C:decouplable} and the first proposition of \ref{C:switchUnstable} have been verified, then the procedure is guaranteed to find a solution to (\ref{eq:backwardstable}) if a solution exists. Equivalently, failure to find a solution implies no solution exists.

This method is recommended over directly attempting to verify the second proposition of \ref{C:switchUnstable} because this second proposition may be conservative, difficult to verify, and yield very limited computational savings over the recommended method.
The conservativeness arises from the choice of set of possible $\statedc^\unstab\kdxplus{1}$ values, $\mathcal{X}^\unstab$.
It is unlikely 
for one
to possess knowledge of the 
solution
$\statedc^\unstab\kdxplus{1}$ values 
prior to actually solving the stable inversion problem, so $\mathcal{X}^\unstab$ may need to be set much larger than necessary to ensure it contains the solution trajectory.
This containment is necessary for truth of the proposition to imply verification of solution existence.
Conversely, the proposition may evaluate to false despite containing a true solution if $\mathcal{X}^\unstab$ also contains unreachable states.
Selection of $\mathcal{X}^\unstab$ may thus be a delicate, challenging task.

The expectation of high computational cost arises from the subtle discrepancy between the capabilities of established \ac{PWA} system verification methods and the second proposition of \ref{C:switchUnstable}.
Multiple methods exist for verifying the reachability/controllability of a \PWA{} system to a target set $\mathcal{X}^\unstab$
\cite{Rakovic2006,Bemporad2000}.
However, these methods typically verify whether 
at least one
element of the target set is reachable, whereas 
\ref{C:switchUnstable} requires
verification that every element of the target set is reachable.
In other words, the computational savings  
one might expect from
existing reachability/controllability verification schemes may not be available.

In short, the second proposition of \ref{C:switchUnstable} is useful for specifying a condition under which a solution is guaranteed to exist, and thus for the derivation and proof of Theorem \ref{thm:SIunstable}. But it is not recommended as a tool for existence verification.
This is not a significant loss, however, because the procedure given in Theorem \ref{thm:SIunstable} for finding a solution is itself a valid verification tool.

\section{Validation: Application to ILC}
\label{sec:val}
This section uses the stable inversion theory of Section \ref{sec:stabinv}, and thus also the conventional inversion theory of Section \ref{sec:exactinverse}, 
to simulate the application of \ILILC{} to a \PWA{} system with \NMP{} component dynamics: an inkjet printhead positioning system.
This system uses feedback and feedforward control simultaneously.
In addition to \ac{ILILC}, for benchmarking purposes a number of other controllers are applied to the system:
\begin{itemize}
    \item
    feedback-only control (i.e. zero feedforward input),
    \item
    learning-free \ac{PWA} stable inversion (i.e. stable inversion without \ac{ILILC}),
    \item
    gradient \ac{ILC}, and
    \item
    P-type \ac{ILC}.
\end{itemize}
P-type \ac{ILC} is used as a benchmark in addition to the gradient \ac{ILC} benchmark introduced in Chapter \ref{ch:4} because P-type \ac{ILC} is among the most common forms of \ac{ILC} used in industry, is considered by many to be a form of ``model-free'' \ac{ILC}, and like gradient \ac{ILC} does not necessarily cause instability when applied to systems with \ac{NMP} dynamics. Details for the formulation of all \ac{ILC} schemes are given in Section \ref{sec:ILCSchemes}.

These simulations are
subject to a variety of model errors and other disturbances to validate
the stable-inversion-supported learning controller's practicality.
In other words, the controller is synthesized from a ``control model'' and applied to a ``truth model'' representing a physical system. The control model features mismatches in parameter values, sample rate, model order, and relative degree.

Additionally, the truth model is subject to copious process noise and measurement noise. While the physical system has virtually no output-measurable noise, the injection here is done as a preliminary test to ensure that noise does not corrupt the learning process beyond the remedial power of conventional filtering.

\subsection{Example System}
The truth model is based on the physical desktop inkjet printhead positioning testbed at the Eindhoven University of Technology, pictured in Figure \ref{fig:photo}.
The input to this system is an applied motor voltage, $\controltotal\kdx$, and the output is the printhead position along a \SI{0.3}{\meter} guide rail, $y^{P}\kdx$, measured by a linear optical encoder with a resolution of \SI{50}{\micro\meter}.
The applied motor voltage $\controltotal\kdx$ is the sum of a feedback component, $\controlfb^{C}\kdx$, and feedforward component, $u\kdx$. Finally, to add additional disturbance to the simulation, Gaussian white process noise $\omega_{\controltotal}$ (zero mean, standard deviation \SI{0.03}{\volt}) and measurement noise $\omega_{y}$ (zero mean, standard deviation \SI{50}{\micro\meter}) are added at the input and output of the plant, respectively. This ultimately results in the block diagram of Figure \ref{fig:blockA3}.

\begin{figure}
    \centering
    \includegraphics[scale=0.28]{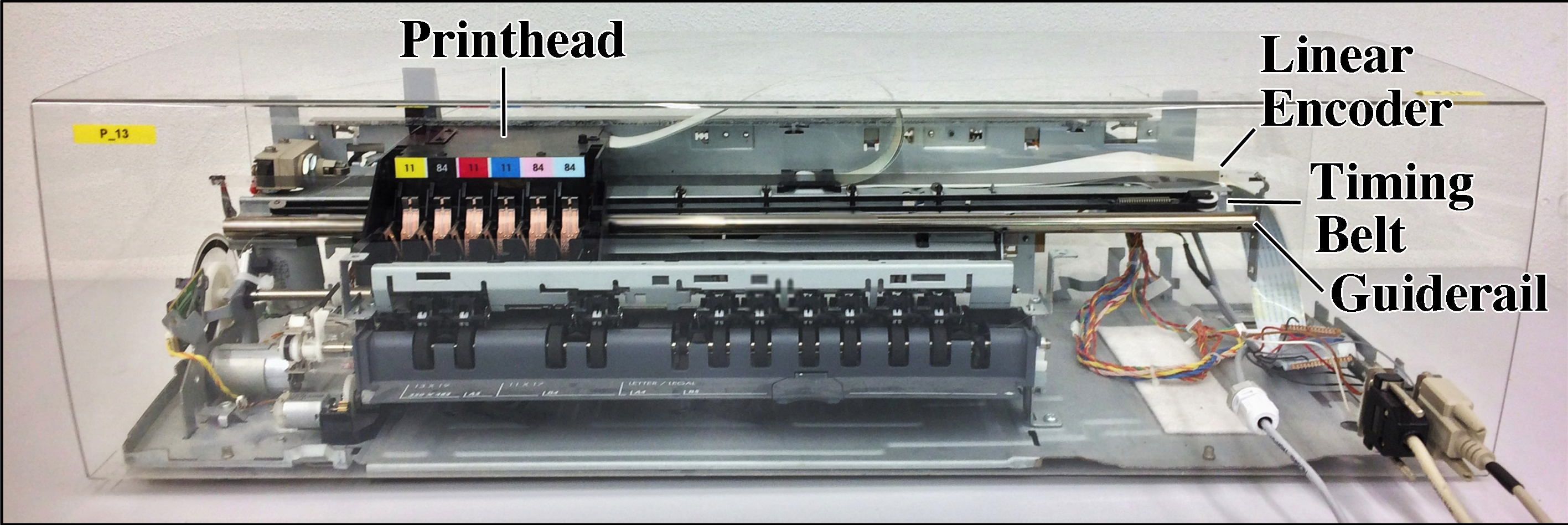}
    \caption{
    Photo of desktop inkjet printer with the case removed.
    The motor actuates the printhead motion along a guide rail via a timing belt, and the motion is measured by a linear optical encoder with resolution of \SI{50}{\micro\meter} (about 600 dots per inch).}
    \label{fig:photo}
\end{figure}

\begin{figure}
    \centering
    \includegraphics[scale=1.1]{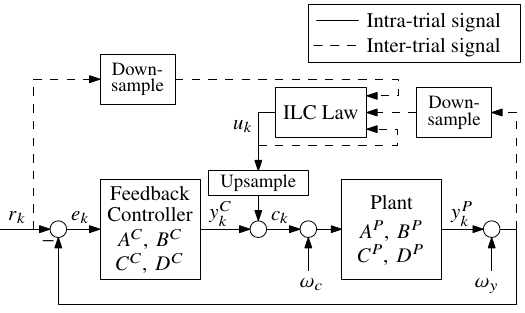}
    \caption{System block diagram. The plant block uses the truth model of the printer system obtained by experimental system identification while the ILC law is synthesized using the control model. The downsample and upsample blocks account for the difference in sample period between the ILC law and the truth model.}
    \label{fig:blockA3}
\end{figure}

System identification of the printer yields a discrete-time LTI 
model,
which is used 
as
the truth model.
Truncation-based model order reduction by 1 (MATLAB function \texttt{balred}), zero-order-hold-based sample period reduction by a factor of 2 (MATLAB function \texttt{d2d}), and random perturbation of model parameters results yields a new LTI model, which is used 
as the control model 
and can be represented by the state space system ($\Acm^{P}$, $\Bcm^{P}$, $\Ccm^{P}$, $\Dcm^{P}$) with state vector $\xcm^{P}\in\real^{n_{\xcm^{P}}}$.

To account for the change in sample period, truth model output signals are decimated by a factor of 2 before being input to the ILC law and ILC law output signals are upsampled by a factor of 2 with a zero order hold before being applied to the truth model.
Parameters for the truth and control models are given in terms of pole, zero, and gain values in Table \ref{table:modelParams}.
A Bode plot of the experimental data, truth model, and control model are given in Figure \ref{fig:bode}.

\begin{figure}
    \centering
    \includegraphics[scale=1.2]{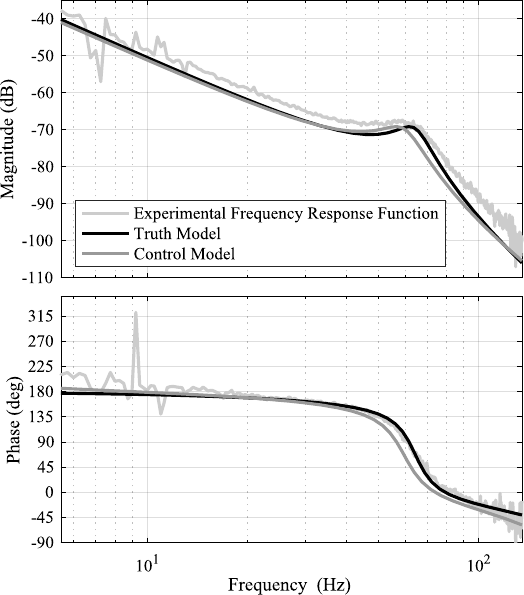}
    \caption{Bode plot of the experimental plant data, truth model of the plant, and control model of the plant.}
    \label{fig:bode}
\end{figure}

The feedback controller is of identical structure and tuning for the truth and control models, but has different parameter values due to the difference in sample period between the two models. In either case, the feedback controller is composed of a second order lowpass filter given by the discrete-time transfer function
\begin{equation}
\label{eq:LP}
    C^{LP}(z) = \frac{b z(z+1)}{z^2+a_1z+a_2}
\end{equation}
in series with a hybrid Proportional-Derivative (PD) controller. The proportional gain $K_p$ is set to a high value when the reference-output error is greater than some magnitude threshold, and is set to a lower value otherwise:
\begin{align}
    C^{PD}(z) &= \frac{
    \left(K_p + \frac{K_d}{T_s}\right)z - \frac{K_d}{T_s}
    }{z}
    \\
    K_p &= \begin{cases}
    K_{p,1} & \left|e\kdxminus{1}\right| \leq e_{\text{switch}}
    \\
    K_{p,2} & \left|e\kdxminus{1}\right| > e_{\text{switch}}
    \end{cases}
\end{align}
where $T_s$ is the sample period in seconds and $K_d$ is the derivative gain.
Note that switching is the error of the previous time step rather than the current time step in order for switching to be state-based, and thus satisfy \ref{C:stateswitch}. Mathematically this switching is made state-based by augmenting the minimal state-space representation of $C^{PD}(z)C^{LP}(z)$ with an extra state that stores the error input to the lowpass filter. In other words, the feedback controller model is given by
\begin{align}
    \Acm^{C} &= \begin{bmatrix}
    0 & 1 & 0 \\
    -a_2 & -a_1 & 0 \\
    0 & 0 & 0
    \end{bmatrix}
    \qquad
    % &
    \Bcm^{C} = \begin{bmatrix}
    0 \\ 1 \\ 1
    \end{bmatrix}
    \\
    \Ccm^{C}&= 
    -b\begin{bmatrix}
    \frac{K_d(1+a_2)}{T_s} + K_pa_2
    &
    \frac{K_da_1}{T_s} + K_p(a_1-1)
    \end{bmatrix}
    \\
    \Dcm^{C}&=\begin{bmatrix} b\left(K_p+\frac{K_d}{T_s}\right)\end{bmatrix}
\end{align}
with state vector $\xcm^{C}\kdx\in\real^{n_{\xcm^{C}}}$ having its final element equal to $e\kdxminus{1}$.
Parameter values for the feedback controller are given in Table \ref{table:modelParams}. For both the truth model and the control model, the lowpass filter has a roll off frequency of \SI{40}{\hertz} and a damping 
ratio
of 0.7.

\begin{table}
    \centering
    \caption{Simulation Model Parameters}
    \label{table:modelParams}
    \setstretch{1.1}
    \begin{tabular}{|c||c|c|}
\hline
            & Truth Model         & Control Model       \\ \hhline{|=|=|=|}
Plant Poles & $0.88\pm 0.37 i$    & $0.67 \pm 0.61 i$   \\ \arrayrulecolor{lightgray}\hline\arrayrulecolor{black} 
            & $1.00$              & $0.99$              \\ \arrayrulecolor{lightgray}\hline\arrayrulecolor{black}
            & $1.00$              & $1.00$              \\ \arrayrulecolor{lightgray}\hline\arrayrulecolor{black}
            & $0$                 & N/A                 \\ \hline
Plant Zeros & $-5.10$             & $33.10$             \\ \arrayrulecolor{lightgray}\hline\arrayrulecolor{black}
            & $-0.44$             & $-2.21$             \\ \arrayrulecolor{lightgray}\hline\arrayrulecolor{black}
            & $0.16$              & $0.16$              \\ \hline
Plant Gain  & \SI{2.42e-7}{}      & \SI{2.38e-7}{}      \\ \hline
$a_1$       & $-1.65$             & $-1.31$             \\ \hline
$a_2$       & $0.70$              & $0.50$              \\ \hline
$b$         & $0.027$             & $0.093$             \\ \hline
$K_d$       & $3$                 & $3$                 \\ \hline
$K_{p,1}$   & $40$                & $40$                \\ \hline
$K_{p,2}$   & $160$               & $160$               \\ \hline
$e_{\text{switch}}$ & \SI{2}{\milli\meter} & \SI{2}{\milli\meter} \\ \hline
$T_s$       & \SI{0.001}{\second} & \SI{0.002}{\second} \\ \hline
\end{tabular}
\end{table}

To perform stable inversion of the system dynamics from the feedforward input to the output,
a monolithic \PWA{} model of the form \sysDef{} is needed. This is given by
\begin{align}
\label{eq:monox}
    \statemod\kdxplus{1} &= \Asummod\kdx\statemod\kdx + \Bsummod\kdx u\kdx + \Fsummod\kdx
    \\
    \label{eq:monoy}
    \ymod\kdx &= \Csummod\kdx\statemod\kdx + \Dsummod\kdx u\kdx + \Gsummod\kdx
\end{align}
where
\begin{align}
    \Asummod\kdx &= 
    \begin{bmatrix}[1.375]
    \Acm^{P} - \Bcm^{P}\Dcm^{C}\Ccm^{P}
    &
    \Bcm^{P}\Ccm^{C}
    \\
    -\Bcm^{C}\Ccm^{P}
    &
    \Acm^{C}
    \end{bmatrix}
    \qquad
    \Bsummod\kdx = 
    \begin{bmatrix}[1.375]
    \Bcm^{P} \\ 0_{n_{\xcm^{C}} \times 1}
    \end{bmatrix}
    \qquad
    \Fsummod\kdx =
    \begin{bmatrix}[1.375]
    \Bcm^{P}\Dcm^{C} \\ \Bcm^{C}
    \end{bmatrix}
    \refr\kdx
    \\
    \Csummod\kdx &=
    \begin{bmatrix}
    \Ccm^{P} & 0_{1\times n_{\xcm^{C}}}
    \end{bmatrix}
    \qquad
    \Dsummod\kdx = 0
    \qquad
    \Gsummod\kdx=0
    \qquad
    \statemod\kdx = 
    \begin{bmatrix}[1.375]
    {\xcm^{P}\kdx} \\ {\xcm^{C}\kdx}
    \end{bmatrix}
    & & & &
\end{align}
This system is \NMP{}, as it has all the zeros of the plant model given in Table \ref{table:modelParams} (as well as additional zeros).

The monolithic control model has two locations based on the switching of $K_p$ in $\Ccm^{C}$ and $\Dcm^{C}$. Let the location $q=1$ correspond to low error with $K_{p,1}$ and $q=2$ correspond to high error and $K_{p,2}$. Then the switching parameters are
\begin{align}
    P &= \begin{bmatrix}
    0_{1\times n_{\xcm^{P}}+n_{\xcm^{C}}-1} & -1
    \\
    0_{1\times n_{\xcm^{P}}+n_{\xcm^{C}}-1} & 1
    \end{bmatrix}
    \qquad
    \offsetvec = \begin{bmatrix}
    -e_{\text{switch}} \\ -e_{\text{switch}}
    \end{bmatrix}
    \\
    \sigvecset_1 &= \left\{ \begin{bmatrix} 1\\1 \end{bmatrix} \right\}
    \qquad
    \sigvecset_2 = \left\{ \, \begin{bmatrix} 1\\0 \end{bmatrix},\,\begin{bmatrix}0\\1\end{bmatrix}\, \right\}
    \label{eq:deltaCellCell}
\end{align}
Note that $\locvec = [ 0, \,\, 0]^T$ is not reachable, and thus need not be included in $\sigvecset_2$.

This monolithic model is of global dynamical relative degree $\dmu=1$ and satisfies \ref{C:reachable}-\ref{C:locIndep}, enabling the use of Corollary \ref{corollary:mu1} for derivation of the conventional inverse.
The resultant inverse system satisfies
\ref{C:boundedSystem}-\ref{C:switchStable},
enabling the use of 
stable inversion
for the generation of stable
inverse state trajectories.
The decoupling similarity transform $V$ is derived by the MATLAB function \texttt{canon} applied to the dynamics of location 1.

Finally, the reference $r\kdx$ for the truth model to track is given in Figure \ref{fig:referenceA3}.

\begin{figure}
    \centering
    \includegraphics[scale=0.65]{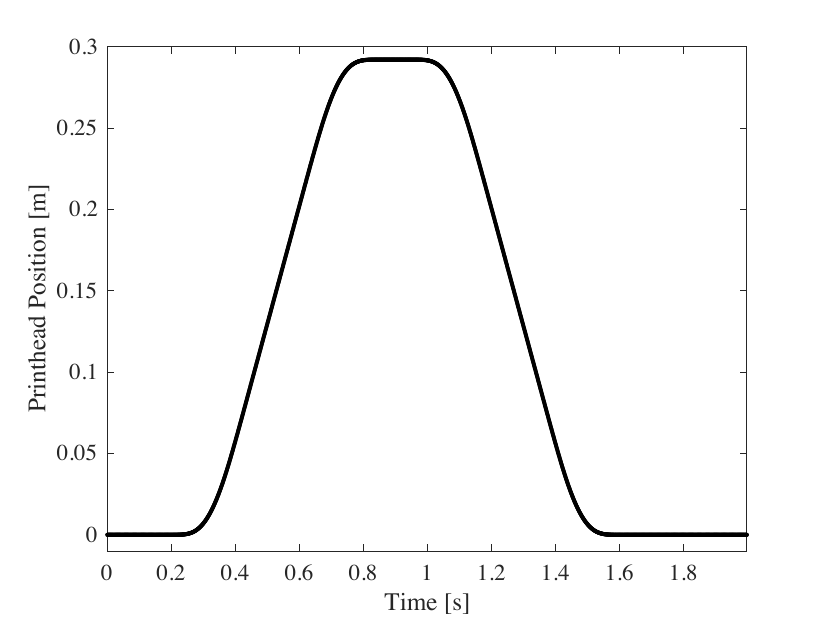}
    \caption{Reference. The reference is 1999 samples long for the truth model, and is downsampled to 1000 samples for the control model.}
    \label{fig:referenceA3}
\end{figure}

\subsection{ILC Schemes}
\label{sec:ILCSchemes}
\subsubsection{Invert-Linearize ILC}
In Chapter
\ref{ch:4},
\ILILC{} is presented as a means to derive the trial-varying learning matrix $\learnMat_\ldx$ of the classical ILC law
\begin{equation}
\eqcomment{\ref{eq:ILCclassic}}
    \uLift_{\ldx+1} = \uLift_{\ldx} + \learnMat_\ldx\left(\rLift - \yLift_\ldx\right)
\end{equation}
where $\ell\in\integer_{\geq 0}$ is the iteration index, $\uLift$, $\rLift$, $\yLift\in\real^{N-\dmu+1}$ are the lifted vectors (i.e. time series vectors)
\begin{align}
    \uLift_\ldx &= \begin{bmatrix}
    u_{\ldx,0} & u_{\ldx,1} & \cdots & u_{\ldx,N-\dmu}
    \end{bmatrix}^T
    \\
    \rLift_\ldx &= \begin{bmatrix}
    r_{\ldx,\dmu} & r_{\ldx,\dmu+1} & \cdots & r_{\ldx,N}
    \end{bmatrix}^T
    \\
    \yLift_\ldx &= \begin{bmatrix}
    y_{\ldx,\dmu} & y_{\ldx,\dmu+1} & \cdots & y_{\ldx,N}
    \end{bmatrix}^T
\end{align}
and $N\in\integer_{>\dmu}$ is the number of time steps in a trial of the output reference tracking task (the number of samples is $N+1$).

To derive $\learnMat_\ldx$, \ILILC{} calls for a lifted input-output model inverse $\ginvLift:\real^{N-\dmu+1}\rightarrow\real^{N-\dmu+1}$ taking in the \emph{measured} output $\yLift$ and outputting the control signal $\uLift$ \emph{predicted} to yield $\yLift$ when input to the true, unknown system.
Equivalently, $\ginvLift$ takes in the \emph{model} output $\yLiftmod$ and outputs the control signal $\uLift$ that yields $\yLiftmod$ when input to the known model \emph{approximating} the true system.

This $\ginvLift$ must be 
closed-form, such that the \ILILC{} learning matrix
\begin{equation}
\eqcomment{\ref{eq:gammanew}}
    \learnMat_\ldx = \jacobian{\ginvLift}{\yLiftmod}\left(\yLift_\ldx\right)
\end{equation}
can be derived via an automatic differentiation tool such as CasADi \cite{Andersson2018}. Here, $\jacobian{\ginvLift}{\yLiftmod}$ is the Jacobian (in numerator layout) of $\ginvLift$ with respect to $\yLiftmod$. Furthermore, as stated in Chapter 
\ref{ch:4},
when the known system model is \NMP, $\jacobian{\ginvLift}{\yLiftmod}$ is likely to be ill-conditioned unless $\ginvLift$ is synthesized using stable inverse trajectories. For a \PWA{} system satisfying \ref{C:boundedSystem}-\ref{C:switchStable}, Theorem \ref{thm:SIstable} provides a method for generating these closed-form state trajectories;
given trial-invariant conditions $\decoup{\state}_{\ldx,0}^{\stab}=0$ and $\decoup{\state}_{\ldx,N-\dmu}^{\unstab}=0$, each element of the time series $\state\kdx$ and $\locvec\kdx$ is a function only of $\yLiftmod$. Then $\ginvLift$ is given by
\begin{equation}
    \ginvLift = \begin{bmatrix}
    \Csummod_{0}\statemod_{0} + \Dsummod_{0}y_{\dmu} + \Gsummod_0
    \\
    \Csummod_{1}\statemod_{1} + \Dsummod_{1}y_{1+\dmu} + \Gsummod_1
    \\
    \vdots
    \\
    \Csummod_{N-\dmu}\statemod_{N-\dmu} + \Dsummod_{N-\dmu}y_{N} + \Gsummod_{N-\dmu}
    \end{bmatrix}
\end{equation}

Finally, for the particular example system studied here, the common
practice
(see, e.g. \cite{Bristow2006}) 
of adding filters to the ILC law is implemented. Two filters are used. First, a zero-phase-shift version of the feedback controller's lowpass filter is applied to the input and output of the ILC law. In lifted form, the feedback controller's filter is given by the lower diagonal, square, Toeplitz matrix $\mathscr{F}$ whose first column is the unit magnitude impulse response of the lowpass filter (\ref{eq:LP}) on $k\in\llbracket0,N-\dmu\rrbracket$. The zero phase shift is achieved by first filtering the signals forwards in time, and then backwards in time. The resultant lifted zero-phase-shift lowpass filter is
\begin{equation}
    \mathscr{Q} = \cancel{I}\mathscr{F}\cancel{I}\mathscr{F}
\end{equation}
where $\cancel{I}$ is a square matrix with ones on the antidiagonal and zeros elsewhere.
Second, to eliminate time series edge effects the first 35 and last 35 samples of the 1000 sample ILC law output are forced to zero. These edge effects may arise because the finite stable inversion trajectories have nearly zero---rather than zero---initial conditions for the unstable modes and similar for the terminal conditions of the stable modes. In lifted form this filter is given by the identity matrix with the first 35 and last 35 diagonal elements set to zero, notated as $\mathscr{E}$.

Thus, the \ILILC{} law used here is ultimately given by
\begin{equation}
    \uLift_{\ldx+1} = \mathscr{E}\mathscr{Q}\left(\uLift_\ldx + \learnMat_\ldx\left(\rLift-\mathscr{Q}\yLift_\ldx\right)\right)
    \label{eq:ILCfiltered}
\end{equation}
with $\uLift_0=0$ and $y_{\ldx,k}=y^{P}_{\ldx,k}+\omega_y$.

\subsubsection{Gradient and Lifted P-Type ILC}
Gradient \ac{ILC} and lifted P-type \ac{ILC} both also use the filtered \ac{ILC} law (\ref{eq:ILCfiltered}), but with different definitions of $\learnMat_\tdx$. 

Gradient \ac{ILC} uses
\begin{equation}
    \learnMat_\tdx = \gamma \jacobian{\gmod}{\uLift}(\uLift_\tdx)^T
\end{equation}
as in Chapter \ref{ch:4}, where $\gamma$ is the learning gain and $\gmod$ is the lifted system input-output model
\begin{equation}
    \yvecmod_\tdx = \gmod(\uLift_\tdx)
\end{equation}
which can be synthesized via (\ref{eq:gq})-(\ref{eqn:recursion}).

P-type \ac{ILC} is typically expressed without lifting as \cite{Ratcliffe2005}
\begin{equation}
u_{\tdx+1,k} = u_{\tdx,k} + \mathscr{P}\left(\refr_{\tdx,k+\mu} - y_{\tdx,k+\mu}\right)
\label{eq:Ptype}
\end{equation}
where $\mathscr{P}$ is a constant scalar learning gain and $\mu$ is the system relative degree (here the global dynamical relative degree $\dmu$ is used).
In lifted form, (\ref{eq:Ptype}) manifests as the trial invariant learning matrix
\begin{equation}
    \learnMat_\tdx = \mathscr{P}I_{N-\dmu+1 \times N-\dmu+1}
    \qquad
    \forall \tdx
\end{equation}
which can be plugged into the filtered learning law (\ref{eq:ILCfiltered}). At this time, there is no literature prescribing a stable synthesis procedure for the P-type \ac{ILC} of \ac{PWA} systems. In fact, the literature lacks application of either gradient \ac{ILC} or P-type \ac{ILC} to \ac{PWA} models.
However, because of its simplicity it is ubiquitous in industry and often synthesized heuristically, much like the PID feedback control for which it was named. Thus, P-type \ac{ILC} makes for an important benchmark.

Tuning of the learning gains for both benchmark methods is described in Section \ref{sec:methods5}.

\subsection{Methods}
\label{sec:methods5}
The presented stable inversion theory's ability to derive the inverse of a non-minimum phase \PWA{} model is tested by applying $\uLift=\ginvLift(\rLift)$ to the control model.
Then, to assess a more practical utility, 
five
independent simulations with the truth model are performed.
First,
a simulation is run with the feedforward input fixed to zero for all time, yielding a feedback-only simulation to serve as a baseline against which the four feedforward controllers can be compared.
Next,
learning-free stable inversion is applied to the truth model.
In other words, a simulation is run with $\uLift=\ginvLift(\rLift)$. 
The remaining three simulations each use one of the \ac{ILC} techniques described Section \ref{sec:ILCSchemes} (\ac{ILILC}, gradient \ac{ILC}, or P-type \ac{ILC}) with 9 trials (8 learning operations). For all \ac{ILC} simulations, $\uLift_0$ is the zero vector.

The primary metric for assessing control performance in a given trial is the normalized root mean square error (NRMSE) of the truth model, defined as
\begin{equation}
    \text{NRMSE} = \frac{\RMS_{k\in\llbracket 0,N\rrbracket}\left(
    e\kdx
    \right)}{\norm{\rLift}_\infty}
\end{equation}
The peak error magnitude $\max_k(|e\kdx |)$ is also considered.

For gradient \ac{ILC} and P-type \ac{ILC}, tuning of the learning gains is done to achieve the most aggressive stable controller possible. Specifically $\gamma$ and $\mathscr{P}$ are chosen as the largest whole number such that the NRMSE decreases monotonically over all trials or drops below the convergence tolerance, set to 0.0005 here. These numbers are found via a line search over many ILC simulations using the given truth and control model, varying only the learning gains. By this method, $\gamma=4255$ (dimensionless) and $\mathscr{P}=\SI{27}{\volt\per\meter}$. If the learning gains are increased above these values, the benchmark \ac{ILC} schemes begin to exhibit instability.

\subsection{Results and Discussion}
When $\uLift=\ginvLift(\rLift)$ is input to the noise-free control model \controlModel{} from which it was derived, the resulting NRMSE and peak error magnitude are 
\SI{1e-7}{}
and
\SI{49}{\nano\meter}. 
This is nearly zero compared to the other simulation errors, tabulated in Table  \ref{table:truthsim}, and what error there is can be attributed to the approximation error expected of finite-time stable inversion, as discussed in Section \ref{sec:practicalSI}.
This validates the fundamental theoretical contributions of this chapter: the stable inversion---and thus also conventional inversion---of \PWA{} systems.

To analyze the more practical application of these techniques to the truth model, Figure \ref{fig:NRMSEA3} plots the evolution of the 
\ac{ILC} schemes' NRMSEs
over the iteration process, and 
Figures \ref{fig:timeseriesA3} and \ref{fig:timeseriesA32} plot the input and error time series for the five simulations. Specifically, Figure \ref{fig:timeseriesA3} compares the learning-free simulations to \ac{ILILC} and Figure \ref{fig:timeseriesA32} compares \ac{ILILC} against the other benchmark learning techniques.
The NRMSE and peak error magnitude for these simulations are tabulated in Table \ref{table:truthsim}.

\begin{table}
\caption{NRMSE and Peak Error Magnitude of Truth Model Simulations}
    \label{table:truthsim}
    \centering
    \setstretch{1.2}
    \begin{tabular}{|c||c|c|}
\hline
                               & NRMSE         & Peak Error Magnitude    \\ \hhline{|=|=|=|}
Feedback-Only                  & \SI{5.6e-3}{} & \SI{3.9}{\milli\meter}  \\ \hline
Learning-Free Stable Inversion & \SI{5.4e-3}{} & \SI{2.8}{\milli\meter}  \\ \hline
P-type ILC - Final Trial       & \SI{3.1e-3}{} & \SI{2.1}{\milli\meter} \\ \hline
Gradient ILC - Final Trial     & \SI{1.2e-3}{} & \SI{1.1}{\milli\meter} \\ \hline
\ILILC{} - Final Trial            & \SI{3.6e-4}{} & \SI{0.34}{\milli\meter} \\ \hline
\end{tabular}
\end{table}

\begin{figure}
    \centering
    \includegraphics[scale=1.25]{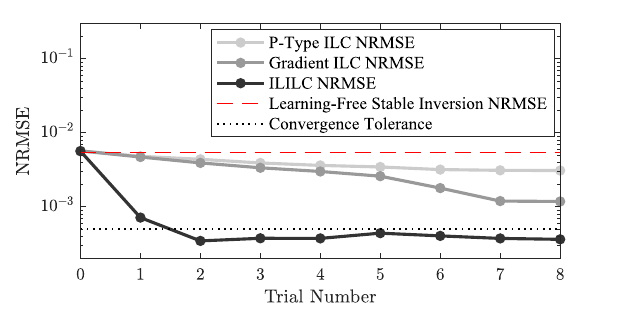}
    \caption{NRMSE of each ILC trial, illustrating convergence of \ILILC{} to a plateau determined by the noise injected to the system, and dramatically surpassing the convergence speed of both benchmark \ac{ILC} techniques. The NRMSE of the learning-free stable inversion simulation is also pictured. It is 4\% smaller than the feedback-only simulation (\ILILC{} trial 0), but is much larger than the performance achievable with learning.
    Because the noise injected in this chapter is of the same distribution and injection location as the noise in Chapter
    \ref{ch:4},
    the same convergence threshold can be used to approximate the minimum NRMSE achievable by \ILILC{}.
    }
    \label{fig:NRMSEA3}
\end{figure}

\begin{figure}
    \centering
    \includegraphics[scale=1.2]{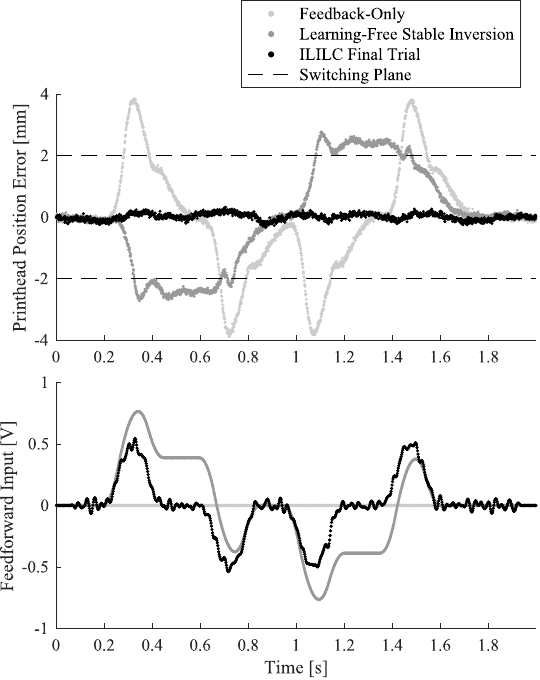}
    \caption{Error (top) and Input (bottom) time series data for the feedback-only simulation, learning-free stable inversion simulation, and the final trial of the \ILILC{} simulation. Both stable inversion and \ILILC{} perform as expected, but due to model error learning is required to reap the full benefit of feedforward control.}
    \label{fig:timeseriesA3}
\end{figure}

\begin{figure}
    \centering
    \includegraphics[scale=1.2]{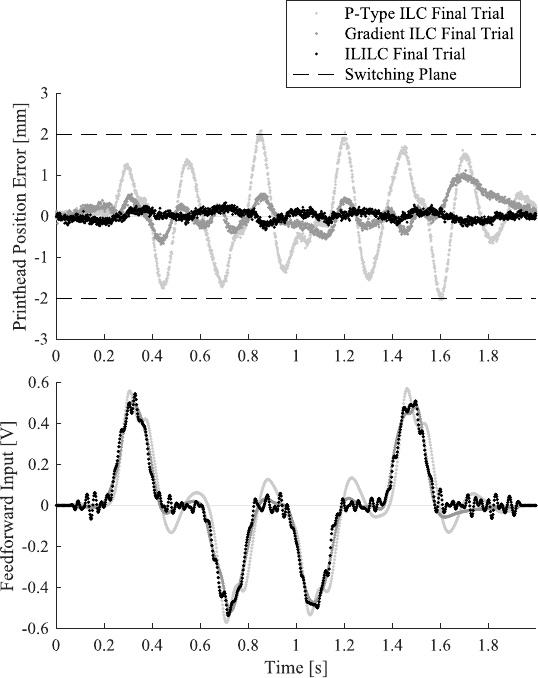}
    \caption{
    Error (top) and Input (bottom) time series data for the three \ac{ILC} simulations. Even the worst-performing \ac{ILC} technique---P-type \ac{ILC}---yields a reduction in maximum error magnitude when compared to the learning-free techniques of Figure \ref{fig:timeseriesA3}, but \ac{ILILC} clearly yields the lowest-error performance. 
    This superiority is in spite of \ac{ILILC} acquiring more high frequency content via learning than the other \ac{ILC} schemes, which appear less noisy but appear to contain higher amplitude, lower frequency oscillations that degrade performance.
    }
    \label{fig:timeseriesA32}
\end{figure}

Because the model error is relatively small, learning-free stable inversion does yield some improvement 
(4\%) 
over the feedback-only NRMSE. (Naturally, one expects that the NRMSE of learning-free stable inversion would grow if the model error increased). Learning-free stable inversion also reduces the peak error magnitude, which is a critical safety criterion in many applications, by 
28\%.
However, because of the model error that does exist, all of the \ac{ILC} schemes defeat learning-free stable inversion in both metrics. P-type \ac{ILC} shows the least improvement, but still yields a 40\% reduction in NRMSE from learning-free stable inversion. This improvement is dwarfed by that of \ac{ILILC}, however, which yields a 70\% NRMSE improvement from gradient \ac{ILC} (the next best technique), or equivalently, a 93\% improvement from learning-free stable inversion.

In addition to these quantitative metrics, inspection of Figure \ref{fig:timeseriesA32} reveals a qualitative comparison worth making between the \ac{ILC} schemes: that of noise acquisition. While even qualitatively, \ac{ILILC} clearly demonstrates the highest quality reference tracking, it also accumulates the most noise in its learned feedforward input. There may be applications in which this noise acquisition is unattractive. However, it must also be noted that gradient \ac{ILC} and P-type \ac{ILC} are not free from unwanted frequency content. P-type \ac{ILC} especially appears to develop mid-frequency oscillations that substantially degrade performance. Gradient \ac{ILC} is more subtle; its main source of error appears to be oscillations arising from overshoot of the feedforward input.
As a final note, it is unclear at this time whether the noise acquisition of the \ac{ILILC} simulation is a result of the Newton-like learning intrinsic to \ac{ILILC} or the stable inversion used to eliminate the inverse instability.

In all, this validation constitutes the first demonstration of ILC applied to any hybrid system with \ac{NMP} component dynamics. Because stable inversion is a required part of the highest performing ILC scheme used here, this validates the utility of this chapter's theoretical contributions for high performance control.

\section{Conclusion}
\label{sec:conc5}
This chapter has derived theory for the inversion of a class of \PWA{} systems.
This includes the implicit inverse formula for systems of any relative degree and the explicit formulas for systems with global dynamical relative degree (a concept introduced in this work) of 0, 1, or 2, along with the proof of sufficient conditions for the inverse of the original \PWA{} system to be explicit.
Additionally, for cases in which the inverse system is unstable, a stable inversion procedure is created, along with proof of the sufficient conditions for the procedure to be applicable.

The ability to analytically produce inverse system models for hybrid systems has multiple applications in controls. Demonstrated here is the newfound ability to apply ILC to \PWA{} systems with unstable inverses to achieve low error output reference tracking.

There are many avenues for future work. Of particular interest is the relaxation of the constraints on relative degree. For some hybrid systems, it may be desirable to have locations in which the input cannot affect the output, and the state is governed by natural dynamics alone. In such cases the global dynamical relative degree would be undefined (infinite), which is not considered here. There may also be more cases in which different locations feature different component relative degrees, which would violate \ref{C:muc}. Relaxing these constraints would dramatically expand the class of systems addressed. Extension to multi-input-multi-output systems and input-based switching would also be significant contributions.

%%% Down with Murphy's law
\addtocontents{toc}{\protect\pagebreak}
%%% :-)

\chapter{Conclusion}
\label{ch:6}
\section{Synthesis of Research Contributions}

Spurred by the long-term goal of achieving high performance droplet volume control in \ac{e-jet} printing, this dissertation makes substantial contributions to both the modeling of \ac{e-jet} printing and the control theory necessary to leverage those models.
At each stage, practical obstacles give rise to novel scientific and mathematical research.

First, while Chapter \ref{ch:2} presents significant new progress in the traditional physics-based modeling of meniscus electrohydrodynamics, the inability of traditional models to completely capture the ejection process from end to end motivates the development of hybrid \ac{e-jet} modeling frameworks.
Likewise, the inability to process the computer-vision-based jet measurements fast enough for real-time feedback control motivates the study of \ac{ILC} for this system.
Finally, the lack of \ac{ILC} theory for hybrid systems 
in the preexisting literature
serves as the impetus for the controls research making up 
Chapters \ref{ch:3}-\ref{ch:5}.

The first of these chapters lays the groundwork for the remaining hybrid systems and \ac{ILC} research by providing a closed-form \ac{PWD} system representation and integrating it with \ac{NILC}. Based on the original convergence analysis for \ac{NILC}, one might expect the performance of this integration to
be mostly uniform
across all \ac{PWD} models. 
However, while the lifted model derivation given in Chapter \ref{ch:3} is guaranteed to have a theoretically invertible Jacobian if the relative degree is constant over the trial, this does not account for the practical ability to compute the inverse. As found in Chapter \ref{ch:4}, this practical ability is compromised when the system model has an unstable inverse.
Such \ac{NMP} systems are not negligible edge cases. They arise when modeling many practical motion control devices---a key application space for \ac{ILC}---from piezoactuators \cite{Schitter2002} to DC motor and tachometer assemblies \cite{Awtar2004}.
Thus, to leave the issue of \ac{NMP} hybrid systems unaddressed would be to deliver more of a minefield than a control theory.

Chapters \ref{ch:4}-\ref{ch:5} sweep this minefield by introducing the new \ac{ILILC} framework, integrating it with stable inversion, and developing the first theory for the stable inversion of hybrid systems.
Ultimately this enables \ac{ILC} synthesis from a model of the physical device that originally sounded the alarm on \ac{NILC} of \ac{NMP} systems, an inkjet printhead positioning system.

\section{Broader Impacts}

Both the hybrid \ac{e-jet} modeling research and the control theory research in this dissertation have ramifications beyond the validations presented here, 
and 
beyond
the intended future use for droplet volume control. 
In fact,
the
broader impacts of 
hybrid \ac{e-jet} modeling
have already begun to manifest.

Within the world of \ac{e-jet} printing, the hybrid models provided by this dissertation serve as an end-to-end process model that is more easily interpreted and analyzed by human researchers than computational multiphysics models.
Because of this, the physics-driven hybrid model is currently being used for the benchmarking and development of such sophisticated partial-differential-equation-based simulations.

This dissertation's modeling contributions have also
gained attention in the broader \ac{AM} community, outside of \ac{e-jet}-specific research.
The promotion and validation of hybrid modeling for \ac{e-jet} printing helped pave the way for the proposal of a more general hybrid \ac{AM} modeling framework, which was used for modeling the multi-level workflow of \ac{FDM} \cite{Balta2019}.
There is also great promise for more thorough hybrid modeling of the physical dynamics of \ac{FDM}: research has been published identifying distinct dynamic regimes \ac{FDM} may occupy depending on the physical state of the printhead and filament 
\cite{Plott2018}.
Additionally, because \ac{AM} processes are mostly open loop at this time, they are compelling candidates for \ac{ILC}.
In other words, there is evidence of a rich field of systems whose modeling would benefit from following in the footsteps of this dissertation, and whose control may benefit directly.

To see the broader impacts of Chapters \ref{ch:3}-\ref{ch:5}, it may be beneficial to revisit the ``castle-of-building-blocks'' visualization from Figure \ref{fig:gardenBasic}. Such a revisitation is given in Figure \ref{fig:gardenBroad}, which illustrates a number of new classes of control 
systems
that may be enabled by the fundamental contributions of this work. Most obviously, where there are contributions to inversion and stable inversion theory, there is opportunity for feedforward control. While this dissertation focuses on \ac{ILC} because of the performance advantages it yields over learning-free methods, one should not discount the importance of ordinary feedforward control, for it can be applied in non-repetitive scenarios where the learning mechanisms of most \ac{ILC} schemes may falter.

% -----------------------------
% PRINT FIGURE VIEW
% -----------------------------
% % \global\pdfpageattr\expandafter{\the\pdfpageattr/Rotate 90}
% \begin{sidewaysfigure}
%     \centering
%     \includegraphics{Figure/gardenBroad06.pdf}
%     \caption{Castle of Control Contributions, revisited for visualization of
    % this dissertation's potential broader impacts.
    % These may take the form of other new classes of control systems leveraging the theoretical contributions presented here.}
%     \label{fig:gardenBroad}
% \end{sidewaysfigure}
% \pagebreak
% % \global\pdfpageattr\expandafter{\the\pdfpageattr/Rotate 0}

% -----------------------------
% COMPUTER FIGURE VIEW
% -----------------------------
\begin{landscape}
\begin{figure}
    \centering
    \includegraphics[scale=1.2]{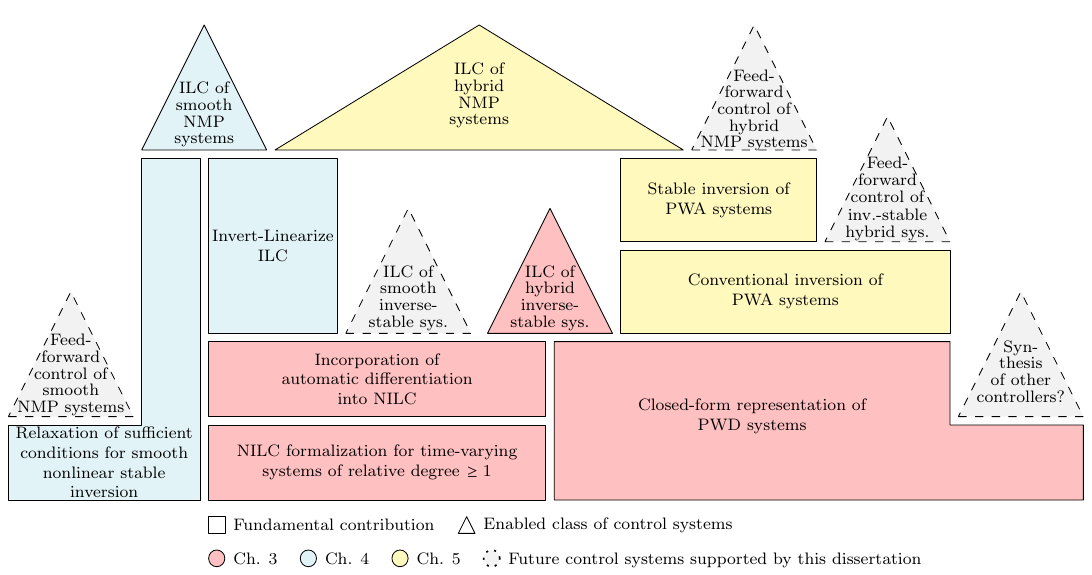}
    \caption{Castle of Control Contributions, revisited for visualization of
    this dissertation's potential broader impacts.
    These may take the form of other new classes of control systems leveraging the theoretical contributions presented here.}
    \label{fig:gardenBroad}
\end{figure}
\end{landscape}
\pagebreak

The ability to generate feedforward control signals can also be useful in scenarios where feedback control is a primary focus. A concrete example is given by \cite{VandeWouw2008}, where a combined feedback/feedforward control scheme for \ac{PWA} system reference tracking is derived, but is limited based on the availability of the feedforward signal. This dissertation alleviates that limitation.

Finally, a further impact on feedback control may be possible with the closed-form \ac{PWD} representation. The new ability to holistically differentiate a hybrid state transition function or output function may facilitate the synthesis of a diverse range of controllers, e.g. feedback linearization.

\section{Future Directions}
\label{sec:futurework}

Clearly, these potential broader impacts themselves constitute a large arena for future work. In particular, investigations into the hybrid modeling of other \ac{AM} technologies and the use of the closed-form \ac{PWD} system representation for feedback control synthesis may be valuable contributions.
Of course, further progress towards the original long-term objective of \ac{e-jet} droplet volume control is of primary interest.

The three main challenges to overcome here are the need for point-to-point \ac{ILC}, managing the trade-off between time-based location transitions and model stability, and encoding the limits of the safe printing region into the model.

Clearly, the droplet volume modeling of the \ac{PWA} \ac{e-jet} model (the second proposed model, Section \ref{sec:controlModeling}) is necessary to achieve droplet volume control. However, it is undesirable to require reference specification for the entire droplet volume time series; only the final droplet volume value matters. \ac{ILC} with reference specification for only a subset of the trial time series exists and is called point-to-point \ac{ILC} \cite{Freeman2013a}, but has not been combined with either \ac{NILC} or hybrid models. This leaves a gap between the prior art and the needs of \ac{e-jet} printing.

Additionally, while the \ac{PWA} \ac{e-jet} model introduces droplet volume modeling, it removes state-based transitions and replaces them with time-based transitions. Time-based transitions cannot be altered by \ac{ILC}, but the time of transition will certainly change in practice if the input voltage changes in magnitude from trial to trial. This makes state-based switching desirable. To achieve state-based transitions, dynamical modeling of the meniscus position is necessary. This presents an issue for \ac{ILC} because a linear model fit to meniscus position data during the approach location is likely to be unstable (this may be deduced from Figure \ref{fig:timeseries01}), and asymptotic model stability is prerequisite to \ac{ILC} in general \cite{Bristow2006}.

Constrained system identification or extension of the physics-focused \ac{e-jet} model's build-up or jetting locations may enable stable modeling of the meniscus position.
However, this is unlikely to completely solve the \ac{e-jet} \ac{ILC} problem.
Stable linear modeling during jetting implies that the meniscus position may be completely controllable during jetting (potentially with all control actions experiencing the same time delay modeled for the physics-focused model's jetting location). There is no evidence to suggest this degree of controllability in practice.
Basing \ac{ILC} on this assumption may have undesirable results. For example, if a droplet volume is desired to be smaller than it was for trial $\ldx$, \ac{ILC} may request a large negative voltage during jetting to arrest the jet or remove material from the substrate. This may be possible according to the \ac{LTI} component models, but may not be possible in reality, and attempts to do so may take the physical system outside the regime in which the model is applicable. Such risks lead to the final gap identified for the \ac{ILC} of \ac{e-jet} printing: encoding of the safe input range.

In other words,
a major area for future work in \ac{e-jet} modeling is the prediction and encoding of the boundaries of the subcritical regime. This is both a performance and safety issue. Attempts to project subcritical regime behavior beyond the regime limits may result in failure to eject, misplaced droplets due to tilted ejection angles, or destruction of the nozzle via flooding or arcing.

Beyond \ac{e-jet} printing, there
are also 
exciting control theory developments to be built directly off the contributions of this dissertation.
Two meaningful areas for future work are identified. First is the investigation of compatibility between the closed-form \ac{PWD} representation and forms of \ac{ILC} not using lifted models. While \ac{ILILC}'s sufficient conditions for convergence are very broad, in some cases it may be desirable to prioritize computational cost. In such cases the large matrix operations in lifted \ac{ILC} may be a disadvantage. Thus, the use of the closed-form \ac{PWD} representation to synthesize filter-based \ac{ILC}, such as that of \cite{Strijbosch2020}, 
may be valuable.

Second is the relaxation of the assumptions under which an \ac{NMP} \ac{PWA} system may be controlled via \ac{ILILC}. Currently, these assumptions are those of \ac{PWA} stable inversion, but \ac{ILILC} can admit other inverse system approximations as well. Thus, assumption relaxation efforts could focus on improving \ac{PWA} stable inversion or introducing a new stable inverse approximation method.

Finally, it must be remembered that 
the research and new engineering tools provided here were developed 
in large part as a
response to the \emph{unexpected} challenges that leapt up from physical systems. Surely nature's surprises are not exhausted.

\appendix

\newcommand{\appATitle}{
Neglect of 
Inverse Instability
by Non-NILC Prior Art
}
\addcontentsline{toc}{chapter}{Appendix A: \appATitle}
\chapter{
\appATitle
}

This appendix demonstrates that the sufficient conditions for convergence proposed by past works \cite{Jang1994,Saab1995,Wang1998,Sun2003} on ILC for discrete-time nonlinear systems are in actuality not sufficient for some cases of systems having unstable inverses.
This is done by running model-error-free ILC simulations that are guaranteed to converge by the past works, and observing them to diverge
instead.

Each of \cite{Jang1994,Saab1995,Wang1998,Sun2003} proposes sufficient conditions for the 
convergence 
$\lim_{\tdx\rightarrow\infty}\evec_\tdx=0_{N-\mu+1}$
of a particular ILC scheme applied to a particular class of nonlinear dynamics.
All of these classes of nonlinear dynamics are supersets of the SISO LTI dynamics
\begin{IEEEeqnarray}{RL}
\eqlabel{eq:LTI}
\IEEEyesnumber
\IEEEyessubnumber*
    x_\tdx(k+1) &= Ax_\tdx(k) + Bu_\tdx(k)
    \\
    y_\tdx(k) &= Cx_\tdx(k)
\end{IEEEeqnarray}
with relative degree $\mu=1$, i.e. $CB\neq 0$. Additionally, assume (\ref{eq:LTI}) is stable and $x_\tdx(0)$ is such that $y_\tdx(0)=r_\tdx(0)$ $\forall\tdx$.
Given a system of this 
structure, the ILC schemes and 
convergence conditions of 
the
past work
reduce to the following.

From \cite{Jang1994} the learning law is
\begin{equation}
    u_{\tdx+1}(k) = u_\tdx(k) + L_\tdx(k)\left(\gamma_1e_\tdx(k+1) + \gamma_0e_\tdx(k) \right)
    \label{eq:JangLearn}
\end{equation}
where $L\in\real$ is a potentially time-varying and trial-varying part of the learning gain and $\gamma_1$, $\gamma_0\in\real$ are trial-invariant, time-invariant learning gains with $\gamma_1\neq 0$. The learning laws of \cite{Saab1995,Wang1998,Sun2003} are special cases of (\ref{eq:JangLearn}): \cite{Saab1995} sets $\gamma_1=1$, $\gamma_0=-1$, \cite{Wang1998} sets $\gamma_1=1$, $\gamma_0=0$, and \cite{Sun2003} sets $\gamma_1=1$ and leaves $\gamma_0$ free.

Each 
work presents a different variation of convergence analysis, but all propose 
a
sufficient condition 
of the form
\begin{enumerate}[label=(CA.\arabic*),leftmargin=*]
    \item 
    $
    |
    1-L_\tdx(k)\gamma_1CB
    |
    <1$ $\forall \,k,\,\tdx$ .
    \label{CJang}
\suspend{conditions03}
\end{enumerate}
In \cite{Jang1994,Wang1998,Sun2003} \ref{CJang} is used exactly, while 
the convergence analysis in \cite{Saab1995} 
implies
the additional sufficient condition
\begin{enumerate}[label=(CA.\arabic*),leftmargin=*]
\resume{conditions03}
\item
$\norm{A}>1$
\label{CSaab}
\end{enumerate}
where any consistent norm may be chosen for $\norm{\cdot}$.

Consider the example system and learning gain
\begin{equation}
    \begin{aligned}
    A &= \begin{bmatrix}
    -0.3 & -0.79 & 0.53
    \\
     0   &  0.5  & 1
     \\
     0   & -0.36 & 0.5
    \end{bmatrix}
    \qquad
    &B &= \begin{bmatrix}
    0 \\ 0 \\ 1.34
    \end{bmatrix}
    \\
    C &= \begin{bmatrix}
    0.7 & 1.1 & -0.74
    \end{bmatrix}
    \qquad 
    &x_\tdx(0)&=0 \,\,\forall\tdx
    \end{aligned}
    \label{eq:LTIex}
\end{equation}
\begin{equation}
    L_\tdx(k)=0.5(CB)^{-1} \quad \forall \, k,\,\tdx
    \label{eq:LTILex}
\end{equation}
with the reference given in Figure \ref{fig:ref}.
This system has an unstable inverse.

The plant (\ref{eq:LTIex}) satisfies \ref{CSaab}, and with (\ref{eq:LTILex}) it satisfies \ref{CJang} for $\gamma_1=1$. Thus, according to \cite{Jang1994,Saab1995,Wang1998,Sun2003} the ILC scheme (\ref{eq:JangLearn}) is guaranteed to 
yield tracking error convergence
in a model-error-free simulation.
However, Figure \ref{fig:neglect} shows that 
the tracking error
diverges
under (\ref{eq:JangLearn}),
meaning that satisfaction of \ref{CJang} and \ref{CSaab} is not actually sufficient for the convergence of all systems (\ref{eq:LTI}) under the learning law (\ref{eq:JangLearn}).
This illustrates that the failure to account for phenomena arising from inverse instability is not unique to \AILC, but rather pervades the literature on ILC with discrete-time nonlinear systems.

\begin{figure}
    \centering
    \includegraphics[scale=0.9,trim={0.14in 0in 0.3in 0.1in},clip]{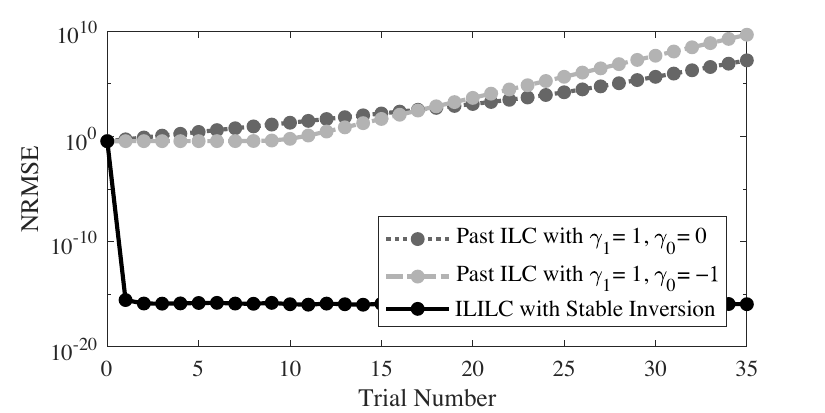}
    \caption{
    NRMSE versus trial number of past works' ILC schemes (\ref{eq:JangLearn}) applied with learning gain (\ref{eq:LTILex})  
    to the system (\ref{eq:LTIex}).
    These NRMSEs monotonically increase, confirming the inability of the past work on ILC with discrete-time nonlinear systems to account for unstable inverses.
    The NRMSE trajectory yielded by the stable-inversion-supported \ILILC{} scheme proposed by this article is also displayed.
    The convergence of this ILC scheme when applied to (\ref{eq:LTIex}) reiterates its ability to control such \nonmin{} phase systems.
    }
    \label{fig:neglect}
\end{figure}

% Works Cited
\bibliographystyle{IEEEtranMod}
\bibliography{spiegelBibDissertation}

\end{document}